\documentclass[12pt]{article}
\usepackage[margin=1in]{geometry}
\usepackage{hyperref}
\usepackage{amsmath}
\usepackage{graphicx}
\usepackage{enumerate}
\usepackage[square,numbers]{natbib}
\usepackage{url} 

\usepackage{mathtools}
\usepackage{empheq}
\usepackage{amsfonts}
\usepackage{amsgen,amsmath,amstext,amsbsy,amsopn,amssymb,amsthm,stmaryrd,bbm,subfigure}
\usepackage{cprotect}
\usepackage{xcolor}

\usepackage{comment}
\usepackage{enumitem}
\usepackage{multirow}
\usepackage[title]{appendix}
\usepackage[linesnumbered,ruled,vlined,longend]{algorithm2e}
\newcommand{\bP}{\mathbb{P}}
\newcommand{\bR}{\mathbb{R}}
\newcommand{\bE}{\mathbb{E}}
\newcommand{\cN}{\mathcal{N}}
\newcommand{\cU}{\mathcal{U}}
\newcommand{\cE}{\mathcal{E}}
\newcommand{\cT}{\mathcal{T}}

\newcommand{\RNum}[1]{\uppercase\expandafter{\romannumeral #1\relax}}

\newcommand{\vertiii}[1]{{\left\vert\kern-0.25ex\left\vert\kern-0.25ex\left\vert #1 
		\right\vert\kern-0.25ex\right\vert\kern-0.25ex\right\vert}}
\DeclareMathOperator*{\argmax}{arg\,max}
\DeclareMathOperator*{\argmin}{arg\,min}

\newtheorem{definition}{Definition}
\newtheorem{thm}{Theorem}
\newtheorem{assump}{Assumption}

\newtheorem{prop}{Proposition}
\newtheorem{cor}{Corollary}
\newtheorem{lemma}{Lemma}
\newtheorem{remark}{Remark}
\newcommand{\mymax}{\max\limits}
\newcommand{\mymin}{\min\limits}
\newcommand{\ones}{{{\mathbbm{1}}}}
\newcommand{\ind}[1]{\ones_{\{#1\}}}

\begin{document}

\def\spacingset#1{\renewcommand{\baselinestretch}%
{#1}\small\normalsize} \spacingset{1}

%%%%%%%%%%%%%%%%%%%%%%%%%%%%%%%%%%%%%%%%%%%%%%%%%%%%%%%%%%%%%%%%%%%%%%%%%%%%%%

\title{\bf Graphical Model Inference with Erosely Measured Data}
  \author{Lili Zheng\hspace{.2cm}\\
    Department of Electrical and Computer Engineering, Rice University\\
    and \\
    Genevera I. Allen \\
    Department of Electrical and Computer Engineering, Rice University,\\
    Department of Computer Science, Rice University,\\
Department of Statistics, Rice University,\\
Department of Pediatrics-Neurology, Baylor College of Medicine,\\
Jan and Dan Duncan Neurological Research Institute, Texas Children’s Hospital}
  \date{}
  \maketitle
  
\begin{abstract}
In this paper, we investigate the Gaussian graphical model inference problem in a novel setting that we call \emph{erose} measurements, referring to irregularly measured or observed data. For graphs, this results in different node pairs having vastly different sample sizes which frequently arises in data integration, genomics, neuroscience, and sensor networks. 
Existing works characterize the graph selection performance using the minimum pairwise sample size, which provides little insights for erosely measured data, and no existing inference method is applicable. 
We aim to fill in this gap by proposing the first inference method that characterizes the different uncertainty levels over the graph caused by the erose measurements, named GI-JOE (\textbf{G}raph \textbf{I}nference when \textbf{J}oint \textbf{O}bservations are \textbf{E}rose).
Specifically, we develop an edge-wise inference method and an affiliated FDR control procedure, where the variance of each edge depends on the sample sizes associated with corresponding neighbors. We prove statistical validity under erose measurements, thanks to careful localized edge-wise analysis and disentangling the dependencies across the graph. Finally, through simulation studies and a real neuroscience data example, we demonstrate the advantages of our inference methods for graph selection from erosely measured data. 
\end{abstract}

\noindent
{\it Keywords:}  Uneven measurements, missing data, graph structure inference, FDR control, graph selection
\vfill

\newpage
\spacingset{1.9} 
\section{Introduction}
Graphical models have been powerful and ubiquitous tools for understanding connection and interaction patterns hidden in large-scale data \citep{koller2009probabilistic}, by exploiting the conditional dependence relationships among a large number of variables. For instance, graphical models have been applied to learn the connectivity among tens of thousands of neurons \citep{vinci2018adjusted}, gene expression networks \citep{allen2012log,dobra2004sparse}, sensor networks \citep{dasarathy2016active,dasarathy2019gaussian}, among many others. The last decade has witnessed a plethora of new statistical methods and theory proposed for various types of models in this area, including the Gaussian graphical models \citep{yuan2007model,ravikumar2011high,friedman2008sparse,meinshausen2006high,cai2011constrained}, graphical models for exponential families and mixed variables \citep{yang2015graphical,yang2014mixed,chen2015selection}, Gaussian copula models \citep{liu2012high,liu2009nonparanormal,dobra2011copula}, etc.  

Despite the abundant literature in this area, most existing methods and theory for graphical models assume even measurements over the graph, where either all variables are measured simultaneously, or they are missing with similar probabilities. However, many real large-scale data sets usually take the form of \emph{erose measurements}, which are irregular over the graph, and different pairs of variables may have \emph{drastically different sample sizes}. Such data sets frequently arise in genetics, neuroscience, sensor networks, among many others, due to various technological limits. 
%%%%%%%%%%%%%%%%%%%%%%%%%%%%%%%%%%%%%%%%%%%%%
\subsection{Problem Setting and Motivating Applications}\label{sec:setting}
Consider the following sparse Gaussian graphical model: $x\sim\mathcal{N}(0,\Sigma^*),\quad \Theta^*=(\Sigma^{*})^{-1},$
where $\Theta^*\in \mathbb{R}^{p\times p}$ is the sparse precision matrix. The graph structure is dictated by the nonzero patterns in $\Theta^*$: $\mathcal{G}=(V,E),\quad V=\{1,\dots,p\},\quad E=\{(i,j):\Theta^*_{ij}\neq 0\},$ where the unknown edge set $E$ is of primary interest.
Suppose that we only have access to 
the following observations: $\{x_{i,V_i}: V_i\subseteq [p]\}_{i=1}^n,$
where $V_i$ is the observed index set of data point $i$. Then the joint observation set for node pair $(j,k)$ is $O_{jk}=\{i:j,k\in V_i\}$ of size $n_{jk}=|O_{jk}|$. 
There are a number of applications where $n_{jk}$ can be drastically different. 

\noindent \textbf{Heterogeneous missingness}:
In a variety of biological experiments, some variables could be missing or have erroneous zero reads (dropouts) much more than others, e.g., the expression levels of certain genes \citep{gan2020correlation,huang2018saver,gong2018drimpute}, or the abundance of some microbes \citep{williams2019microbiomedasim}. Figure \ref{fig:scRNA_seq_samplesizes} shows the observational patterns and pairwise sample sizes of two real single-cell RNA sequencing (scRNA-seq) data sets, which is far from uniform.
\begin{figure}[!htbp]
    \centering
    \subfigure[Real observational patterns]{
    \includegraphics[height=3.3cm]{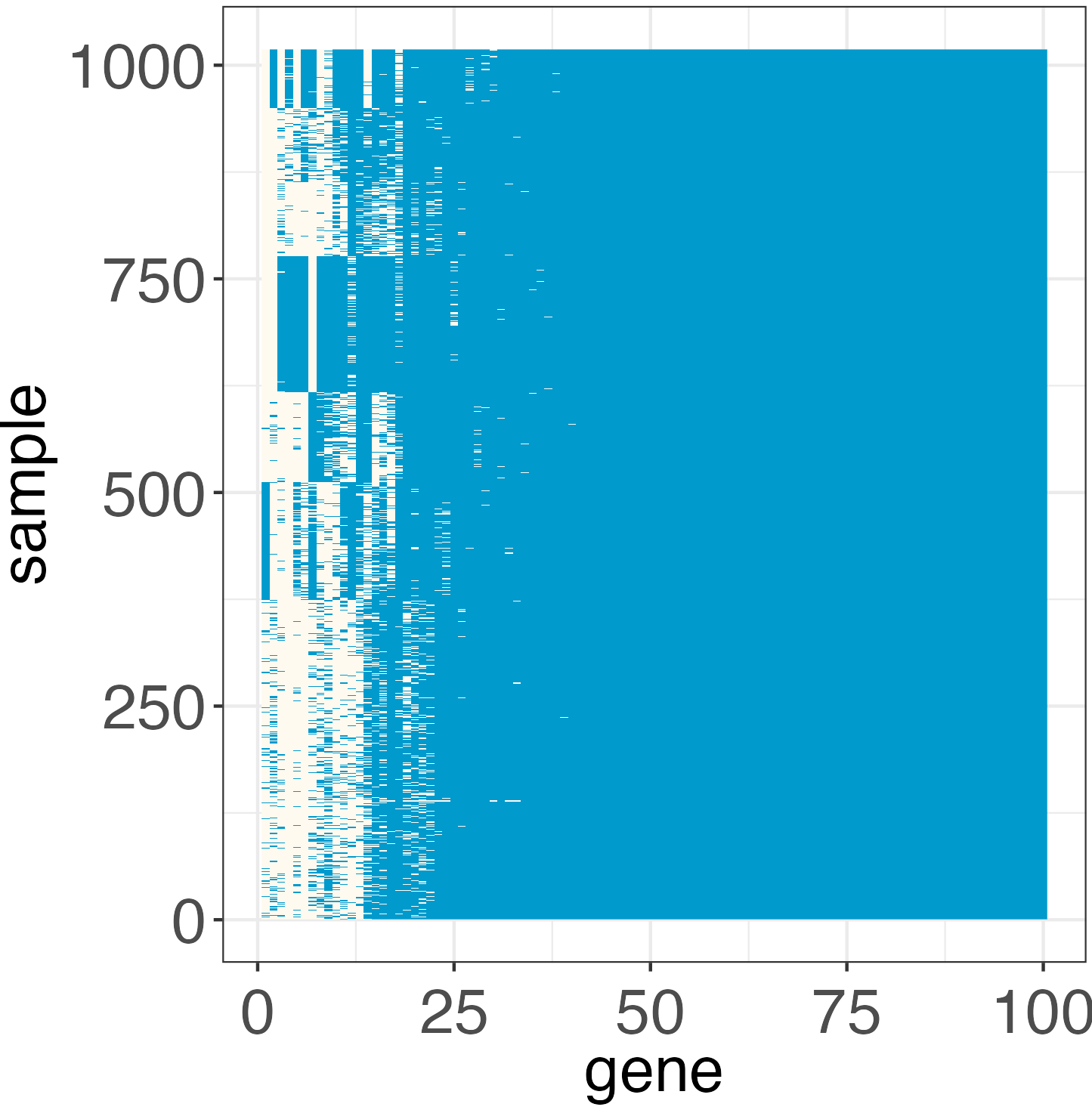}
    \includegraphics[height=3.3cm]{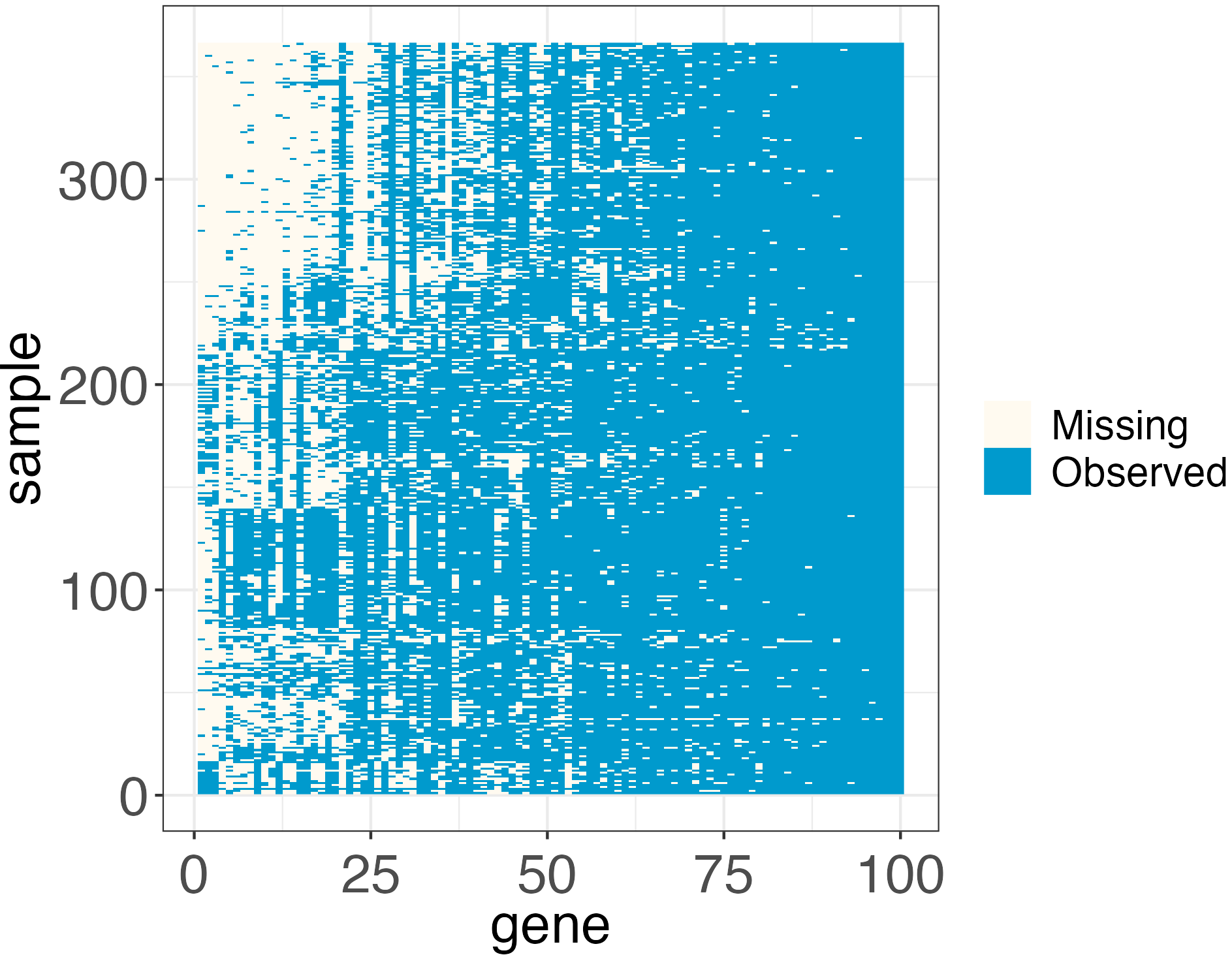}}
    \subfigure[Real pairwise sample sizes]{
    \includegraphics[height=3.3cm]{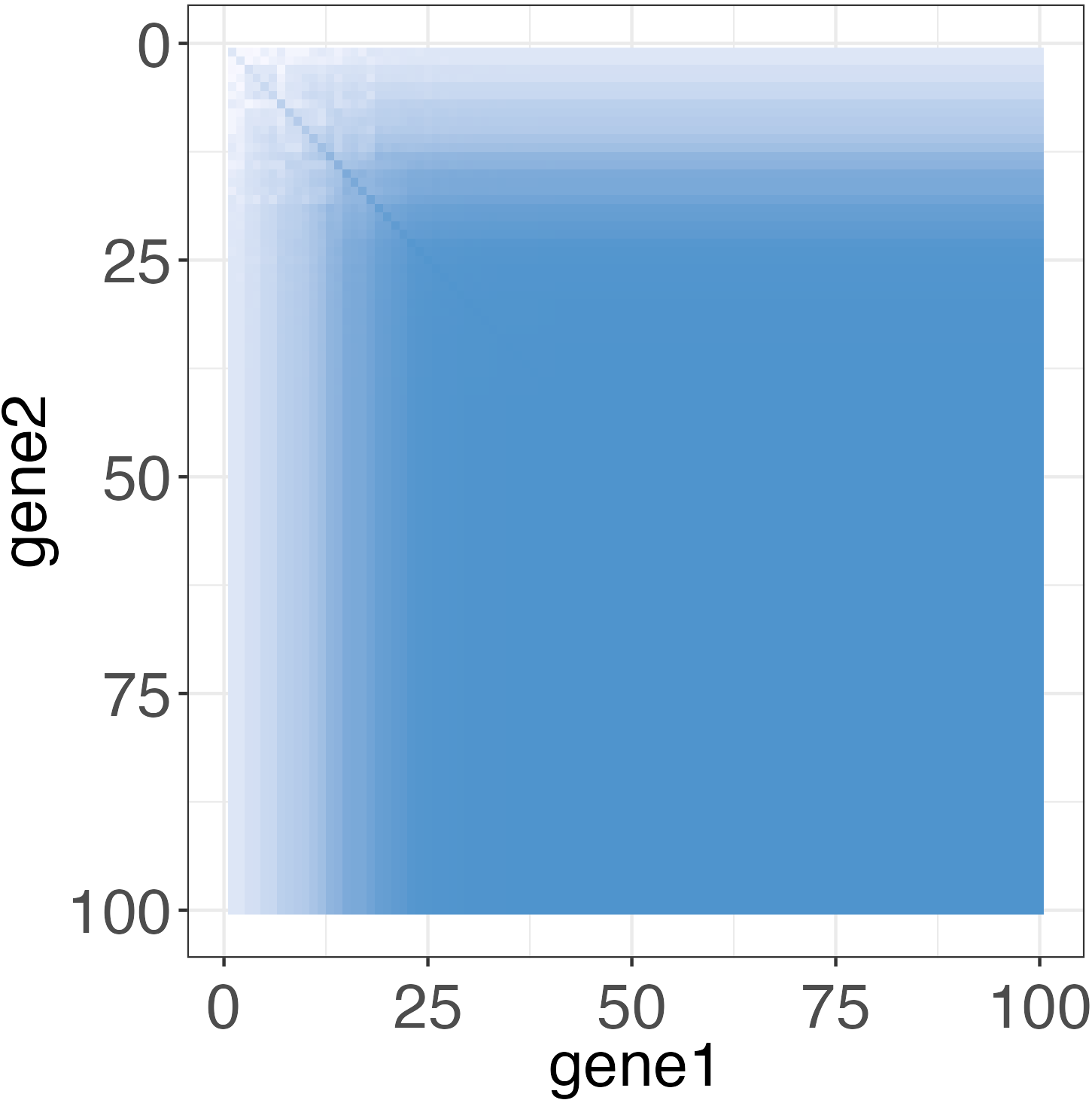}
    \includegraphics[height=3.3cm]{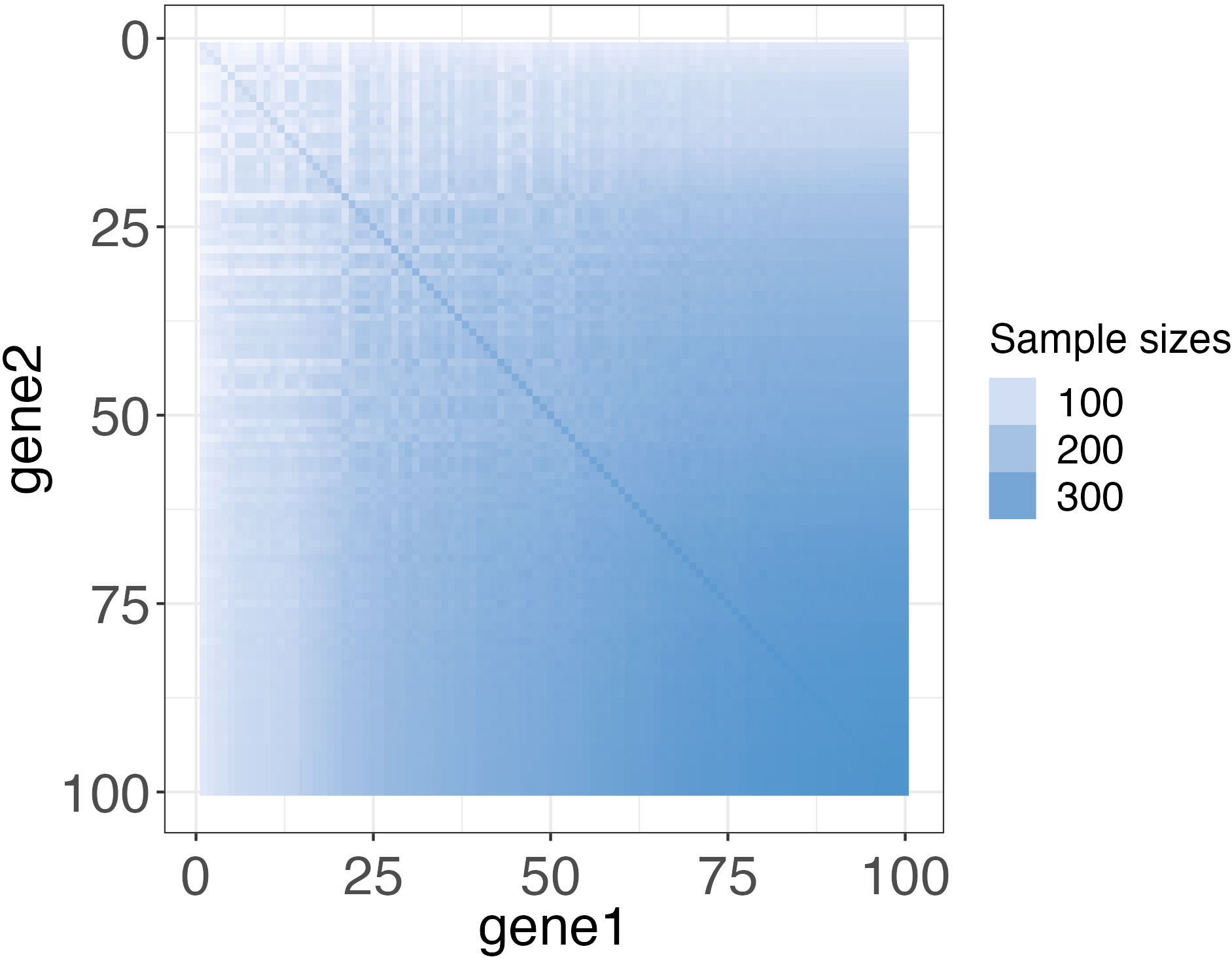}}
    \subfigure[GI-JOE (FDR)]{\includegraphics[height=3.5cm]{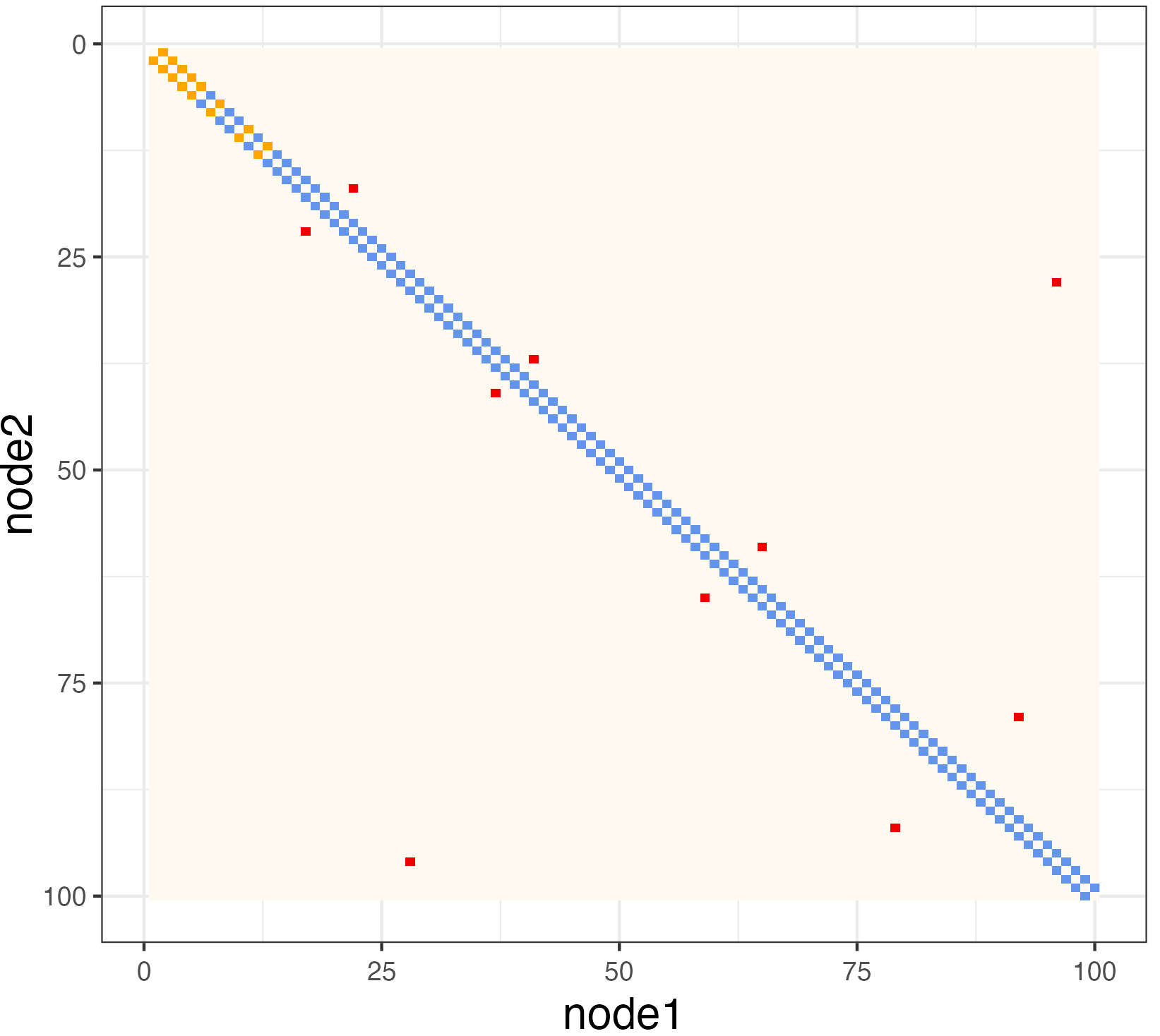}}
    \subfigure[Naive FDR control with minimum sample size]{
    \includegraphics[height = 3.5cm]{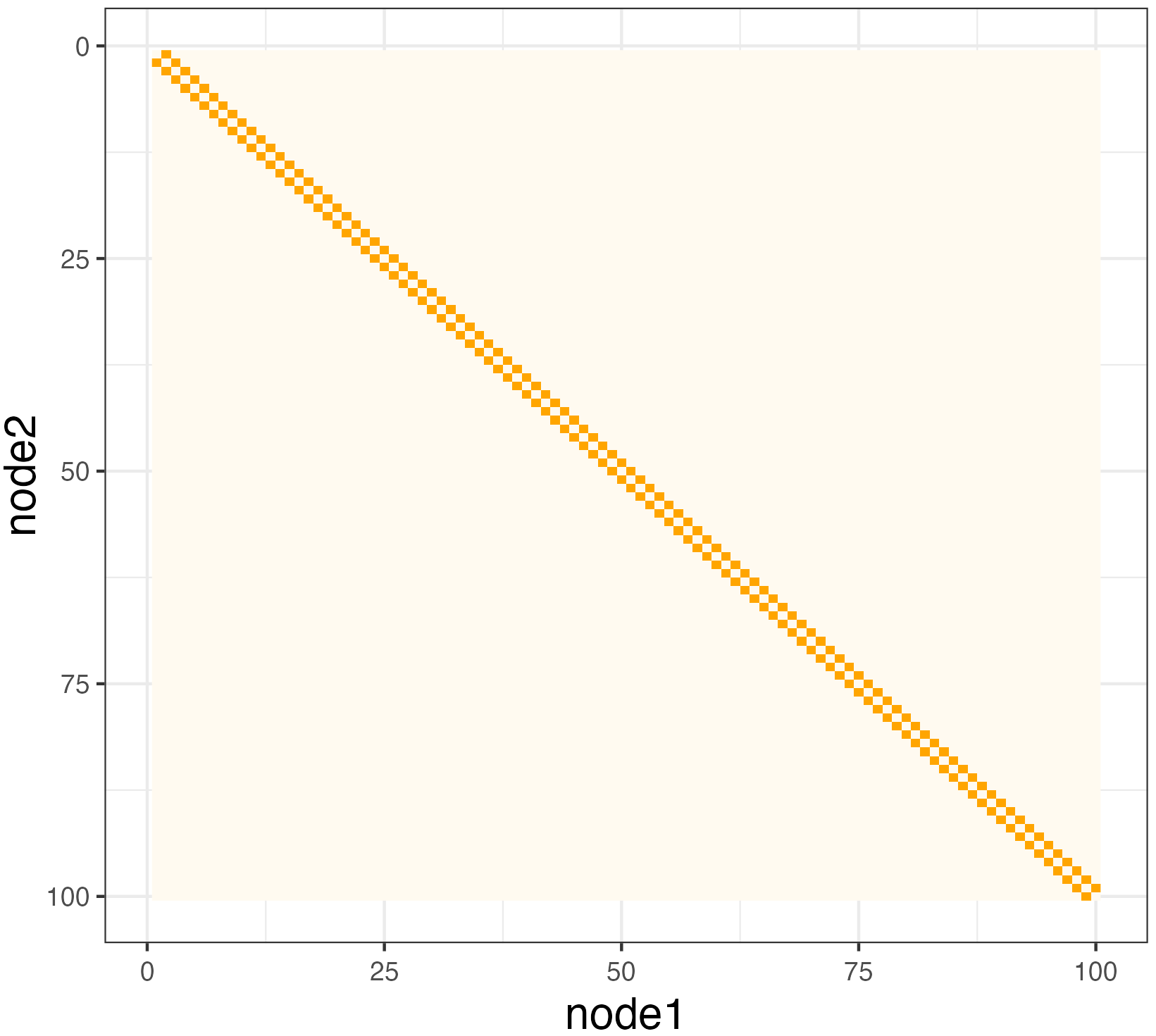}}
    \subfigure[Baseline estimator]{
    \includegraphics[height = 3.5cm]{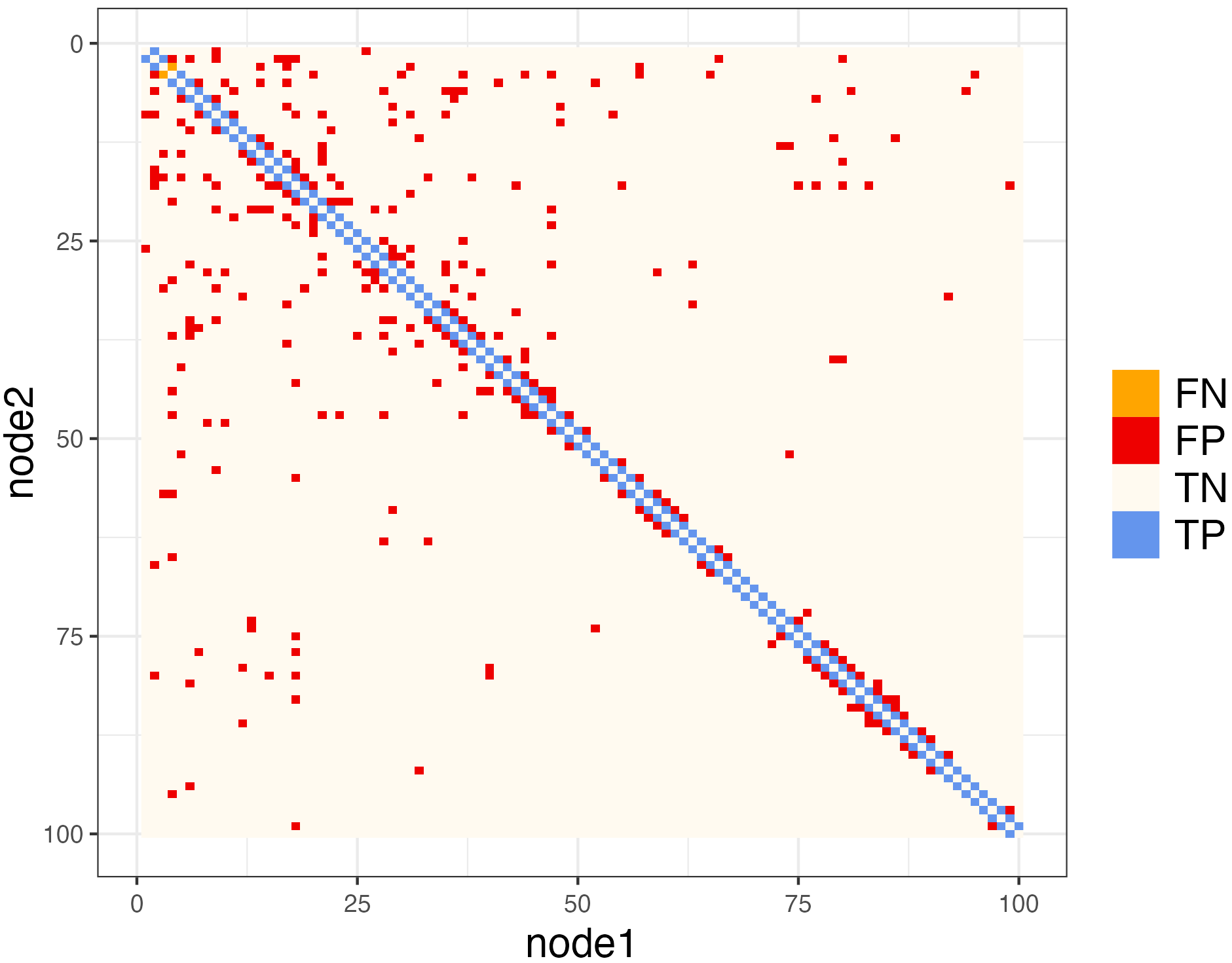}}
    \caption{\small Two erose measurement patterns in real scRNA-seq data sets \citep{chu2016single, darmanis2015survey} are presented in (a), (b), including the top 100 genes with the highest variances. The pairwise sample sizes range from 0 to 1018 (\emph{chu} data, left) and from 12 to 366 (\emph{darmanis}, right). (c)-(e) present the graph selection and inference results for a chain graph, when the data has the \emph{darmanis} measurement pattern. (c) is selected by our GI-JOE (FDR) approach and is the most accurate; (d) is obtained by an ad hoc implementation of the debiased graphical lasso \citep{jankova2015confidence} that plugs in the minimum pairwise sample size, which is too conservative and identifies no edge at all; (e) is the estimated graph by a baseline approach \citep{kolar2012estimating}, which plugs in a covariance estimate into the graphical lasso, and the many false positives suggest that the graph selection problem with such data set is non-trivial.}
    \label{fig:scRNA_seq_samplesizes}
\end{figure}

\vspace{-2mm}
\noindent \textbf{Data integration / size-constrained measurements}:
Non-simultaneous and uneven measurements also frequently arise from data integration and size-constrained measurements. For instance, to better understand the neuronal circuits from neuronal functional activities, one promising strategy is to estimate a large neuronal network \citep{vinci2018adjusted,chang2021extreme} from in vivo calcium imaging data sets. However, to ensure a sufficient temporal resolution of the recording, the spatial resolution is limited, putting a constraint on the number of neurons simultaneouly measured \citep{bae2021functional,zheng2022low},
and neuron pairs that are further from each other are less likely to be measured together. In genome-wide association studies (GWAS), it is also desirable to integrate genomic data across multiple sources due to the limited sample sizes of each data set, while these different sources might have different genomic coverage \citep{cai2016structured}. Similar measurement constraints also arise in sensor networks where it is extremely expensive to synchronize a large number of sensors \citep{dasarathy2016active,dasarathy2019gaussian}. 

%%%%%%%%%%%%%%%%%%%%%%%%%%%%%%%%%%%%%%%%%%%%
\subsection{Limitations of Existing Works for Erose Measurements}
To learn graphical models from erosely measured data, one might want to leverage the current literature on graphical models with missing data \citep{stadler2012missing,kolar2012estimating,wang2014gaussian,park2021estimating}. However, most of these works assume the variables are missing independently with the same missing probability. While \cite{park2021estimating} allows for arbitrary missing probabilities and dependency in their problem formulation, their theoretical guarantees still hinge on the minimum observational probability.
Using the minimum pairwise sample size over the whole graph to characterize the performance of the graph learning result can be too coarse and provides little insights to erosely measured data sets. Interestingly, one recent work \citep{zheng2022learning} provides a localized theoretical guarantee for neighborhood selection consistency, requiring only sample size conditions imposed upon the corresponding neighbors instead of all node pairs. Such theoretical results suggest that the estimation accuracy should vary over the graph when measurements are erose, and a coarse characterization based on the minimum sample size would only provide insights for the worst part of the graph estimate. 

Inspired by this intuition, here arises one natural question: \emph{can we develop a statistical inference method that quantifies the different uncertainty levels over the graph arising from the erose measurements?} 
Over the last decade, significant efforts have been devoted to the statistical inference in high-dimensional settings, including techniques such as the debiased Lasso \citep{van2014asymptotically,zhang2014confidence,javanmard2014confidence}, post-selection inference approaches \citep{lee2016exact,tibshirani2016exact}, knockoff methods \citep{barber2015controlling,candes2018panning}, and various other FDR control methods \citep{javanmard2019false,liu2013gaussian}. These techniques have been applied in regression or classification problems, as well as in graphical models. However, these prior works mainly consider simultaneous measurements across all variables \citep{jankova2015confidence,ren2015asymptotic,gu2015local,yu2020simultaneous,liu2013gaussian,jankova2017honest}, which, in the context of graphical models, would result in the same sample size across the entire graph; or they consider the missing data setting where all variables are missing independently with the same missing probability \citep{belloni2017confidence}, still leading to approximately the same sample sizes. To the best of our knowledge, there is no applicable statistical inference method for the general observational patterns and erose measurements that we are considering. 
If practitioners want to apply these existing inference methods with erosely measured data, they have to come up with one single sample size quantity $n$ to determine the uncertainty levels for each edge. To ensure the validity of the test, one ad hoc way might be to plug in the minimum pairwise sample size, which can be extremely conservative and has no power (see Figure \ref{fig:scRNA_seq_samplesizes}(d)). 

The rest of the paper is organized as follows. We first review the set-ups and neighborhood selection results from \cite{zheng2022learning} in Section \ref{sec:estimation_selection}, which serves as an inspiration and basis of our graph inference method under erose measurements; Our key contribution, the GI-JOE approach, is introduced in Section \ref{sec:debiased_nb_lasso} and \ref{sec:FDR}. In particular, Section \ref{sec:debiased_nb_lasso} is devoted to the edge-wise inference method, and for any node pair, we characterize its type I error and power based on the sample sizes involving the node pair's neighbors. Section \ref{sec:FDR} focuses on the FDR control procedure, also shown to be theoretically valid under appropriate conditions. The synthetic and real data experiments are included in Sections \ref{sec:numeric}. We conclude with discussion of some open questions in Section \ref{sec:discussion}.

\noindent \textbf{Notations:} For any matrix $A\in \mathbb{R}^{p_1\times p_2}$, let $\|A\|_{\infty}=\max_{j,k}|A_{j,k}|$, $\|A\|=\sup_{\|u\|_2=1}\|Au\|_2$ be its spectral norm, and $\vertiii{A}_{\infty}=\max_{j=1,\dots,p_1}\sum_{k=1}^p|A_{j,k}|$ be the matrix-operator $\ell_{\infty}$ to $\ell_{\infty}$ norm.  
For any tensor $\mathcal{T}\in \mathbb{R}^{p_1\times p_2\times p_3\times p_4}$ and matrix $A\in \mathbb{R}^{p_1,q_1}$ define the tensor-matrix/vector product $\mathcal{T}\times_1 A\in \mathbb{R}^{q_1\times p_2\times p_3\times p_4}$ as follows: $(\mathcal{T}\times_1 A)_{i_1,i_2,i_3,i_4}=\sum_{j_1=1}^{p_1}A_{j_1,i_1}\mathcal{T}_{j_1,i_2,i_3,i_4}$. Similarly we extend this definition of tensor-matrix product to other modes.
\section{Graph Selection with Erose Measurements}\label{sec:estimation_selection}
In this section, we review the set-up and neighborhood selection theory in \cite{zheng2022learning}, as it underpins our own inference procedure and theory in Section \ref{sec:debiased_nb_lasso}. In particular, we follow \cite{zheng2022learning} and study a variant of the neighborhood lasso method instead of other graph estimation methods \citep{yuan2007model,cai2011constrained}, since its form makes it easier to disentangle the effects of different parts of the graph on each other. 

The neighborhood lasso algorithm proposed in \cite{zheng2022learning} consists of two steps: estimating the true covariance $\Sigma^*$ and plugging the estimate into a neighborhood lasso estimator. An unbiased estimate $\widehat{\Sigma}$ is defined as follows: given observations $\{x_{i,V_i}\}_{i=1}^n$, for each entry $(j,k)$, $\widehat{\Sigma}_{j,k}=\frac{1}{n_{j,k}}\sum_{i: j,k\in V_i}x_{i,j}x_{i,k}$. However, $\widehat{\Sigma}$ is not guaranteed to be positive semi-definite, resulting in both optimization and statistical issues in neighborhood lasso. To ensure convexity and preserve the entry-wise error bounds for $\widehat{\Sigma}_{j,k}-\Sigma^*_{j,k}$, an additional projection step upon the positive semi-definite cone is considered: 
\begin{equation}\label{eq:proj_Sigma}
\widetilde{\Sigma}=\argmin_{\Sigma\succeq 0}\max_{j,k}\sqrt{n_{j,k}}|\Sigma_{j,k}-\widehat{\Sigma}_{j,k}|,
\end{equation} 
where $n_{j,k}$ is the pairwise sample size associated with node pair $(j,k)$, defined as in Section \ref{sec:setting}.
The projection problem \eqref{eq:proj_Sigma} can be solved by the ADMM, and we include the detailed optimization steps in Appendix \ref{sec:proj_alg}.

Given the covariance estimate $\widetilde{\Sigma}$, for any target node $a$ of which we want to estimate the neighborhood, consider the following neighborhood regression problem: 
\begin{equation}\label{eq:nb_lasso}
\widehat{\theta}^{(a)}=\argmin_{\theta\in \mathbb{R}^p,\theta_a=0}\frac{1}{2}\theta^\top \widetilde{\Sigma}\theta-\widetilde{\Sigma}_{a,:}\theta +\sum_{j=1}^p\lambda^{(a)}_j|\theta_j|,
\end{equation}
where $\lambda^{(a)}=(\lambda^{(a)}_1,\dots,\lambda^{(a)}_p)^\top\in \mathbb{R}^p$ is a vector of tuning parameters, with each entry $\lambda^{(a)}_j$ corresponding to a potential edge connecting node $j$ and $a$. The solution $\widehat{\theta}^{(a)}$ serves as an estimate for
$\theta^{(a)*}=\argmin_{\theta\in \mathbb{R}^p,\theta_a=0}\frac{1}{2}\theta^\top \Sigma^*\theta-\Sigma^*_{a,:}\theta,$
which satisfies 
$\theta^{(a)*}_{\backslash a}=(\Sigma^*_{\backslash a,\backslash a})^{-1}\Sigma^*_{\backslash a,a}=\frac{1}{\Theta^*_{a,a}}\Theta^*_{\backslash a, a},$
and hence the support set of $\theta^{(a)*}$ equals the true neighborhood of node $a$: $\mathcal{N}_a=\{j\neq a:\Theta^*_{a,j}\neq 0\}$. Then one can estimate $\mathcal{N}_a$ by the support of $\widehat{\theta}^{(a)}$: $\widehat{\mathcal{N}}_a=\{j\neq a:\widehat{\theta}^{(a)}\neq 0\}$. It was shown in \cite{zheng2022learning} that the neighborhood selection consistency is guaranteed with sample size conditions involving the neighbors of node $a$. Here, we present a similar theoretical result, with only a slight modification on the tuning parameter choice. 
Let $\gamma_{a}=\frac{\mymax_{j\in \overline{\cN}_a^c}\mymin_k n_{j,k}}{\mymin_{j\in \cN_a}\mymin_{k}n_{j,k}}$ be the sample size ratio between $a$'s non-neighbors and neighbors.
\begin{thm}[Neighborhood Selection Consistency, Similar to \cite{zheng2022learning}]\label{thm:nblasso_support}
	Consider the Gaussian graphical model with erose measurement setting described in Section~\ref{sec:setting} and the estimator $\widehat{\theta}^{(a)}$ defined in~\eqref{eq:nb_lasso}.
	Suppose Assumption \ref{assump:incoh} in Appendix \ref{append:nbconsistency_theory} (the mutual incoherence condition) holds, and the tuning parameters $\lambda^{(a)}_j$'s in \eqref{eq:nb_lasso} satisfy $\lambda^{(a)}_j \asymp$ $\|\Sigma^*\|_{\infty}\frac{\|\Theta^*_{:,a}\|_1}{\Theta^*_{a,a}}$ $\sqrt{\frac{\log p}{\mymin_{k} n_{j,k}}}$. If $\gamma_a\leq C$ for some $C>0$ depending on the incoherence parameter,
	\begin{equation}\label{eq:nblasso_samplesize}
	\min_{j\in \mathcal{N}_a}\min_k n_{j,k}\geq C(\Sigma^*)\|\Sigma^*\|^2_{\infty}\left[d_a^2+(\theta^{(a)}_{\min})^{-2}\right]\log p,
	\end{equation}
	where the constant $C(\Sigma^*)$ depends on $\Sigma^*$, then $\widehat{\mathcal{N}}_a=\{j:\widehat{\theta}^{(a)}_j\neq 0\}=\mathcal{N}_a$ with probability at least $1-p^{-c}$ for some absolute constants $c>0$. 
\end{thm}
\noindent The complete version of Theorem \ref{thm:nblasso_support}, additional $\ell_1$ and $\ell_2$ error bounds for $\widehat{\theta}^{(a)}-\theta^{(a)*}$, a pictorial illustration of the sample size condition \eqref{eq:nblasso_samplesize}, and the proofs can be found in Appendix \ref{append:proofs}. The localized characterization of the graph estimation performance in Theorem \ref{thm:nblasso_support} inspires us to develop an inference method that quantifies the uneven uncertainty levels over the graph.

\section{Edge-wise Inference: Quantifying Uncertainties from Erose Measurements}\label{sec:debiased_nb_lasso}
In this section, we propose our GI-JOE method for edge-wise inference with erose data. The key idea follows the debiased lasso \citep{van2014asymptotically}, while the main challenge and innovation is characterizing the uncertainty level associated with each edge-wise statistic. We first introduce our edge-wise debiased statistic $\widetilde{\theta}^{(a)}_b$ in Section \ref{sec:debiased_nb_lasso}, and characterize its asymptotic distribution in Section \ref{sec:normal_approx}; We further propose a consistent estimator of its variance and establish statistical validity of edge-wise inference in Section \ref{sec:var_est}.
\subsection{Debiased Neighborhood Lasso}\label{sec:debiasing}
First we introduce the key idea and intuition behind how we construct our debiased test statistic. Recall that our neighborhood regression estimator $\widehat{\theta}^{(a)}$ was defined as in \eqref{eq:nb_lasso}, and by the Karush–Kuhn–Tucker (KKT) condition, we know that it satisfies
\begin{equation*}
    \widetilde{\Sigma}_{\backslash a,\backslash a}\widehat{\theta}^{(a)}_{\backslash a} - \widetilde{\Sigma}_{\backslash a,a}+(\lambda^{(a)}\circ \widehat{Z})_{\backslash a}=0,
\end{equation*}
where $\circ$ represents element-wise multiplication, and $\widehat{Z}\in \mathbb{R}^p$ satisfies that $\|\widehat{Z}\|_{\infty}\leq 1$ and $\widehat{Z}_j=\mathrm{sgn}(\widehat{\theta}^{(a)}_j)$ if $\widehat{\theta}^{(a)}_j\neq 0$. Noting the fact that $\Sigma^*\theta^{(a)*}-\Sigma^*_{:,a}=0$, and the relationship between $\theta^{(a)*}$ and $\Theta^*_{:,a}$, we can use some rearrangements to obtain the following:
\begin{equation*}
\begin{split}
    \widetilde{\Sigma}_{\backslash a,\backslash a}(\widehat{\theta}^{(a)}-\theta^{(a)*})_{\backslash a} +(\lambda^{(a)}\circ \widehat{Z})_{\backslash a}&=(\Theta^*_{a,a})^{-1}(\widetilde{\Sigma}-\Sigma^*)_{\backslash a,:}\Theta^*_{:,a},\\
    \widetilde{\Sigma}_{\backslash a,\backslash a}(\widehat{\theta}^{(a)}-\theta^{(a)*})_{\backslash a} +\widetilde{\Sigma}_{\backslash a,a}-\widetilde{\Sigma}_{\backslash a,\backslash a}\widehat{\theta}^{(a)}_{\backslash a} &= (\Theta^*_{a,a})^{-1}(\widetilde{\Sigma}-\Sigma^*)_{\backslash a,:}\Theta^*_{:,a}.
\end{split}
\end{equation*}
The derivation above follows similar arguments for the debiased lasso in \cite{van2014asymptotically}, and ideally, we would hope the RHS of the equation above has (asymptotically) normal distribution and can serve as a basis for our inference. However, since $\widetilde{\Sigma}$ is the solution of a weighted $\ell_{\infty}$ projection onto the positive semi-definite cone, the CLT is not directly applicable. Instead, we want to change it to a function or $\widehat{\Sigma}$,
whose entries can be written as independent sums. As will be shown in our proofs, substituting $\widetilde{\Sigma}$ by $\widehat{\Sigma}$ in the debiasing terms above can help us achieve this goal: $\widetilde{\Sigma}_{\backslash a,\backslash a}(\widehat{\theta}^{(a)}-\theta^{(a)*})_{\backslash a} +\widehat{\Sigma}_{\backslash a,a}-\widehat{\Sigma}_{\backslash a,\backslash a}\widehat{\theta}^{(a)}_{\backslash a} \approx (\Theta^*_{a,a})^{-1}(\widehat{\Sigma}-\Sigma^*)_{\backslash a,:}\Theta^*_{:,a}.$
Furthermore, to invert the factor $\widetilde{\Sigma}_{\backslash a,\backslash a}$, we need a good approximation of $(\Sigma^*_{\backslash a,\backslash a})^{-1}\in \mathbb{R}^{(p-1)\times (p-1)}$. Define the debiasing matrix $\Theta^{(a)*} \in \bR^{p\times p}$, which satisfies $\Theta^{(a)*}_{a,:}=0$, $\Theta^{(a)*}_{:,a}=0$, and $\Theta^{(a)*}_{\backslash a, \backslash a}= (\Sigma^*_{\backslash a,\backslash a})^{-1}$. Then suppose we have a good estimate $\Theta^{(a)}$, one would be able to show 
\begin{equation}\label{eq:debias_approx}
\begin{split}
    \widehat{\theta}^{(a)}-\theta^{(a)*} +\Theta^{(a)}(\widehat{\Sigma}_{:,a}-\widehat{\Sigma}_{:,\backslash a}\widehat{\theta}^{(a)}_{\backslash a}) \approx (\Theta^*_{a,a})^{-1}\Theta^{(a)*}(\widehat{\Sigma}-\Sigma^*)\Theta^*_{:,a}.
\end{split}
\end{equation}
Given a node pair $(a,b)$ for $a\neq b$, this motivates us to consider an edge-wise test statistic of the form $\widehat{\theta}^{(a)}_b+\Theta^{(a)}_{b,:}(\widehat{\Sigma}_{:,a}-\widehat{\Sigma}_{:,\backslash a}\widehat{\theta}^{(a)}_{\backslash a})$, where $\Theta^{(a)}_{b,:}$ is an appropriate estimate for $\Theta^{(a)*}_{b,:}$.

Throughout the rest of this section, suppose that we are interested in testing whether there is an edge between node $a\neq b$. 
Now we introduce our estimates for $\Theta^{(a)*}_{b,:}$. Denote by $\cN_b^{(a)}$ the support set of $\Theta^{(a)*}_{b,:}$ and $\overline{\cN}_b^{(a)}=\cN_b^{(a)}\cup j$. By block matrix inverse formula, $\Theta^{(a)*}_{b,:}=\Theta^*_{b,:}-(\Theta^*_{a,a})^{-1}\Theta^*_{b, a}\Theta^*_{a,:}$ and hence is also sparse with $d_b^{(a)}:=|\cN_b^{(a)}|\leq d_a+d_b$. Therefore, we can estimate $\Theta^{(a)*}_{b,:}$ by performing another neighborhood regression. Let
\begin{equation}\label{eq:debias_term}
\begin{split}
\widehat{\theta}^{(a,b)} = &\argmin_{\theta\in \bR^{p},\theta_a=\theta_b=0}\frac{1}{2}\theta^\top \widetilde{\Sigma}\theta - \widetilde{\Sigma}_{b,:}\theta +\sum_{k=1}^{p}\lambda^{(a,b)}_k|\theta_k|,\\ \widehat{\overline{\theta}}^{(a,b)}_b=&1,\, \widehat{\overline{\theta}}^{(a,b)}_{\backslash b}=-\widehat{\theta}^{(a,b)}_{\backslash b},
\end{split}
\end{equation}
where $\lambda^{(a,b)}_k$'s are tuning parameters depending on the pairwise sample sizes $\mymin_{i\in [p]}n_{i,b}$. 
Then $\widehat{\overline{\theta}}^{(a,b)}$ serves as an estimate of $(\Theta^{(a)*}_{b,b})^{-1}\Theta^{(a)*}_{b,:}$. To estimate $\Theta^{(a)*}_{b,b}$, we note the fact that $\Theta^{(a)*}_{b,b}=[\Sigma^*_{b,:}(\Theta^{(a)*}_{b,b})^{-1}\Theta^{(a)*}_{:,b}]^{-1}=[(\Theta^{(a)*}_{b,b})^{-2}\Theta^{(a)*}_{b,:}\Sigma^*\Theta^{(a)*}_{:,b}]^{-1}$. Hence either of the following two estimators can serve appropriately for estimating $\Theta^{(a)*}_{b,:}$:
\begin{equation}\label{eq:debias_Theta_est}
\begin{split}
\widehat{\Theta}^{(a)}_{b,b} = &(\widetilde{\Sigma}_{b,:}\widehat{\overline{\theta}}^{(a,b)})^{-1},\, \widehat{\Theta}^{(a)}_{b,:} = \widehat{\Theta}^{(a)}_{b,b}\widehat{\overline{\theta}}^{(a,b)},\\
\widetilde{\Theta}^{(a)}_{b,b} = &(\widehat{\overline{\theta}}^{(a,b)\top}\widetilde{\Sigma}\widehat{\overline{\theta}}^{(a,b)})^{-1},\, \widetilde{\Theta}^{(a)}_{b,:} = \widetilde{\Theta}^{(a)}_{b,b}\widehat{\overline{\theta}}^{(a,b)}.
\end{split}
\end{equation}
As we will show in Lemma \ref{lem:debias_nb_lasso_err}, both estimators are consistent and lead to sufficiently good statistical error bounds. Based on some empirical investigation (details presented in Section F), we propose to use $\widehat{\Theta}^{(a)}_{b,:}$ for the debiasing step, but would revisit $\widetilde{\Theta}^{(a)}_{b,:}$ for variance estimation in Section \ref{sec:var_est}.
Then the debiased neighborhood lasso estimator for node pair $(a,b)$ is 
\begin{equation}\label{eq:debiased_nb_lasso}
\widetilde{\theta}^{(a)}_b=\widehat{\theta}^{(a)}_b-\widehat{\Theta}^{(a)}_{b,:}(\widehat{\Sigma}\widehat{\theta}^{(a)}-\widehat{\Sigma}_{:,a}). 
\end{equation}

\subsection{Normal Approximation of Debiased Edge-wise Statistic}\label{sec:normal_approx}
Although the edge-wise statistic $\widetilde{\theta}^{(a)}_b$ defined in \eqref{eq:debiased_nb_lasso} is similar to the debiased lasso in the literature, its asymptotical normality is not readily present due to the erose measurement setting we are concerned with. In the following, we present a novel characterization of  $\widetilde{\theta}^{(a)}_b$ that consists of a bias term and an asymptotically normal error term, each term depending on one pairwise sample size quantity, respectively. Before presenting the main theorem, we first define and discuss these two key sample size quantities.

Given the target node pair $(a,b)$, define two sets of node pairs involving $a,\,b$'s neighbors: $S_1(a,b)=\{(j,k): j\text{ or }k\in \cN_a\cup \overline{\cN}_b^{(a)}\}$, $S_2(a,b)=\{(j,k): \Theta^{(a)*}_{j,b}\Theta^*_{k,a}+\Theta^{(a)*}_{k,b}\Theta^*_{j,a}\neq 0\}$, where $\overline{\cN}_b^{(a)}$ and matrix $\Theta^{(a)*}$ are defined in the beginning of Section \ref{sec:debiased_nb_lasso}. Here the order of $a$ and $b$ matters since we first apply neighborhood lasso for node $a$ and then debias its entry $\widehat{\theta}^{(a)}_b$. Proposition \ref{prop:indexsets} characterizes the index set $S_2(a, b)$ and $\overline{\cN}_b^{(a)}$ through their relationships with $\overline{\cN}_a$ and $\overline{\cN}_b$. Figure \ref{fig:inference_n_example} also gives a pictorial illustration of $S_1(a,b)$ and $S_2(a,b)$ for a chain graph. The two key sample size quantities are then defined as the minimum pairwise sample sizes within these two sets:
$n_1^{(a,b)}=\mymin_{(j,k)\in S_1(a,b)}n_{j,k}$, $n_2^{(a,b)}=\mymin_{(j,k)\in S_2(a,b)}n_{j,k},$ which will be shown to determine the bias and variance of the edge-wise statistic. The intuition behind these two node pair sets can be traced back to our main idea for constructing the debiased edge-wise statistic in Section \ref{sec:debiased_nb_lasso}. As shown in \eqref{eq:debias_approx}, our debiased test statistic for $(a,b)$ can be well approximated by $\langle\widehat{\Sigma}-\Sigma^*,\frac{\Theta^*_{:,a}\Theta^{(a)*}_{b,:}}{\Theta^*_{a,a}}\rangle=\langle\widehat{\Sigma}-\Sigma^*,\frac{\Theta^*_{:,a}\Theta^{(a)*}_{b,:}+\Theta^{(a)*}_{:,b}\Theta^*_{a,:}}{2\Theta^*_{a,a}}\rangle$, which can be intuitively understood as a first order Taylor's expansion of the estimation error around $\Sigma^*$. The approximation errors constitute our bias term, which mainly depends on how well we estimate $\theta^{(a)*}$ and $\Theta^{(a)*}_{b,:}$ using neighborhood regression. Similar to the neighborhood selection theory presented in Section \ref{sec:estimation_selection}, we can show that the estimation error for $\theta^{(a)*}$ and $\Theta^{(a)*}_{b,:}$ depend on the pairwise sample sizes $n_{j,k}$ for $j$ or $k$ in the neighborhood sets $\mathcal{N}_a$ and $\overline{\mathcal{N}}_b^{(a)}$, respectively, and hence this leads to our definition of set $S_1(a,b)$. On the other hand, since the matrix $\Theta^*_{:,a}\Theta^{(a)*}_{b,:}+\Theta^{(a)*}_{:,b}\Theta^*_{a,:}$ has support set $S_2(a,b)$, the variance of $\langle\widehat{\Sigma}-\Sigma^*,\frac{\Theta^*_{:,a}\Theta^{(a)*}_{b,:}+\Theta^{(a)*}_{:,b}\Theta^*_{a,:}}{2\Theta^*_{a,a}}\rangle$ is then dominated by the minimum sample sizes in $S_2(a,b).$
\begin{prop}\label{prop:indexsets}
    For any given support set $\overline{E}\subseteq [p]\times [p]$, the following holds except when $\Theta^*_{\overline{E}}\in \mathbb{R}^{|\overline{E}|}$ falls in a measure zero set:
    (i) $S_2(a, b)=(\overline{\cN}_a\times \overline{\cN}_b^{(a)})\cup (\overline{\cN}_b^{(a)}\times \overline{\cN}_a)$. (ii) If $b\in \cN_a$, $\overline{\cN}_b^{(a)}=\overline{\cN}_a\cup\overline{\cN}_b$; otherwise, $\overline{\cN}_b^{(a)}=\overline{\cN}_b$.
\end{prop}

\begin{figure}[!htb]
    \centering
    \includegraphics[width=0.38\textwidth,height = 2cm]{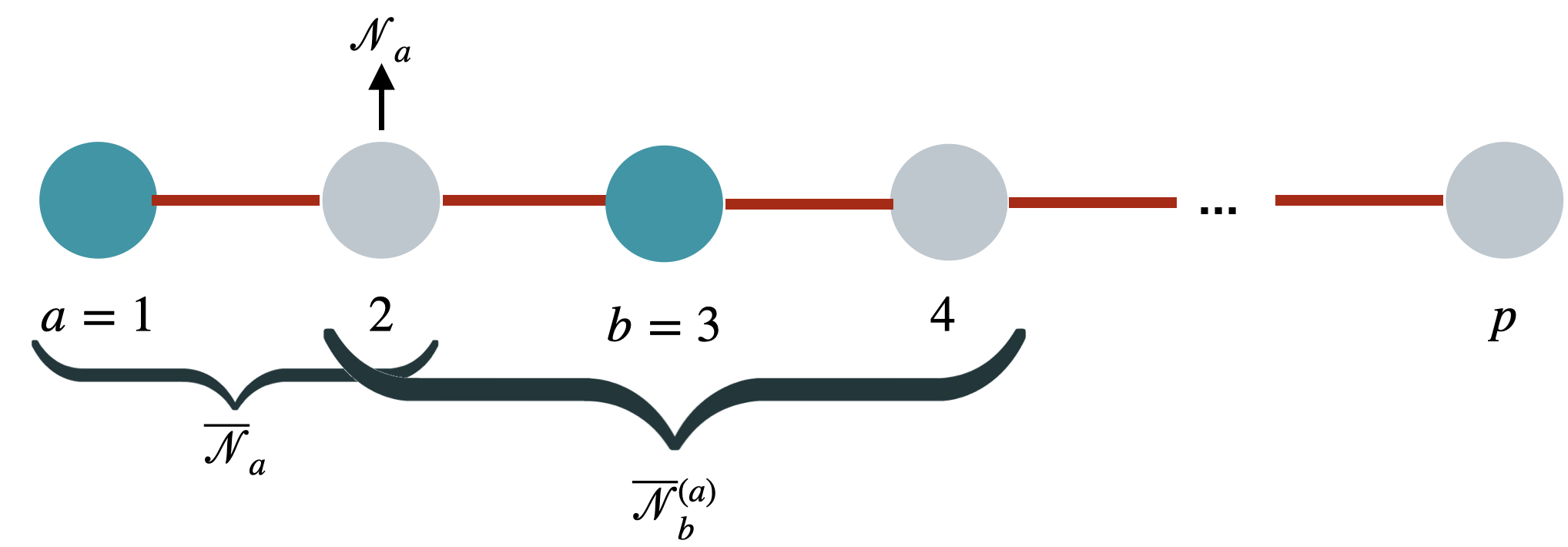}
    \includegraphics[width=0.3\textwidth]{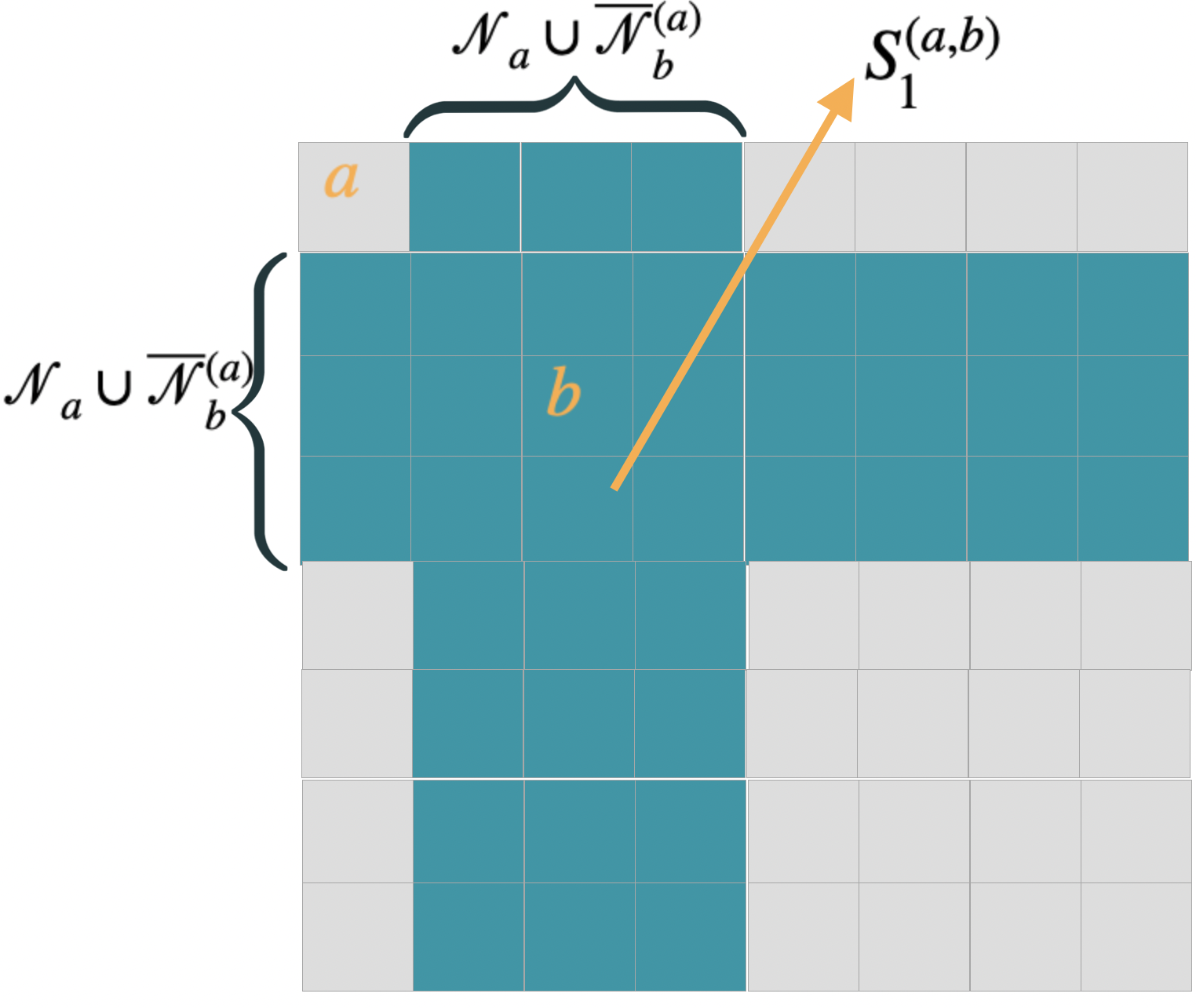}
    \includegraphics[width=0.3\textwidth]{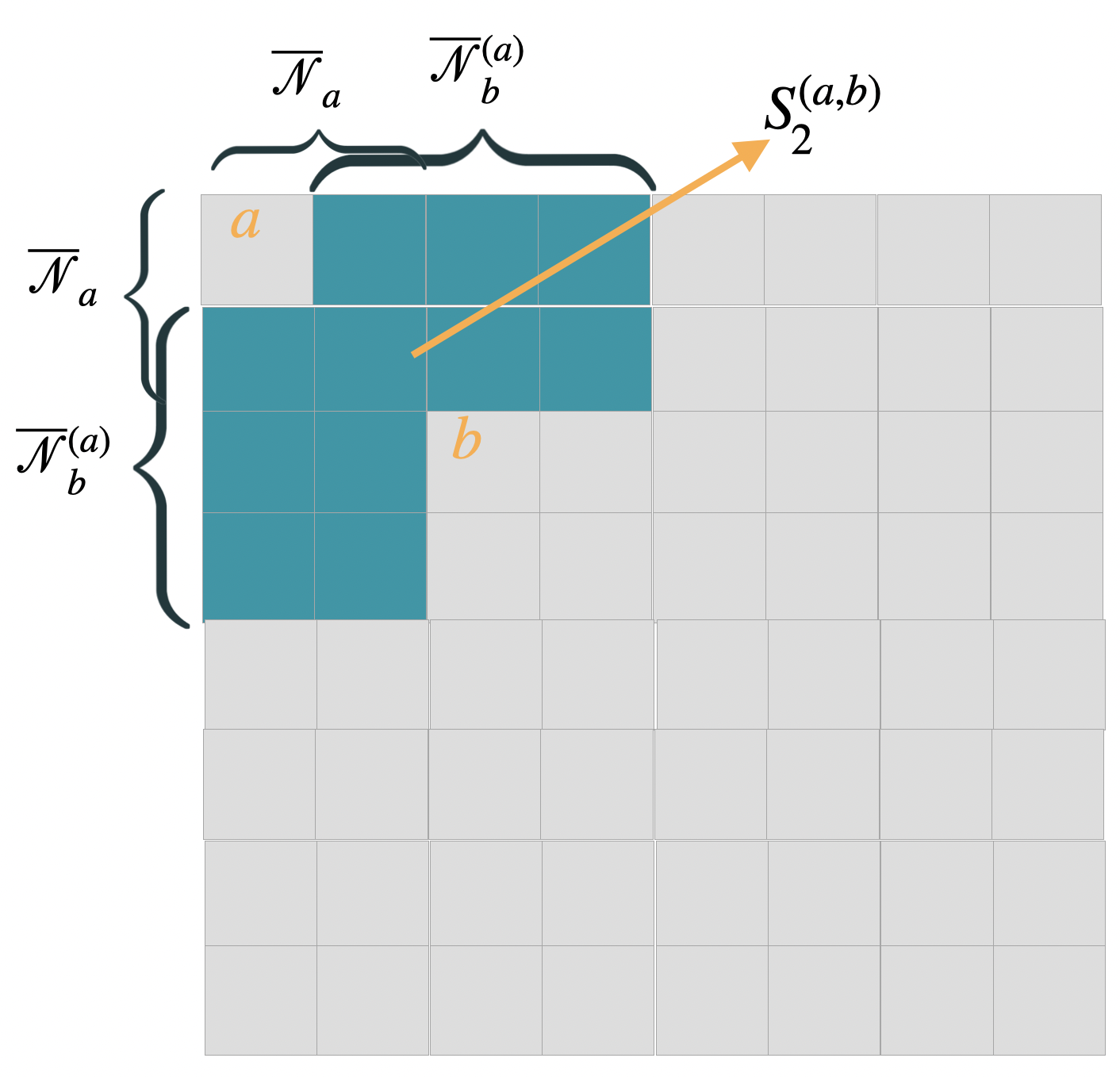}
    \caption{\small An illustration of the set $S_1(a,b)$ and $S_2(a,b)$ 
    in a chain graph, when $a=1$ and $b=3$. The minimum sample size in $S_1(a,b)$ determines the bias for estimating edge $(a,b)$, while the minimum sample size in $S_2(a,b)$ determines the variance for estimating edge $(a,b)$.
    }
    \label{fig:inference_n_example}
\end{figure}

Similar to the the support recovery guarantee in Theorem \ref{thm:nblasso_support}, here we also define the sample size ratio for node $b$ here by $\gamma^{(a)}_b=\frac{\max_{j\in \overline{\cN}_b^{(a)c}}\min_k n_{j,k}}{\min_{j\in \cN_b^{(a)}}\min_k n_{j,k}}$. The following covariance parameters are also useful: let $\mathcal{T}^*,\mathcal{T}^{(n)*}\in \mathbb{R}^{p\times p\times p\times p}$ satisfy $$\mathcal{T}^*_{j,k,j',k'}=\mathrm{Cov}(X_jX_k,X_{j'}X_{k'})=\Sigma^*_{j,j'}\Sigma^*_{k,k'} + \Sigma^*_{j,k'}\Sigma^*_{k,j'}$$ for $1\leq j,k,j',k'\leq p$, and $(\mathcal{T}^{(n)*})_{j,k,j',k'} = \mathcal{T}^*_{j,k,j',k'}\frac{n_{j,k,j',k'}}{n_{j,k}n_{j',k'}}$, where $n_{j,k,j',k'}=|\{i:j,k,j',k'\in V_i\}|$ is the number of joint measurements for $j,k,j',k'$. 
\vspace{-3mm}
\begin{assump}[Sample size condition for accurate estimation]\label{assump:inference_n_B}
    \begin{align*}
	    n_1^{(a,b)}\geq C\frac{\|\Sigma^*\|_{\infty}^2}{\lambda_{\mymin}^2(\Sigma^*)}(\kappa_{\Sigma^*}^2+\gamma_a+\gamma^{(a)}_b)(d_a+d_b+1)^2\log p,
	\end{align*} 
\end{assump}
Assumption \ref{assump:inference_n_B} is similar to the sample size condition in Theorem \ref{thm:nblasso_support}, while the only difference lies that here $\min_kn_{j,k}$ needs to be large as long as $j\in \cN_a\cup\overline{\cN}_b^{(a)}$ instead of $\cN_a$ only, so that both $\widehat{\theta}^{(a)}$ and $\widehat{\Theta}^{(a)}_{b,:}$ are accurate estimators for $\theta^{(a)*}$ and $\Theta^{(a)*}_{b,:}$. 
\vspace{-3mm}
\begin{assump}[Sample size condition for normal approximation]\label{assump:inference_n_E}
    $C_{\varepsilon}(\Sigma^*)(d_a+d_b+1)^{2+\varepsilon}=o(n_2^{(a,b)})$ for some constant $\varepsilon>0$, where $C_{\varepsilon}(\Sigma^*)=\left(\frac{C(1+2/\epsilon)\|\Sigma^*\|_{\infty}}{\lambda_{\mymin}(\Sigma^*)}\right)^{2+\epsilon}$.
\end{assump}
Due to the erose measurements, establishing the Lyapunov condition is much more complicated than the same sample size setting. Assumption \ref{assump:inference_n_E} is a technical assumption we need in this step so that the CLT can be applied to derive asymptotic normality results.
\vspace{-3mm}
\begin{assump}[Sample size condition for controlling bias]\label{assump:inference_n_BE}
    \begin{equation}\label{eq:inference_n_BE}
        n_1^{(a,b)} \gg C^2(\Theta^*;a,b)(\kappa_{\Theta^*}^4+\gamma_a+\gamma^{(a)}_b)[(d_a+d_b+1)\log p]^2\frac{n_2^{(a,b)}}{n_1^{(a,b)}},
    \end{equation}
    where $C(\Theta^*;a,b)=\frac{C\kappa_{\Theta^*}^3\|\Theta^*_{:,a}\|_1\|\Theta^{(a)*}_{:,b}\|_1}{\mymin_{(j,k)\in S_2(a, b)}\left|\Theta^{(a)*}_{b,j}\Theta^{*}_{a,k}+\Theta^{(a)*}_{b,k}\Theta^{*}_{a,j}\right|}$.
\end{assump}
\begin{remark}
Rearranging \eqref{eq:inference_n_BE}, we can also write the this sample size condition as
$$
n_1^{(a,b)} \gg C(\Theta^*;a,b)(\kappa_{\Theta^*}^2+\sqrt{\gamma_a}+\sqrt{\gamma^{(a)}_b})[(d_a+d_b+1)\log p]\sqrt{n_2^{(a,b)}}.
$$
\end{remark}
\begin{remark}
One may be confused when seeing $n_1^{(a,b)}$ both on the left and right hand side of \eqref{eq:inference_n_BE}. In fact, we present it this way in order to connect it to the same sample size setting where $n_2^{(a,b)}=n_1^{(a,b)}=n$, and thus \eqref{eq:inference_n_BE} becomes
$n \gg C^2(\Theta^*;a,b)\kappa_{\Theta^*}^4[(d_a+d_b+1)\log p]^2$.
This is similar to the requirement in prior results on debiased lasso and debiased graphical lasso \citep{van2014asymptotically,zhang2014confidence,jankova2015confidence} with the same sample sizes, which requires $n\gg d^2\log^2 p$. The additional price we paid for uneven sample sizes is reflected in the sample size ratios $\gamma_a$, $\gamma^{(a)}_b$ and $\frac{n_2^{(a,b)}}{n_1^{(a,b)}}$. 
\end{remark}
\vspace{-4mm}
\begin{remark}[Effect of $\gamma_a$, $\gamma^{(a)}_b$ and $\frac{n_2^{(a,b)}}{n_1^{(a,b)}}$]
$\gamma_a$, $\gamma^{(a)}_b$ are the sample size ratios between the most well measured non-neighbor and the worst measured neighbor of $a$ and $b$. These two quantities have a negative effect on our theory, as when the sample sizes for the non-neighbors are all much larger than the neighbors, the neighbors would suffer from much stronger regularization than non-neighbors. While for the sample size ratio $\frac{n_2^{(a,b)}}{n_1^{(a,b)}}$, note that when $n_2^{(a,b)}$ grows too much more quickly than $n_1^{(a,b)}$, the bias term would dominate the variance term and then the normal approximation of $\widetilde{\theta}^{(a)}_b$ would not hold. 
\end{remark}
\vspace{-3mm}
The following theorem establishes the asymptotic normality of $\widetilde{\theta}^{(a)}_{b}+\frac{\Theta^*_{a,b}}{\Theta^*_{a,a}}$ under these three sample size assumptions, and Corollary \ref{cor:nb_lasso_debias_decomp} presents its direct consequence when all pairwise sample sizes are equal ($n_{i,j}=n$), with simplified sample size assumptions that is comparable to prior literature \citep{van2014asymptotically,jankova2015confidence}. 
\vspace{-3mm}
\begin{thm}[Asymptotic Normality]\label{thm:nb_lasso_debias_decomp}
	Consider the Gaussian graphical model with erose measurement setting described in Section \ref{sec:setting} and the debiased edge-wise statistic $\widetilde{\theta}^{(a)}_b$ defined in \eqref{eq:debiased_nb_lasso}. Suppose that $\lambda^{(a)}$ in \eqref{eq:nb_lasso} is chosen as in Theorem~\ref{thm:nblasso_support}, and $\lambda^{(a,b)}$ in \eqref{eq:debias_term} satisfies
	$\lambda^{(a,b)}_k \asymp \|\Sigma^*\|_{\infty}\frac{\|\Theta^{(a)*}_{b,:}\|_1}{\Theta^{(a)*}_{b,b}}\sqrt{\frac{\log p}{\mymin_{j\in [p]}n_{j,k}}}$. Then we have the following decomposition:
	\begin{equation}\label{eq:nb_lasso_debias_decomp}
	    \widetilde{\theta}^{(a)}_b=-\frac{\Theta^{*}_{a,b}}{\Theta^*_{a,a}}+B+E.
	\end{equation}
	If Assumption~\ref{assump:inference_n_B} holds,
	then with probability at least $1-Cp^{-c}$, $|B|\leq C(\Theta^*,\gamma_a,\gamma^{(a)}_b)\frac{(d_a+d_b+1)\log p}{n_1^{(a,b)}},$
	where $C(\Theta^*,\gamma_a,\gamma^{(a)}_b)=C\kappa_{\Sigma^*}(\kappa_{\Sigma^*}^2+\sqrt{\gamma_a}+\sqrt{\gamma_b^{(a)}})\|\Sigma^*\|_{\infty}^2\|\Theta^*_{:,a}\|_1\|\Theta^{(a)*}_{:,b}\|_1$. If Assumption~\ref{assump:inference_n_E} holds, $\sigma_n^{-1}(a,b)E\overset{\text{d}}{\rightarrow}\cN(0,1)$ with $\sigma_n^2(a,b)= \frac{1}{\Theta_{aa}^{*2}}\mathcal{T}^{(n)*}\times_1\Theta^{(a)*}_{:,b}\times_2\Theta^*_{:,a}\times_3\Theta^{(a)*}_{:,b}\times_4\Theta^*_{:,a}$.
	 Furthermore, if Assumptions~\ref{assump:inference_n_B}-\ref{assump:inference_n_BE} hold,
	$$
	\sigma^{-1}_n(a,b)\left(\widetilde{\theta}^{(a)}_b+\frac{\Theta^{*}_{a,b}}{\Theta^*_{aa}}\right)\overset{\text{d}}{\rightarrow}\cN(0,1).
	$$
\end{thm}
\noindent{The proof of Theorem \ref{thm:nb_lasso_debias_decomp} is deferred to Appendix \ref{append:proofs}.}
\vspace{-2mm}
\begin{remark}[Bias-Variance decomposition]
As suggested by \eqref{eq:nb_lasso_debias_decomp}, the error of the debiased lasso estimator can be decomposed into a bias term ($B$) and a variance term ($E$), where $B$ depends on the minimum pairwise sample size $n_1^{(a,b)}$ between any nodes and the neighbors of nodes $a,b$, while $E$ depends only on the sample size $n_2^{(a,b)}$ for nodes within the neighborhoods of $a,b$ (See Figure \ref{fig:inference_n_example}). When $C(\Theta^*,\gamma_a,\gamma_b^{(a)})$ is viewed as a constant, then $|B|\asymp \frac{(d_a+d_b+1)\log p}{n_1^{(a,b)}}$, and the term $E$ scales as the asymptotic standard deviation $\sigma_n(a,b)$, which is further characterized by Proposition \ref{prop:var_bnds}.
\end{remark}
\begin{prop}[Variance characterization]\label{prop:var_bnds}
    The variance term $\sigma_n^2(a,b)$ satisfies
    \begin{align*}
       \sigma_n(a,b)\leq& \frac{\sqrt{2}\lambda_{\max}(\Sigma^*)\|\Theta^{(a)*}_{:,b}\|_2\|\Theta^{*}_{:,a}\|_2}{\Theta^*_{a,a}}(n_2^{(a,b)})^{-\frac{1}{2}}\leq \sqrt{2}\kappa_{\Sigma^*}^2(n_2^{(a,b)})^{-\frac{1}{2}},\\
    \sigma_n(a,b)\geq& \frac{\sqrt{2}\lambda_{\min}(\Sigma^*)\mymin_{(j,k)\in S_2(a, b)}\left|\Theta^{(a)*}_{b,j}\Theta^{*}_{a,k}+\Theta^{(a)*}_{b,k}\Theta^{*}_{a,j}\right|}{2\Theta^*_{a,a}}(n_2^{(a,b)})^{-\frac{1}{2}}.
    \end{align*}
\end{prop}
\noindent When $C_1\leq \frac{\lambda_{\min}(\Sigma^*)\mymin_{(j,k)\in S_2(a, b)}\left|\Theta^{(a)*}_{b,j}\Theta^{*}_{a,k}+\Theta^{(a)*}_{b,k}\Theta^{*}_{a,j}\right|}{\Theta^*_{a,a}}\leq\kappa_{\Sigma^*}^2\leq C_2$, Proposition \ref{prop:var_bnds} suggests that $\sigma_n(a,b)\asymp (n_2^{(a,b)})^{-1/2}$. 

\begin{cor}[Normal Approximation with the Same Sample Size]\label{cor:nb_lasso_debias_decomp}
    Consider the same model, edge-wise statistic and tuning parameters as in Theorem \ref{thm:nb_lasso_debias_decomp}. When the pairwise sample sizes are all equal: $n_{j,k}=n$, then if for some $\epsilon>0$,$n \gg C^2(\Theta^*;a,b)\kappa_{\Theta^*}^4(d_a+d_b+1)^2\log^2 p+C_{\epsilon}(\Sigma^*)(d_a+d_b+1)^{2+\epsilon},$
  we have
    $\sigma_n^{-1}(a,b)\left(\widetilde{\theta}^{(a)}_b+\frac{\Theta^{*}_{a,b}}{\Theta^*_{aa}}\right)\overset{\text{d}}{\rightarrow}\cN(0,1)$, where $C(\Theta^*;a,b)$ and $C_{\epsilon}(\Sigma^*)$ are as defined in Assumption \ref{assump:inference_n_E} and \ref{assump:inference_n_BE}, depending only on $\Sigma^*$ and $\Theta^*$. In addition, if the sample size of all quadrupples $n_{j,k,j',k'}=n$, $\sigma_n^2(a,b) =\frac{1}{n} \frac{\Theta^*_{a,a}\Theta^*_{b,b}-(\Theta^*_{a,b})^2}{(\Theta^*_{a,a})^2}.$ 
\end{cor}
\begin{remark}
Corollary \ref{cor:nb_lasso_debias_decomp} is a direct consequence of Theorem \ref{thm:nb_lasso_debias_decomp}. If $d_a+d_b+1\leq (\log p)^c$ for some $c>0$, the sample size condition is the same as the prior literature on debiased lasso and debiased graphical lasso \citep{van2014asymptotically,jankova2015confidence}. Note that Corollary \ref{cor:nb_lasso_debias_decomp} does not require all variables are measured simultaneously, and hence we can also apply it to the settings where only pairwise measurements or general size-constrained measurements are available \citep{dasarathy2019gaussian}. 
\end{remark}

\subsection{Variance Estimation and Edge-wise Inference}\label{sec:var_est}
In this section, we propose our GI-JOE method for edge-wise statistical inference. That is, we test the null hypothesis: $H_0: \Theta^*_{a,b}=0$ against $H_1:\Theta^*_{a,b}\neq 0$. With the aid of Theorem \ref{thm:nb_lasso_debias_decomp}, we still need to estimate the unknown variance $\sigma_n^2(a,b)$ so that we can construct a test statistic with known distribution under $H_0$. 
\begin{algorithm}[ht!]
\noindent{\textbf{Input}}: Data set $\{x_{i,V_i}:V_i\subset [p]\}_{i=1}^n$, pairwise sample sizes $\{n_{j,k}\}_{j,k=1}^p$, node pair $(a,b)$ for testing with $a\neq b$, significant level $\alpha$
\vspace{-1mm}
 \begin{enumerate}[leftmargin=*]
    \item Compute the entry-wise estimate of the covariance matrix $\widehat{\Sigma}\in \mathbb{R}^{p\times p}$: $\widehat{\Sigma}_{j,k}=\frac{1}{n_{j,k}}\sum_{j,k\in V_i}x_{i,j}x_{i,k}$
    \vspace{-2mm}
    \item Project $\widehat{\Sigma}$ onto the positive semi-definite cone: compute $\widetilde{\Sigma}$ as in \eqref{eq:proj_Sigma}.
    \vspace{-2mm}
    \item Perform neighborhood regression for node $a$: compute $\widehat{\theta}^{(a)}$ as in \eqref{eq:nb_lasso}
    \vspace{-2mm}
    \item Estimate the debiasing matrix by performing neighborhood regression for node $b$ upon nodes $[p]\backslash \{a,b\}$: compute $\widehat{\Theta}^{(a)}_{b,:}$ as in \eqref{eq:debias_term} and \eqref{eq:debias_Theta_est}.
    \vspace{-2mm}
    \item Debias the neighborhood lasso estimate: $\widetilde{\theta}^{(a)}_b=\widehat{\theta}^{(a)}_b-\widehat{\Theta}^{(a)}_{b,:}(\widehat{\Sigma}\widehat{\theta}^{(a)}-\widehat{\Sigma}_{:,a})$.
    \vspace{-2mm}
    \item Estimate the variance: $\widehat{\sigma}_n^2 = \widehat{\mathcal{T}}^{(n)}\times_1\widetilde{\Theta}^{(a)}_{b,:}\times_2\widehat{\overline{\theta}}^{(a)}\times_3\widetilde{\Theta}^{(a)}_{b,:}\times_4\widehat{\overline{\theta}}^{(a)}$, where 
    $$
    (\widehat{\mathcal{T}}^{(n)})_{j,k,j',k'}=(\widetilde{\Sigma}_{j,j'}\widetilde{\Sigma}_{k,k'}+\widetilde{\Sigma}_{j,k'}\widetilde{\Sigma}_{k,j'})\frac{n_{j,k,j',k'}}{n_{j,k}n_{j',k'}},
    $$
$\widehat{\overline{\theta}}^{(a)}$ is defined as $\widehat{\overline{\theta}}^{(a)}_a = 1$ and $\widehat{\overline{\theta}}^{(a)}_{\backslash a}=-\widehat{\theta}^{(a)}_{\backslash a}$, and $\widetilde{\Theta}^{(a)}_{b,:}$ is computed as in \eqref{eq:debias_term} and \eqref{eq:debias_Theta_est}.
\vspace{-2mm}
    \item Compute $p$-value $p_{a,b}=2(1-\Phi(\frac{\widetilde{\theta}^{(a)}_b}{\widehat{\sigma}_n(a,b)}))$ where $\Phi(\cdot)$ is the distribution function of standard Gaussian; confidence interval $\widehat{\mathbb{C}}_{\alpha}^{a,b}=[\widetilde{\theta}^{(a)}_b-z_{\alpha/2}\widehat{\sigma}_n(a,b),\widetilde{\theta}^{(a)}_b+z_{\alpha/2}\widehat{\sigma}_n(a,b)]$
\end{enumerate}
\vspace{-2mm}
 \textbf{Output}: $p$-value $p_{a,b}$, confidence interval $\widehat{\mathbb{C}}_{\alpha}^{a,b}$ for $\theta^{(a)*}_{a,b}-\frac{\Theta^*_{a,b}}{\Theta^*_{a,a}}$.
\caption{GI-JOE: edge-wise inference}
 \label{alg:GI_JOE_edgewise}
\end{algorithm}
Recall the definition of $\sigma_n^2(a,b)$ in Theorem \ref{thm:nb_lasso_debias_decomp}, and the fact that $\mathcal{T}^*_{j,k,j',k'}=\Sigma^*_{j,j'}\Sigma^*_{k,k'} + \Sigma^*_{j,k'}\Sigma^*_{k,j'}$, $(\mathcal{T}^{(n)*})_{j,k,j',k'}=\mathcal{T}^*_{j,k,j',k'}\frac{n_{j,k,j',k'}}{n_{j,k}n_{j',k'}}$, 
here we define an estimator $\widehat{\sigma}_n^2(a,b)$ as follows:
\begin{equation}\label{eq:var_est}
    \widehat{\sigma}_n^2(a,b)=\widehat{\mathcal{T}}^{(n)}\times_1\widetilde{\Theta}^{(a)}_{b,:}\times_2\widehat{\overline{\theta}}^{(a)}\times_3\widetilde{\Theta}^{(a)}_{b,:}\times_4\widehat{\overline{\theta}}^{(a)},
\end{equation}
where 
$\widehat{\mathcal{T}}^{(n)}$ is an estimator for $\mathcal{T}^{(n)*}$:
$(\widehat{\mathcal{T}}^{(n)})_{j,k,j',k'}=(\widetilde{\Sigma}_{j,j'}\widetilde{\Sigma}_{k,k'}+\widetilde{\Sigma}_{j,k'}\widetilde{\Sigma}_{k,j'})\frac{n_{j,k,j',k'}}{n_{j,k}n_{j',k'}}$;
$\widehat{\overline{\theta}}^{(a)}\in \bR^p$ serves as an estimate for $\frac{\Theta^*_{:,a}}{\Theta^*_{a,a}}$ and it satisfies $\widehat{\overline{\theta}}^{(a)}_a = 1$ and $\widehat{\overline{\theta}}^{(a)}_{\backslash a}=-\widehat{\theta}^{(a)}_{\backslash a}$; $\widetilde{\Theta}^{(a)}_{b,:}$ is defined in Section \ref{sec:debiasing} and serves an estimate for $\Theta^{(a)*}_{b,:}$.
\vspace{-2mm}
\begin{assump}[Sample size condition for variance estimation]\label{assump:var_est}
    $$n_1^{(a,b)} \gg \frac{C^4(\Theta^*;a,b)}{\kappa_{\Theta^*}^4}(d_a+d_b+1)^2\log p\left(\frac{n_2^{(a,b)}}{n_1^{(a,b)}}\right)^2,$$
    where $C(\Theta^*;a,b)$ is defined as in Assumption \ref{assump:inference_n_BE}.
\end{assump}
\vspace{-2mm}
\begin{prop}[Estimation consistency of variance]\label{prop:var_est}
    Under Assumptions~\ref{assump:inference_n_B}, \ref{assump:inference_n_BE} and \ref{assump:var_est}, if the tuning parameters are as specified in Theorem \ref{thm:nb_lasso_debias_decomp}, 
   then \eqref{eq:var_est} satisfies
    $\frac{\widehat{\sigma}^{-1}_n(a,b)}{\sigma^{-1}_n(a,b)}\overset{p}{\rightarrow} 1$.
\end{prop}
\vspace{-2mm}
\begin{thm}[Normal approximation with unknown variance]\label{thm:normal_approx_var_est}
With appropriately chosen tuning parameters as in Theorem \ref{thm:nb_lasso_debias_decomp}, if Assumptions \ref{assump:inference_n_B}-\ref{assump:var_est} hold,
    $\widehat{\sigma}^{-1}_n(a,b)(\widetilde{\theta}^{(a)}_b-\theta^{(a)*}_b)\overset{\text{d}}{\rightarrow}\cN(0,1)$ for $\widehat{\sigma}^{2}_n(a,b)$ defined in \eqref{eq:var_est} and $\widetilde{\theta}^{(a)}_b$ defined in \eqref{eq:debiased_nb_lasso}.
\end{thm}
\vspace{-2mm}
\noindent{The proof of Theorem \ref{thm:normal_approx_var_est} can be found in Appendix \ref{append:proofs}.}
Theorem \ref{thm:normal_approx_var_est} suggests us to construct the test statistic $\widehat{z}(a,b)=\widehat{\sigma}_n^{-1}(a,b)\widetilde{\theta}^{(a)}_b$, and for a desired type I error $\alpha$, wereject $H_0: \Theta^*_{a,b}=0$ if $|\widehat{z}(a,b)|\geq z_{\alpha/2}$, where $z_{\alpha/2}$ is the $1-\frac{\alpha}{2}$ quantile of standard Gaussian distribution. 
Our full GI-JOE (edge-wise inference) procedure is summarized in Algorithm \ref{alg:GI_JOE_edgewise}, with its type I error and power characterized by the following theorem.
\vspace{-2mm}
\begin{thm}[Type I error and Power Analysis]\label{thm:typeI_power}
    Consider the Gaussian graphical model with erose measurement setting described in Section \ref{sec:setting} and let $p_{a,b}$ be the $p$-value given by Algorithm \ref{alg:GI_JOE_edgewise} for node pair $(a,b)$. If all conditions in Theorem~\ref{thm:normal_approx_var_est} hold so that as $n,\,p\rightarrow \infty$, $p,\, d_a,\,d_b$, $n_1^{(a,b)}$, $n_2^{(a,b)}$ scale as in Assumption~\ref{assump:inference_n_BE} and Assumption~\ref{assump:var_est}, then the following holds:
    \vspace{-10mm}
    \begin{enumerate}[leftmargin=*]
        \item Under the null hypothesis $H_0: \Theta^*_{a,b}=0$,
        $\lim_{n,p\rightarrow\infty}\bP(p_{a,b}\leq \alpha)=\alpha$;
    \vspace{-2mm}
    \item Under the alternative hypothesis $H_1: \frac{\Theta^*_{a,b}}{\Theta^*_{a,a}}=\delta_n$, 
    \vspace{-2mm}
    \begin{enumerate}[leftmargin=*]
    \item if $\lim_{n,p\rightarrow\infty}\frac{\delta_n}{\sigma_n(a,b)}=0$, $\lim_{n,p\rightarrow\infty}\bP(p_{a,b}\leq \alpha)=\alpha$;
        \item if $\lim_{n,p\rightarrow\infty}\frac{\delta_n}{\sigma_n(a,b)}=\delta$ for $\delta\neq 0$, $\lim_{n,p\rightarrow\infty}\bP(p_{a,b}\leq \alpha)\geq \Phi(|\delta|-z_{\alpha/2})$, where $\Phi(\cdot)$ is the distribution function of standard Gaussian $\cN(0,1)$;
        \vspace{-2mm}
        \item if $\lim_{n,p\rightarrow\infty}\frac{\delta_n}{\sigma_n(a,b)}=+\infty$, $\lim_{n,p\rightarrow\infty}\bP(p_{a,b}\leq \alpha)=1$.
    \end{enumerate}
    \end{enumerate}
\end{thm}
\vspace{-2mm}
\noindent{The proof of Theorem \ref{thm:typeI_power} is deferred to Appendix \ref{append:proofs}.}
Theorem \ref{thm:typeI_power} suggests that when all conditions of Theorem \ref{thm:normal_approx_var_est} hold, the type I error of this test is asymptotically $\alpha$. Furthermore, as long as the signal strength $\frac{\Theta^*_{a,b}}{\Theta^*_{a,a}}$ shrinks no faster than $\sigma_n(a,b)\asymp (n_2^{(a,b)})^{-1/2}$, we can still reject the null with constant or high probability. Other than hypothesis testing, we can also construct asymptotically valid confidence intervals for each entry of the precision matrix ($\Theta^*_{a,b}$), under similar assumptions to Theorem \ref{thm:normal_approx_var_est}. More details can be found in Appendix \ref{append:CI}.
\section{FDR Control for Graph Inference with Erose Measurements}\label{sec:FDR}
In many application scenarios, the inference of the full graph may be of more interest than the inference of one particular edge. Confronted with a multiple testing problem, we can simply apply Holm's correction upon the $p$-values of all $\frac{p(p-1)}{2}$ node pairs $(a,b)$ for $a<b$, and hence control the family-wise error rate. However, as this approach can be too conservative, here we also propose an FDR control procedure. We leverage the ideas from \cite{javanmard2019false, liu2013gaussian} which consider the FDR control for the debiased lasso and Gaussian graphical models. In particular, for any $0\leq \rho\leq 1$, let $R(\rho)=\sum_{i<j}\ind{p_{i,j}\leq \rho}$ be the number of significant edges when $\rho$ is the threshold for $p$-values. Also define $t_p=(2\log(p(p-1)/2)-2\log\log(p(p-1)/2))^{\frac{1}{2}}$, and if there exists $2(1-\Phi(t_p))\leq \rho\leq 1$ such that $\frac{p(p-1)\rho}{2R(\rho)\vee 1}\leq \alpha$, the nominal level,
 then we would define $\rho_0 = \sup_{2(1-\Phi(t_p))\leq \rho\leq 1}\left\{\rho: \frac{p(p-1)\rho}{2(R(\rho)\vee 1)}\leq \alpha\right\}$;
otherwise, $\rho_0=2(1-\Phi(\sqrt{2\log(p(p-1)/2)}))$. The significant edge set is then defined as $\widetilde{E}=\{(j,k): p_{j,k}\leq \rho_0\}$. This is similar to the Benjamini-Hochberg procedure with only an extra truncation step. The full procedure is summarized in Algorithm 2 in Appendix \ref{sec:FDR_Alg}.

In the following, we provide theoretical guarantees for our GI-JOE (FDR) approach. Define $\epsilon_i(a,b)=\sum_{j,k}(x_{i,j}x_{i,k}-\Sigma^*_{j,k})\frac{\delta^{(i)}_{j,k}}{n_{j,k}}\frac{\Theta^{(a)*}_{j,b}\Theta^*_{k,a}+\Theta^{(a)*}_{k,b}\Theta^*_{j,a}}{2\Theta^*_{a,a}}$ as the error for estimating edge $(a,b)$, contributed by the $i$th sample; $\xi_i(a,b)$ is the normalized error: $\xi_i(a,b)=\frac{\epsilon_i(a,b)}{\sigma_n(a,b)}$. One technical quantity that is useful in our proofs is $\alpha(\Theta^*,\{V_i\}_{i=1}^n) = \sup_{(a,b),(a',b')}\frac{\|\xi_i(a,b)\|_{\psi_1}^2\vee \|\xi_i(a',b')\|_{\psi_1}^2}{\lambda_{\min}(\mathrm{Cov}((\xi_i(a,b),\xi_i(a',b')))}$. We want $\alpha(\Theta^*,\{V_i\}_{i=1}^n)$ to be not too large, similar to assuming $(\xi_i(a,b),\xi_i(a',b'))$ to be far from degenerate. Also define the covariance between the test statistics of different edges: $\sigma_n^2(a,b,a',b')=\frac{1}{\Theta^*_{a,a}\Theta^*_{a',a'}}\Gamma^{(n)*}\times_1\Theta^{(a)*}_{:,b}\times_2\Theta^{*}_{:,a}\times_3\Theta^{(a' )*}_{:,b'}\times_4\Theta^*_{:,a'}$, and the correlation $\rho_n(a,b,a',b')=\frac{\sigma_n^2(a,b,a',b')}{\sigma_n(a,b)\sigma_n(a',b')}$.
\vspace{-2mm}
\begin{assump}\label{assump:FDR_n}
    For any edge $(a,b)\in [p]\times [p]$, $n_2^{(a,b)}\geq \frac{C\|\Sigma^*\|_{\infty}^6\alpha^2(\Theta^*  ,\{V_i\}_{i=1}^n)}{\lambda_{\min}^6(\Sigma^*)}(d+1)^6(\log p)^6$, and $n_1^{(a,b)}\gg Cd^2(\log p)^5\log\log p\left(\frac{n_2^{(a,b)}}{n_1^{(a,b)}}\right)^2$, where $d=\max_{a\in[p]}d_a$. 
\end{assump}
\vspace{-4mm}
\noindent Assumption \ref{assump:FDR_n} is stronger than the sample size requirements in Assumptions \ref{assump:inference_n_BE} and \ref{assump:var_est} for edge-wise inference, since the proof for asymptotically valid FDR control needs stronger normal approximation guarantees, especially at the tail.
\begin{remark}
When all variables are simultaneously measured with sample size $n$, Assumption \ref{assump:FDR_n} reduces to $n\geq C(d+1)^6(\log p)^6$. This is much weaker than the condition in the literature of the graphical model FDR control  \citep{liu2013gaussian}: $p\leq n^r$ for some constant $r>0$. They used this assumption to show the tail probability of the test statistics can be well approximated by the Gaussian tail, while we use a different proof that exploits the sub-exponential properties of $\widehat{\Sigma}_{j,k}$ (second moments of Gaussian variables are sub-exponential). 
\end{remark}
\vspace{-4mm}
\begin{assump}\label{assump:FDR_correlation}
    For any $0<\rho<1$, $\gamma>0$, let $\mathcal{A}_1(\rho)=\{(a,b,a',b')\in [p]\times [p]: \Theta^*_{a,b}=\Theta^*_{a',b'}=0, |\rho_n(a, b, a', b')|>\rho\}$, and $\mathcal{A}_2(\rho,\gamma)=\{(a, b, a', b')\in [p]\times [p]:  \Theta^*_{a,b}=\Theta^*_{a',b'}=0, (\log p)^{-2-\gamma}<|\rho_n(a,b,a',b')|\leq\rho\}$. There exist $0<\rho_0<1$ and $\gamma>0$ such that, 
    $|\mathcal{A}_1(\rho_0)|\leq Cp^2,\,
    |\mathcal{A}_2(\rho_0,\gamma)|\ll p^{\frac{4}{1+\rho_0}}(\log p)^{\frac{2\rho_0}{1+\rho_0}-\frac{1}{2}}(\log\log p)^{-\frac{1}{2}}$.
\end{assump}
\vspace{-3mm}
Assumption \ref{assump:FDR_correlation} enforces that most edge-wise test statistics are only weakly correlated, and similar assumptions have also appeared in \cite{liu2013gaussian} and \cite{javanmard2019false}. To better understand this, recall the definition of the edge-pair covariance $\sigma_n^2(a, b, a', b')$ before Assumption \ref{assump:FDR_n} and $\mathcal{T}^{(n)}$ in Section \ref{sec:normal_approx}, and note that the covariance $\sigma_n^2(a,b,a',b')$ between edge pair $(a,b)$ and $(a',b')$ sums over the covariance $\Sigma^*_{j,j'}\Sigma^*_{k,k'}+\Sigma^*_{j,k'}\Sigma^*_{k,j'}$ across $j\in \overline{\mathcal{N}}_b$, $k\in \mathcal{N}_a$, $j'\in \mathcal{N}_b'$, $k'\in \mathcal{N}_a'$. In other words, Assumption \ref{assump:FDR_correlation} requires that for most edge pairs, the correlations passed from each node's neighbors are weak. To make this intuition more concrete, consider the simultaneous measurement and same sample size setting as a special example, then this assumption would hold if each node has a constant number of strongly connected neighbors (see Proposition \ref{prop:assump7_simultaneous} in Appendix \ref{sec:FDR_Alg} for a more rigorous statement).
The following theorem suggests that our GI-JOE (FDR) procedure can successfully control the false discovery proportion both in expectation and in probability when Assumptions \ref{assump:FDR_n} and \ref{assump:FDR_correlation} hold.
\vspace{-3mm}
\begin{thm}[Validity of FDR control]\label{thm:FDR_valid}
Consider the GI-JOE (FDR) procedure described in Section \ref{sec:FDR} and suppose the significant edge set under a given nominal level $\alpha$ is given by $\widetilde{E}$. Let $\mathrm{FDP} = \frac{\sum_{(a,b)\in \mathcal{H}_0}\ind{(a,b)\in \widetilde{E}}}{|\widetilde{E}|\vee 1}$, where $\mathcal{H}_0=\{(i,j)\in [p]\times [p]:\Theta^*_{i,j}=0\}$; Also let $\mathrm{FDR} = \mathbb{E}\mathrm{FDP}$. If Assumptions \ref{assump:FDR_n} and \ref{assump:FDR_correlation} hold, then we have $\limsup_{n,p\rightarrow \infty}\mathrm{FDR}\leq \alpha$,
and for any $\epsilon>0$,
$\lim_{n,p\rightarrow \infty}\mathbb{P}(\mathrm{FDP}>\alpha + \epsilon)=0$.
\end{thm}
\vspace{-3mm}
\noindent{The proof of Theorem \ref{thm:FDR_valid} can be found in Appendix \ref{append:proofs}.}
This is the first theoretical guarantee for FDR control with erosely measured data. Although we still require sufficient pairwise sample sizes for all pairs of nodes, we allow $n_{j,k}$ to be of different order. As shown in Appendix \ref{append:empiricalResults}, our GI-JOE FDR approach indeed controls the FPR empirically in a wide range of erose measurement settings and graph structures, even when the sample size condition might be violated. 
\section{Empirical Studies}\label{sec:numeric}
In this section, we present empirical studies to validate our GI-JOE approach for both edge-wise inference and full graph inference. We first verify the validity of our edge-wise inference procedure in Section \ref{sec:sim_edge}; then we compare the graph selection performance using both our GI-JOE approaches and various baseline estimation and inference methods in Section \ref{sec:sim_graph}. A real data example is included in Section \ref{sec:real_data}. Software and code for reproducing the empirical results can be found at \href{https://github.com/Lili-Zheng-stat/GI-JOE}{https://github.com/Lili-Zheng-stat/GI-JOE}.

\subsection{Simulations for Edge-wise Inference: Validating Theory}\label{sec:sim_edge}
Here, we investigate the type I error and power of GI-JOE for testing one node pair. 
To study the effect of different pairwise sample sizes, here we consider the pairwise measurement scenario where each sample only consists of two variables. 
When the given node pair for inference is $(a,b)$, we set the pairwise sample size as follows:
$n_{j,k}=\sum_{i=1}^n \ind{j,k\in V_i}=n_1$ if $(j,k)\in S_1(a,b)\backslash S_2(a,b)$, $n_{j,k}=n_2$ if $(j,k)\in S_2(a,b)$, and $n_{j,k}=50$ otherwise.
The precision matrix $\Theta^*=\Sigma^{*-1}$ is generated with three graph structures: a chain graph, a three-star graph, and an Erd\H{o}s–R\'{e}nyi graph with connection probability $\frac{3}{p-1}$. Then we study the type I error for testing an unconnected node pair, and the power for testing an edge with different signal strengths. More details on the set-ups and the implementation of our GI-JOE approach can be found in Appendix \ref{append:NumericDetails}.

Figure \ref{fig:type1_err_power} summarizes the type I error rate and power averaged over 200 replicates when the confidence level is set as $0.95$, under three graph structures. In the Type I error plot (a), we consider three network sizes $p=50,\, 100,\, 200$, and a range of $n_1^{(a,b)},\, n_2^{(a,b)}$. We can see that the type I error rates are close to $0.05$ with moderately large, although differing, pairwise sample sizes. In the power plot, The dimension $p$ is fixed as $200$, sample sizes $n_2^{(a,b)}=n2\in\{125, 250, 500\}$, $n_1^{(a,b)}=n_1\in n_2/\{1, 1.2, 1.5\}$. The dependence of power on the SNR $\frac{\Theta^*_{a,b}}{\Theta^*_{a,a}\sigma_n(a,b)}$ is similar across different sample sizes, supporting Theorem \ref{thm:typeI_power}.  
\begin{figure}[!ht]
    \centering
    \subfigure[Type I Error]{
    \includegraphics[height = 5.4cm]{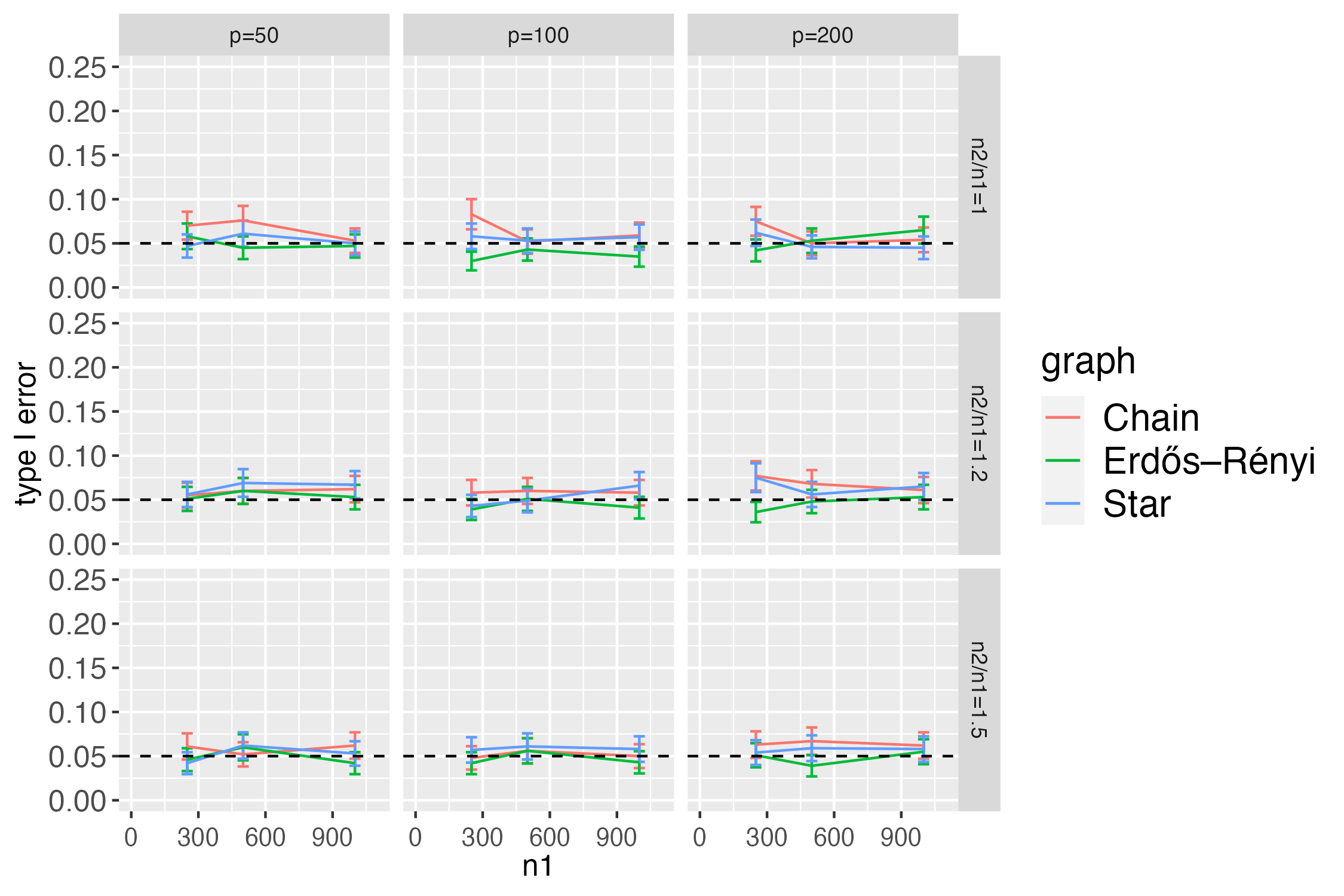}}
    \subfigure[Power]{
    \includegraphics[height = 5.9cm]{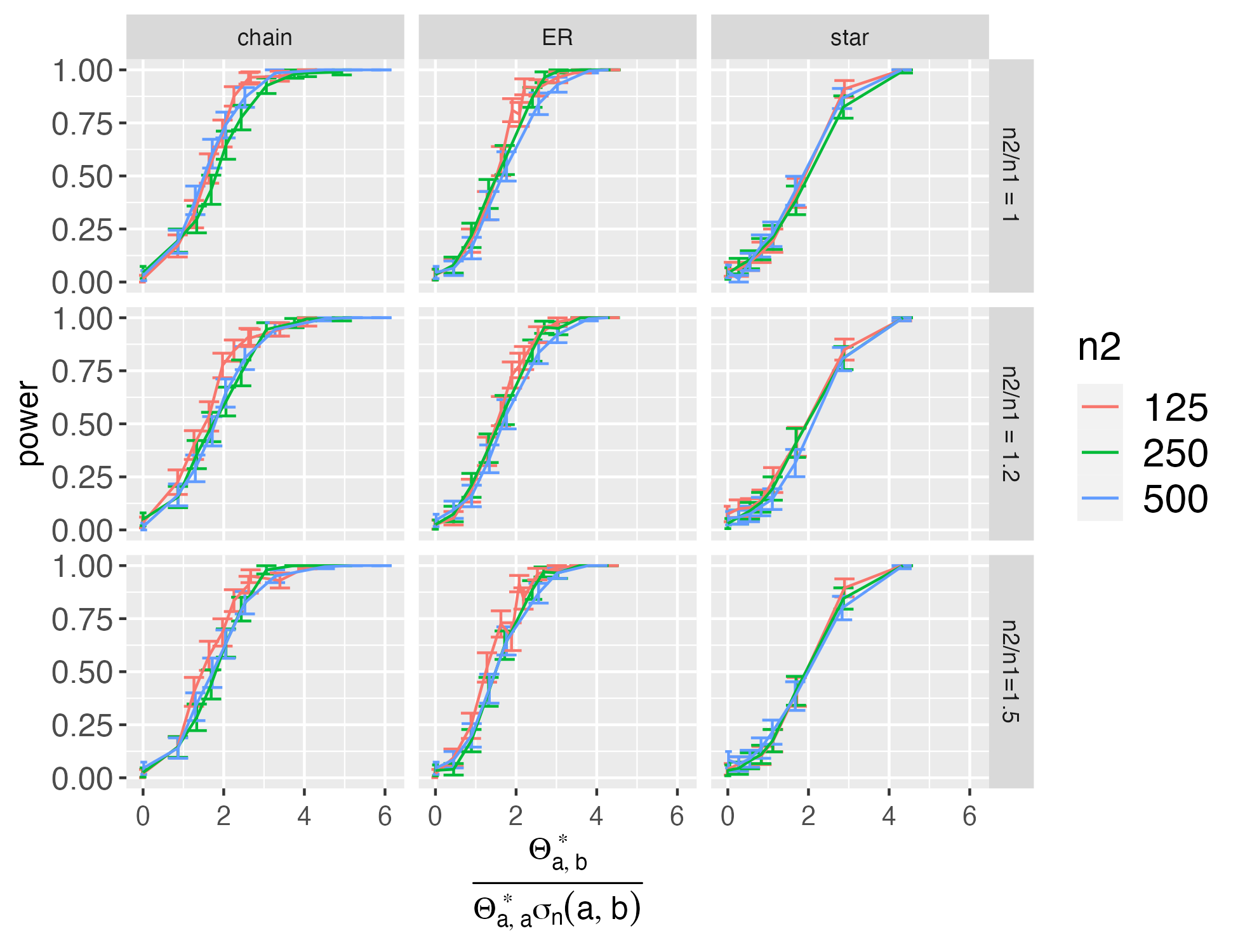}}
    \vspace{-2mm}
    \caption{\small Type I error rates vs. sample size, averaged over 1000 replicates; Power vs. signal to noise ratio, averaged over 200 replicates. Target type I error rate is set as $\alpha=0.05$, and the error bars represent the 95\% confidence interval. The $x$-axis in (b) is the signal-to-noise ratio $\frac{\Theta^*_{a,b}}{\Theta^*_{a,a}\sigma_n(a,b)}$, which determines the asymptotic power of our test (see Theorem \ref{thm:typeI_power}).}
    \label{fig:type1_err_power}
\end{figure}

\subsection{Graph Selection Study and Comparison with Baselines}\label{sec:sim_graph}
Now we study the graph selection performance of our GI-JOE (Holm) and GI-JOE (FDR) approaches under different erose measurements, compared with some baselines.
In particular, we consider four estimation methods and four inference methods. The estimation methods include standard plug-in methods \citep{kolar2012estimating,park2021estimating} with neighborhood lasso (Nlasso), graphical lasso (Glasso), CLIME, and the variant plug-in method described in Section \ref{sec:estimation_selection} (Nlasso-JOE). The only difference between Nlasso-JOE and Nlasso is the former uses a different tuning parameter for each node that depends on its own sample sizes, as explained in Section \ref{sec:estimation_selection}. The inference methods include GI-JOE (Holm), GI-JOE (FDR), and also ad hoc implementations of the debiased graphical lasso \citep{jankova2015confidence}. Specifically, since there are no applicable inference methods designed for erose measurements, the only baseline we can implement is applying existing inference methods for simultaneous measurement settings and plug in the minimum pairwise sample size for computing the variance of each edge. This is not a method one would ever use in practice, since it is too conservative and has no theoretical guarantees, but we still present the results of such baselines, just to prove the concept that considering the different sample sizes over the graph is important.
Although \cite{jankova2015confidence} is only concerned with edge-wise inference, we still add a Holm's correction and FDR control procedure on top of its edge-wise $p$-values for a fair comparison, and we denote these two procedures by DB-Glasso (Holm) and DB-Glasso (FDR). Some additional implementation details of these methods can be found in Appendix \ref{append:NumericDetails}.

The comparative studies here include both synthetic and real data-inspired simulations, and we first present the synthetic set-up. For graph structures, we consider the chain graph, 10-star graph, and Erd\H{o}s–R\'{e}nyi graph, with dimension $p = 200$. We experiment with three synthetic erose measurement patterns, but due to space limit, here we only present the results for two measurement patterns and defer more complete results to Appendix \ref{append:empiricalResults}. In measurement 1, each node is independently missing with low, moderate, or high probabilities; measurement 2 is the size-constrained measurements scenario, where each sample consists of randomly selected $20$ nodes, and each node is sampled with probability weight positively depending on its degree. 
For each measurement scenario, we consider two different total sample sizes. Table \ref{tab:sim2_F1_summary_chain_star} summarizes the F1 scores of our GI-JOE method and some other baseline estimation and inference methods, averaged over 20 independent runs (standard deviations included in parentheses). For both Nlasso and Nlasso-JOE, we present the results based on the AND rule; The results for OR rule and more detailed results on true positive rate, true negative rate, and true discovery rate can be found in Appendix \ref{append:empiricalResults}. In summary, among all different measurement scenarios, sample sizes, and graph structures, GI-JOE (FDR) is either the best or comparable to the best among all inference and estimation methods, in terms of the F1 score. The Nlasso-JOE is usually the best among the estimation methods, suggesting that using different tuning parameters that accommodate for the different pairwise sample sizes may be a simple yet effective trick.

\begin{table}[!ht]
     \centering
     \scalebox{0.7}{
     \begin{tabular}{|c|c|c|c|c|c|c|c|c|}
     \hline
     &\multicolumn{4}{|c|}{Chain graph}&\multicolumn{4}{|c|}{10-Star graph}\\
     \hline
        \multirow{2}{*}{Method}  
        &\multicolumn{2}{c|}{Measurement 1} & \multicolumn{2}{c|}{Measurement 2}&\multicolumn{2}{c|}{Measurement 1} & \multicolumn{2}{c|}{Measurement 2} \\
        \cline{2-9}
        &n=600&n=800&n=20000&n=30000&n=600&n=800&n=20000&n=30000\\
        \hline
       Nlasso &0.62(0.18)&0.70(0.02)&0.41(0.01)&0.34(0.01)&0.35(0.11)& 0.40(0.10)&0.28(0.01)&0.29(0.01)\\
       \hline
       Glasso &0.37(0.01)&0.35(0.01)&0.31(0.01)&0.47(0.21)&0.49(0.02)&0.47(0.01)&0.35(0.01)&0.39(0.25)\\
       \hline
       CLIME &0.40(0.04)&0.32(0.01)&0.27(0.00)&0.44(0.01)&0.25(0.01)&0.39(0.03)&0.28(0.03)&0.43(0.02)\\
       \hline
       Nlasso-JOE &\textbf{0.64}(0.12)&\textbf{0.91}(0.05)&\textbf{0.68}(0.20)&\textbf{0.88}(0.01)&\textbf{0.82}(0.07)&\textbf{0.86}(0.02)&\textbf{0.80}(0.02)&\textbf{0.90}(0.01)\\
        \hline
       \hline
       DB-Glasso (Holm) &0.01(0.01)& 0.05(0.02)&0.07(0.02)&0.42(0.03)&0.00(0.01)&   0.00(0.01)&0.00(0.01)&0.04(0.01)\\
       \hline
       DB-Glasso (FDR)&0.00(0.01)&0.01(0.01)&0.01(0.01) &0.06(0.02)&0.00(0.01)&0.01(0.01)&0.01(0.01) &0.06(0.02)\\
       \hline
       \textbf{GI-JOE (Holm)}&0.78(0.01)&0.81(0.01)& 0.37(0.12)& 0.67(0.02)& 0.32(0.04)& 0.50(0.01)& 0.96(0.01)&0.97(0.01)\\
        \hline
        \textbf{GI-JOE (FDR)}&\textbf{0.81}(0.01)&\textbf{0.86}(0.01)& \textbf{0.74}(0.08)&\textbf{0.92}(0.01)&\textbf{0.51}(0.03)&\textbf{0.61}(0.01)&\textbf{0.99}(0.01)&\textbf{0.98}(0.01)\\
        \hline
     \end{tabular}}
     \caption{\small F1 scores of the graphs selected by estimation methods (first three) and inference methods (last four) under two measurement scenarios. The average sample size $\frac{1}{p^2}\sum_{j,k}n_{j,k}$ ranges from $50$ to $350$.}
     \label{tab:sim2_F1_summary_chain_star}
 \end{table}

While for the real data-inspired simulations, either the graph structure is adopted from real neuroscience data or the measurement patterns are adopted from real gene expression data sets. Due to space limit, here we only present the latter (real erose measurements) but leave the real graph simulation results Appendix \ref{append:empiricalResults}.

For the real measurement patterns, we take two publicly available single-cell RNA sequencing data sets \citep{darmanis2015survey,chu2016single}, and focus on the top $200$ genes with highest variances. The \emph{chu} and \emph{darmanis} measurement patterns have pairwise sample sizes ranging from 0 to 1018 and from 5 to 366. The graphs for data generation are scale-free graphs and small-world graphs with $200$ nodes. The F1-scores summarized in Table \ref{tab:scRNA_sim_F1_summary} also suggest the efficacy of the GI-JOE (FDR) approach. Visualizations of the real graph and measurement patterns, and specifics of the simulated graphs can be found in Appendix \ref{append:RealSim_Setting}.
\begin{table}[!ht]
     \centering
     \scalebox{0.7}{
     \begin{tabular}{|c|c|c|c|c|}\hline
        Method  & \multicolumn{2}{c|}{\emph{chu} measurement} & \multicolumn{2}{c|}{\emph{darmanis} measurement} \\
        \cline{2-5}
        &scale-free graph& small-world graph & scale-free graph & small-world graph\\
        \hline
        Nlasso&0.54(0.25)&0.57(0.28)&0.31(0.02)&0.34(0.11)\\
       \hline
       Glasso&\textbf{0.61}(0.19)&0.70(0.17)&\textbf{0.59}(0.03)&0.49(0.10)\\
       \hline
       CLIME&0.41(0.07)&0.40(0.02)&0.26(0.04)&0.35(0.12)\\
       \hline
       Nlasso-JOE&\textbf{0.61}(0.30)&\textbf{0.92}(0.01)&0.51(0.03)&\textbf{0.81}(0.01)\\
        \hline
       \hline
       DB-Glasso (Holm) &0.00(0.00)&0.00(0.00)&0.00(0.00) &0.00(0.00)\\
       \hline
       DB-Glasso (FDR)&0.00(0.00)& 0.00(0.00)& 0.00(0.00) &0.00(0.00)\\
       \hline
       \textbf{GI-JOE (Holm)}&\textbf{0.96}(0.01)&\textbf{0.93}(0.01)&0.44(0.04)&0.55(0.03)\\
        \hline
        \textbf{GI-JOE (FDR)} &0.94(0.03)&0.93(0.01)&\textbf{0.75}(0.03)&\textbf{0.76}(0.02)\\
        \hline
     \end{tabular}}
     \caption{\small F1 scores of estimation methods (first three) and inference methods (last four) with the ground truth graphs being a scale-free graph and a small-world graph with 200 nodes, under two real measurement patterns from single-cell RNA sequencing data sets (the chu data \citep{chu2016single} and darmanis data \citep{darmanis2015survey}). Our GI-JOE methods always have the highest F1-scores, and GI-JOE (FDR) is better for the the darmanis measurement (average sample size $250$) while GI-JOE (Holm) is better for the chu measurement (average sample size $850$).}
     \label{tab:scRNA_sim_F1_summary}
 \end{table}

 \subsection{Real Data Example: Application to Calcium Imaging Data}\label{sec:real_data}
 The two-photon calcium imaging technology can record in vivo functional activities of thousands of neurons \citep{stringer2019spontaneous}, and such data sets can be used to understand the neuronal circuits with the help of graphical model techniques \citep{wang2021thresholded,vinci2019graph}. In this section, we investigate the potential of our GI-JOE approach on a real calcium imaging data set from the Allen Institute \citep{lein2007genome}, which contains the functional recordings of $p=227$ neurons in a mouse's visual cortex, when different visual stimuli or no stimulus were presented to the mouse. Here, we focus on the raw fluorescence traces in one spontaneous session where there is no external stimulus. This session includes $n=8931$ samples of the trace data associated with all $227$ neurons. We manually mask the data for some neurons to create erose measurements, and then validate our methods by comparing the tested graph based on masked data with the tested graph based on the full data set. In particular, the recorded neurons all lie on the same vertical plane in a mouse's visual cortex (see Figure \ref{fig:ABA_graphs} for the physical locations of the neurons in $x$ and $y$ axis). To manually create erose measurements, we divide the neurons into three subsets based on their location on the $x$-axis (marked in different colors in Figure \ref{fig:ABA_graphs}), and neurons in each subset are randomly observed with high ($\sqrt{0.9}$), moderate ($\sqrt{0.5}$), and low ($\sqrt{0.1}$) probabilities. 
 
 However, when inference methods are applied on the full data set, we find that the tested graph is always dense. This might be due to the huge amount of latent neurons in the mouse's brain since latent variables are known to lead to dense graph structures in graphical models \citep{wang2021thresholded}. Since most of these edges have small edge weights, here we consider testing if the edge weights are stronger than a threshold instead of testing if they are zero. Specifically, for any node pair $(a,b)$, $H_{0,(a,b)}: |\frac{\Theta^*_{a,b}}{\Theta^*_{a,a}}|\leq 0.12$. Our GI-JOE approach can be directly extended to test such hypothesis and the detailed procedures are included in Appendix \ref{append:NumericDetails}.
 For validation purposes, we also apply a special version of our GI-JOE (FDR) approach on the full data set, where all pairwise sample sizes are equal. As suggested by Figure \ref{fig:ABA_graphs}, our GI-JOE (FDR) approach works well, especially for the neuron set 1 (red) as they have larger sample sizes; it identifies the same hub neuron in set 1 as the tested graph (a) with the full data set. The specific F1-scores of each method for each neuron set can be found in Appendix \ref{append:empiricalResults}.

 \begin{figure}
     \centering
     \subfigure[FDR-selected graph with full data]{
    \includegraphics[height = 4cm]{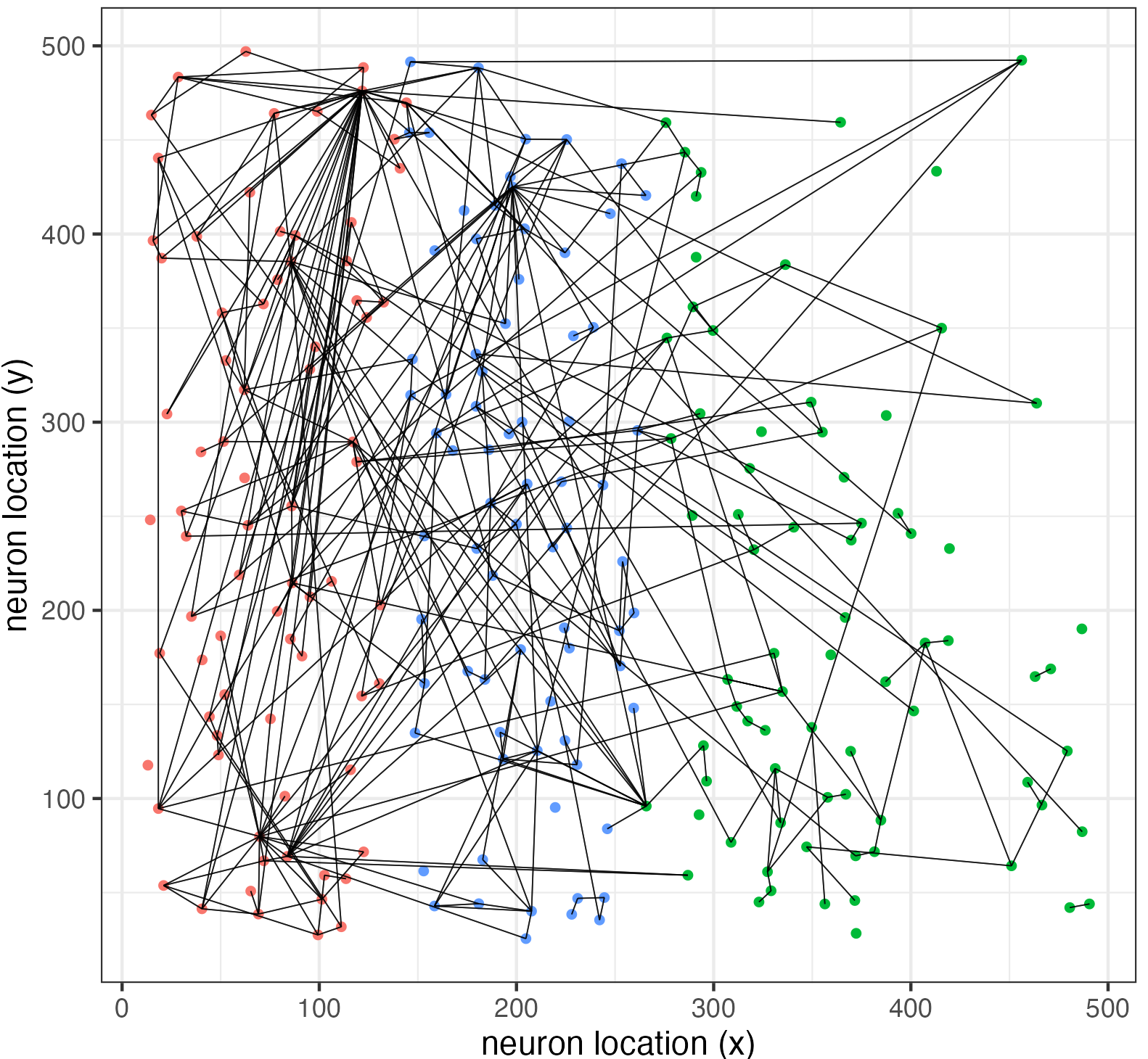}}
    \subfigure[GI-JOE (FDR), applied to data with erose measurements]{
    \includegraphics[height = 4cm]{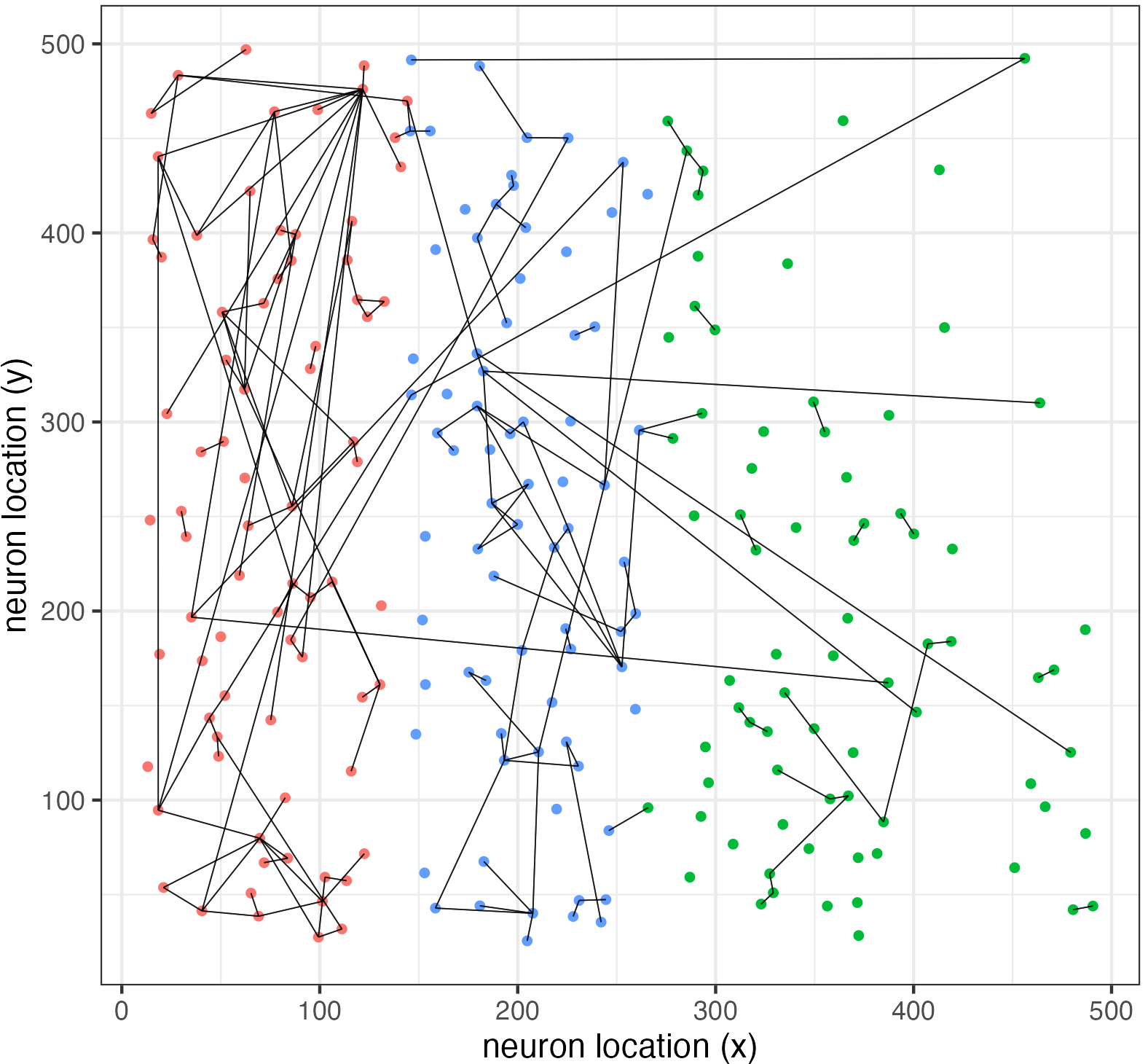}}
    \subfigure[Debiased graphical lasso with minimum sample size, applied to data with erose measurements]{
    \includegraphics[height = 4cm]{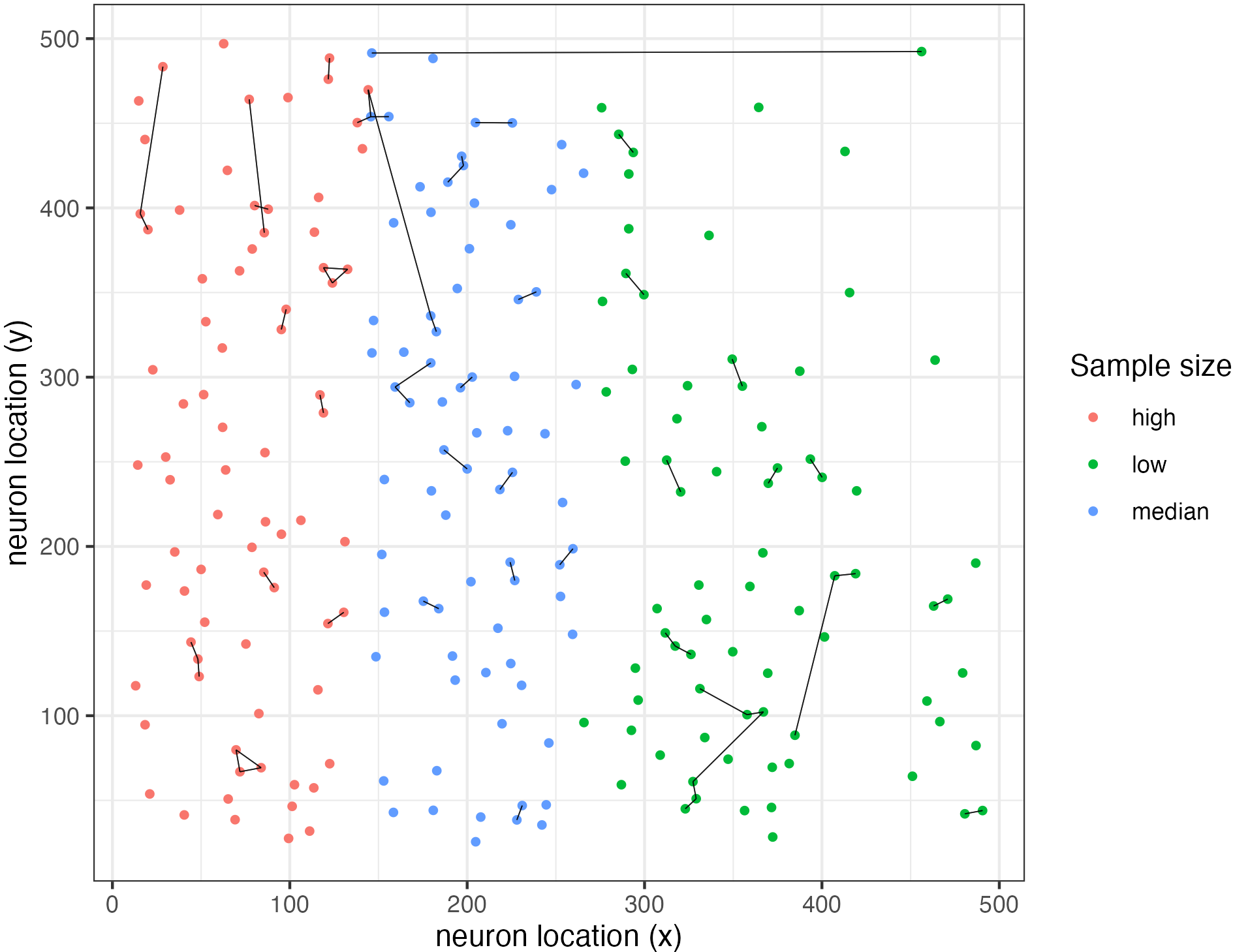}}
     \caption{\small Tested functional connectivity graphs among neurons, using the Allen Brain Atlas data set. The neurons marked in red, blue, and green have high, moderate and low sample sizes. Our GI-JOE (FDR) approach works reasonably well, especially for red neurons, while the debiased graphical lasso with minimum sample size is too conservative to find any edge.}
     \label{fig:ABA_graphs}
 \end{figure}

\section{Discussion}\label{sec:discussion}
In this paper, we propose the \emph{GI-JOE} (\textbf{G}raph \textbf{I}nference when \textbf{J}oint \textbf{O}bservations are \textbf{E}rose) approach to address the graphical model inference problem under the \emph{erose} measurement setting, where irregular observational patterns lead to vastly different sample sizes for node pairs. Our GI-JOE approach quantifies the different uncertainty levels over the graph induced by the erose measurements, including both an edge-wise inference method and an FDR control procedure. We characterize the Type I error and power for testing any node pair $(a,b)$ by considering the sample sizes involving $a,b$'s neighbors; We also guarantee the valid FDR control of GI-JOE (FDR) under appropriate conditions. Finally, our experiments with synthetic and real data demonstrate the efficacy of our approach for different graphs and measurement patterns.

Graph learning with erose measurements is a common but understudied problem, and there are still many open questions for future investigation. For instance, it is possible to extend our theory and methods to general sub-Gaussian data or semiparametric graphical models. Our problem setting is closely related to the latent variable graphical models when some pairwise sample sizes are extremely low, and some ideas in this line of work might be useful to further improve our method with relaxed sample size conditions. In addition, our current results are based on a variant of the neighborhood lasso, while one may also consider a variant of the graphical lasso or CLIME and investigate their potential in the erose measurement setting. Another practical but challenging setting not considered in this paper is data-dependent erose measurements, which requires novel methods and theory since the sample covariance matrix would be biased in this setting and the plug-in type approach no longer works. Furthermore, erosely measured data sometimes take the form of a time series, and the temporal dependence could have an effect on the edge-wise variance computation, calling for new inference methodologies.
\section*{Acknowledgments}
The authors gratefully acknowledge support by NSF NeuroNex-1707400, NIH 1R01GM140468, and NSF DMS-2210837. The authors also thank Gautam Dasarathy for helpful discussion during the preparation of this manuscript.
\clearpage
%\appendix
\begin{appendices}
\section{Optimization Algorithm for the Weighted \texorpdfstring{$\ell_{\infty}$}{Lg} Norm Projection}
\label{sec:proj_alg}
We provide an ADMM algorithm to solve the following problem:
\begin{equation}
    \widetilde{\Sigma}=\argmin_{\Sigma\succeq 0}\max_{i,j}\sqrt{n_{ij}}|\Sigma_{ij}-\widehat{\Sigma}_{ij}|.
\end{equation} 
Similar to~\cite{datta2017cocolasso}, we introduce an additional variable $B\in \mathbb{R}^{p\times p}$, choose a small $\varepsilon>0$ and solve \begin{equation*}
    (\widetilde{\Sigma},\widehat{B})=\argmin_{\Sigma\succeq \varepsilon I, B=\Sigma-\widehat{\Sigma}}\max_{i,j}\sqrt{n_{ij}}|B_{ij}|.
\end{equation*}
Consider the augmented Lagrangian function of the objective above:
\begin{equation*}
\begin{split}
    &f(\Sigma,B,\Lambda)\\
    = &\frac{1}{2}\max_{i,j}\sqrt{n_{ij}}|B_{ij}|-\langle \Lambda,\Sigma-B-\widehat{\Sigma}\rangle+\frac{1}{2\mu}\|\Sigma-B-\widehat{\Sigma}\|_{F}^2,
\end{split}
\end{equation*}
where $\Lambda\in \mathbb{R}^{p\times p}$ is the Lagrangian variable, and $\mu$ is a penalty parameter. Then the ADMM algorithm~\citep{boyd2011distributed} takes the following steps at the $i$th iteration:
\begin{align}
        \Sigma^{(i+1)}=&\argmin_{\Sigma\succeq \varepsilon I}f(\Sigma,B^{(i)},\Lambda^{(i)}),\label{eq:ADMM_Sigma_update}\\
        B^{(i+1)}=&\argmin_{B}f(\Sigma^{(i+1)},B,\Lambda^{(i)}),\label{eq:ADMM_B_update}\\
        \Lambda^{(i+1)}=&\Lambda^{(i)}-\frac{\Sigma^{(i+1)}-B^{(i+1)}-\widehat{\Sigma}}{\mu}.\label{eq:ADMM_Lambda_update}
\end{align}
Now we discuss how to solve \eqref{eq:ADMM_Sigma_update} and \eqref{eq:ADMM_B_update}. One can show that
\begin{equation*}
    \begin{split}
        \Sigma^{(i+1)}=&\argmin_{\Sigma\succeq \varepsilon I}f(\Sigma,B^{(i)},\Lambda^{(i)})\\
        =&\argmin_{\Sigma\succeq \varepsilon I}\|\Sigma-B^{(i)}-\mu\Lambda^{(i)}-\widehat{\Sigma}\|_F^2\\
        =&\mathcal{P}_{\mathcal{S}^{++}_{\varepsilon}}(B^{(i)}+\mu\Lambda^{(i)}+\widehat{\Sigma}),
    \end{split}
\end{equation*}
where $\mathcal{P}_{\mathcal{S}^{++}_{\varepsilon}}(\cdot)$ is the projection operator upon $\mathcal{S}^{++}_{\varepsilon}=\{A\in \mathbb{R}^{p\times p}:A=A^\top, \lambda_{\min}(A)\geq \varepsilon\}$, w.r.t. the Frobenius norm. We can perform an eigenvalue thresholding to find this projection.

While for $B^{(i+1)}$, the following lemma shows how to solve \eqref{eq:ADMM_B_update}:
\begin{lemma}\label{lem:ADMM_B_update}
    Let $\omega\in \mathbb{R}^{p\times p}$ with $\omega_{ij}=\sqrt{n_{ij}}$. The solution for ~\eqref{eq:ADMM_B_update} is:
    $$B^{(i+1)}=\Sigma^{(i+1)}-\mu\Lambda^{(i)}-\widehat{\Sigma}-\mathcal{P}_{\mathbb{B}_{\omega^{-1},1}(\frac{\mu}{2})}(\Sigma^{(i+1)}-\mu\Lambda^{(i)}-\widehat{\Sigma}),$$
    where $\mathcal{P}_{\mathbb{B}_{\omega^{-1},1}(\frac{\mu}{2})}(\cdot)$ is the projection operator upon the weighted $\ell_1$ ball $$\mathbb{B}_{\omega^{-1},1}(\frac{\mu}{2})=\{A\in \mathbb{R}^{p\times p}: \sum_{i,j}\omega^{-1}_{ij}|A_{ij}|\leq \frac{\mu}{2}\},$$ w.r.t. the Frobenius norm.
\end{lemma}

\begin{proof}[Proof of Lemma~\ref{lem:ADMM_B_update}]
    Let $A=\Sigma^{(i+1)}-\mu\Lambda^{(i)}-\widehat{\Sigma}$. First we can write
    \begin{equation*}
        B^{(i+1)}=\argmin_{B}\|B-A\|_F^2+\mu\|\omega\circ B\|_{\infty}.
    \end{equation*}
    Due to the convexity of $g(B)=\|B-A\|_F^2+\mu\|\omega\circ B\|_{\infty}$, it suffices to show that $0_{p\times p}$ is one sub-gradient of $g(B)$ at $\widehat{B}:=A-\mathcal{P}_{\mathbb{B}_{\omega^{-1},1}(\frac{\mu}{2})}(A)$. 
    
    Let $\delta=\mathcal{P}_{\mathbb{B}_{\omega^{-1},1}(\mu/2)}(A)$, and $z=\frac{2}{\mu}\delta$. It is straightforward to verify that
    $(\nabla_B\|B-A\|_F^2)_{B=\widehat{B}}+\mu z=0$, thus we only need to show that $z$ is a sub-gradient of $\|\omega\circ B\|_{\infty}$ at $B=\widehat{B}$, that is, (i) $\|\omega^{-1}\circ z\|_1\leq 1$ and (ii) $\langle \widehat{B},z\rangle =\|\omega\circ \widehat{B}\|_{\infty}$. 
    
    First note that $\|\omega^{-1}\circ \delta\|_1<\frac{\mu}{2}$ if and only if $\|\omega^{-1}\circ A\|_1<\frac{\mu}{2}$, $\delta=A$ and $\widehat{B}=0$, which implies $\|\omega^{-1}\circ z\|_1< 1$ and $\langle\widehat{B},z\rangle=\|\omega\circ \widehat{B}\|_{\infty}=0$. Now we focus on the case where $\|\omega^{-1}\circ \delta\|_1=\frac{\mu}{2}$. In this case, we can still easily verify (i) since $\|\omega^{-1}\circ z\|_1=\frac{2}{\mu}\|\delta\|_{\omega^{-1},1}= 1$. While for (ii), note that by established results for projection on weighted $\ell_1$ ball~\citep{slavakis2010adaptive}, there exists $c>0$ such that    $\delta_{ij}=\max\{|A_{ij}|-c\omega_{ij}^{-1},0\}\mathrm{sgn}(A_{ij}).$ Then one can show that $\|\omega\circ \widehat{B}\|_{\infty}=\max_{i,j}\omega_{ij}|A_{ij}-\delta_{ij}|= c$. Furthermore, 
    $z_{ij}>0$ implies $A_{ij}>c\omega_{ij}^{-1}$, and $z_{ij}<0$ implies $A_{ij}<-c\omega_{ij}^{-1}$. Hence we have
    \begin{equation*}
        \begin{split}
            \langle \widehat{B},z\rangle=&\sum_{i,j}z_{ij}c\omega_{ij}^{-1}\mathrm{sgn}(A_{ij})\\
            =&\sum_{i,j}c\omega_{ij}^{-1}|z_{ij}|\\
            = &c\\
            = &\|\omega\circ \widehat{B}\|_{\infty}.
        \end{split}
    \end{equation*}
\end{proof}
Hence solving for $B^{(i+1)}$ only requires the projection upon a weighted $\ell_1$ ball, and this projection can be done by applying the algorithm proposed in~\cite{slavakis2010adaptive}. We summarize the full ADMM algorithm for solving \eqref{eq:proj_Sigma} in Algorithm~\ref{alg:nb_lasso}.
\begin{algorithm}[!h]
\SetAlgoLined
 Input: $\widehat{\Sigma}$, $n_{ij}, 1\leq i,j\leq p$, $\varepsilon>0$, penalty parameter $\mu>0$, initial values $B_0$, $\Lambda_0$\\
 $\omega_{ij}=\sqrt{n_{ij}}$\;
 \For{$i=0,1,\dots,K$}{
 $\Sigma^{(i+1)}=\mathcal{P}_{\mathcal{S}^{++}_{\varepsilon}}(B^{(i)}+\widehat{\Sigma}+\mu\Lambda^{(i)})$\;
 $A=\Sigma^{(i+1)}-\mu\Lambda^{(i)}-\widehat{\Sigma}$\;
$B^{(i+1)}=A-\mathcal{P}_{\mathbb{B}_{\omega,1}(\mu/2)}(A)$\;
 $\Lambda^{(i+1)}=\Lambda^{(i)}-\frac{\Sigma^{(i+1)}-B^{(i+1)}-\widehat{\Sigma}}{\mu}$\;
 }
 Output: $\Sigma^{(K+1)}$.
 \caption{ADMM for weighted $\ell_{\infty}$ norm projection on positive semi-definite cone}
 \label{alg:nb_lasso}
\end{algorithm}\vspace{-2mm}

\section{Additional Theoretical Results and Discussion for Neighborhood Regression}\label{append:nbconsistency_theory}
Here, we first restate our neighborhood recovery guarantee (Theorem 1 in the main paper), and discuss the consequences of this result. 
The neighborhood lasso estimator defined in \eqref{eq:nb_lasso} for each node $a$ is:
\begin{equation}
\widehat{\theta}^{(a)}=\argmin_{\theta\in \mathbb{R}^p}\frac{1}{2}\theta^\top \widetilde{\Sigma}\theta-\widetilde{\Sigma}_{a,:}\theta +\sum_{j=1}^p\lambda^{(a)}_j|\theta_j|,
\end{equation}
\begin{assump}[Mutual incoherence condition]\label{assump:incoh}
	The population covariance $\Sigma^*$ satisfies $\vertiii{\Sigma^*_{(\overline{\mathcal{N}}_a)^c,\mathcal{N}_a}(\Sigma^*_{\mathcal{N}_a,\mathcal{N}_a})^{-1}}_{\infty}\leq 1-\gamma$ for some $0<\gamma\leq 1$.
\end{assump}
\begin{thm}[Neighborhood Selection Consistency]\label{thm:nblasso_support_full}
	Consider the model setting described in Section 1.1 and the estimator $\widehat{\theta}^{(a)}$ defined in~\eqref{eq:nb_lasso}.
	Suppose Assumption \ref{assump:incoh} holds, and the tuning parameters $\lambda^{(a)}_j$'s in \eqref{eq:nb_lasso} satisfy $\lambda^{(a)}_j \asymp \|\Sigma^*\|_{\infty}\frac{\|\Theta^*_{:,a}\|_1}{\Theta^*_{a,a}}\sqrt{\frac{\log p}{\mymin_{k} n_{j,k}}}$. If $\gamma_a\leq (\frac{2-c\gamma}{2-\gamma})^2$ for some constant $c>0$, and
	\begin{equation}\label{eq:nblasso_samplesize_full}
	\min_{j\in \mathcal{N}_a}\min_k n_{j,k}\geq \frac{C_2\|\Sigma^*\|^2_{\infty}\kappa_2^2}{\gamma^{2}}\left[(\kappa_3^2+1)d_a^2+\frac{(\gamma+4)^2(\kappa_1+1)^2}{(\theta^{(a)}_{\min})^2}\right]\log p,
	\end{equation}
	then $\widehat{\mathcal{N}}_a=\{j:\widehat{\theta}^{(a)}_j\neq 0\}=\mathcal{N}_a$ with probability at least $1-p^{-c}$ for some absolute constants $c,C_1,C_2>0$.
\end{thm}
\begin{remark}
	Theorem~\ref{thm:nblasso_support_full} suggests that, even when many pairs of nodes are only measured few times together, as long as the tuning parameters are chosen carefully w.r.t. the pairwise sample sizes, we are still able to recover the neighborhood of node $a$. In particular, we only need to collect sufficient samples for a pair of nodes if at least one node in this pair is a neighbor of $a$. This is not a trivial result, since the estimator~\eqref{eq:nb_lasso} has to perform a selection from its neighbors and a large number of non-neighbors which are not measured well. 
\end{remark}
For ease of presentation, here we have assumed that $n_{j,k}>0, \forall j,k$ so that $\lambda^{(a)}_k$ is well defined. However, with a slight modification of our proof, the $\ell_2$ and $\ell_1$ error bounds in Theorem~\ref{thm:nblasso_support_full} would still hold even if $n_{j,k}=0$ for some $j,k\in \overline{\cN}_a^c$, and if we simply define $\lambda^{(a)}_j \asymp \|\Sigma^*\|_{\infty}\frac{\|\Theta^*_{:,a}\|_1}{\Theta^*_{a,a}}\sqrt{\frac{\log p}{\mymax\{1,\mymin_k n_{j,k}\}}}$.

\begin{remark}[Effect of differing sample sizes]
Theorem \ref{thm:nblasso_support_full} also imposes a constraint on $\gamma_a$, the sample size ratio between the most well measured non-neighbor and the worst measured neighbor of $a$. The is due to that when the sample sizes for the non-neighbors are all much larger than the neighbors (large $\gamma_a$), the neighbors would suffer from much stronger regularization than non-neighbors.
\end{remark}

\begin{remark}[Comparison with the literature]
 When all pairwise sample sizes are all equal, $\gamma_a=1$ and this requirement becomes $\mymin_{j\in \cN_a}\mymin_{k\in [p]}n_{j,k}\geq C(\Sigma^*)d_a^2\log p$, which is similar to the standard sample size requirement for neighborhood lasso in \cite{meinshausen2006high,wainwright2009sharp}, except for an additional factor $d_a$. This additional factor is due to technical challenges brought by non-simultaneous measurements of all variables.
\end{remark}
\begin{thm}\label{thm:nb_lasso_err}
	If the tuning parameters $\lambda_j$'s in the neighborhood regression estimator satisfy $$\lambda^{(a)}_j \asymp \|\Sigma^*\|_{\infty}\frac{\|\Theta^*_{:,a}\|_1}{\Theta^*_{a,a}}\sqrt{\frac{\log p}{\mymin_{k} n_{j,k}}},$$ $$
\mymin_{j\in \cN_a}\mymin_k n_{j,k}\geq \frac{C\|\Sigma^*\|_{\infty}^2(\gamma_a+1)}{\lambda_{\mymin}^2(\Sigma^*)}d_a^2\log p,
$$
then with probability at least $1-Cp^{-c}$,
\begin{align*}
    \|\widehat{\theta}^{(a)}-\theta^{(a)*}\|_2\leq& \frac{C\|\Sigma^*\|_{\infty}}{\lambda_{\mymin}(\Sigma^*)}\frac{\|\Theta^*_{:,a}\|_1}{\Theta^*_{a,a}}\sqrt{\frac{d_a\log p}{\mymin_{j\in \cN_a}\mymin_k n_{j,k}}},\\
\|\widehat{\theta}^{(a)}-\theta^{(a)*}\|_1\leq& \frac{C\|\Sigma^*\|_{\infty}(\sqrt{\gamma_a}+1)}{\lambda_{\mymin}(\Sigma^*)}\frac{\|\Theta^*_{:,a}\|_1}{\Theta^*_{a,a}}d_a\sqrt{\frac{\log p}{\mymin_{j\in \cN_a}\mymin_k n_{j,k}}},\\
\|\widehat{\theta}^{(a)}-\theta^{(a)*}\|_{\lambda^{(a)}_j,1}\leq& \frac{C\|\Sigma^*\|^2_{\infty}}{\lambda_{\mymin}(\Sigma^*)}\frac{\|\Theta^*_{:,a}\|^2_1}{(\Theta^*_{a,a})^2}\frac{d_a\log p}{\mymin_{j\in \cN_a}\mymin_k n_{j,k}}.
\end{align*}
\end{thm}
One key sample size quantity in Theorem \ref{thm:nb_lasso_err} is $\mymin_{j\in \cN_a}\mymin_kn_{j,k}$, which is illustrated in Figure~\ref{fig:nb_lasso_err_n}. With appropriately chosen regularization parameters $\lambda_{j}^{(a)}$ that reflect the corresponding sample sizes for each node, Theorem \ref{thm:nb_lasso_err} shows that when $\mymin_{j\in \cN_a}\mymin_kn_{j,k}$ is sufficiently large, we can establish the both $\ell_2$ and $\ell_1$ consistency of the neighborhood lasso estimator $\widehat{\theta}^{(a)}$. \emph{The number of node pairs involved is $d_a p$, much smaller than the total number of node pairs $p^2$ when the degree of node $a$ is small.}
\begin{figure}[!htb]
    \centering
    \includegraphics[width=0.3\textwidth]{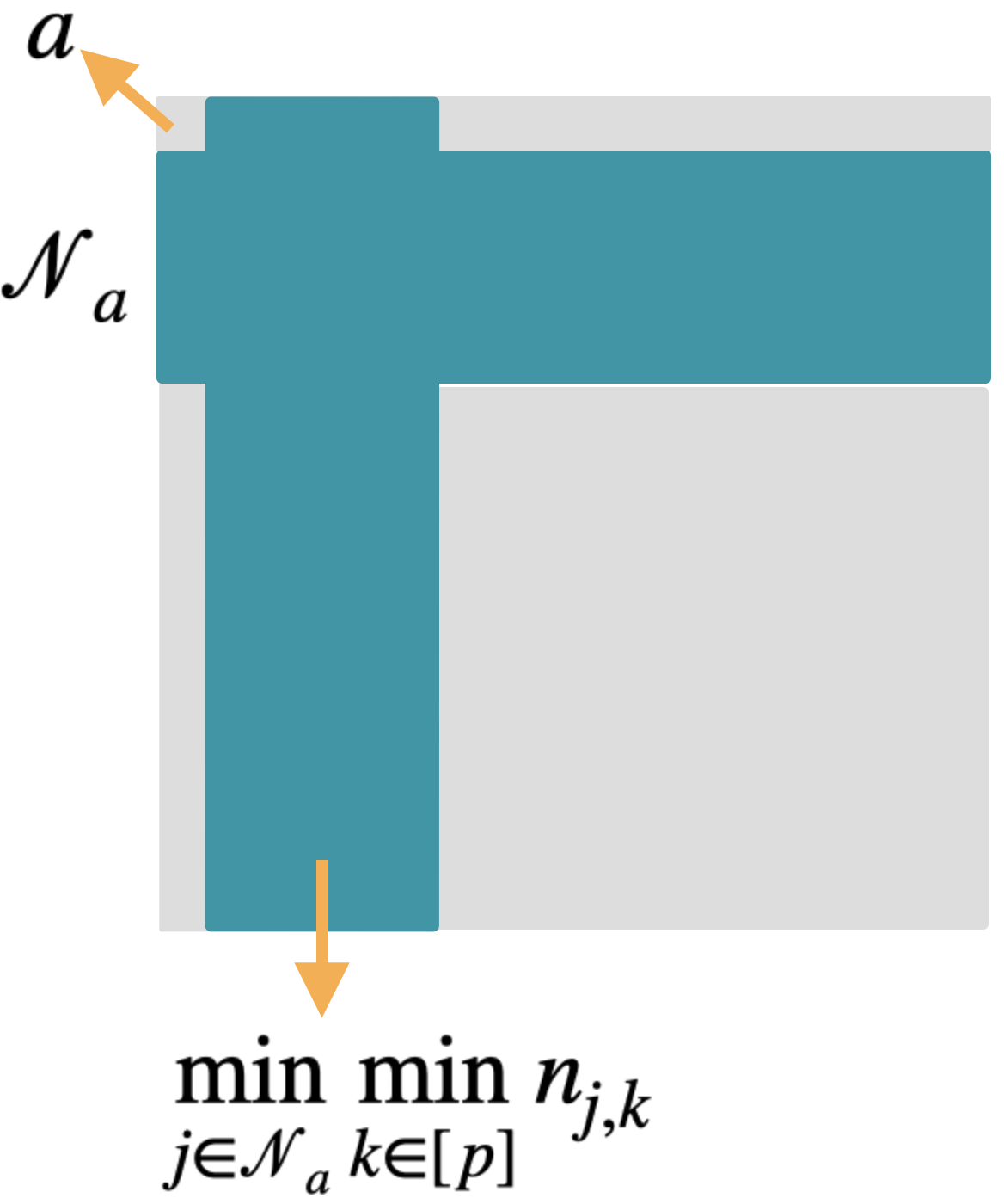}
    \caption{Illustration of node pairs whose pairwise sample sizes play key roles in the estimation error rate $\|\widehat{\theta}^{(a)}-\theta^{(a)*}\|_2$, $\|\widehat{\theta}^{(a)}-\theta^{(a)*}\|_2$.}
    \label{fig:nb_lasso_err_n}
\end{figure}
\section{GI-JOE (FDR) Algorithm}\label{sec:FDR_Alg}
We summarize our full GI-JOE (FDR) procedure in Algorithm \ref{alg:GI_JOE_FDR}. It calls the edge-wise inference algorithm for each node pair $(a,b)$ for $1\leq a<b\leq p$. 
\begin{algorithm}[ht!]
\noindent{\textbf{Input}}: Data set $\{x_{i,V_i}:V_i\subset [p]\}_{i=1}^n$, pairwise sample sizes $\{n_{j,k}\}_{j,k=1}^p$, significance level $\alpha$\\
 \For{$1\leq a<b\leq p$}{Run Algorithm 1 (edge-wise inference) in the main paper for node pair $(a,b)$ and obtain $p$-value $p_{a,b}$}
   Sort the $p$-values of $m:=\frac{p(p-1)}{2}$ pairs as $\boldsymbol{p}_{i_1} \leq \dots\leq \boldsymbol{p}_{i_m}$. \\
  Let $k=m$, $t_p=(2\log m-2\log\log m)^{\frac{1}{2}}$.\\
    \While{$\boldsymbol{p}_{i_k}\geq 2(1-\Phi(t_p))$ and $\boldsymbol{p}_{i_k}> \frac{\alpha k}{m}$}{
        $k=k-1$}
    \If{$\boldsymbol{p}_{i_k}< 2(1-\Phi(t_p))$}{
        $k=\argmax_l \boldsymbol{p}_{i_l} \leq 2(1-\Phi(\sqrt{2\log m})$.
    }
 \textbf{Output}: Set $\widetilde{E}$ of node pairs $(a,b)$ associated with the $k$ smallest $p$-values.
\caption{GI-JOE: FDR control}
 \label{alg:GI_JOE_FDR}
\end{algorithm}
Recall that in the main paper, we have presented a theoretical guarantee for our GI-JOE (FDR) approach, under Assumptions \ref{assump:FDR_n} and \ref{assump:FDR_correlation}.
Here, we specifically give an example for when Assumption \ref{assump:FDR_correlation} holds in the simultaneous measurement setting.
\begin{prop}\label{prop:assump7_simultaneous}
    Consider the special case where all $p$ variables are observed simultaneously with $n$ samples, and let $\Omega^* = \mathrm{diag}(\Theta^*)^{-\frac{1}{2}}\Theta^*\mathrm{diag}(\Theta^*)^{-\frac{1}{2}}$ be the normalized precision matrix. Then Assumption \ref{assump:FDR_correlation} can be implied by the following conditions: there exists $0<\rho_0<1$ such that
    \begin{enumerate}
        \item for $1\leq j\leq p$, $|\{k:|\Omega_{j,k}|>\rho_0\}|\leq C$ for some constant $C$;
        \item the maximum degree of the graph satisfies $d\ll p^{\frac{1-\rho_0}{1+\rho_0}}(\log p)^{\frac{\rho_0}{1+\rho_0}-\frac{1}{4}}(\log\log p)^{-\frac{1}{4}}$.
    \end{enumerate}
\end{prop}
Proposition \ref{prop:assump7_simultaneous} suggests that, for the special setting where all variables are simultaneously observed, Assumption \ref{assump:FDR_correlation} can be directly implied if  (i) each node only has constant number of strongly connected neighbors; (ii) a polynomial of the maximum degree of the graph is much smaller than the total number of variables $p$. Unfortunately, in the more general setting considered in this paper, it is extremely hard to simplify $\rho_n(a,b,a',b')$ due to the complicated structures induced by the erose measurements. However, as we will see in our empirical results, for many of the erose measurement settings and graph types where the sample size assumptions may even be violated, the FDP is still controlled, which may suggest that our Assumption \ref{assump:FDR_correlation} is not overly stringent.
\section{Confidence Intervals for Each Edge}\label{append:CI}
In this section, we extend our GI-JOE edge-wise testing method to construct confidence intervals for each edge. Specifically, we propose confidence intervals both for $\theta^{(a)*}_b = -\frac{\Theta^*_{a,b}}{\Theta^*_{a,a}}$ and $\Theta^*_{a,b}$, where the former is a simple extension of the edge-wise testing algorithm, while the latter requires new method and normal approximation theory. 

For $\theta^{(a)*}_b = -\frac{\Theta^*_{a,b}}{\Theta^*_{a,a}}$, we consider the following confidence interval: 
\begin{equation}\label{eq:CI}
\widehat{\mathbb{C}}_{\alpha}^{a,b}=[\widetilde{\theta}^{(a)}_b-z_{\alpha/2}\widehat{\sigma}_n(a,b),\widetilde{\theta}^{(a)}_b+z_{\alpha/2}\widehat{\sigma}_n(a,b)].
\end{equation} 
The following theorem shows the validity of the confidence interval \eqref{eq:CI}, a direct consequence of Theorem \ref{thm:normal_approx_var_est}.
\begin{thm}[Coverage and Width of Confidence Intervals]\label{thm:coverage}
Under the same assumptions as in Theorem \ref{thm:typeI_power}, the confidence interval $\widehat{\mathbb{C}}$ defined in \eqref{eq:CI} satisfies
$\lim_{n,p\rightarrow \infty}\bP(\theta^{(a)*}_b\in \widehat{\mathbb{C}}_{\alpha}^{a,b})=1-\alpha$, $|\widehat{\mathbb{C}}_{\alpha}^{a,b}|\overset{p}{\rightarrow}2z_{\alpha/2}\sigma_n(a,b)$, where $|\widehat{\mathbb{C}}_{\alpha}^{a,b}|$ is the width of the confidence interval.
\end{thm}

Although the inference target $\theta^{(a)*}_b$ covered by the confidence interval above shares the same non-zero pattern as $\Theta^*_{a,b}$, sometimes it might be of interest to directly construct a confidence interval for $\Theta^*_{a,b}$ instead of $-\frac{\Theta_{a,b}^*}{\Theta_{a,a}^*}$. We first consider the case where $a\neq b$, and discuss our confidence interval for each diagonal element $\Theta^*_{a,a}$ later. To achieve this purpose, here we propose the following test statistic for $\Theta^*_{a,b}$ when $a\neq b$. Noting that $\Theta^*_{a,a}=(\frac{\Theta^*_{a,:}}{\Theta^*_{a,a}}\Sigma^*\frac{\Theta^*_{:,a}}{\Theta^*_{a,a}})^{-1}$, we can estimate $\Theta^*_{a,a}$ by  $\widehat{\Theta}_{a,a}:=(\widehat{\overline{\theta}}^{(a)\top} \widehat{\Sigma}\widehat{\overline{\theta}}^{(a)})^{-1}$, where $\widehat{\Sigma}$ is the entry-wise estimate of the covariance matrix, and $\widehat{\overline{\theta}}^{(a)}\in \bR^p$ serves as an estimate for $\frac{\Theta^*_{:,a}}{\Theta^*_{a,a}}$, satisfying $\widehat{\overline{\theta}}^{(a)}_a = 1$ and $\widehat{\overline{\theta}}^{(a)}_{\backslash a}=-\widehat{\theta}^{(a)}_{\backslash a}$. Then we consider the test statistic 
\begin{equation}\label{eq:test_stat_Theta}
    \widetilde{\Theta}_{a,b}:=-\widehat{\Theta}_{a,a}\widetilde{\theta}^{(a)}_b,
\end{equation} 
with $\widetilde{\theta}^{(a)}_b$ being the debiased neighborhood lasso estimate, and construct a confidence interval centered around it. Similar to the proof of Theorem \ref{thm:nb_lasso_debias_decomp}, some computations can show that the asymptotic variance of $\widetilde{\Theta}_{a,b}$ is $\sigma_{n,\Theta}^2(a,b) = \mathcal{T}^{(n)*}\times_1\Theta^*_{b,:}\times_2 \Theta^*_{a,:}\times_3\Theta^*_{b,:}\times_4 \Theta^*_{a,:}$, and hence we estimate it by 
\begin{equation}\label{eq:var_est_Theta}
    \widehat{\sigma}_{n,\Theta}^2(a,b):=\widehat{\mathcal{T}}^{(n)}\times_1\widehat{\Theta}_{b,:}\times_2 \widehat{\Theta}_{a,:}\times_3\widehat{\Theta}_{b,:}\times_4 \widehat{\Theta}_{a,:},
\end{equation}
where $\widehat{\mathcal{T}}^{(n)}$ is an estimator for $\mathcal{T}^{(n)*}$:
$(\widehat{\mathcal{T}}^{(n)})_{j,k,j',k'}=(\widetilde{\Sigma}_{j,j'}\widetilde{\Sigma}_{k,k'}+\widetilde{\Sigma}_{j,k'}\widetilde{\Sigma}_{k,j'})\frac{n_{j,k,j',k'}}{n_{j,k}n_{j',k'}}$; and $\widehat{\Theta}_{a,:}$ and $\widehat{\Theta}_{b,:}$ are appropriately constructed estimators with several options. 

\paragraph{Choices of $\widehat{\Theta}_{:,a}$ and $\widehat{\Theta}_{:,b}$} For $\widehat{\Theta}_{a,:}$, note that $\widehat{\overline{\theta}}^{(a)}$ defined by $\widehat{\overline{\theta}}^{(a)}_a=1$ and $\widehat{\overline{\theta}}^{(a)}_{\backslash a} = -\widehat{\theta}^{(a)}_{\backslash a}$ can serve as an estimate for $\frac{1}{\Theta^*_{a,a}}\Theta^*_{:,a}$. Hence we can consider $\widehat{\Theta}_{:,a}=\widehat{\overline{\theta}}^{(a)}\widehat{\Theta}_{a,a}$, where $\widehat{\Theta}_{a,a}=(\widehat{\overline{\theta}}^{(a)\top}\widehat{\Sigma}\widehat{\overline{\theta}}^{(a)})^{-1}$ was defined in the beginning of this section. One can also substitute the entry-wise estimate $\widehat{\Sigma}$ in $\widehat{\Theta}_{a,a}$ by $\widetilde{\Sigma}$, the projected positive semi-definite estimate for $\Sigma^*$, and obtain $\widetilde{\Theta}_{a,a}=(\widehat{\overline{\theta}}^{(a)\top}\widetilde{\Sigma}\widehat{\overline{\theta}}^{(a)})^{-1}$. As we will show in Lemma \ref{lem:SampleCov_entry_err}, both $\widehat{\Sigma}$ and $\widetilde{\Sigma}$ satisfy the same entry-wise error bound, and hence using either of the following two estimates
\begin{equation}\label{eq:Theta_a_est}
\widehat{\Theta}_{:,a} = \widehat{\overline{\theta}}^{(a)}\widehat{\Theta}_{a,a},\text{ or }
\widehat{\Theta}_{:,a} = \widehat{\overline{\theta}}^{(a)}\widetilde{\Theta}_{a,a}
\end{equation}
would lead to the same estimation error bound for the variance. While for $\widehat{\Theta}_{:,b}$, one can either construct 
it by another node-wise regression for node $b$ upon the other $p-1$
 nodes, or one can make use of the estimators for $\Theta^{(a)*}_{:,b}$ previously constructed in the debiasing step. For computational efficiency, here we focus on the latter approach. In particular, note that $\Theta^*_{:,b} = \Theta^{(a)*}_{:,b} + \frac{\Theta^*_{a,b}}{\Theta^*_{a,a}}\Theta^*_{:,a}$, and hence one can estimate $\Theta^*_{:,b}$ by
\begin{equation}\label{eq:Theta_b_est}
    \widehat{\Theta}_{:,b} = \widehat{\Theta}_{:,b}^{(1)}+\widehat{\Theta}_{:,b}^{(2)},
\end{equation}
where $\widehat{\Theta}_{:,b}^{(1)}$ is an estimate for $\Theta^{(a)*}_{:,b}$ and $\widehat{\Theta}_{:,b}^{(2)}$ is an estimate for $\frac{\Theta^*_{a,b}}{\Theta^*_{a,a}}\Theta^*_{:,a}$. $\widehat{\Theta}_{:,b}^{(1)}$ can be set as $\widehat{\Theta}^{(a)}_{b,:}$ or $\widetilde{\Theta}^{(a)}_{b,:}$ previously defined in Section 3.1 of the main paper:
\begin{equation*}
\begin{split}
\widehat{\Theta}^{(a)}_{b,b} = &(\widetilde{\Sigma}_{b,:}\widehat{\overline{\theta}}^{(a,b)})^{-1},\, \widehat{\Theta}^{(a)}_{b,:} = \widehat{\Theta}^{(a)}_{b,b}\widehat{\overline{\theta}}^{(a,b)},\\
\widetilde{\Theta}^{(a)}_{b,b} = &(\widehat{\overline{\theta}}^{(a,b)\top}\widetilde{\Sigma}\widehat{\overline{\theta}}^{(a,b)})^{-1},\, \widetilde{\Theta}^{(a)}_{b,:} = \widetilde{\Theta}^{(a)}_{b,b}\widehat{\overline{\theta}}^{(a,b)},
\end{split}
\end{equation*}
where $\widehat{\overline{\theta}}^{(a,b)}$ was based on the node-wise regression of node $b$ upon $[p]\backslash \{a, b\}$, defined in the main paper. One can also substitute $\widetilde{\Sigma}$ by $\widehat{\Sigma}$.
While for estimating $\frac{\Theta^*_{a,b}}{\Theta^*_{a,a}}\Theta^*_{:,a}$, since $\widehat{\overline{\theta}}^{(a)}$ is an estimate for $\frac{1}{\Theta^*_{a,a}}\Theta_{:,a}$, $\widehat{\theta}^{(a)}_b$ is an estimate for $-\frac{\Theta^*_{a,b}}{\Theta^*_{a,a}}$, $\widehat{\Theta}_{:,b}^{(2)}$ can be set as
\begin{equation}
    \widehat{\Theta}_{:,b}^{(2)} =  -\widehat{\Theta}_{a,a}\widehat{\theta}^{(a)}_b\widehat{\overline{\theta}}^{(a)},\text{ or }\widehat{\Theta}_{:,b}^{(2)} =  -\widetilde{\Theta}_{a,a}\widehat{\theta}^{(a)}_b\widehat{\overline{\theta}}^{(a)}.
\end{equation}
Similar to estimating $\Theta^*_{:,a}$, these options above for estimating $\Theta^*_{:,b}$ use either $\widehat{\Sigma}$ or $\widetilde{\Sigma}$, and can satisfy the same statistical error bounds due to the same entry-wise error guarantee of $\widehat{\Sigma}$ and $\widetilde{\Sigma}$ in Lemma \ref{lem:SampleCov_entry_err}.

Then when $a\neq b$, we construct our confidence interval for $\Theta^*_{a,b}$ by \begin{equation}\label{eq:CI_Theta}
    \widehat{\mathbb{C}}_{\Theta,\alpha}^{a,b}=[\widetilde{\Theta}_{a,b}-z_{\alpha/2}\widehat{\sigma}_{n,\Theta}(a,b),\widetilde{\Theta}_{a,b}+z_{\alpha/2}\widehat{\sigma}_{n,\Theta}(a,b)].
\end{equation}
We can also provide inference for each diagonal element $\Theta^*_{a,a}$ based on test statistic $\widehat{\Theta}_{a,a}$. Let $\widehat{\sigma}_{n,\Theta}(a,a) = \widehat{\mathcal{T}}^{(n)}\times_1\widehat{\Theta}_{:,a}\times_2\widehat{\Theta}_{:,a}\times_3\widehat{\Theta}_{:,a}\times_4\widehat{\Theta}_{:,a}$. Then we can construct the confidence interval for $\widehat{\Theta}_{a,a}$ as follows:
\begin{equation}\label{eq:CI_Theta_aa}
    \widehat{\mathbb{C}}_{\Theta,\alpha}^{a,a} = [\widehat{\Theta}_{a,a}-z_{\alpha/2}\widehat{\sigma}_{n,\Theta}(a,a), \widehat{\Theta}_{a,a}+z_{\alpha/2}\widehat{\sigma}_{n,\Theta}(a,a)].
\end{equation}
Note here that all extra computations needed for the inference of $\Theta^*_{a,b}$ are elementary and basically free. 

In the following, we present theoretical guarantees to show the asymptotic validity of these confidence intervals under similar conditions to Assumptions \ref{assump:inference_n_B}-\ref{assump:var_est}. Define $S_1'(a,b) = S_1(a,b)\cup \{(a,a)\}$ for $a\neq b$ and $S_1'(a,a)=\{(j,k): j\text{ or }k\in \overline{\mathcal{N}}_a\}$;  $S_2'(a,b)= (\overline{\mathcal{N}}_a\times \overline{\mathcal{N}}_b)\cup (\overline{\mathcal{N}}_b\times \overline{\mathcal{N}}_a)$ for $a\neq b$ and $S_2'(a,a)=\overline{\mathcal{N}}_a\times \overline{\mathcal{N}}_a$. Let $n_1'(a,b) = \min_{(j,k)\in S_1'(a,b)}n_{j,k}$ and $n_2'^{(a,b)} = \min_{(j,k)\in S_2'(a,b)}n_{j,k}$ for arbitrary $a,\,b\in [p]$. We also define $\gamma_a^{(a)}=\gamma_a$, and key constant terms
\begin{equation}\label{eq:constant_Theta_inference}
\begin{split}
    C'(\Theta^*;a,b)=\begin{cases}\frac{C\kappa_{\Theta^*}^2\|\Theta^*_{:,a}\|_1^2(\|\Theta^{(a)*}_{:,b}\|_1+|\Theta^*_{a,b}|)}{\lambda_{\min}(\Theta^*)\min_{(j,k)\in S_2'(a,b)}|\Theta^*_{a,j}\Theta^*_{b,k}+\Theta^*_{a,k}\Theta^*_{b,j}|},&a\neq b;\\
    \frac{C\kappa_{\Theta^*}^2\|\Theta^*_{:,a}\|_1^2}{\min_{j\in \overline{\mathcal{N}}_a}|\Theta^*_{j,a}|^2}, & a=b.
    \end{cases}
\end{split}
\end{equation}
\begin{assump}\label{assump:inference_n_B2}
        \begin{align*}
	    n_1'^{(a,b)} = \min_{(j,k)\in S_1'(a,b)}n_{j,k}\geq C\frac{\|\Sigma^*\|_{\infty}^2}{\lambda_{\mymin}^2(\Sigma^*)}(\kappa_{\Sigma^*}^4+\gamma_a+\gamma_b^{(a)})(d_a+d_b+1)^2\log p.
	\end{align*} 
\end{assump}
\begin{remark}[$n_1^{(a,b)} = n_1'^{(a,b)}$]
    Note that since $n_{a,a} = \sum_i \ind{a\in V_i} \geq \sum_i\ind{a,j\in V_i}=n_{a,j}$ for all $j$, we immediately have that $n_{a,a}\geq \min_{j\in \mathcal{N}_a\cup \mathcal{N}_b^{(a)}}n_{a,j}=\min_{(j,k)\in S_1(a,b)}n_{j,k}$. Therefore, $n_1'^{(a,b)} = n_1^{(a,b)}$ always holds, and it is equivalent to change $n_1'^{(a,b)}$ in Assumption \ref{assump:inference_n_B2}, \ref{assump:inference_n_BE2}, and \ref{assump:var_est2} to $n_1^{(a,b)}$.
\end{remark}
\begin{assump}\label{assump:inference_n_E2}
    $C_{\varepsilon}(\Sigma^*)(d_a+d_b+1)^{2+\varepsilon}=o(n_2'^{(a,b)})$ for some constant $\varepsilon>0$, where $C_{\varepsilon}(\Sigma^*)=\left(\frac{C(1+2/\epsilon)\|\Sigma^*\|_{\infty}}{\lambda_{\mymin}(\Sigma^*)}\right)^{2+\epsilon}$.
\end{assump}
\begin{assump}\label{assump:inference_n_BE2}
    \begin{equation}\label{eq:inference_n_BE2}
        n_1'^{(a,b)} \gg C'^2(\Theta^*;a,b)(\kappa_{\Theta^*}^4+\gamma_a+\gamma^{(a)}_b)[(d_a+d_b+1)\log p]^2\frac{n_2'^{(a,b)}}{n_1'^{(a,b)}},
    \end{equation}
    where $C'(\Theta^*;a,b)$ is defined in \eqref{eq:constant_Theta_inference}. 
\end{assump}
\begin{assump}\label{assump:var_est2}
    $$n_1'^{(a,b)} \gg C'^4(\Theta^*;a,b)(d_a+d_b+1)^2\log p\left(\frac{n_2'^{(a,b)}}{n_1'^{(a,b)}}\right)^2,$$
    where $C'(\Theta^*;a,b)$ is defined in \eqref{eq:constant_Theta_inference}. 
\end{assump}
\begin{thm}\label{thm:coverage_Theta}
Consider the inference for $\Theta^*_{a,b}$ for any nodes $a,\,b\in [p]$. Suppose we construct the confidence interval $\widehat{\mathbb{C}}_{\Theta,\alpha}^{a,b}$ as in \eqref{eq:CI_Theta} if $a\neq b$, and as in \eqref{eq:CI_Theta_aa} if $a=b$. Then with appropriately chosen tuning parameters as specified in Theorem \ref{thm:nb_lasso_debias_decomp}, if Assumptions \ref{assump:inference_n_B2}-\ref{assump:var_est2} hold, the confidence interval $\widehat{\mathbb{C}}_{\Theta,\alpha}^{a,b}$ satisfies $\lim_{n,p\rightarrow \infty}\mathbb{P}(\Theta^*_{a,b}\in \widehat{\mathbb{C}}_{\Theta,\alpha}^{a,b})=1-\alpha$.
\end{thm}
\noindent The detailed assumptions and proof of Theorem \ref{thm:coverage_Theta} can be found in Section \ref{sec:proof_CLT_Theta}.
Assumptions \ref{assump:inference_n_B2}-\ref{assump:var_est2} are very similar to Assumptions \ref{assump:inference_n_B}-\ref{assump:var_est}, the sample size conditions required for the valid inference of $\frac{\Theta^*_{a,b}}{\Theta^*_{a,a}}$. One main difference is that Assumptions \ref{assump:inference_n_E2}-\ref{assump:var_est2} consider a slightly different sample size quantity $n_2'^{(a,b)}$, which by Proposition \ref{prop:indexsets} is larger than $n_2^{(a,b)}$. Therefore, under similar assumptions and without additional computations, our procedure can also lead to an asymptotically valid confidence interval for $\Theta^*_{a,b}$.

Now we present some simulation results that validate our confidence intervals for the entries of the precision matrix. In particular, for $\widehat{\sigma}_{n,\Theta}^2(a,b)$ defined in \eqref{eq:var_est_Theta}, we choose $\widehat{\Theta}_{:,a} = \widehat{\overline{\theta}}^{(a)}\widehat{\Theta}_{a,a}$, $\widehat{\Theta}_{:,b}$ as in \eqref{eq:Theta_b_est} with $\widehat{\Theta}_{:,b}^{(1)} = (\widehat{\overline{\theta}}^{(a,b)\top}\widehat{\Sigma}\widehat{\overline{\theta}}^{(a,b)})^{-1}\widehat{\overline{\theta}}^{(a,b)}$, and $\widehat{\Theta}_{:,b}^{(2)} = -\widehat{\Theta}_{a,a}\widehat{\theta}^{(a)}_b\widehat{\overline{\theta}}^{(a)}$. That is, we use $\widehat{\Sigma}$ instead of $\widetilde{\Sigma}$ throughout for our variance estimation.
We follow the same pairwise measurement set-up as in Section 5.1 of the main paper, and chain graphs, multi-star graphs with three stars, and Erd\H{o}s–R\'{e}nyi graphs are considered. Since the minimum pairwise sample size conditions in Assumptions \ref{assump:inference_n_B2}-\ref{assump:var_est2} imposed on edge sets $S_1'(a,b)$ and $S_2'(a,b)$ can be implied by or equivalent to the same sample size conditions on $S_1(a,b)$ and $S_2(a,b)$, here we still set $n_{j,k}=n1$ if $(j,k)\in S_1(a,b)\backslash S_2(a,b)$, and set $n_{j,k}=n2$ if $(j,k)\in S_2(a,b)$. We provide confidence intervals for $\Theta^*_{2,4}=0$ and $\Theta^*_{2,3}\neq 0$ with different signal strength. The number of nodes is fixed as $p=200$, and $n2/n1=1.5$. Figure \ref{fig:coverage} presents the coverage result of our confidence interval over $1000$ replicates, where the error bars represent $95\%$ confidence interval for the coverage rate. The dashed line is the target coverage, and we can see that the coverage of our confidence interval is close to the target rate $95\%$.
\begin{figure}
    \centering
\includegraphics[height = 5.5cm]{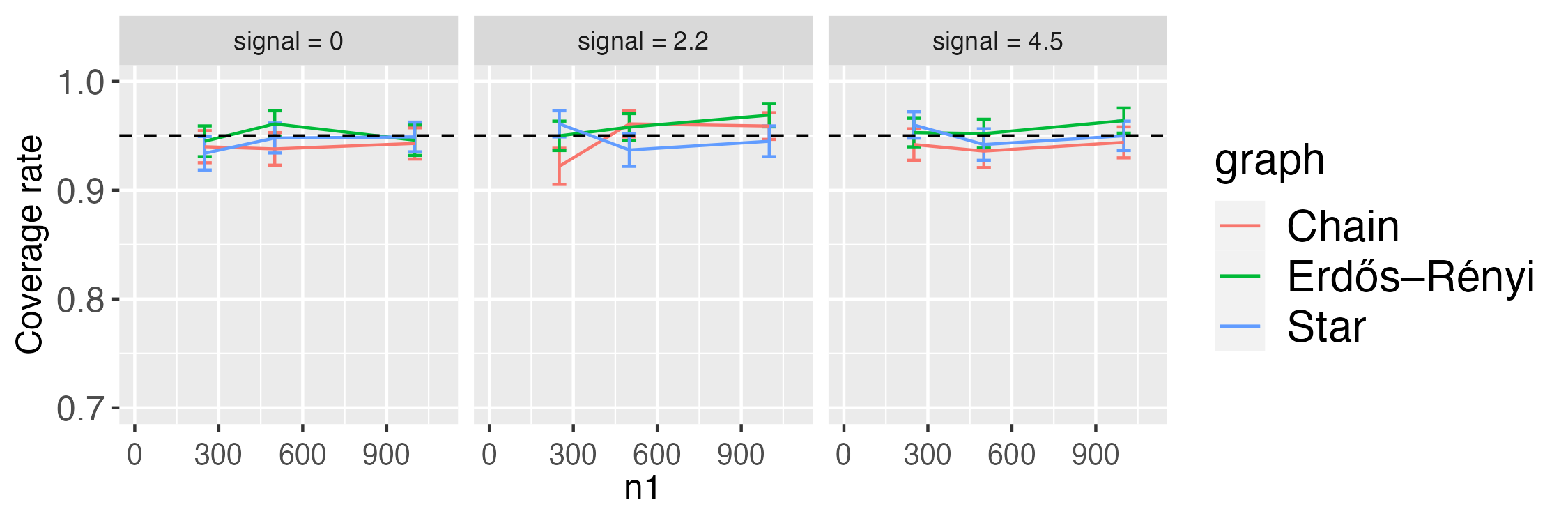}
\caption{Average coverage of the proposed confidence intervals for $\Theta^*_{a,b}$ over $1000$ replicates. Error bars represent $95\%$ confidence interval for the coverage rate, and the dashed line is the target coverage. Each panel considers one signal strength $\sqrt{n_2}\Theta^*_{a,b}=0, 2.2, 4.5$.}
    \label{fig:coverage}
\end{figure}
\section{Numerical Details}\label{append:NumericDetails}
\subsection{Tuning Parameter Selection in the Simulations on Edge-wise Inference}
Our estimators $\widehat{\theta}^{(a)}$ and $\widehat{\Theta}_{b,:}^{(a)}$ depend on tuning parameters $\lambda^{(a)}_j$ and $\lambda^{(a,b)}_j$ for $1\leq j\leq p$. Based on the theoretical scaling, we set $\lambda^{(a)}_j=C_1\sqrt{\frac{\log p}{\min_kn_{j,k}}}$, and $\lambda^{(a,b)}_j=C_2\sqrt{\frac{\log p}{\min_kn_{j,k}}}$, where $C_1,\,C_2$ are chosen based on neighborhood selection stability \citep{liu2010stability} over $20$ random subsamples for each corresponding neighborhood, with the stability threshold set as $0.05$, as suggested by \cite{liu2010stability}. For each random subsample, each sample is included with probability $0.8$ independently. The candidates for the tuning constant $C_1$ and $C_2$ are set as a geometric sequence of length $20$, spaced between $[C_{1,\max}/10, C_{1,\max}]$ and $[C_{2,\max}/10, C_{2,\max}]$. $C_{1,\max},\,C_{2,\max}$ are chosen as the regularization scaling that leads to a zero neighborhood regression estimate. For each experimental setting, 200 or 1000 independent replicates are run, but we only perform stability selection for one replicate and use the same tuning parameter for all replicates. 
\subsection{Implementation Details in Graph Selection Comparison}
\paragraph{Estimation Methods:} For baseline plug-in type methods, we directly plug in the covariance matrix estimate $\widetilde{\Sigma}$ (the positive semi-definite matrix solved by the ADMM algorithm) into the graphical lasso, neighborhood lasso, and CLIME algorithms to estimate the graph structure. This is slightly different from the original plug-in type methods where unbiased estimate $\widehat{\Sigma}$ is in use. We choose $\widetilde{\Sigma}$ to ensure convexity or algorithmic stability. We choose the tuning parameters for all methods using stability selection. $20$ random subsamplings are used for Nlasso, Glasso, Nlasso-JOE, and $10$ random subsamplings are used for CLIME to save computational time. For each random subsampling, each sample is included with probability $0.8$, and the stability threshold is $0.05$. 

\paragraph{Implementation Details of Inference Methods:}
The inference methods include GI-JOE (Holm), GI-JOE (FDR), DB-Glasso (Holm), and DB-Glasso (FDR). For GI-JOE (Holm) and GI-JOE (FDR), the tuning parameters in the neighborhood regression problems are set as $\lambda_j^{(a)}, \lambda_j^{(a,b)} = C\sqrt{\frac{\log p}{\min_{k}n_{j,k}}}$, and the constant $C$ is the same over the whole graph. We choose $C$ by stability selection on the graph, with $20$ random subsampling and stability threshold $0.05$. Same as the tuning procedure for edge-wise simulations, for each random subsampling, each sample is included with probability $0.8$ independently. The candidate set for $C$ is also a geometric sequence spaced between $[C_{\max}/10,C_{\max}]$, where $C_{\max}$ is the smallest constant that leads to empty graph estimate. this data-drive tuning procedure was done for each random replicate in the experiments. For DB-Glasso (Holm) and DB-Glasso (FDR), we first plug in $\widetilde{\Sigma}$ into graphical lasso to obtain an estimate $\widehat{\Theta}$ of the precision matrix; then we compute the debiased graphical lasso statistic:
$\widehat{T}=2\widehat{\Theta} -\widehat{\Theta}\widehat{\Sigma}\widehat{\Theta}$.
The variance of each edge $(j,k)$ is estimated by $\frac{1}{\min_{j,k}n_{j,k}}\widehat{\sigma}_{j,k}^2=\widehat{\Theta}_{j,j}\widehat{\Theta}_{k,k}+\widehat{\Theta}_{j,k}^2$, where we plug in the minimum pairwise sample size. Then we normalize the edge statistic to $\frac{\sqrt{\min_{j,k}n_{j,k}}\widehat{T}_{j,k}}{\widehat{\sigma}_{j,k}}$ and compute its $p$-value $p_{j,k}=2(1-\Phi(\frac{\sqrt{\min_{j,k}n_{j,k}}|\widehat{T}_{j,k}|}{\widehat{\sigma}_{j,k}}))$. Then we add Holm's correction and FDR control procedure on top of these edge-wise $p$-values, similar to what GI-JOE (Holm) and GI-JOE (FDR). 
\subsection{Testing against a Threshold}
Suppose we would like to test the hypothesis $H_{0,(a,b)}: |\frac{\Theta^{*}_{a,b}}{\Theta^*_{a,a}}|\leq \epsilon$ for some positive value $\epsilon$. We propose to find the $p$-value for $H_{0,(a,b)}$ as follows: 
$p_{a,b} = \min\{1,2(1-\Phi(\frac{|\widetilde{\theta}^{(a)}_b|-\epsilon}{\widehat{\sigma}_n(a,b)}))\}$. Now we show that the validity of this $p$-value is directly implied by our current theoretical results. Note that under $H_{0,(a,b)}$, for any $t>0$, 
\begin{align*}
    &\mathbb{P}(\frac{|\widetilde{\theta}^{(a)}_b|-\epsilon}{\widehat{\sigma}_n(a,b)}>t)\\
    \leq&\mathbb{P}(\frac{|\widetilde{\theta}^{(a)}_b|-|\frac{\Theta^{*}_{a,b}}{\Theta^*_{a,a}}|}{\widehat{\sigma}_n(a,b)}>t)\\
    \leq &\mathbb{P}(\frac{|\widetilde{\theta}^{(a)}_b+\frac{\Theta^{*}_{a,b}}{\Theta^*_{a,a}}|}{\widehat{\sigma}_n(a,b)}>t)\rightarrow 2(1-\Phi(t)),
\end{align*}
where the last line is due to Theorem 3 in the main paper. After getting the $p$-values for all node pairs, we can further apply the Holm's correction and the FDR control procedure upon these $\frac{p(p-1)}{2}$ $p$-values.

\section{Comparison between Different Variance Estimates for Edge-wise Inference}\label{append:varest_comp}
As mentioned in Section 3.1 of the main paper, there are different potential estimates for $\Theta^{(a)*}_{b,:}$, including $\widehat{\Theta}^{(a)}_{b,:}$ and $\widetilde{\Theta}^{(a)}_{b,:}$ defined as follows:
\begin{equation*}
\begin{split}
\widehat{\Theta}^{(a)}_{b,b} = &(\widetilde{\Sigma}_{b,:}\widehat{\overline{\theta}}^{(a,b)})^{-1},\, \widehat{\Theta}^{(a)}_{b,:} = \widehat{\Theta}^{(a)}_{b,b}\widehat{\overline{\theta}}^{(a,b)},\\
\widetilde{\Theta}^{(a)}_{b,b} = &(\widehat{\overline{\theta}}^{(a,b)\top}\widetilde{\Sigma}\widehat{\overline{\theta}}^{(a,b)})^{-1},\, \widetilde{\Theta}^{(a)}_{b,:} = \widetilde{\Theta}^{(a)}_{b,b}\widehat{\overline{\theta}}^{(a,b)}.
\end{split}
\end{equation*}
One may estimate the variance of the edge-wise statistic by
\begin{equation}\label{eq:var_est1}
    \widehat{\sigma}_n^2(a,b)=\widehat{\mathcal{T}}_n\times_1\widehat{\Theta}^{(a)}_{b,:}\times_2\widehat{\overline{\theta}}^{(a)}\times_3\widehat{\Theta}^{(a)}_{b,:}\times_4\widehat{\overline{\theta}}^{(a)},
\end{equation}
or 
\begin{equation}\label{eq:var_est2}
    \widehat{\sigma}_n^2(a,b)=\widehat{\mathcal{T}}_n\times_1\widetilde{\Theta}^{(a)}_{b,:}\times_2\widehat{\overline{\theta}}^{(a)}\times_3\widetilde{\Theta}^{(a)}_{b,:}\times_4\widehat{\overline{\theta}}^{(a)}.
\end{equation}
Both estimators satisfy the same theoretical error bounds, but we specifically recommend practitioners to use \eqref{eq:var_est2}, since it controls the Type I error rate better. In particular, we report the variance estimation results of both \eqref{eq:var_est1} and \eqref{eq:var_est2} in our edge-wise inference simulations in Figure \ref{fig:var_est_comparison}. We observe that both estimates tend to underestimate the true variance, while \eqref{eq:var_est2} is slightly better than \eqref{eq:var_est1}. As a consequence, constructing the test statistic with variance estimate in \eqref{eq:var_est2} leads to better Type I error control, as demonstrated in Figure \ref{fig:type1_err_comparison}. We have also empirically investigated using $\widetilde{\Theta}^{(a)}_{b,:}$ for the debiasing step, while it turns out that the combination that uses $\widehat{\Theta}^{(a)}_{b,:}$ for debiasing and $\widetilde{\Theta}^{(a)}_{b,:}$ for variance estimate results in the best type I error control.
\begin{figure}
    \centering
    \includegraphics[height = 5.5cm]{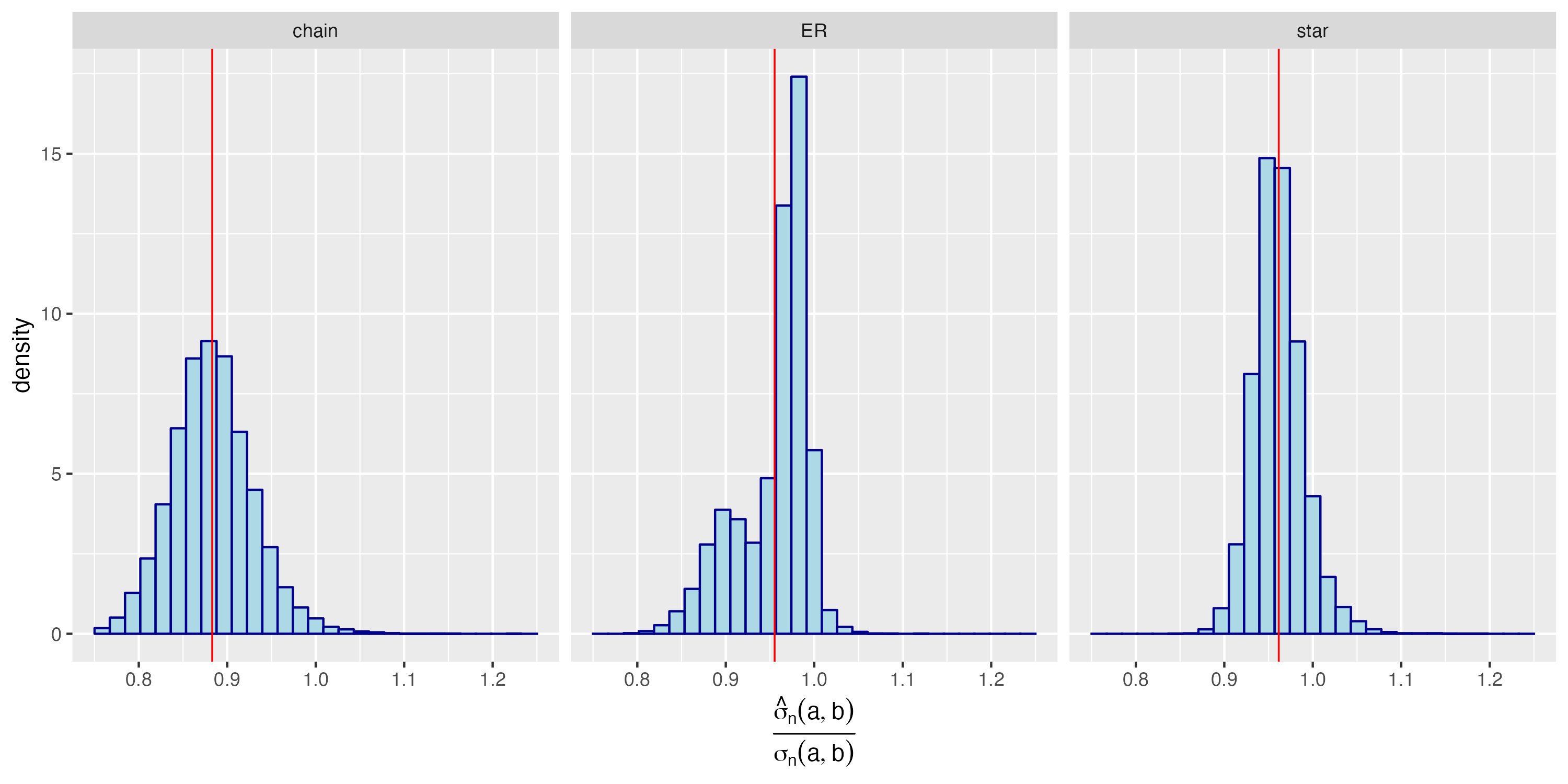}
    \includegraphics[height = 5.5cm]{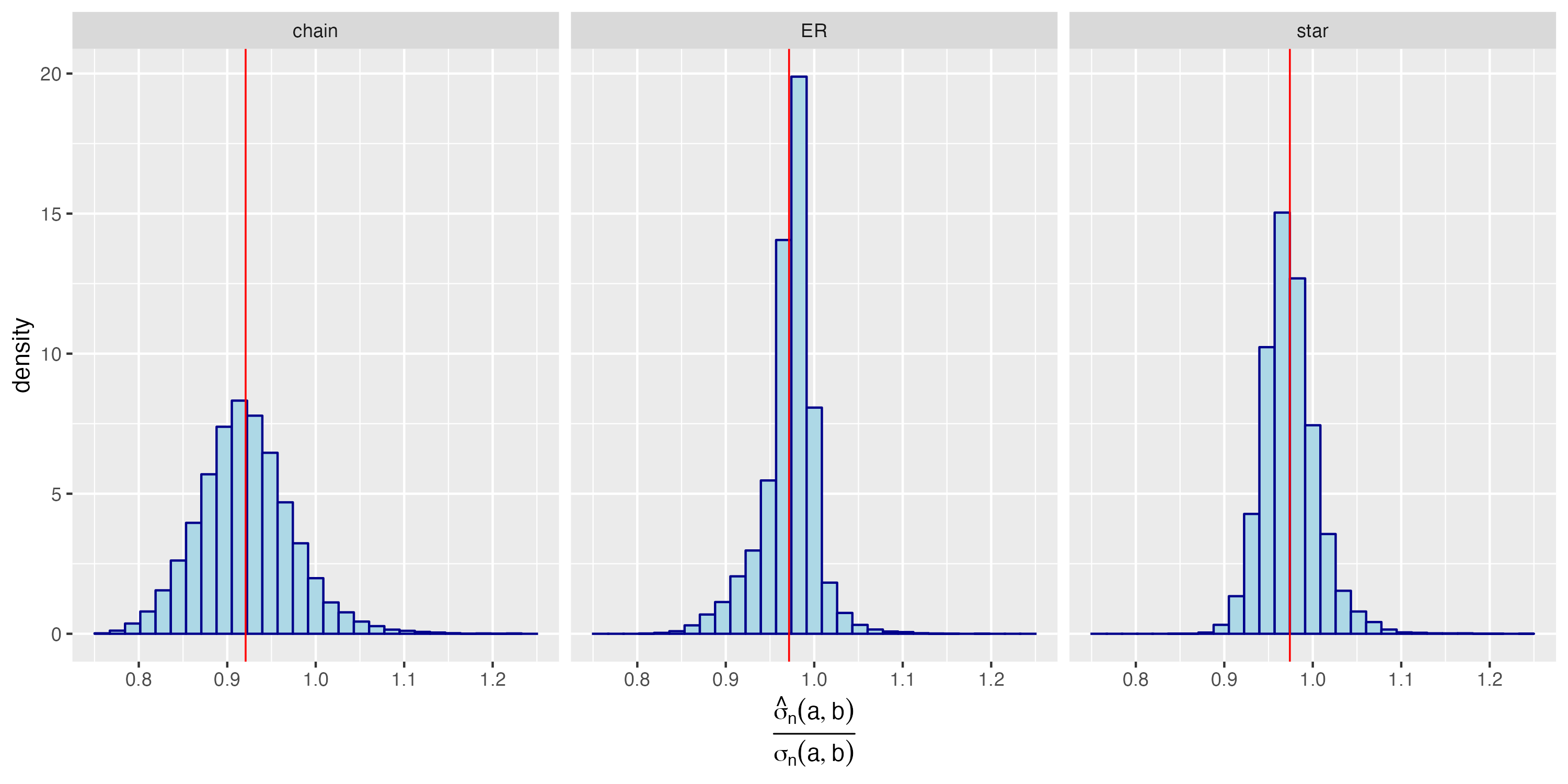}
\caption{Variance estimation results of $1000$ replicates in the edge-wise inference simulations when $\Theta^*_{a,b}=0$. The top figure presents $\widehat{\sigma}_n(a,b)/\sigma_n(a,b)$ when $\widehat{\sigma}_n(a,b)$ is chosen as in \eqref{eq:var_est1}, while the bottom figure presents the ratio when $\widehat{\sigma}_n(a,b)$ is as in \eqref{eq:var_est2}. The red line represents the mean of $\widehat{\sigma}_n(a,b)/\sigma_n(a,b)$ over $1000$ replicates. We can see that both variance estimates tend to be smaller than the true variance ($\widehat{\sigma}_n(a,b)/\sigma_n(a,b)<1$ most of the time), choosing $\widehat{\sigma}_n(a,b)$ as in \eqref{eq:var_est2} leads to slightly better estimate, especially for chain graphs.}
    \label{fig:var_est_comparison}
\end{figure}
\begin{figure}
    \centering
    \includegraphics[height = 5.5cm]{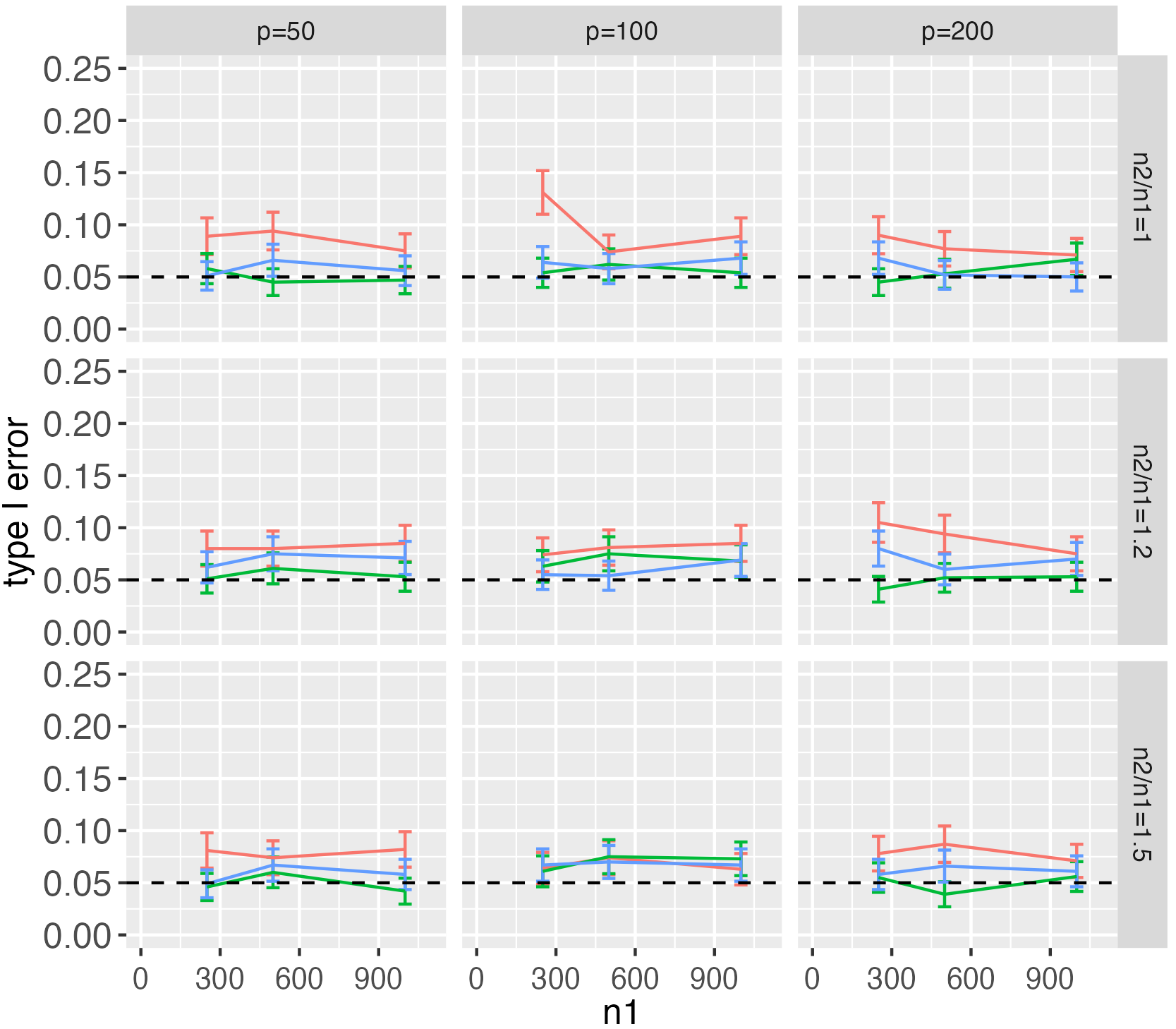}
    \includegraphics[height = 5.5cm]{fig/type1_err_varnewpsd.png}
    \caption{Type I error rates in edge-wise inference simulations. The left figure is when the test statistics is constructed with the variance estimate in \eqref{eq:var_est1}, and the right figure is using \eqref{eq:var_est2}. Hence we recommend the practitioners to use \eqref{eq:var_est2} for better type I error control.}
    \label{fig:type1_err_comparison}
\end{figure}

\section{Additional Empirical Results}\label{append:empiricalResults}
In this section, we present some additional empirical results, including the detailed true positive rates (TPR), true negative rates (TNR), true discovery rate (TDR) and F1 scores in the graph selection studies (Table \ref{tab:sim2_summary_chain_mssc1}-\ref{tab:scRNA_sim_summary_small-world}). Specifically, TPR is the ratio of selected true edges and total number of true edges; TNR is the ratio of unselected nonedges and total number of nonedges; TDR is the ratio of selected true edges and selected edges. In summary, estimation methods tend to be much more liberal (higher TPR but lower TDR) while inference methods are more conservative (higher TDR but lower TPR). The debiased graphical lasso methods with minimum sample size are often too conservative and select no edge at all, hence their TDR are sometimes NA.
\begin{table}[!ht]
     \centering
     \scalebox{0.6}{
     \begin{tabular}{|c|c|c|c|c|c|c|c|c|}
     \hline
     \multirow{2}{*}{Method}&\multicolumn{4}{|c|}{$n=600$}&\multicolumn{4}{|c|}{$n=800$}\\
     \cline{2-9}
        & TPR & TNR & TDR & F1 score  & TPR & TNR & TDR & F1 score  \\
        \hline
       Nlasso (AND) &0.977(0.012)& 0.984(0.016) & 0.475(0.172) & 0.619(0.181) & 0.991(0.007) & 0.991(0.001) & 0.537(0.019) & 0.697(0.016) \\ 
       \hline
       Nlasso (OR) &0.986(0.009) & 0.973(0.011) & 0.289(0.064) & 0.443(0.082) & 0.996(0.004) & 0.975(0.001) & 0.285(0.009) & 0.443(0.011) \\ 
       \hline
       Glasso &0.993(0.005) & 0.966(0.001) & 0.228(0.008) & 0.371(0.010) & 0.999(0.002) & 0.963(0.001) & 0.214(0.005) & 0.352(0.006)\\
       \hline
       CLIME &0.974(0.010) & 0.971(0.003) & 0.256(0.029) & 0.405(0.035) & 0.991(0.005) & 0.957(0.002) & 0.189(0.008) & 0.317(0.011)\\
       \hline
       Nlasso-JOE (AND) &0.961(0.021) & 0.989(0.005) & 0.499(0.174) & \textbf{0.641}(0.118) & 0.966(0.015) & 0.998(0.002) & 0.867(0.082) & \textbf{0.911}(0.053)\\ 
        \hline
        Nlasso-JOE (OR) &0.979(0.013) & 0.929(0.0270 & 0.170(0.152) & 0.267(0.177)& 0.982(0.012) & 0.986(0.012) & 0.453(0.077) & 0.615(0.091)\\
       \hline
       \hline
       DB-Glasso (Holm) &0.007(0.004) & 1.000(0.000) & NA & 0.013(0.008) & 0.028(0.011) & 1.000(0.000) & 1.000(0.000) & 0.054(0.021)\\
       \hline
       DB-Glasso (FDR)& 0.009(0.005) & 1.000(0.000) & NA & 0.017(0.010) & 0.040(0.011) & 1.000(0.000) & 1.000(0.000) & 0.078(0.021)\\
       \hline
       \textbf{GI-JOE (Holm)}&0.638(0.018) & 1.000(1.13e-05) & 1.000(0.002) & 0.779(0.013) & 0.680(0.009) & 1.000(0.000) & 1.000(0.000) & 0.810(0.006) \\ 
        \hline
        \textbf{GI-JOE (FDR)} &0.687(0.022) & 1.000(3.47e-05) & 0.997(0.005) & \textbf{0.813}(0.015) & 0.765(0.020) & 1.000(4.66e-05) & 0.996(0.006) & \textbf{0.865}(0.012)\\ 
        \hline
     \end{tabular}}
     \caption{Graph selection results on chain graph, measurement scenario 1. The highest F1 scores among estimation methods and among inference methods are both in bold.}
     \label{tab:sim2_summary_chain_mssc1}
 \end{table}

 \begin{table}[!ht]
     \centering
     \scalebox{0.6}{
     \begin{tabular}{|c|c|c|c|c|c|c|c|c|}
     \hline
     \multirow{2}{*}{Method}&\multicolumn{4}{|c|}{$n=20000$}&\multicolumn{4}{|c|}{$n=30000$}\\
     \cline{2-9}
        & TPR & TNR & TDR & F1 score  & TPR & TNR & TDR & F1 score  \\
        \hline
       Nlasso (AND) &0.997(0.004) & 0.971(0.001) & 0.258(0.010) & 0.410(0.012) & 0.999(0.002) & 0.961(0.002) & 0.205(0.007) & 0.340(0.010) \\ 
       \hline
       Nlasso (OR) &0.997(0.004) & 0.949(0.002) & 0.164(0.004) & 0.281(0.006) & 1.000(0.001) & 0.956(0.002) & 0.185(0.006) & 0.312(0.008) \\ 
       \hline
       Glasso & 0.999(0.002) & 0.954(0.002) & 0.180(0.005) & 0.305(0.007) & 0.999(0.002) & 0.969(0.019) & 0.328(0.194) & 0.466(0.207)\\
       \hline
       CLIME &0.996(0.004) & 0.947(0.001) & 0.158(0.003) & 0.273(0.005) & 0.999(0.002) & 0.975(0.001) & 0.283(0.008) & 0.441(0.010)\\
       \hline
       Nlasso-JOE (AND) &0.987(0.010) & 0.988(0.009) & 0.562(0.255) & \textbf{0.683}(0.204) & 0.996(0.004) & 0.997(3.39e-4) & 0.790(0.022) & \textbf{0.881}(0.014) \\ 
        \hline
        Nlasso-JOE (OR) &0.993(0.007) & 0.966(0.027) & 0.368(0.248) & 0.491(0.263) & 0.998(0.002) & 0.992(0.001) & 0.551(0.021) & 0.710(0.018)\\
       \hline
       \hline
       DB-Glasso (Holm) &0.036(0.013) & 1.000(0.000) & 1.000(0.000) & 0.070(0.023) & 0.263(0.027) & 1.000(0.000) & 1.000(0.000) & 0.416(0.034) \\ 
       \hline
       DB-Glasso (FDR)& 0.137(0.082) & 1.000(0.000) & 1.000(0.000) & 0.232(0.125) & 0.706(0.045) & 1.000(0.000) & 0.999(0.003) & 0.827(0.031)\\ 
       \hline
       \textbf{GI-JOE (Holm)}&0.237(0.090) & 1.000(0.000) & 1.000(0.000) & 0.375(0.116) & 0.504(0.024) & 1.000(0.000) & 1.000(0.000) & 0.670(0.021)\\ 
        \hline
        \textbf{GI-JOE (FDR)} &0.593(0.099) & 1.000(4.42e-05) & 0.993(0.008) & \textbf{0.738}(0.078) & 0.863(0.018) & 1.000(7.35e-05) & 0.987(0.008) & \textbf{0.921}(0.011)\\ 
        \hline
     \end{tabular}}
     \caption{Graph selection results on chain graph, measurement scenario 2. The highest F1 scores among estimation methods and among inference methods are both in bold.}
     \label{tab:sim2_summary_chain_mssc2}
 \end{table}
\begin{table}[!ht]
     \centering
     \scalebox{0.6}{
     \begin{tabular}{|c|c|c|c|c|c|c|c|c|}
     \hline
     \multirow{2}{*}{Method}&\multicolumn{4}{|c|}{$n=1500$}&\multicolumn{4}{|c|}{$n=3000$}\\
     \cline{2-9}
        & TPR & TNR & TDR & F1 score  & TPR & TNR & TDR & F1 score  \\
        \hline
       Nlasso (AND) &0.997(0.003) & 0.955(0.002) & 0.183(0.005) & 0.310(0.007) & 1.000(0.000) & 0.978(0.022) & 0.514(0.313) & 0.621(0.297) \\ 
       \hline
       Nlasso (OR) &0.997(0.003) & 0.965(0.001) & 0.221(0.005) & 0.362(0.007) & 1.000(0.000) & 0.978(0.013) & 0.376(0.145) & 0.530(0.158) \\ 
       \hline
       Glasso & 0.998(0.003) & 0.964(0.011) & 0.245(0.113) & 0.382(0.125) & 1.000(0.000) & 0.986(0.001) & 0.417(0.014) & 0.588(0.014)\\ 
       \hline
       CLIME &0.996(0.005) & 0.974(0.001) & 0.282(0.011) & 0.439(0.013) & 1.000(0.000) & 0.973(0.001) & 0.272(0.008) & 0.427(0.010)\\ 
       \hline
       Nlasso-JOE (AND) &0.991(0.008) & 0.986(0.007) & 0.462(0.203) & \textbf{0.608}(0.158) & 1.000(0.000) & 0.998(0.000) & 0.806(0.015) & \textbf{0.893}(0.009) \\ 
        \hline
        Nlasso-JOE (OR) &0.996(0.005) & 0.954(0.021) & 0.250(0.200) & 0.368(0.210) & 1.000(0.000) & 0.990(0.001) & 0.501(0.015) & 0.668(0.013) \\ 
       \hline
       \hline
       DB-Glasso (Holm) &0.009(0.007) & 1.000(0.000) & NA & 0.017(0.013) & 0.539(0.029) & 1.000(0.000) & 1.000(0.000) & 0.700(0.024) \\ 
       \hline
       DB-Glasso (FDR)& 0.020(0.011) & 1.000(0.000) & 1.000(0.000) & 0.038(0.020) & 0.925(0.018) & 1.000(0.000) & 0.999(0.002) & 0.960(0.010)\\ 
       \hline
       \textbf{GI-JOE (Holm)}&0.211(0.069) & 1.000(0.000) & 1.000(0.000) & 0.344(0.090) & 0.766(0.027) & 1.000(0.000) & 1.000(0.000) & 0.867(0.017)\\ 
        \hline
        \textbf{GI-JOE (FDR)} &0.591(0.096) & 1.000(0.000) & 0.997(0.004) & \textbf{0.738}(0.072) & 0.966(0.011) & 1.000(0.000) & 0.994(0.004) & \textbf{0.980}(0.005)\\ 
        \hline
     \end{tabular}}
     \caption{Graph selection results on chain graph, measurement scenario 3. The highest F1 scores among estimation methods and among inference methods are both in bold.}
     \label{tab:sim2_summary_chain_mssc3}
 \end{table}

\begin{table}[!ht]
     \centering
     \scalebox{0.6}{
     \begin{tabular}{|c|c|c|c|c|c|c|c|c|}
     \hline
     \multirow{2}{*}{Method}&\multicolumn{4}{|c|}{$n=600$}&\multicolumn{4}{|c|}{$n=800$}\\
     \cline{2-9}
        & TPR & TNR & TDR & F1 score  & TPR & TNR & TDR & F1 score  \\
        \hline
       Nlasso (AND) &0.778(0.041) & 0.969(0.016) & 0.235(0.099) & 0.348(0.111) & 0.833(0.027) & 0.974(0.013) & 0.266(0.078) & 0.396(0.095) \\ 
       \hline
       Nlasso (OR) &0.920(0.017) & 0.960(0.013) & 0.193(0.049) & 0.317(0.066) & 0.949(0.015) & 0.963(0.006) & 0.201(0.024) & 0.332(0.033) \\ 
       \hline
       Glasso & 0.993(0.006) & 0.981(0.001) & 0.329(0.014) & 0.494(0.016) & 0.997(0.004) & 0.979(0.001) & 0.312(0.012) & 0.475(0.014) \\
       \hline
       CLIME & 0.833(0.019) & 0.952(0.003) & 0.144(0.008) & 0.246(0.012) & 0.836(0.025) & 0.976(0.003) & 0.251(0.024) & 0.385(0.030)\\
       \hline
       Nlasso-JOE (AND) &0.759(0.031) & 0.999(0.003) & 0.897(0.126) & \textbf{0.815}(0.073) & 0.821(0.017) & 0.999(0.000) & 0.913(0.035) & \textbf{0.864}(0.021) \\ 
        \hline
        Nlasso-JOE (OR) &0.972(0.017) & 0.991(0.014) & 0.603(0.114) & 0.736(0.122) & 0.989(0.007) & 0.993(0.001) & 0.584(0.026) & 0.734(0.021) \\ 
       \hline
       \hline
       DB-Glasso (Holm) &0.001(0.002) & 1.000(0.000) & NA & 0.003(0.005) & 0.002(0.003) & 1.000(0.000) & NA & 0.004(0.005) \\ 
       \hline
       DB-Glasso (FDR)& 0.001(0.002) & 1.000(0.000) & NA & 0.003(0.005) & 0.004(0.003) & 1.000(0.000) & NA & 0.007(0.007)\\
       \hline
       \textbf{GI-JOE (Holm)}&0.193(0.027) & 1.000(0.000) & 1.000(0.000) & 0.323(0.041) & 0.330(0.013) & 1.000(0.000) & 1.000(0.000) & 0.496(0.014) \\ 
        \hline
        \textbf{GI-JOE (FDR)} &0.341(0.028) & 1.000(3.76e-05) & 0.987(0.011) & \textbf{0.506}(0.032) & 0.443(0.017) & 1.000(6.40e-05) & 0.985(0.014) & \textbf{0.61}1(0.015) \\ 
        \hline
     \end{tabular}}
     \caption{Graph selection results on star graph, measurement scenario 1. The highest F1 scores among estimation methods and among inference methods are both in bold.}
     \label{tab:sim2_summary_star_mssc1}
 \end{table}
  \begin{table}[!ht]
     \centering
     \scalebox{0.6}{
     \begin{tabular}{|c|c|c|c|c|c|c|c|c|}
     \hline
     \multirow{2}{*}{Method}&\multicolumn{4}{|c|}{$n=20000$}&\multicolumn{4}{|c|}{$n=30000$}\\
     \cline{2-9}
        & TPR & TNR & TDR & F1 score  & TPR & TNR & TDR & F1 score  \\
        \hline
       Nlasso (AND) &0.923(0.020) & 0.954(0.001) & 0.162(0.005) & 0.275(0.007) & 0.954(0.012) & 0.956(0.001) & 0.174(0.005) & 0.294(0.007) \\ 
       \hline
       Nlasso (OR) &0.999(0.002) & 0.952(0.001) & 0.166(0.003) & 0.284(0.004) & 1.000(0.000) & 0.957(0.001) & 0.184(0.004) & 0.311(0.006) \\ 
       \hline
       Glasso &1.000(0.000) & 0.965(0.001) & 0.213(0.004) & 0.351(0.006) & 1.000(0.000) & 0.960(0.017) & 0.285(0.281) & 0.392(0.248)\\
       \hline
       CLIME & 0.898(0.026) & 0.956(0.007) & 0.168(0.021) & 0.282(0.030) & 0.942(0.014) & 0.976(0.002) & 0.275(0.022) & 0.425(0.024)\\
       \hline
       Nlasso-JOE (AND) &0.659(0.020) & 1.000(0.000) & 1.000(0.002) & 0.794(0.015) & 0.821(0.023) & 1.000(0.000) & 1.000(0.001) & 0.901(0.014) \\ 
        \hline
        Nlasso-JOE (OR) &1.000(0.000) & 1.000(0.000) & 0.962(0.008) & \textbf{0.981}(0.004) & 1.000(0.000) & 0.999(0.000) & 0.946(0.018) & \textbf{0.972}(0.010) \\ 
       \hline
       \hline
       DB-Glasso (Holm) &0.001(0.004) & 1.000(1.13e-05) & NA & 0.003(0.007) & 0.018(0.007) & 1.000(1.13e-05) & 0.988(0.056) & 0.035(0.014)\\
       \hline
       DB-Glasso (FDR)& 0.003(0.005) & 1.000(1.13e-05) & NA & 0.006(0.010) & 0.032(0.011) & 1.000(1.13e-05) & 0.983(0.056) & 0.061(0.021)\\
       \hline
       \textbf{GI-JOE (Holm)}&0.917(0.018) & 1.000(1.11e-04) & 1.000(0.001) & 0.957(0.010) & 0.938(0.017) & 1.000(0.000) & 1.000(0.000) & 0.968(0.009) \\ 
        \hline
        \textbf{GI-JOE (FDR)} &0.990(0.008) & 1.000(0.000) & 0.981(0.011) & \textbf{0.986}(0.007) & 0.993(0.005) & 1.000(1.10e-04) & 0.977(0.011) & \textbf{0.985}(0.006)\\
        \hline
     \end{tabular}}
     \caption{Graph selection results on star graph, measurement scenario 2. The highest F1 scores among estimation methods and among inference methods are both in bold.}
     \label{tab:sim2_summary_star_mssc2}
 \end{table}
 \begin{table}[!ht]
     \centering
     \scalebox{0.6}{
     \begin{tabular}{|c|c|c|c|c|c|c|c|c|}
     \hline
     \multirow{2}{*}{Method}&\multicolumn{4}{|c|}{$n=1500$}&\multicolumn{4}{|c|}{$n=3000$}\\
     \cline{2-9}
        & TPR & TNR & TDR & F1 score  & TPR & TNR & TDR & F1 score  \\
        \hline
       Nlasso (AND) &0.867(0.020) & 0.951(0.001) & 0.144(0.005) & 0.247(0.008) & 0.948(0.017) & 0.950(0.001) & 0.153(0.004) & 0.263(0.006) \\ 
       \hline
       Nlasso (OR) &0.975(0.012) & 0.954(0.001) & 0.170(0.004) & 0.289(0.006) & 1.000(0.001) & 0.957(0.009) & 0.187(0.027) & 0.314(0.040) \\ 
       \hline
       Glasso & 1.000(0.001) & 0.967(0.017) & 0.274(0.131) & 0.414(0.156) & 1.000(0.000)& 0.980(0.001) & 0.327(0.013) & 0.493(0.015)\\
       \hline
       CLIME & 0.290(0.025) & 0.983(0.001) & 0.140(0.010) & 0.189(0.014) & 0.888(0.023) & 0.951(0.001) & 0.147(0.004) & 0.253(0.006) \\
       \hline
       Nlasso-JOE (AND) &0.594(0.152) & 1.000(0.000) & 0.973(0.017) & \textbf{0.723}(0.157) & 0.535(0.031) & 1.000(0.000) & 1.000(0.000) & 0.697(0.026) \\ 
        \hline
        Nlasso-JOE (OR) &0.952(0.003) & 0.988(0.004) & 0.465(0.170) & 0.611(0.118) & 1.000(0.000) & 0.999(1.56e-4) & 0.942(0.015) & \textbf{0.970}(0.008)\\ 
       \hline
       \hline
       DB-Glasso (Holm) &0.000(0.000) & 1.000(0.000) & NA & 0.000(0.000) & 0.000(0.000) & 1.000(0.000) & NA & 0.000(0.000) \\ 
       \hline
       DB-Glasso (FDR)& 0.000(0.000) & 1.000(0.000) & NA & 0.000(0.000) & 0.000(0.000) & 1.000(0.000) & NA & 0.000(0.000) \\ 
       \hline
       \textbf{GI-JOE (Holm)}&0.460(0.131) & 1.000(0.000) & 1.000(0.000) & 0.622(0.102) & 0.994(0.005) & 1.000(0.000) & 1.000(0.000) & \textbf{0.997}(0.003) \\ 
        \hline
        \textbf{GI-JOE (FDR)} &0.774(0.064) & 1.000(3.76e-05) & 0.994(0.005)& \textbf{0.869}(0.039) & 1.000(0.001) & 1.000(7.94e-05) & 0.987(0.008) & 0.993(0.004) \\ 
        \hline
     \end{tabular}}
     \caption{Graph selection results on star graph, measurement scenario 3. The highest F1 scores among estimation methods and among inference methods are both in bold.}
     \label{tab:sim2_summary_star_mssc3}
 \end{table}
\begin{table}[!ht]
\centering
\centering
     \scalebox{0.6}{
     \begin{tabular}{|c|c|c|c|c|c|c|c|c|}
     \hline
     \multirow{2}{*}{Method}&\multicolumn{4}{|c|}{$n=1500$}&\multicolumn{4}{|c|}{$n=3000$}\\
     \cline{2-9}
       & TPR & TNR & TDR & F1 score  & TPR & TNR & TDR & F1 score  \\
        \hline
    Nlasso (AND) & 0.905(0.019) & 0.960(0.004) & 0.268(0.017) & 0.413(0.022) & 0.943(0.016) & 0.960(0.011) & 0.289(0.084) & 0.436(0.085) \\ 
    \hline
    Nlasso (OR) & 0.924(0.016) & 0.960(0.006) & 0.271(0.026) & 0.418(0.032) & 0.960(0.010) & 0.958(0.005) & 0.268(0.024) & 0.419(0.028) \\ 
    \hline
    Glasso & 0.895(0.029) & 0.968(0.007) & 0.316(0.032) & 0.466(0.037) & 0.951(0.010) & 0.967(0.001) & 0.316(0.007) & 0.474(0.009) \\ 
    \hline
    CLIME & 0.869(0.034) & 0.968(0.019) & 0.347(0.115) & 0.483(0.120) & 0.932(0.011) & 0.967(0.004) & 0.314(0.030) & 0.469(0.032) \\ 
    \hline
    Nlasso-JOE (AND) & 0.385(0.095) & 0.997(0.005) & 0.846(0.121) & \textbf{0.514}(0.038) & 0.431(0.014) & 0.999(2.85e-4) & 0.886(0.020) & \textbf{0.580}(0.014) \\ 
    \hline
    Nlasso-JOE (OR) & 0.526(0.086) & 0.982(0.019) & 0.503(0.102) & 0.499(0.063) & 0.548(0.008) & 0.987(0.001) & 0.521(0.014) & 0.534(0.009) \\ 
    \hline
    \hline
    DB-Glasso (Holm) & 0.000(0.000) & 1.000(0.000) & NA & 0.000(0.000) & 4.75e-4(0.001) & 1.000(0.000) & NA & 0.001(0.002) \\ 
    \hline
    DB-Glasso (FDR) & 0.000(0.000) & 1.000(0.000) & NA & 0.000(0.000) & 0.001(0.001) & 1.000(0.000) & NA & 0.002(0.003) \\ 
    \hline
    \textbf{GI-JOE (Holm)} & 0.180(0.025) & 1.000(0.000) & 1.000(0.000) & 0.305(0.038) & 0.249(0.010) & 1.000(0.000) & 1.000(0.002) & 0.399(0.012) \\ 
    \hline
    \textbf{GI-JOE (FDR)} & 0.316(0.017) & 1.000(1.04e-4) & 0.975(0.012) & \textbf{0.477}(0.020) & 0.378(0.013) & 1.000(9.31e-5) & 0.976(0.009) & \textbf{0.545}(0.013) \\ 
   \hline
\end{tabular}}
     \caption{Graph selection results on Erd\H{o}s–R\'{e}nyi graph, measurement scenario 1. The highest F1 scores among estimation methods and among inference methods are both in bold.}
     \label{tab:sim2_summary_ER_mssc1}
 \end{table}

\begin{table}[!ht]
\centering
\centering
     \scalebox{0.6}{
     \begin{tabular}{|c|c|c|c|c|c|c|c|c|}
     \hline
     \multirow{2}{*}{Method}&\multicolumn{4}{|c|}{$n=1500$}&\multicolumn{4}{|c|}{$n=3000$}\\
     \cline{2-9}
       & TPR & TNR & TDR & F1 score  & TPR & TNR & TDR & F1 score  \\
        \hline
    Nlasso (AND) & 0.924(0.013) & 0.954(0.003) & 0.244(0.014) & 0.386(0.017) & 0.969(0.008) & 0.955(0.001) & 0.257(0.007) & 0.407(0.009) \\ 
    \hline
    Nlasso (OR) & 0.931(0.013) & 0.960(0.002) & 0.271(0.009) & 0.420(0.011) & 0.974(0.007) & 0.963(0.001) & 0.298(0.005) & 0.456(0.006) \\ 
    \hline
    Glasso & 0.942(0.009) & 0.954(0.001) & 0.246(0.006) & 0.391(0.008) & 0.952(0.039) & 0.971(0.015) & 0.407(0.190) & 0.543(0.159) \\ 
    \hline
    CLIME & 0.902(0.014) & 0.959(0.001) & 0.261(0.007) & 0.405(0.008) & 0.940(0.011) & 0.979(0.001) & 0.420(0.013) & 0.580(0.013) \\ 
    \hline
    Nlasso-JOE (AND) & 0.405(0.092) & 0.992(0.010) & 0.755(0.266) & \textbf{0.485}(0.026) & 0.428(0.010) & 0.999(0.000) & 0.946(0.011) & 0.589(0.011) \\ 
    \hline
    Nlasso-JOE (OR) & 0.507(0.064) & 0.975(0.030) & 0.568(0.279) & 0.479(0.138) & 0.538(0.006) & 0.995(0.000) & 0.734(0.016) & \textbf{0.621}(0.008) \\ 
    \hline
    \hline
    DB-Glasso (Holm) & 0.000(0.000) & 1.000(0.000) & NA & 0.000(0.000) & 0.000(0.000) & 1.000(0.000) & NA & 0.000(0.000) \\ 
    \hline
    DB-Glasso (FDR) & 0.000(0.000) & 1.000(0.000) & NA & 0.000(0.000) & 0.000(0.000) & 1.000(0.000) & NA & 0.000(0.000)  \\ 
    \hline
    \textbf{GI-JOE (Holm)} & 0.091(0.044) & 1.000(0.000) & 1.000(0.000) & 0.165(0.076) & 0.212(0.011) & 1.000(0.000) & 1.000(0.002) & 0.349(0.015) \\ 
    \hline
    \textbf{GI-JOE (FDR)} & 0.269(0.095) & 1.000(0.000) & 0.974(0.012) & \textbf{0.412}(0.123) & 0.438(0.012) & 1.000(0.000) & 0.971(0.011) & \textbf{0.604}(0.012) \\ 
   \hline
\end{tabular}}
     \caption{Graph selection results on Erd\H{o}s–R\'{e}nyi graph, measurement scenario 2. The highest F1 scores among estimation methods and among inference methods are both in bold.}
     \label{tab:sim2_summary_ER_mssc2}
 \end{table}
 
\begin{table}[!ht]
\centering
\centering
     \scalebox{0.6}{
     \begin{tabular}{|c|c|c|c|c|c|c|c|c|}
     \hline
     \multirow{2}{*}{Method}&\multicolumn{4}{|c|}{$n=1500$}&\multicolumn{4}{|c|}{$n=3000$}\\
     \cline{2-9}
       & TPR & TNR & TDR & F1 score  & TPR & TNR & TDR & F1 score  \\
        \hline
    Nlasso (AND) & 0.980(0.008) & 0.954(0.001) & 0.256(0.005) & 0.406(0.007) & 0.998(0.002) & 0.954(0.001) & 0.259(0.005) & 0.411(0.006) \\ 
    \hline
    Nlasso (OR) & 0.980(0.007) & 0.960(0.001) & 0.282(0.005) & 0.438(0.006) & 0.998(0.003) & 0.962(0.001) & 0.298(0.006) & 0.459(0.007) \\ 
    \hline
    Glasso & 0.975(0.028) & 0.956(0.010) & 0.279(0.113) & 0.423(0.091) & 0.982(0.005) & 0.994(0.000) & 0.725(0.016) & \textbf{0.834}(0.010) \\  
    \hline
    CLIME & 0.972(0.009) & 0.960(0.001) & 0.279(0.005) & 0.434(0.007) & 0.997(0.004) & 0.974(0.003) & 0.381(0.027) & 0.551(0.030)\\ 
    \hline
    Nlasso-JOE (AND) & 0.456(0.012) & 1.000(0.000) & 0.969(0.010) & 0.620(0.011) & 0.570(0.005) & 1.000(0.000) & 0.971(0.007) & 0.718(0.004) \\ 
    \hline
    Nlasso-JOE (OR) & 0.551(0.007) & 0.995(0.000) & 0.736(0.016) & \textbf{0.630}(0.008) & 0.606(0.005) & 0.993(0.001) & 0.692(0.016) & 0.646(0.007) \\ 
    \hline
    \hline
    DB-Glasso (Holm) & 0.000(0.000) & 1.000(0.000) & NA & 0.000(0.000) & 0.000(0.000) & 1.000(0.000) & NA & 0.000(0.000) \\
    \hline
    DB-Glasso (FDR) & 0.000(0.001) & 1.000(0.000) & NA & 0.000(0.001) & 0.000(0.001) & 1.000(0.000) & NA & 0.000(0.001)  \\ 
    \hline
    \textbf{GI-JOE (Holm)} & 0.334(0.012) & 1.000(0.000) & 1.000(0.000) & 0.500(0.013) & 0.516(0.005) & 1.000(0.000) & 1.000(0.001) & 0.681(0.005) \\ 
    \hline
    \textbf{GI-JOE (FDR)} & 0.509(0.008) & 1.000(0.000) & 0.976(0.008) & \textbf{0.669}(0.008) & 0.588(0.006) & 1.000(0.000) & 0.977(0.007) & \textbf{0.734}(0.005) \\ 
   \hline
\end{tabular}}
     \caption{Graph selection results on Erd\H{o}s–R\'{e}nyi graph, measurement scenario 3. The highest F1 scores among estimation methods and among inference methods are both in bold.}
     \label{tab:sim2_summary_ER_mssc3}
 \end{table}
 
\begin{table}[!ht]
     \centering
     \scalebox{0.6}{
     \begin{tabular}{|c|c|c|c|c|c|c|c|c|}
     \hline
     \multirow{2}{*}{Method}&\multicolumn{4}{|c|}{$n=8000$}&\multicolumn{4}{|c|}{$n=12000$}\\
     \cline{2-9}
        & TPR & TNR & TDR & F1 score  & TPR & TNR & TDR & F1 score  \\
        \hline
       Nlasso (AND) &0.550(0.028) & 0.998(0.000) & 0.611(0.033) & 0.579(0.027) & 0.713(0.028) & 0.998(0.000) & 0.662(0.026) & 0.686(0.021)\\
       \hline
       Nlasso (OR) &0.623(0.030) & 0.998(0.000) & 0.601(0.027) & 0.612(0.025) & 0.782(0.024) & 0.998(0.000) & 0.648(0.026) & 0.709(0.021)\\
       \hline
       Glasso &0.736(0.024) & 0.994(0.000) & 0.427(0.019)&0.540(0.020) & 0.864(0.015) & 0.994(0.000) & 0.461(0.019) & 0.601(0.016)\\
       \hline
       CLIME & 0.931(0.020) & 0.954(0.001) & 0.104(0.003) & 0.187(0.006) & 0.684(0.034) & 0.999(0.000) & 0.757(0.035) & 0.718(0.029)\\
       \hline
       Nlasso-JOE (AND) &0.482(0.029) & 1.000(0.000) & 0.959(0.017) & 0.641(0.026) & 0.568(0.166) & 1.000(0.000) & 0.974(0.017) & 0.700(0.157) \\
        \hline
        Nlasso-JOE (OR) &0.693(0.028) & 0.999(0.000) & 0.859(0.030) & \textbf{0.767}(0.023) & 0.746(0.187) & 0.999(0.000) & 0.894(0.059) & \textbf{0.792}(0.123)\\
       \hline
       \hline
       DB-Glasso (Holm) &0.126(0.010) & 1.000(0.000) & 1.000(0.000) & 0.224(0.016) & 0.178(0.012) & 1.000(0.000) & 1.000(0.000) & 0.302(0.017)\\
       \hline
       DB-Glasso (FDR)& 0.142(0.009) & 1.000(0.000) & 0.998(0.010) & 0.249(0.014) & 0.247(0.046) & 1.000(0.000) & 0.999(0.006) & 0.394(0.059)\\
       \hline
       \textbf{GI-JOE (Holm)}&0.550(0.023) & 1.000(0.000) & 1.000(0.000) & 0.709(0.020) & 0.693(0.027) & 1.000(0.000) & 0.998(0.004) & 0.818(0.018)\\
        \hline
        \textbf{GI-JOE (FDR)} &0.729(0.030) & 1.000(0.000) & 0.950(0.023) & \textbf{0.825}(0.023) & 0.847(0.022) & 1.000(0.000) & 0.952(0.022) & \textbf{0.896}(0.017)\\
        \hline
     \end{tabular}}
     \caption{Graph selection results with the ground truth graph being the neuronal functional network estimated from a real calcium imaging data \cite{lein2007genome}, measurement scenario 1.}
     \label{tab:ABA_sim_summary_mssc1}
 \end{table}
\begin{table}[!ht]
     \centering
     \scalebox{0.6}{
     \begin{tabular}{|c|c|c|c|c|c|c|c|c|}
     \hline
     \multirow{2}{*}{Method}&\multicolumn{4}{|c|}{$n=80000$}&\multicolumn{4}{|c|}{$n=120000$}\\
     \cline{2-9}
        & TPR & TNR & TDR & F1 score  & TPR & TNR & TDR & F1 score  \\
        \hline
       Nlasso (AND) &0.440(0.018) & 1.000(0.000) & 0.846(0.034) & 0.579(0.016) & 0.573(0.025) & 1.000(0.000) & 0.879(0.029) & 0.694(0.022) \\
       \hline
       Nlasso (OR) &0.508(0.017) & 0.999(0.000) & 0.850(0.035) & 0.635(0.015) & 0.654(0.026) & 1.000(0.000) & 0.884(0.027) & 0.751(0.021)\\
       \hline
       Glasso &0.620(0.035) & 0.999(0.000) & 0.714(0.034) & 0.663(0.029) & 0.765(0.034) & 0.999(0.000) & 0.759(0.041) & 0.762(0.034)\\
       \hline
       CLIME &0.891(0.031) & 0.987(0.000) & 0.284(0.011) & 0.430(0.015) & 0.959(0.015) & 0.986(0.000) & 0.288(0.007) & 0.443(0.009)\\
       \hline
       Nlasso-JOE (AND) &0.513(0.026) & 1.000(0.000) & 0.974(0.015) & 0.672(0.022) & 0.639(0.030) & 1.000(0.000) & 0.982(0.009) & 0.774(0.022)\\
        \hline
        Nlasso-JOE (OR) &0.735(0.035) & 1.000(0.000) & 0.924(0.015) & \textbf{0.818}(0.023) & 0.851(0.027) & 1.000(0.000) & 0.920(0.022) & \textbf{0.884}(0.022)\\
       \hline
       \hline
       DB-Glasso (Holm) &0.127(0.012) & 1.000(0.000) & 1.000(0.000) & 0.225(0.019) & 0.175(0.010) & 1.000(0.000) & 1.000(0.000) & 0.298(0.014)\\
       \hline
       DB-Glasso (FDR)& 0.144(0.010) & 1.000(0.000) & 1.000(0.000) & 0.251(0.016) & 0.260(0.043) & 1.000(0.000) & 0.999(0.005) & 0.411(0.055)\\
       \hline
       \textbf{GI-JOE (Holm)}&0.590(0.030) & 1.000(0.000) & 0.998(0.003) & 0.741(0.024) & 0.708(0.023) & 1.000(0.000) & 0.997(0.002) & 0.828(0.016)\\
        \hline
        \textbf{GI-JOE (FDR)} &0.767(0.037) & 1.000(0.000) & 0.953(0.018) & \textbf{0.849}(0.025) & 0.853(0.025) & 1.000(0.000) & 0.949(0.017) & \textbf{0.898}(0.016)\\
        \hline
     \end{tabular}}
     \caption{Graph selection results with the ground truth graph being the neuronal functional network estimated from a real calcium imaging data \cite{lein2007genome}, measurement scenario 2.}
     \label{tab:ABA_sim_summary_mssc2}
 \end{table}
 \begin{table}[!ht]
     \centering
     \scalebox{0.6}{
     \begin{tabular}{|c|c|c|c|c|c|c|c|c|}
     \hline
     \multirow{2}{*}{Method}&\multicolumn{4}{|c|}{$n=5000$}&\multicolumn{4}{|c|}{$n=8000$}\\
     \cline{2-9}
        & TPR & TNR & TDR & F1 score  & TPR & TNR & TDR & F1 score  \\
        \hline
       Nlasso (AND) &0.468(0.039) & 0.988(0.001) & 0.182(0.016) & 0.262(0.022) & 0.586(0.033) & 0.988(0.001) & 0.219(0.019) & 0.318 (0.021) \\
       \hline
       Nlasso (OR) &0.466(0.034) & 0.990(0.001) & 0.202(0.017) & 0.282(0.023) & 0.591(0.029) & 0.990(0.001) & 0.244(0.019) & 0.346(0.023)\\
       \hline
       Glasso &0.492(0.036) & 0.989(0.001) & 0.210(0.018) & 0.295(0.023) & 0.613(0.031) & 0.989(0.001) & 0.250(0.020) & 0.355(0.024)\\
       \hline
       CLIME &0.370(0.026) & 0.993(0.000) & 0.242(0.021) & 0.292(0.022) & 0.537(0.024) & 0.992(0.001) & 0.284(0.022) & 0.371(0.022)\\
       \hline
       Nlasso-JOE (AND) &0.523(0.043) & 0.994(0.001) & 0.353(0.042) & \textbf{0.419}(0.029) & 0.315(0.069) & 1.000(0.001) & 0.935(0.099) & 0.461(0.027)\\
        \hline
        Nlasso-JOE (OR) &0.719(0.047) & 0.982(0.002) & 0.183(0.020) & 0.291(0.024) & 0.465(0.073) & 0.999(0.003) & 0.855(0.146) & \textbf{0.585}(0.051)\\
       \hline
       \hline
       DB-Glasso (Holm) &0.027(0.004) & 1.000(0.000) & 1.000(0.000) & 0.053(0.007) & 0.049(0.005) & 1.000(0.000) & 1.000(0.000) & 0.094(0.009)\\
       \hline
       DB-Glasso (FDR)& 0.029(0.003) & 1.000(0.000) & 1.000(0.000) & 0.056(0.005) & 0.056(0.008) & 1.000(0.000) & 1.000(0.000) & 0.107(0.014)\\
       \hline
       \textbf{GI-JOE (Holm)}&0.273(0.020) & 1.000(0.000) & 1.000(0.000) & 0.428(0.024) & 0.420(0.022) & 1.000(0.000) & 0.999(0.004) & 0.591(0.022)\\
        \hline
        \textbf{GI-JOE (FDR)} &0.396(0.029) & 1.000(0.000) & 0.965(0.028) & \textbf{0.561}(0.028) & 0.571(0.032) & 1.000(0.000) & 0.957(0.018) & \textbf{0.715}(0.026)\\
        \hline
     \end{tabular}}
     \caption{Graph selection results with the ground truth graph being the neuronal functional network estimated from a real calcium imaging data \cite{lein2007genome}, measurement scenario 3.}
     \label{tab:ABA_sim_summary_mssc3}
 \end{table}

 \begin{table}[!ht]
     \centering
     \scalebox{0.6}{
     \begin{tabular}{|c|c|c|c|c|c|c|c|c|}
     \hline
     \multirow{2}{*}{Method}&\multicolumn{4}{|c|}{\emph{chu} measurement}&\multicolumn{4}{|c|}{\emph{darmanis} measurement}\\
     \cline{2-9}
        & TPR & TNR & TDR & F1 score  & TPR & TNR & TDR & F1 score  \\
        \hline
       Nlasso (AND) &0.996(0.004) & 0.974(0.020) & 0.408(0.252) & 0.536(0.252) & 0.899(0.026) & 0.961(0.005) & 0.189(0.018) & 0.312(0.024)\\
       \hline
       Nlasso (OR) &0.997(0.003) & 0.970(0.021) & 0.343(0.187) & 0.483(0.206) & 0.912(0.022) & 0.959(0.002) & 0.185(0.009) & 0.307(0.014)\\
       \hline
       Glasso &0.551(0.314) & 0.999(0.001) & 0.823(0.051) & 0.606(0.189) & 0.529(0.036) & 0.997(0.001) & 0.672(0.046) & 0.591(0.031)\\
       \hline
       CLIME &0.998(0.003) & 0.969(0.011) & 0.260(0.051) & 0.410(0.067) & 0.908(0.028) & 0.947(0.008) & 0.152(0.035) & 0.259(0.043)\\
       \hline
       Nlasso-JOE (AND) &0.509(0.342) & 1.000(0.000) & 1.000(0.000) & 0.610(0.301) & 0.343(0.026) & 1.000(0.000) & 0.994(0.015) & 0.510(0.029)\\
        \hline
        Nlasso-JOE (OR) &0.593(0.327) & 0.999(0.000) & 0.915(0.018) & \textbf{0.670}(0.242) & 0.514(0.035) & 0.999(0.000) & 0.805(0.048) & \textbf{0.626}(0.033)\\
       \hline
       \hline
       DB-Glasso (Holm) &0.000(0.000) & 1.000(0.000) & NA & 0.000(0.000) & 0.000(0.000) & 1.000(0.000) & NA & 0.000(0.000)\\
       \hline
       DB-Glasso (FDR)&0.000(0.000) & 1.000(0.000) & NA & 0.000(0.000) & 0.000(0.000) & 1.000(0.000) & NA & 0.000(0.000)\\
       \hline
       \textbf{GI-JOE (Holm)}&0.927(0.011) & 1.000(0.000) & 0.997(0.005) & \textbf{0.960}(0.006) & 0.295(0.034) & 1.000(0.000) & 0.999(0.003) & 0.455(0.040)\\
        \hline
        \textbf{GI-JOE (FDR)} &0.963(0.006) & 0.999(0.001) & 0.911(0.050) & 0.936(0.025) & 0.603(0.040) & 1.000(0.000) & 0.978(0.015) & \textbf{0.745}(0.033)\\
        \hline
     \end{tabular}}
     \caption{Graph selection results with the ground truth graph being a scale-free graph with 200 nodes, under two real measurement patterns from single-cell RNA sequencing data sets (the chu data \citep{chu2016single} and darmanis data \citep{darmanis2015survey}).}
     \label{tab:scRNA_sim_summary_scale-free}
 \end{table}
  \begin{table}[!htbp]
     \centering
     \scalebox{0.6}{
     \begin{tabular}{|c|c|c|c|c|c|c|c|c|}
     \hline
     \multirow{2}{*}{Method}&\multicolumn{4}{|c|}{\emph{chu} measurement}&\multicolumn{4}{|c|}{\emph{darmanis} measurement}\\
     \cline{2-9}
        & TPR & TNR & TDR & F1 score  & TPR & TNR & TDR & F1 score  \\
        \hline
       Nlasso (AND) &0.997(0.003) & 0.975(0.022) & 0.457(0.288) & 0.574(0.281) & 0.953(0.018) & 0.960(0.008) & 0.213(0.113) & 0.339(0.114)\\
       \hline
       Nlasso (OR) &0.999(0.003) & 0.973(0.020) & 0.372(0.199) & 0.512(0.216) & 0.965(0.015) & 0.964(0.007) & 0.222(0.072) & 0.356(0.080)\\
       \hline
       Glasso &0.983(0.021) & 0.990(0.007) & 0.579(0.234) & 0.701(0.172) & 0.936(0.017) & 0.980(0.004) & 0.341(0.120) & 0.492(0.098)\\
       \hline
       CLIME &0.999(0.002) & 0.969(0.002) & 0.247(0.013) & 0.395(0.016) & 0.947(0.018) & 0.958(0.018) & 0.218(0.094) & 0.345(0.118)\\
       \hline
       Nlasso-JOE (AND) &0.912(0.043) & 0.999(0.000) & 0.939(0.043) & \textbf{0.924}(0.008) & 0.757(0.020) & 0.999(0.000) & 0.881(0.019) & \textbf{0.814}(0.014)\\
        \hline
        Nlasso-JOE (OR) &0.935(0.023) & 0.996(0.002) & 0.746(0.152) & 0.820(0.078) & 0.860(0.016) & 0.992(0.001) & 0.518(0.024) & 0.646(0.018)\\
       \hline
       \hline
       DB-Glasso (Holm) &0.000(0.000) & 1.000(0.000) & NA & 0.000(0.000) & 0.000(0.000) & 1.000(0.000) & NA & 0.000(0.000)\\
       \hline
       DB-Glasso (FDR)&0.000(0.000) & 1.000(0.000) & NA & 0.000(0.000) & 0.000(0.000) & 1.000(0.000) & NA & 0.000(0.000)\\
       \hline
       \textbf{GI-JOE (Holm)}&0.883(0.014) & 1.000(0.000) & 0.988(0.009) & \textbf{0.933}(0.009) & 0.382(0.025) & 1.000(0.000) & 1.000(0.000) & 0.553(0.027)\\
        \hline
        \textbf{GI-JOE (FDR)} &0.928(0.011) & 0.999(0.000) & 0.933(0.027) & 0.930(0.010) & 0.626(0.025) & 1.000(0.000) & 0.983(0.012) & \textbf{0.764}(0.019)\\
        \hline
     \end{tabular}}
     \caption{Graph selection results with the ground truth graph being a small-world graph with 200 nodes, under two real measurement patterns from single-cell RNA sequencing data sets (the chu data \citep{chu2016single} and darmanis data \citep{darmanis2015survey}).}
     \label{tab:scRNA_sim_summary_small-world}
 \end{table}
 \begin{table}[!ht]
     \centering
     \scalebox{0.8}{
     \begin{tabular}{|c|c|c|c|c|c|c|c|c|c|c|c|c|}
     \hline
     \multirow{2}{*}{Method} & \multicolumn{4}{|c|}{Neuron set 1}&\multicolumn{4}{|c|}{Neuron set 2}&\multicolumn{4}{|c|}{Neuron set 3}\\
     \cline{2-13}
     &TPR & TNR & TDR & F1 &TPR & TNR & TDR & F1 &TPR & TNR & TDR & F1\\ 
  \hline
    GI-JOE(FDR) & 0.608 & 0.998 & 0.936 & 0.737 & 0.557 & 0.999 & 0.917 & 0.693 & 0.396 & 1.000 & 0.950 & 0.559\\
\hline
    GI-JOE(Holm) & 0.483 & 1.000 & 0.983 & 0.648 & 0.380 & 1.000 & 1.000 & 0.550 & 0.312 & 1.000 & 1.000 & 0.476 \\ 
    \hline
    DB-Glasso(FDR) & 0.142 & 1.000 & 1.000 & 0.248 & 0.177 & 1.000 & 0.933 & 0.298 & 0.333 & 1.000 & 0.941 & 0.492 \\ 
    \hline
    DB-Glasso(Holm) & 0.133 & 1.000 & 1.000 & 0.235 & 0.139 & 1.000 & 1.000 & 0.244 & 0.250 & 1.000 & 1.000 & 0.400\\ 
   \hline
     \end{tabular}}
     \caption{Comparison of the tested sub-graph on a real calcium imaging data set, when the sub-graph only consists of neurons in one of the three sets. Neuron set 1, 2, 3 are observed with high, median and low probabilities, respectively. For neuron set 3, all methods don't work well due to the small sample sizes, while for neuron sets 1 and 2, GI-JOE approaches have much higher true positive rate and F1-scores than the debiased graphical lasso with minimum sample size.}
     \label{tab:ABA_real_subgraph}
 \end{table}
  \begin{table}[!ht]
     \centering
     \begin{tabular}{|c|c|c|c|c|}
     \hline
     Method & TPR & TNR & TDR & F1\\ 
  \hline
    GI-JOE(FDR) & 0.438 & 0.999 & 0.920 & 0.594\\ 
\hline
    GI-JOE(Holm) & 0.339 & 1.000 & 0.974 & 0.492\\ 
    \hline
    DB-Glasso(FDR)& 0.150 & 1.000 & 0.962 & 0.260\\ 
    \hline
    DB-Glasso(Holm) & 0.124 & 1.000 & 1.000 & 0.220\\ 
   \hline
     \end{tabular}
     \caption{Comparison of the tested full graph on a real calcium imaging data set.}
     \label{tab:ABA_real_full_graph}
 \end{table}

\section{The Graphs and Measurement Patterns in Real Data-inspired Simulations}\label{append:RealSim_Setting}
In Section 5.3 of the main paper, we present simulation results (i) when the underlying graph for data generation is estimated from a real calcium imaging data set; and (ii) when the measurement patterns are the same as the ones in two single-cell RNA sequencing data sets. We present here the estimated graph structure from neuroscience and the measurement patterns from gene expression data in Figure \ref{fig:real_graph} and Figure \ref{fig:real_Obs}.
\begin{figure}[!ht]
    \centering
    \includegraphics[width = 0.5\textwidth]{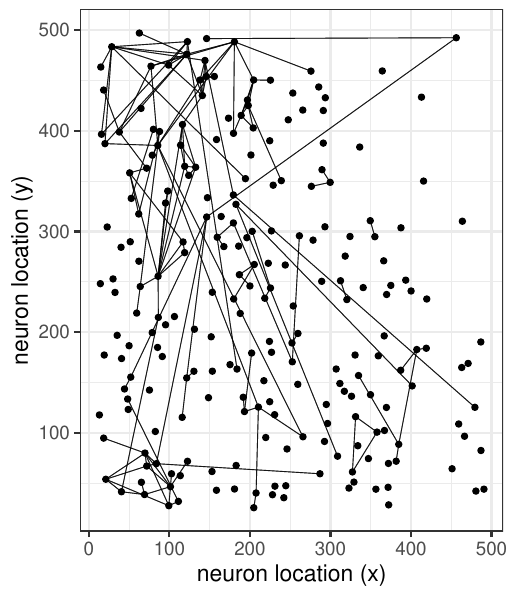}
    \caption{The estimated graph from the calcium imaging data set \citep{lein2007genome}. We generate data from this graph in our first set of simulations presented in Section 5.3.}
    \label{fig:real_graph}
\end{figure}
\begin{figure}[!ht]
    \centering
    \subfigure[Observational patterns]{
    \includegraphics[height=6cm]{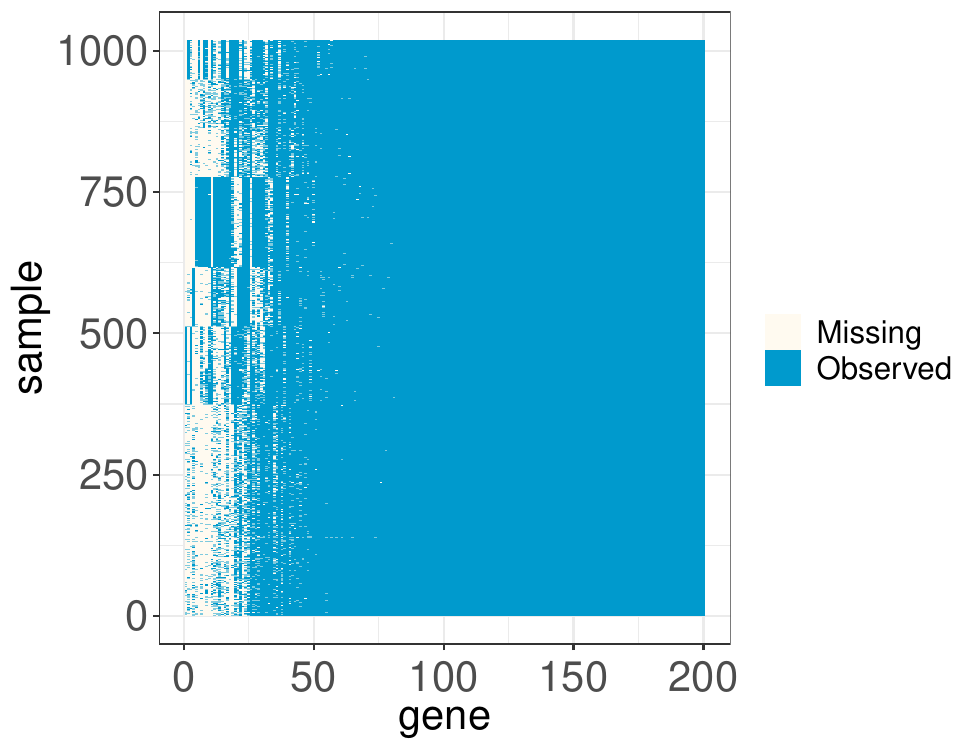}
    \includegraphics[height=6cm]{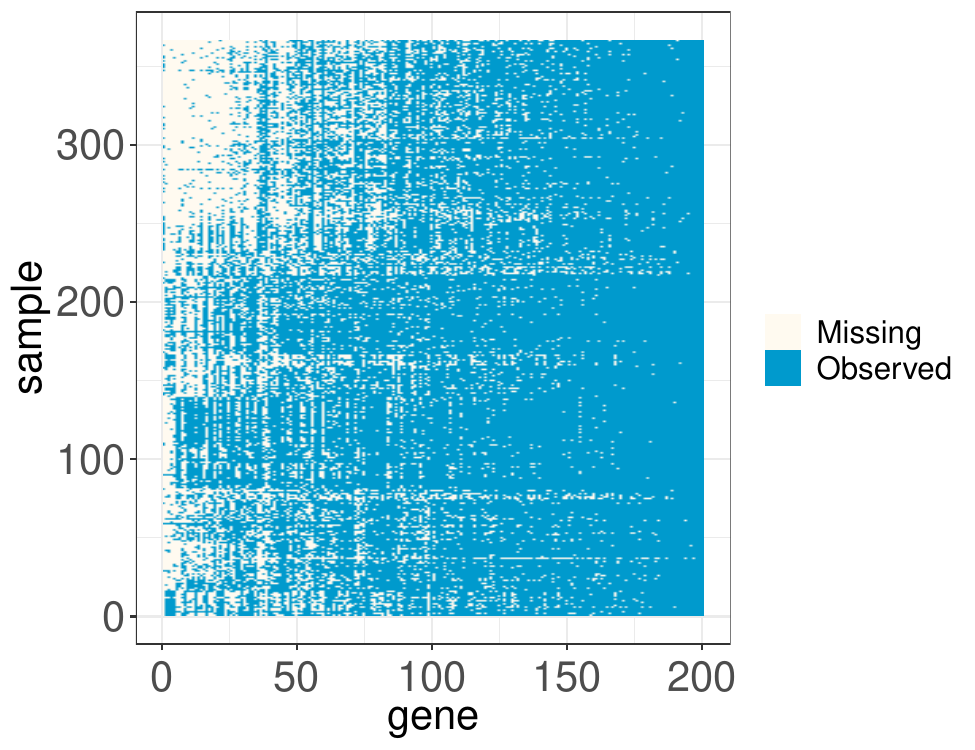}}
    \subfigure[Pairwise sample sizes]{
    \includegraphics[height=6cm]{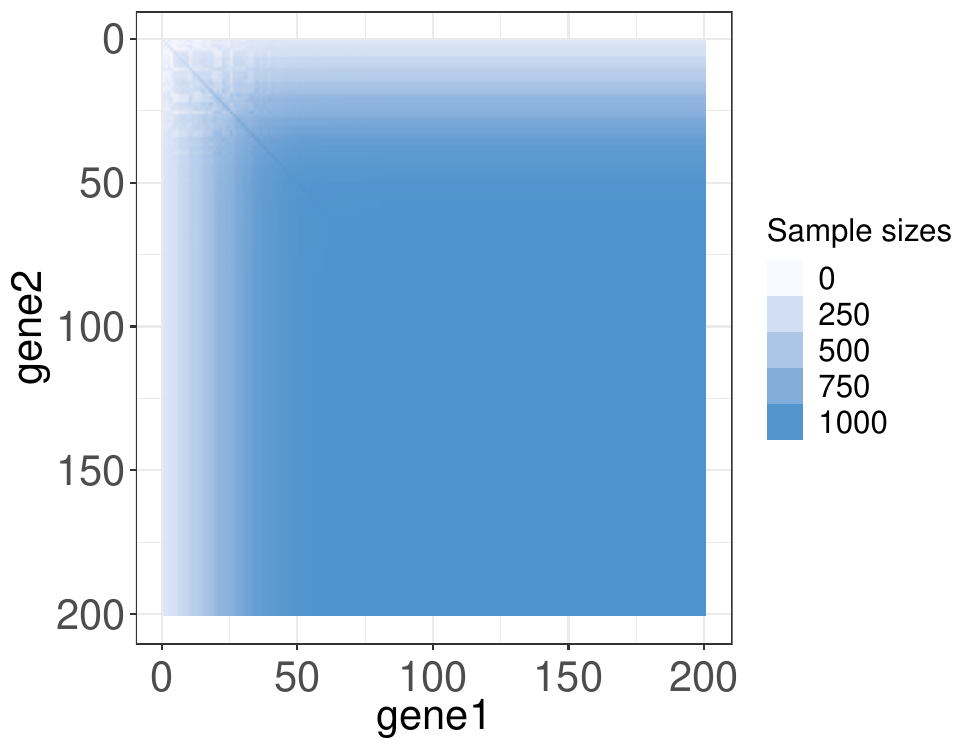}
    \includegraphics[height=6cm]{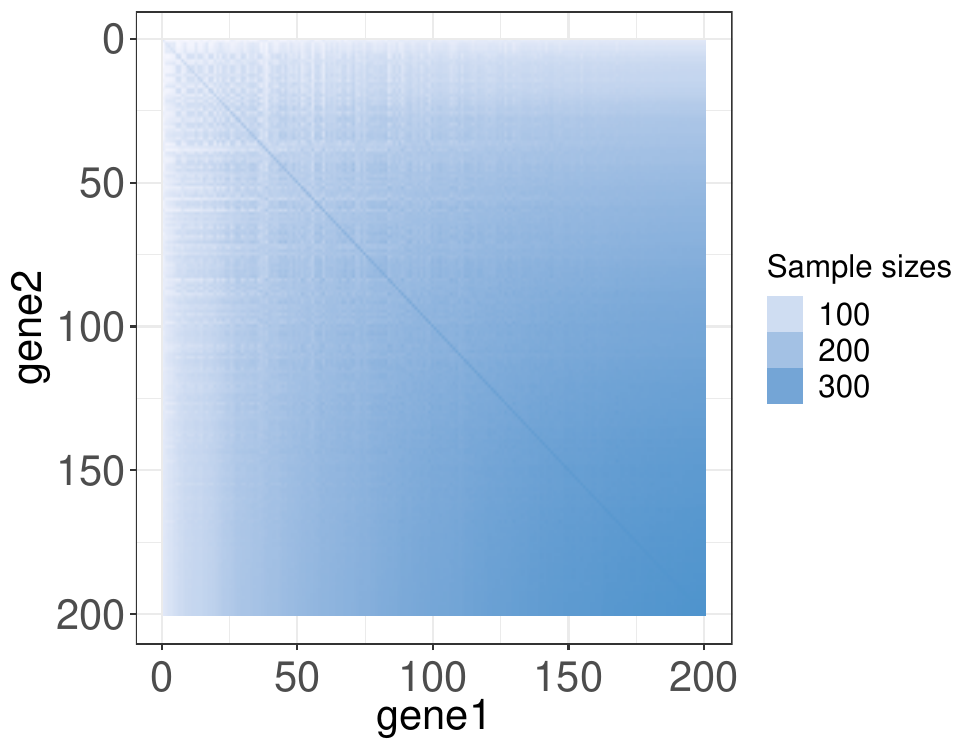}}
    \caption{\small Observational patterns and pairwise sample sizes of two real scRNA-seq data sets \citep{chu2016single, darmanis2015survey}, including top 200 genes with the highest variances. The pairwise sample sizes range from 0 to 1018 (\emph{chu} data, left) and from 5 to 366 (\emph{darmanis}, right). These are the observational patterns used in our second set of simulations presented in Section 5.3.}
    \label{fig:real_Obs}
\end{figure}

\paragraph{Generation of scale-free and small-world graphs:} The scale-free graph is generated from the Barabasi-Albert model with $200$ nodes, and for each new node sequentially added to the model, one edge between it and the existing nodes is randomly created. The small-world graph is generated from the Watts-Strogatz model with degree $2$ and rewiring probability $0.5$. 

\section{Proofs}\label{append:proofs}

For convenience, here we restate some of our key notations presented in the main paper. Recall the definition of node pair index sets $S_1(a,b)$ and $S_2(a,b)$: \begin{equation}\label{eq:index_sets}
\begin{split}
    S_1(a,b)=&\{(j,k): j\text{ or }k\in \cN_a\cup \overline{\cN}_b^{(a)}\},\\
    S_2(a,b)=&\{(j,k): \Theta^{(a)*}_{j,b}\Theta^*_{k,a}+\Theta^{(a)*}_{k,b}\Theta^*_{j,a}\neq 0\},
\end{split}
\end{equation}
and sample size quantities $n_1(a,b)=\min_{(j,k)\in S_1(a,b)}n_{j,k}$, $n_2(a,b)=\min_{(j,k)\in S_2(a,b)}n_{j,k}$.
\subsection{Proof of Theorem~\ref{thm:nb_lasso_err}}
\begin{proof}[Proof of Theorem~\ref{thm:nb_lasso_err}]
By the definition of $\widehat{\theta}^{(a)}$, we have
\begin{equation*}
\begin{split}
    &\frac{1}{2}\widehat{\theta}^{(a)\top}\widetilde{\Sigma}\widehat{\theta}^{(a)}-\widetilde{\Sigma}_{a,:}\widehat{\theta}^{(a)}+\sum_{j=1}^p\lambda_j^{(a)}|\widehat{\theta}^{(a)}_j|\\
    \leq &\frac{1}{2}\theta^{(a)*\top}\widetilde{\Sigma}\widehat{\theta}^{(a)*}-\widetilde{\Sigma}_{a,:}\theta^{(a)*}+\sum_{j=1}^p\lambda_j^{(a)}|\theta^{(a)*}_j|,
\end{split}
\end{equation*}
which implies 
\begin{align}
    &\frac{1}{2}(\widehat{\theta}^{(a)}-\theta^{(a)*})^\top\widetilde{\Sigma}(\widehat{\theta}^{(a)}-\theta^{(a)*})\notag\\
    \leq &(\widetilde{\Sigma}_{:,a}-\widetilde{\Sigma}\theta^{(a)*})^\top(\widehat{\theta}^{(a)}-\theta^{(a)*})+\sum_{j=1}^p\lambda_j^{(a)}(|\theta^{(a)*}_j|-|\widehat{\theta}^{(a)}_j|)\notag\\
    \leq&\sum_{j=1}^p\|\widetilde{\Sigma}_{j,:}-\Sigma^*_{j,:}\|_{\infty}(\|\theta^{(a)*}\|_1+1)|\widehat{\theta}^{(a)}_j-\theta^{(a)*}_j|+\sum_{j=1}^p\lambda_j^{(a)}(|\theta^{(a)*}_j|-|\widehat{\theta}^{(a)}_j|)\\
    =&\sum_{j=1}^p\|\widetilde{\Sigma}_{j,:}-\Sigma^*_{j,:}\|_{\infty}\frac{\|\Theta^{*}_{:,a}\|_1}{\Theta^*_{a,a}}|\widehat{\theta}^{(a)}_j-\theta^{(a)*}_j|+\sum_{j=1}^p\lambda_j^{(a)}(|\theta^{(a)*}_j|-|\widehat{\theta}^{(a)}_j|).\label{eq:l2_err_standard_ineq}
\end{align}
The following lemma provides an upper bound for $\|\widetilde{\Sigma}_{j,:}-\Sigma^*_{j,:}\|_{\infty}$:
\begin{lemma}[Entry-wise error bounds for $\widehat{\Sigma}-\Sigma^*$ and $\widetilde{\Sigma}$]\label{lem:SampleCov_entry_err}
	With probability at least $1-p^{-c}$, for $1\leq i,j\leq p$,
	\begin{align*}
	|\widehat{\Sigma}_{i,j}-\Sigma^*_{i,j}|\leq &C\|\Sigma^*\|_{\infty}\sqrt{\frac{\log p}{n_{i,j}}},\\
	|\widetilde{\Sigma}_{i,j}-\Sigma^*_{i,j}|
	\leq &C\|\Sigma^*\|_{\infty}\sqrt{\frac{\log p}{n_{i,j}}}.
	\end{align*}
\end{lemma}
Lemma~\ref{lem:SampleCov_entry_err} suggests that $$\|\widetilde{\Sigma}_{j,:}-\Sigma^*_{j,:}\|_{\infty}(\|\theta^{(a)*}\|_1+1)\leq C\|\Sigma^*\|_{\infty}\frac{\|\Theta^{*}_{:,a}\|_1}{\Theta^*_{a,a}}\sqrt{\frac{\log p}{\mymin_{k}n_{j,k}}}\leq \frac{\lambda^{(a)}_j}{2}$$ with high probability. Together with \eqref{eq:l2_err_standard_ineq}, we have 
\begin{align*}
    &\frac{1}{2}(\widehat{\theta}^{(a)}-\theta^{(a)*})^\top\widetilde{\Sigma}(\widehat{\theta}^{(a)}-\theta^{(a)*})\notag\\
    \leq &\sum_{j=1}^p\frac{\lambda_j}{2}|\widehat{\theta}^{(a)}_j-\theta^{(a)*}_j|+\lambda_j^{(a)}(|\theta^{(a)*}_j|-|\widehat{\theta}^{(a)}_j|)\\
    \leq &\sum_{j\in \cN_a}\frac{3}{2}\lambda^{(a)}_j|\widehat{\theta}^{(a)}_j-\theta^{(a)*}_j|-\sum_{j\in \overline{\cN}_a^c}\frac{1}{2}\lambda^{(a)}_j|\widehat{\theta}^{(a)}_j-\theta^{(a)*}_j|,
\end{align*}
where $\cN_a=\{j\neq a:\Theta^*_{j,a}\neq 0\}=\{j:\theta^{(a)*}_{j}\neq 0\}$ and $\overline{\cN}_a=\{a\}\cup \cN_a$. By construction, $\widetilde{\Sigma}$ is positive semi-definite, which implies that 
\begin{align}\label{eq:nb_lasso_err_cone1}
    \sum_{j\in \overline{\cN}_a^c}\lambda_j^{(a)}|\widehat{\theta}_j^{(a)}-\theta^{(a)*}_j|\leq 3\sum_{j\in \cN_a}\lambda_j^{(a)}|\widehat{\theta}_j^{(a)}-\theta^{(a)*}_j|. 
\end{align}
Our next step is to show a lower bound of the L.H.S. of the inequality above. Let $\Delta = \widehat{\theta}^{(a)}-\theta^{(a)*}$. First note that 
\begin{align*}
\frac{1}{2}\Delta^\top\widetilde{\Sigma}\Delta\geq \frac{1}{2}\|\Delta\|_2^2\lambda_{\mymin}(\Sigma^*)-\frac{1}{2}|\Delta^\top(\widetilde{\Sigma}-\Sigma^*)\Delta|,
\end{align*}
where the latter term can be further bounded as follows:
\begin{align*}
    |\Delta^\top(\widetilde{\Sigma}-\Sigma^*)\Delta|\leq&\sum_{j,k}\lambda^{(a)}_j|\Delta_j|\left|\frac{\widetilde{\Sigma}_{j,k}-\Sigma^*_{j,k}}{\lambda^{(a)}_j}\right||\Delta_k|\\
    \leq&\left(\sum_{j=1}^p\lambda^{(a)}_j|\Delta_j|\right)\|\Delta\|_1\mymax_{j,k}\left|\frac{\widetilde{\Sigma}_{j,k}-\Sigma^*_{j,k}}{\lambda^{(a)}_j}\right|\\
    \leq&\frac{\Theta^*_{a,a}}{2\|\Theta^*_{:,a}\|_1}\left(\sum_{j=1}^p\lambda^{(a)}_j|\Delta_j|\right)\|\Delta\|_1,
\end{align*}
where the last line is due to Lemma~\ref{lem:SampleCov_entry_err}. By \eqref{eq:nb_lasso_err_cone1}, 
\begin{equation}\label{eq:nb_lasso_l1_l2}
\begin{split}
    \|\Delta\|_1=& \sum_{j\in \cN_a}|\Delta_j|+\sum_{j\in \overline{\cN}_a^c}|\Delta_j|\\
    \leq&\sum_{j\in \cN_a}|\Delta_j|+\frac{1}{\mymin_{j\in \overline{\cN}_a^c}\lambda_j^{(a)}}\sum_{j\in \overline{\cN}_a^c}\lambda_j^{(a)}|\Delta_j|\\
    \leq &\sum_{j\in \cN_a}|\Delta_j|+\frac{3\mymax_{j\in \cN_a}\lambda_j^{(a)}}{\mymin_{j\in \overline{\cN}_a^c}\lambda_j^{(a)}}\sum_{j\in \cN_a}|\Delta_j|\\
    \leq&\left(3\sqrt{\gamma_a}+1\right)\sqrt{d_a}\|\Delta\|_2,
\end{split}
\end{equation}
and 
\begin{align*}
    \sum_{j=1}^p\lambda_j^{(a)}|\Delta_j|\leq&4\sum_{j\in \cN_a}\lambda_j^{(a)}|\Delta_j|\\
    \leq&4\mymax_{j\in \cN_a}\lambda_j^{(a)}\sqrt{d_a}\|\Delta\|_2\\
    \leq&\frac{C\|\Sigma^*\|_{\infty}\|\Theta^*_{:,a}\|_1}{\Theta^*_{a,a}}\sqrt{\frac{d_a\log p}{\mymin_{j\in \cN_a}\mymin_{k}n_{j,k}}}\|\Delta\|_2.
\end{align*}
Hence we can further bound the term above as follows:
\begin{align*}
    &|\Delta^\top(\widetilde{\Sigma}-\Sigma^*)\Delta|\\
    \leq&C\|\Sigma^*\|_{\infty}\sqrt{\frac{d_a^2\log p}{\mymin_{j\in \cN_a}\mymin_k n_{j,k}}}\left(\sqrt{\gamma_a}+1\right)\|\Delta\|_2^2\\
    \leq&\frac{\lambda_{\mymin}(\Sigma^*)}{2}\|\Delta\|_2^2,
\end{align*}
where the last line is due to our condition that
$$
\mymin_{j\in \cN_a}\mymin_k n_{j,k}\geq C\frac{C\|\Sigma^*\|_{\infty}^2(\gamma_a+1)}{\lambda_{\mymin}^2(\Sigma^*)}d_a^2\log p.
$$
Therefore, 
\begin{align*}
    \|\Delta\|_2^2\leq &\frac{6}{\lambda_{\mymin}(\Sigma^*)}\sum_{j\in \cN_a}\lambda^{(a)}_j|\Delta_j|\\
    \leq &\frac{C\|\Sigma^*\|_{\infty}}{\lambda_{\mymin}(\Sigma^*)}\frac{\|\Theta^*_{:,a}\|_1}{\Theta^*_{a,a}}\sqrt{\frac{d_a\log p}{\mymin_j\in \cN_a\mymin_k n_{j,k}}}\|\Delta\|_2,
\end{align*}
which implies
\begin{align*}
    \|\Delta\|_2\leq \frac{C\|\Sigma^*\|_{\infty}}{\lambda_{\mymin}(\Sigma^*)}\frac{\|\Theta^*_{:,a}\|_1}{\Theta^*_{a,a}}\sqrt{\frac{d_a\log p}{\mymin_{j\in \cN_a}\mymin_k n_{j,k}}}.
\end{align*}
While for the $\ell_1$ norm error, one can apply \eqref{eq:nb_lasso_l1_l2} to obtain that
\begin{align*}
    \|\Delta\|_1\leq &\frac{C\|\Sigma^*\|_{\infty}(\sqrt{\gamma_a}+1)}{\lambda_{\mymin}(\Sigma^*)}\frac{\|\Theta^*_{:,a}\|_1}{\Theta^*_{a,a}}\sqrt{\frac{d_a^2\log p}{\mymin_{j\in \cN_a}\mymin_k n_{j,k}}},
\end{align*}
\begin{align*}
    \sum_{j=1}^p\lambda^{(a)}_j|\Delta_j|\leq &4\sum_{j\in \cN_a}\lambda_j^{(a)}|\Delta_j|\\
    \leq&\frac{C\|\Sigma^*\|_{\infty}^2}{\lambda_{\mymin}(\Sigma^*)}\frac{\|\Theta^*_{:,a}\|_1^2}{(\Theta^*_{a,a})^2}\frac{d_a\log p}{\mymin_{j\in \cN_a}\mymin_k n_{j,k}}.
\end{align*}
\end{proof}
\begin{proof}[Proof of Theorem \ref{thm:nblasso_support_full}]
    We first would like to show that there is no false positive: $\widehat{\theta}^{(a)}_{(\mathcal{N}_a)^c}=0$, and the second step is to show no false negative in $\mathcal{N}_a$.
	Let $$\widetilde{\theta}=\arg\min_{\theta_{(\mathcal{N}_a)^c=0}}\frac{1}{2}\theta^\top \widetilde{\Sigma}\theta-\widetilde{\Sigma}_{:,a}\theta+\|\lambda^{(a)}\circ \theta\|_1,$$ and $\widetilde{\Delta}=\widetilde{\theta}-\theta^*$. Since $\widetilde{\Sigma}$ is positive semi-definite, $\widetilde{\theta}$ exists and satisfies the KKT condition,
	\begin{equation}\label{eq:nblasso:KKT1}
	\widetilde{\Sigma}_{\mathcal{N}_a,\mathcal{N}_a}\Delta_{\mathcal{N}_a}+\widetilde{\Sigma}_{\mathcal{N}_a,\mathcal{N}_a}\theta^*_{\mathcal{N}_a}-\widetilde{\Sigma}_{\mathcal{N}_a,a}+\lambda^{(a)}_{\mathcal{N}_a}\circ \widetilde{Z}_{\mathcal{N}_a}=0,
	\end{equation}
	for some $\widetilde{Z}_{j}=\begin{cases}
	\mathrm{sgn}(\widetilde{\theta}_j),& j\in \mathcal{N}_a,\widetilde{\theta}_j\neq 0,\\
	\in [-1,1], &j\in \mathcal{N}_a,\widetilde{\theta}_j=0.
	\end{cases}$
	We define $\widetilde{Z}_j, j\in \overline{\mathcal{N}}_a^c$ to satisfy
	\begin{equation}\label{eq:nblasso:KKT2}
	\widetilde{\Sigma}_{\overline{\mathcal{N}}_a^c,\mathcal{N}_a}\Delta_{\mathcal{N}_a}+\widetilde{\Sigma}_{\overline{\mathcal{N}}_a^c,\mathcal{N}_a}\theta^*_{\mathcal{N}_a}-\widetilde{\Sigma}_{\overline{\mathcal{N}}_a^c,a}+\lambda^{(a)}_{\overline{\mathcal{N}}_a^c}\circ \widetilde{Z}_{\overline{\mathcal{N}}_a^c}=0,
	\end{equation}
	and we would like to show $\|\widetilde{Z}_{\overline{\mathcal{N}}_a^c}\|_{\infty}\leq 1$, so that $\widetilde{\theta}$ also satisfies the KKT condition of 
    $$
    \min_\theta \frac{1}{2}\theta^\top \widetilde{\Sigma}\theta-\widetilde{\Sigma}_{:,a}\theta+\|\lambda^{(a)}\circ \theta\|_1.
    $$
    This would complete our proof for the no false positive statement.

	By \eqref{eq:nblasso:KKT1} and \eqref{eq:nblasso:KKT2}, we have that
	\begin{equation*}
	\begin{split}
	    \widetilde{Z}_{\overline{\mathcal{N}}_a^c}=&\frac{1}{\lambda^{(a)}_{\overline{\mathcal{N}}_a^c}}\circ\Big[\widetilde{\Sigma}_{\overline{\mathcal{N}}_a^c,a}-\widetilde{\Sigma}_{\overline{\mathcal{N}}_a^c,\mathcal{N}_a}\theta^*_{\mathcal{N}_a}\\
	    &-\widetilde{\Sigma}_{\overline{\mathcal{N}}_a^c,\mathcal{N}_a}(\widetilde{\Sigma}_{\mathcal{N}_a,\mathcal{N}_a})^{-1}(\widetilde{\Sigma}_{\mathcal{N}_a,a}-\widetilde{\Sigma}_{\mathcal{N}_a,\mathcal{N}_a}\theta^*_{\mathcal{N}_a}-\lambda^{(a)}_{\mathcal{N}_a}\circ \widetilde{Z}_{\mathcal{N}_a})\Big]
	\end{split}
	\end{equation*}
	In the following, we first bound $|\widetilde{\Sigma}_{j,a}-\widetilde{\Sigma}_{j,\mathcal{N}_a}\theta^*_{\mathcal{N}_a}|$ for any $j\neq a$, then bound $\vertiii{\widetilde{\Sigma}_{\overline{\mathcal{N}}_a^c,\mathcal{N}_a}(\widetilde{\Sigma}_{\mathcal{N}_a,\mathcal{N}_a})^{-1}}_{\infty}$.
	
	Let $\widehat{W}=\widetilde{\Sigma}-\Sigma^*$, then we have 
	\begin{equation}\label{eq:nblasso_primal_witness}
	\begin{split}
	    \left|\widetilde{\Sigma}_{j,a}-\widetilde{\Sigma}_{j,\mathcal{N}_a}\theta^*_{\mathcal{N}_a}\right|=&\left|\Sigma^*_{j,a}-\Sigma^*_{j,\mathcal{N}_a}\theta^*_{\mathcal{N}_a}+\widehat{W}_{j,a}-\widehat{W}_{j,\mathcal{N}_a}\theta^*_{\mathcal{N}_a}\right|\\
	    =&\left|\widehat{W}_{j,a}-\widehat{W}_{j,\mathcal{N}_a}\theta^*_{\mathcal{N}_a}\right|,
	\end{split}
	\end{equation}
	where the last equation is due to the definition $$\theta^*_{\backslash a}=(\Sigma^*_{\backslash a,\backslash a})^{-1}\Sigma^*_{\backslash a,a}.$$
	By Proposition 1, with probability at least $1-p^{-c}$, 
	\begin{equation*}
	|\widehat{W}_{ij}|\leq C\|\Sigma^*\|_{\infty}\sqrt{\frac{\log p}{n_{ij}}}
	\end{equation*}
	holds for all $1\leq i,j\leq p$. Hence 
	\begin{equation}\label{eq:nblasso_dev_bnd}
	\left|\widetilde{\Sigma}_{j,a}-\widetilde{\Sigma}_{j,\mathcal{N}_a}\theta^*_{\mathcal{N}_a}\right|\leq \|\widehat{W}_{j,\overline{\mathcal{N}}_a}\|_{\infty}(\kappa_1+1)\leq \frac{\gamma \lambda^{(a)}_j}{4}
	\end{equation}
	for $j\neq a$. 
	
	On the other hand, for any $j\in \overline{\mathcal{N}}^c_a$, 
	\begin{equation}\label{eq:nblasso_incoherence_bnd}
	\begin{split}
	&\|(\widetilde{\Sigma}_{\mathcal{N}_a,\mathcal{N}_a})^{-1}\widetilde{\Sigma}_{\mathcal{N}_a,j}\|_1\\
	\leq& \|(\Sigma^*_{\mathcal{N}_a,\mathcal{N}_a})^{-1}\Sigma^*_{\mathcal{N}_a,j}\|_1+d_a\vertiii{(\widetilde{\Sigma}_{\mathcal{N}_a,\mathcal{N}_a})^{-1}}_{\infty}\left\|\widetilde{\Sigma}_{\mathcal{N}_a,j}-\Sigma^*_{\mathcal{N}_a,j}\right\|_{\infty}\\
	&+d_a\|(\widetilde{\Sigma}_{\mathcal{N}_a,\mathcal{N}_a})^{-1}-(\Sigma^*_{\mathcal{N}_a,\mathcal{N}_a})^{-1}\|_{\infty}\|\Sigma^*_{\mathcal{N}_a,j}\|_1\\
	\leq& 1-\gamma+\kappa_3 d_a\|(\widetilde{\Sigma}_{\mathcal{N}_a,\mathcal{N}_a})^{-1}-(\Sigma^*_{\mathcal{N}_a,\mathcal{N}_a})^{-1}\|_{\infty}\\
	&+Cd_a\left(\kappa_2+d_a\left\|(\widetilde{\Sigma}_{\mathcal{N}_a,\mathcal{N}_a})^{-1}-(\Sigma^*_{\mathcal{N}_a,\mathcal{N}_a})^{-1}\right\|_{\infty}\right)\\
	&*\max_{i\in \mathcal{N}_a}\|\Sigma^*\|_{\infty}\sqrt{\frac{\log p}{n_{ij}}}\\
	\leq&1-\frac{5\gamma}{6}+(\kappa_3+\frac{\gamma}{6\kappa_2}) d_a\|(\widetilde{\Sigma}_{\mathcal{N}_a,\mathcal{N}_a})^{-1}-(\Sigma^*_{\mathcal{N}_a,\mathcal{N}_a})^{-1}\|_{\infty}.
	\end{split}
	\end{equation}
	Since 
	\begin{equation*}
	\begin{split}
	\vertiii{(\Sigma^*_{\mathcal{N}_a,\mathcal{N}_a})^{-1}\widehat{W}_{\mathcal{N}_a,\mathcal{N}_a}}_{\infty}\leq& \vertiii{(\Sigma^*_{\mathcal{N}_a,\mathcal{N}_a})^{-1}}_{\infty}\vertiii{\widehat{W}_{\mathcal{N}_a,\mathcal{N}_a}}_{\infty}\\
	\leq&\kappa_2d_a\|\widehat{W}_{\mathcal{N}_a,\mathcal{N}_a}\|_{\infty}\\
	\leq& C\kappa_2d_a\max_{i,j\in\mathcal{N}_a}\|\Sigma^*\|_{\infty}\sqrt{\frac{\log p}{n_{ij}}}\\
	\leq &\frac{1}{2},
	\end{split}
	\end{equation*}
	where the last line is due to the sample size requirement in Theorem 1. Thus the following matrix expansion holds:
	\begin{equation}
	\begin{split}
	&(\widetilde{\Sigma}_{\mathcal{N}_a,\mathcal{N}_a})^{-1}-(\Sigma^*_{\mathcal{N}_a,\mathcal{N}_a})^{-1}\\
	=&(\Sigma^*_{\mathcal{N}_a,\mathcal{N}_a}+\widehat{W}_{\mathcal{N}_a,\mathcal{N}_a})^{-1}-(\Sigma^*_{\mathcal{N}_a,\mathcal{N}_a})^{-1}\\
	=&\sum_{k=1}^{\infty}(-1)^k[(\Sigma^*_{\mathcal{N}_a,\mathcal{N}_a})^{-1}\widehat{W}_{\mathcal{N}_a,\mathcal{N}_a}]^k(\Sigma^*_{\mathcal{N}_a,\mathcal{N}_a})^{-1}\\
	=&-(\Sigma^*_{\mathcal{N}_a,\mathcal{N}_a})^{-1}\widehat{W}_{\mathcal{N}_a,\mathcal{N}_a}J(\Sigma^*_{\mathcal{N}_a,\mathcal{N}_a})^{-1}.
	\end{split}
	\end{equation}
	Here $J=\sum_{k=0}^{\infty}(-1)^k[(\Sigma^*_{\mathcal{N}_a,\mathcal{N}_a})^{-1}\widehat{W}_{\mathcal{N}_a,\mathcal{N}_a}]^k$ and satisfies $\vertiii{J}_{\infty}\leq \sum_{k=0}^{\infty}(\frac{1}{2})^k=2$. Hence we have 
	\begin{equation}\label{eq:nblasso_inv_bnd}
	\begin{split}
	&\|(\widetilde{\Sigma}_{\mathcal{N}_a,\mathcal{N}_a})^{-1}-(\Sigma^*_{\mathcal{N}_a,\mathcal{N}_a})^{-1}\|_{\infty}\\
	\leq &2\kappa_2^2\|\widehat{W}_{\mathcal{N}_a,\mathcal{N}_a}\|_{\infty}\\
	\leq &2C\kappa_2^2\|\Sigma^*\|_{\infty}\max_{i,j\in \mathcal{N}_a}\sqrt{\frac{\log p}{n_{ij}}}\\
	\leq &\min\{\frac{\kappa_2}{d_a},\frac{\gamma}{6\kappa_3d_a}\}.
	\end{split}
	\end{equation}\
	Combining \eqref{eq:nblasso_inv_bnd} and \eqref{eq:nblasso_incoherence_bnd} implies
	\begin{equation}\label{eq:nblasso_incoherence_finalbnd}
	\|(\widetilde{\Sigma}_{\mathcal{N}_a,\mathcal{N}_a})^{-1}\widetilde{\Sigma}_{\mathcal{N}_a,j}\|_1\leq 1-\frac{\gamma}{2}.
	\end{equation}
	By \eqref{eq:nblasso_primal_witness}, \eqref{eq:nblasso_dev_bnd}, \eqref{eq:nblasso_incoherence_finalbnd}, we have
	\begin{equation*}
	    \begin{split}
	        &\|\widetilde{Z}_{\overline{\mathcal{N}}_a^c}\|_{\infty}\\
	        \leq &\frac{\gamma}{4}+(1-\frac{\gamma}{2})\frac{\max_{j\in \mathcal{N}_a}\lambda^{(a)}_j}{\min_{j\in \overline{\mathcal{N}_a}^c}\lambda^{(a)}_j}\\
	        \leq& c\gamma+(1-\frac{\gamma}{2})\sqrt{\gamma_a}\\
	        \leq& 1,
	    \end{split}
	\end{equation*}
	where the last inequality is due to the condition on $\gamma_a$. 
	
	Now we only need to upper bound $\|\Delta_{\mathcal{N}_a}\|_{\infty}$. Based on \eqref{eq:nblasso:KKT1}, \eqref{eq:nblasso_dev_bnd} and \eqref{eq:nblasso_inv_bnd}, it is straightforward to see that
	\begin{equation*}
	\|\Delta_{\mathcal{N}_a}\|_{\infty}\leq 2\kappa_2(1+\frac{\gamma}{4})\max_{j\in \mathcal{N}_a}\lambda^{(a)}_j\leq \frac{\theta^{(a)}_{\min}}{2},
	\end{equation*}
	 which implies that $\{j:\widehat{\theta}^{(a)}_j\neq 0\}=\mathcal{N}_a$ and $\mathrm{sgn}(\widehat{\theta}^{(a)}_j)=\mathrm{sgn}(\theta^{(a)*}_j)=\mathrm{sgn}(\Theta^*_{ja})$.
\end{proof}

\subsection{Proof of Theorem~\ref{thm:nb_lasso_debias_decomp}}
We start by presenting the following lemma that guarantees the estimation performance of $\widehat{\Theta}^{(a)}$ and hence helps establish our normal approximation result of the debiased neighborhood lasso. Define the sample size ratio for node $j$ as $$\gamma_j^{(a)}=\frac{\mymax_{i\in \overline{\cN}^{(a)c}_j}\mymin_kn_{i,k}}{\mymin_{i\in \cN^{(a)}_j}\mymin_kn_{i,k}}.$$ Also recall the neighborhood lasso estimator used for debiasing in \eqref{eq:debias_term}:
\begin{equation}
\begin{split}
\widehat{\theta}^{(a,b)} = &\argmin_{\theta\in \bR^{p},\theta_a=\theta_b=0}\frac{1}{2}\theta^\top \widetilde{\Sigma}\theta - \widetilde{\Sigma}_{b,:}\theta +\sum_{k=1}^{p}\lambda^{(a,b)}_k|\theta_k|,\\ \widehat{\tau}^{(a,b)} = &(\widetilde{\Sigma}_{b,b}-\widetilde{\Sigma}_{b,:}\widehat{\theta}^{(a,b)})^{-1},
\end{split}
\end{equation}
\begin{lemma}\label{lem:debias_nb_lasso_err}
	If the tuning parameters $\lambda^{(a,j)}_k$'s in \eqref{eq:debias_term} satisfy 
	\begin{align*}
		\lambda^{(a,j)}_k \asymp& \|\Sigma^*\|_{\infty}\frac{\|\Theta^{(a)*}_{j,:}\|_1}{\Theta^{(a)*}_{j,j}}\sqrt{\frac{\log p}{\mymin_{i\in [p]}n_{i,k}}},\\
		\mymin_{i\in \overline{\cN}_j^{(a)}}\mymin_{k\in [p]}n_{i,k}\geq&\frac{C\|\Sigma^*\|^2_{\infty}(\gamma_j^{(a)}+\kappa_{\Sigma^*}^2)}{\lambda_{\mymin}^2(\Sigma^*)}d_j^{(a)2}\log p,
	\end{align*} then with probability at least $1-Cp^{-c}$,
	\begin{align*}
	\|\widehat{\Theta}^{(a)}_{j,:}-\Theta^{(a)*}_{j,:}\|_2\leq 
	&\frac{C\kappa_{\Sigma^*}\|\Sigma^*\|_{\infty}}{\lambda_{\mymin}(\Sigma^*)}\|\Theta^{(a)*}_{j,:}\|_1\sqrt{\frac{d_j^{(a)}\log p}{\mymin_{k\in \overline{\cN}_k^{(a)}}\mymin_{i\in [p]}n_{i,k}}},\\
	\|\widehat{\Theta}^{(a)}_{j,:}-\Theta^{(a)*}_{j,:}\|_1\leq 
	&\frac{C(\kappa_{\Sigma^*}+\sqrt{\gamma_j^{(a)}})\|\Sigma^*\|_{\infty}}{\lambda_{\mymin}(\Sigma^*)}\|\Theta^{(a)*}_{j,:}\|_1d^{(a)}_j\sqrt{\frac{\log p}{\mymin_{k\in\overline{\cN}_j^{(a)}}\mymin_{i\in[p]}n_{i,k}}},\\
    \|\widehat{\Theta}^{(a)}_{j,:}-\Theta^{(a)*}_{j,:}\|_{\lambda^{(a,j)},1}\leq &\frac{C\kappa_{\Sigma^*}\|\Sigma^*\|_{\infty}^2}{\lambda_{\mymin}(\Sigma^*)}\frac{\|\Theta^{(a)*}_{j,:}\|_1^2}{\Theta^{(a)*}_{j,j}}\frac{d_j^{(a)}\log p}{\mymin_{k\in \cN_k^{(a)}}\mymin_{i\in [p]}n_{i,k}},
	\end{align*}
 where $d_j^{(a)}=|\{k:\Theta^{(a)*}_{j,k}\neq 0\}|$. Furthermore, the bounds above also hold when substituting $\widehat{\Theta}^{(a)}_{j,:}$ with $\widetilde{\Theta}^{(a)}_{j,:}$.
\end{lemma}

\begin{proof}[Proof of Lemma~\ref{lem:debias_nb_lasso_err}]
    By the definition of $\widehat{\Theta}^{(a)}_{j,:}$ in \eqref{eq:debias_Theta_est}:
    \begin{equation}
    \widehat{\Theta}_{j,:}^{(a)}=\widehat{\Theta}^{(a)}_{j,j} \widehat{\overline{\theta}}^{(a,j)},
\end{equation}
we have 
    \begin{equation}\label{eq:debias_Theta_err}
    \begin{split}
        &\|\widehat{\Theta}^{(a)}_{j,:}-\Theta^{(a)*}_{j,:}\|\\
        \leq &|\widehat{\Theta}^{(a)}_{j,j}-\Theta^{(a)*}_{j,j}|\frac{\|\Theta^{(a)*}_{j,:}\|}{\Theta^{(a)*}_{j,j}}+\Theta^{(a)*}_{j,j}\|\widehat{\theta}^{(a,j)}-\theta^{(a,j)*}\|\\
        &+|\widehat{\Theta}^{(a)}_{j,j}-\Theta^{(a)*}_{j,j}|\|\widehat{\theta}^{(a,j)}-\theta^{(a,j)*}\|,
    \end{split}
    \end{equation}
    where $\|\cdot\|$ can either be $\ell_2$ norm $\|\cdot\|_2$, $\ell_1$ norm $\|\cdot\|_1$, or the weighted $\ell_1$ norm $\|\cdot\|_{\lambda^{(a,j)},1}$ defined as $\|\theta\|_{\lambda^{(a,j)},1}=\sum_{k=1}^p\lambda^{(a,j)}_k|\theta_k|$, and 
    \begin{align*}
       \Theta^{(a)*}_{j,j}&=(\Sigma^*_{j,j}-\Sigma^*_{j,\backslash\{a,j\}}(\Sigma^*_{\backslash\{a,j\},\backslash\{a,j\}})^{-1}\Sigma^*_{\backslash\{a,j\},j})^{-1},\\
    \theta^{(a,j)*} &= (\Sigma^*_{\backslash\{a,j\},\backslash\{a,j\}})^{-1}\Sigma^*_{\backslash\{a,j\},j}.
    \end{align*}
    For $\|\widehat{\theta}^{(a,j)}-\theta^{(a,j)*}\|_2$, $\|\widehat{\theta}^{(a,j)}-\theta^{(a,j)*}\|_1$, and $\|\widehat{\theta}^{(a,j)}-\theta^{(a,j)*}\|_{\lambda^{(a,j)},1}$, 
    we can bound them by the same arguments as the proof of Theorem~\ref{thm:nb_lasso_err}. The only difference lies that we substitute $\Sigma^*$ and $\widetilde{\Sigma}$ by $\Sigma^*_{\backslash a, \backslash a}$ and $\widetilde{\Sigma}_{\backslash a, \backslash a}$, and hence with probability at least $1-Cp^{-c}$, we have
    \begin{equation}\label{eq:debias_theta_err}
        \begin{split}
            \|\widehat{\theta}^{(a,j)}-\theta^{(a,j)*}\|_2\leq &\frac{C\|\Sigma^*\|_{\infty}}{\lambda_{\mymin}(\Sigma^*)}\frac{\|\Theta^{(a)*}_{j,:}\|_1}{\Theta^{(a)*}_{j,j}}\sqrt{\frac{d_j^{(a)}\log p}{\mymin_{k\in \cN_k^{(a)}}\mymin_{i\in [p]}n_{i,k}}},\\
            \|\widehat{\theta}^{(a,j)}-\theta^{(a,j)*}\|_1\leq &\frac{C\|\Sigma^*\|_{\infty}(\sqrt{\gamma^{(a)}_j}+1)}{\lambda_{\mymin}(\Sigma^*)}\frac{\|\Theta^{(a)*}_{j,:}\|_1}{\Theta^{(a)*}_{j,j}}\sqrt{\frac{d_j^{(a)2}\log p}{\mymin_{k\in \cN_j^{(a)}}\mymin_{i\in [p]}n_{i,k}}},\\
            \|\widehat{\theta}^{(a,j)}-\theta^{(a,j)*}\|_{\lambda^{(a,j)},1}\leq &\frac{C\|\Sigma^*\|_{\infty}^2}{\lambda_{\mymin}(\Sigma^*)}\frac{\|\Theta^{(a)*}_{j,:}\|_1^2}{(\Theta^{(a)*}_{j,j})^2}\frac{d_j^{(a)}\log p}{\mymin_{k\in \cN_k^{(a)}}\mymin_{i\in [p]}n_{i,k}}.
        \end{split}
    \end{equation}
    While for $|\widehat{\Theta}^{(a)}_{j,j}-\Theta^{(a)*}_{j,j}|$, we have
    \begin{align}\label{eq:tau_err}
        |\widehat{\Theta}^{(a)}_{j,j}-\Theta^{(a)*}_{j,j}|\leq& |\widehat{\Theta}^{(a)}_{j,j}\Theta^{(a)*}_{j,j}||(\widehat{\Theta}^{(a)}_{j,j})^{-1}-(\Theta^{(a)*}_{j,j})^{-1}|,
    \end{align}
    and by the definition of $\widehat{\Theta}^{(a)}_{j,j}$,
    \begin{align*}
        |(\widehat{\Theta}^{(a)}_{j,j})^{-1}-(\Theta^{(a)*}_{j,j})^{-1}|=&|\widetilde{\Sigma}_{j,j}-\Sigma^*_{j,j}|+|(\widetilde{\Sigma}_{j,:}-\Sigma^*_{j,:})\widehat{\theta}^{(a,j)}|\\
        &+|\Sigma^*_{j,:}(\widehat{\theta}^{(a,j)}-\theta^{(a,j)*})|\\
        \leq &\|(\widetilde{\Sigma}-\Sigma^*)_{j,:}\|_{\infty}(1+\|\theta^{(a,j)*}\|_1+\|\widehat{\theta}-\theta^{(a,j)*}\|_1)\\
        &+\|\Sigma^*_{j,:}\|_2\|\widehat{\theta}^{(a,j)}-\theta^{(a,j)*}\|_2.
    \end{align*}
    Since $\|\theta^{(a,j)*}\|_1+1=\frac{\|\Theta^{(a)*}_{j,:}\|_1}{\Theta^{(a)*}_{j,j}}$, $\|\widehat{\theta}^{(a,j)}-\theta^{(a,j)*}\|_1\leq \frac{\|\Theta^{(a)*}_{j,:}\|_1}{\Theta^{(a)*}_{j,j}}$ due to \eqref{eq:debias_theta_err} and the sample size condition in Lemma~\ref{lem:debias_nb_lasso_err}, and $\|\Sigma^*_{j,:}\|_2\geq \lambda_{\mymin}(\Sigma^*)$, the error term $|(\widehat{\Theta}^{(a)}_{j,j})^{-1}-(\Theta^{(a)*}_{j,j})^{-1}|$ can be further bounded by
    \begin{align}\label{eq:tau_inv_err}
        &|(\widehat{\Theta}^{(a)}_{j,j})^{-1}-(\Theta^{(a)*}_{j,j})^{-1}|\\
        \leq& \frac{C\|\Sigma^*\|_{\infty}\|\Sigma^*_{j,:}\|_2\|\Theta^{(a)*}_{j,:}\|_1}{\lambda_{\mymin}(\Sigma^*)\Theta^{(a)*}_{j,j}}\sqrt{\frac{d_j^{(a)}\log p}{\mymin_{k\in\overline{\cN}_j^{(a)}}\mymin_{i\in[p]}n_{i,k}}}\\
        \leq&\frac{C\kappa_{\Sigma^*}\|\Sigma^*\|_{\infty}}{\lambda_{\min}(\Sigma^*)\Theta^{(a)*}_{j,j}}\sqrt{\frac{d_j^{(a)2}\log p}{\mymin_{k\in\overline{\cN}_j^{(a)}}\mymin_{i\in[p]}n_{i,k}}}\\
        \leq& \frac{1}{2\Theta^{(a)*}_{j,j}},
    \end{align}
    where the last line is due to the fact that $\|\Theta^{(a)*}_{j,:}\|_1\leq \sqrt{d_j^{(a)}}\|\Theta^{(a)*}_{j,:}\|_2\leq \sqrt{d_j^{(a)}}\lambda_{\max}(\Theta^{(a)*})\leq \sqrt{d_j^{(a)}}\lambda_{\min}^{-1}(\Sigma^{*})$, $\|\Sigma^*_{j,:}\|_2\leq \lambda_{\max}(\Sigma^*)$, and the sample size condition in Lemma~\ref{lem:debias_nb_lasso_err}. 
    Plug in \eqref{eq:tau_inv_err} into \eqref{eq:tau_err}, one can show that
    \begin{align*}
        &|\widehat{\Theta}^{(a)}_{j,j}-\Theta^{(a)*}_{j,j}|\\
        \leq &C(\Theta^{(a)*}_{j,j})^2\frac{\|\Sigma^*\|_{\infty}\|\Sigma^*_{j,:}\|_2\|\Theta^{(a)*}_{j,:}\|_1}{\lambda_{\mymin}(\Sigma^*)\Theta^{(a)*}_{j,j}}\sqrt{\frac{d_j^{(a)}\log p}{\mymin_{k\in\overline{\cN}_j^{(a)}}\mymin_{i\in[p]}n_{i,k}}}\\
        =&\frac{C\|\Sigma^*\|_{\infty}\|\Sigma^*_{j,:}\|_2\Theta^{(a)*}_{j,j}\|\Theta^{(a)*}_{j,:}\|_1}{\lambda_{\mymin}(\Sigma^*)}\sqrt{\frac{d_j^{(a)}\log p}{\mymin_{k\in\overline{\cN}_j^{(a)}}\mymin_{i\in[p]}n_{i,k}}},
    \end{align*}
    and 
    \begin{align*}
        |\widehat{\Theta}^{(a)}_{j,j}-\Theta^{(a)*}_{j,j}|\leq&\frac{1}{2}\widehat{\Theta}^{(a)}_{j,j}\\
        \leq&\frac{1}{2}[(\Theta^{(a)*}_{j,j})^{-1}-|(\widehat{\Theta}^{(a)}_{j,j})^{-1}-(\Theta^{(a)*}_{j,j})^{-1}|]^{-1}\leq \Theta^{(a)*}_{j,j}.
    \end{align*}
    Therefore, combining the bound above with \eqref{eq:debias_theta_err} and \eqref{eq:debias_Theta_err}, we have
    \begin{equation*}
        \begin{split}
            \|\widehat{\Theta}^{(a)}_{j,:}-\Theta^{(a)*}_{j,:}\|_2\leq &C\frac{\|\Sigma^*\|_{\infty}\|\Sigma^*_{j,:}\|_2\|\Theta^{(a)*}_{j,:}\|_1\|\Theta^{(a)*}_{j,:}\|_2}{\lambda_{\mymin}(\Sigma^*)}\sqrt{\frac{d_j^{(a)}\log p}{\mymin_{k\in\overline{\cN}_j^{(a)}}\mymin_{i\in[p]}n_{i,k}}}\\
            &+\frac{C\|\Sigma^*\|_{\infty}}{\lambda_{\mymin}(\Sigma^*)}\|\Theta^{(a)*}_{j,:}\|_1\sqrt{\frac{d_j^{(a)}\log p}{\mymin_{k\in \cN_k^{(a)}}\mymin_{i\in [p]}n_{i,k}}}
            \\
            \leq& \frac{C\kappa_{\Sigma^*}\|\Sigma^*\|_{\infty}}{\lambda_{\mymin}(\Sigma^*)}\|\Theta^{(a)*}_{j,:}\|_1\sqrt{\frac{d_j^{(a)}\log p}{\mymin_{k\in \overline{\cN}_k^{(a)}}\mymin_{i\in [p]}n_{i,k}}},
        \end{split}
    \end{equation*}
    where the last line is due to that $\|\Sigma^*_{j,:}\|_2\leq \lambda_{\mymax}(\Sigma^*)$ and $\|\Theta^{(a)*}_{j,:}\|_2\leq \lambda_{\mymin}^{-1}(\Sigma^*)$;
    \begin{equation*}
        \begin{split}
            \|\widehat{\Theta}^{(a)}_{j,:}-\Theta^{(a)*}_{j,:}\|_1\leq &C\frac{\|\Sigma^*\|_{\infty}\|\Sigma^*_{j,:}\|_2\|\Theta^{(a)*}_{j,:}\|_1^2}{\lambda_{\mymin}(\Sigma^*)}\sqrt{\frac{d_j^{(a)}\log p}{\mymin_{k\in\overline{\cN}_j^{(a)}}\min_{i\in [p]}n_{i,k}}}\\
            &+\frac{C\|\Sigma^*\|_{\infty}(\sqrt{\gamma_j^{(a)}}+1)}{\lambda_{\mymin}(\Sigma^*)}\|\Theta^{(a)*}_{j,:}\|_1\sqrt{\frac{d_j^{(a)2}\log p}{\mymin_{k\in \cN_k^{(a)}}\mymin_{i\in [p]}n_{i,k}}}\\
            \leq&\frac{C\|\Sigma^*\|_{\infty}\|\Theta^{(a)*}_{j,:}\|_1}{\lambda_{\mymin}(\Sigma^*)}(\kappa_{\Sigma^*}+\sqrt{\gamma_j^{(a)}})d^{(a)}_j\sqrt{\frac{\log p}{\mymin_{k\in\overline{\cN}_j^{(a)}}\mymin_{i\in[p]}n_{i,k}}},
        \end{split}
    \end{equation*}
    where we have applied the fact that $\|\Theta^{(a)*}_{j,:}\|_1\leq \sqrt{d^{(a)}_j}\|\Theta^{(a)*}_{j,:}\|_2$. While for the $\|\cdot\|_{\lambda^{(a,j)},1}$ error, note that $$\frac{\|\Theta^{(a)*}_{j,:}\|_{\lambda^{(a,j)},1}}{\Theta^{(a)*}_{j,j}}\leq \frac{\|\Sigma^*\|_{\infty}\|\Theta^{(a)*}_{j,:}\|_1^2}{(\Theta^{(a)*}_{j,j})^2}\sqrt{\frac{\log p}{\mymin_{k\in \overline{\cN}_j^{(a)}}\mymin_{i\in [p]}n_{i,k}}},$$
    and hence we have
    \begin{equation*}
        \begin{split}
            \|\widehat{\Theta}^{(a)}_{j,:}-\Theta^{(a)*}_{j,:}\|_{\lambda^{(a,j)},1}\leq &C\frac{\|\Sigma^*\|_{\infty}^2\|\Sigma^*_{j,:}\|_2\|\Theta^{(a)*}_{j,:}\|_1^3}{\lambda_{\mymin}(\Sigma^*)\Theta^{(a)*}_{j,j}}\frac{\sqrt{d_j^{(a)}}\log p}{\mymin_{k\in\overline{\cN}_j^{(a)}}\mymin_{i\in[p]}n_{i,k}}\\
            &+\frac{C\|\Sigma^*\|_{\infty}^2}{\lambda_{\mymin}(\Sigma^*)}\frac{\|\Theta^{(a)*}_{j,:}\|_1^2}{\Theta^{(a)*}_{j,j}}\frac{d_j^{(a)}\log p}{\mymin_{k\in \cN_k^{(a)}}\mymin_{i\in [p]}n_{i,k}}\\
            \leq&C\kappa_{\Sigma^*}\frac{\|\Sigma^*\|_{\infty}^2}{\lambda_{\mymin}(\Sigma^*)}\frac{\|\Theta^{(a)*}_{j,:}\|_1^2}{\Theta^{(a)*}_{j,j}}\frac{d_j^{(a)}\log p}{\mymin_{k\in \cN_k^{(a)}}\mymin_{i\in [p]}n_{i,k}}.
        \end{split}
    \end{equation*}

    While for the corresponding bounds for $\widetilde{\Theta}^{(a)}_{j,:}-\Theta^{(a)*}_{j,:}$, we only need to show that $|(\widetilde{\Theta}^{(a)}_{j,j})^{-1}-(\Theta^{(a)*}_{j,j})^{-1}|$ satisfies similar bounds to the ones we showed for $|(\widehat{\Theta}^{(a)}_{j,j})^{-1}-(\Theta^{(a)*}_{j,j})^{-1}|$. Specifically,
    \begin{equation*}
        \begin{split}
            |(\widetilde{\Theta}^{(a)}_{j,j})^{-1}-(\Theta^{(a)*}_{j,j})^{-1}|=&\left|\widehat{\overline{\theta}}^{(a,j)\top}\widetilde{\Sigma}\widehat{\overline{\theta}}^{(a,j)} - \overline{\theta}^{(a,j)*\top}\Sigma^*\overline{\theta}^{(a,j)*}\right|\\
            \leq & |2(\widehat{\overline{\theta}}^{(a,j)}-\overline{\theta}^{(a,j)*})^{\top}\Sigma^*\overline{\theta}^{(a,j)*}|+ |\widehat{\overline{\theta}}^{(a,j)\top}(\widetilde{\Sigma}-\Sigma^*)\widehat{\overline{\theta}}^{(a,j)}|\\
            &+|(\widehat{\overline{\theta}}^{(a,j)}-\overline{\theta}^{(a,j)*})^{\top}\Sigma^*(\widehat{\overline{\theta}}^{(a,j)}-\overline{\theta}^{(a,j)*})|. 
        \end{split}
    \end{equation*}
    Since $\Sigma^*\overline{\theta}^{(a,j)*} = e_j$, the $j$th canonical vector in $
    \mathbb{R}^p$, and $\widehat{\overline{\theta}}^{(a,j)}_j=\overline{\theta}^{(a,j)*}_j=1$, we know that the first term above is zero. On the other hand, Assumption \ref{assump:inference_n_B} and \eqref{eq:debias_theta_err} suggest that $\|\widehat{\overline{\theta}}^{(a,j)}\|_1\leq \|\overline{\theta}^{(a,j)*}\|_1=\frac{\|\Theta^{(a)*}_{j,:}\|_1}{\Theta^{(a)*}_{j,j}}$, and suggest the following bounds:
    \begin{equation*}
        \begin{split}
            |(\widetilde{\Theta}^{(a)}_{j,j})^{-1}-(\Theta^{(a)*}_{j,j})^{-1}|\leq &\lambda_{\max}(\Sigma^*)\|\widehat{\overline{\theta}}^{(a,j)}-\overline{\theta}^{(a,j)*}\|_2^2+C\|\overline{\theta}^{(a,j)*}\|_1\sum_k\|\widetilde{\Sigma}_{k,:}-\Sigma^*_{k,:}\|_{\infty}|\widehat{\overline{\theta}}^{(a,j)}_k|\\
            \leq&\lambda_{\max}(\Sigma^*)\|\widehat{\overline{\theta}}^{(a,j)}-\overline{\theta}^{(a,j)*}\|_2^2+\|\widehat{\overline{\theta}}^{(a,j)}-\overline{\theta}^{(a,j)*}\|_{\lambda^{(a,j)},1}\\
            &+ C\|\Sigma^*\|_{\infty}\|\overline{\theta}^{(a,j)*}\|_1^2\sqrt{\frac{\log p}{\mymin_{k\in \overline{\cN}_j^{(a)}}\mymin_{i\in [p]}n_{i,k}}}\\
            \leq &\frac{C\kappa_{\Sigma^*}\|\Sigma^*\|_{\infty}\|\Theta^{(a)*}_{j,:}\|_1}{\Theta^{(a)*}_{j,j}}\sqrt{\frac{d_j^{(a)}\log p}{\mymin_{k\in \overline{\cN}_j^{(a)}}\mymin_{i\in [p]}n_{i,k}}},
        \end{split}
    \end{equation*}
    where we have applied the error bound for $\widetilde{\Sigma}-\Sigma^*$ in Lemma \ref{lem:SampleCov_entry_err} and the condition on tuning parameters $\lambda^{(a,j)}$ in the second line. The last line is due to the fact that $\|\overline{\theta}^{(a,j)*}\|_1=\frac{\|\Theta^{(a)*}_{j,:}\|_1}{\Theta^{(a)*}_{j,j}}\leq \sqrt{d^{(a)}_j}\frac{\|\Theta^{(a)*}_{j,:}\|_2}{\Theta^{(a)*}_{j,j}}\leq \sqrt{d^{(a)}_j}\kappa_{\Theta^*}$. In addition, since $\|\Theta^{(a)*}\|_1\leq \sqrt{d_j^{(a)}}\|\Theta^{(a)*}\|_2\leq \frac{\sqrt{d_j^{(a)}}}{\lambda_{\min}(\Theta^{(a)*}}$, Assumption \ref{assump:inference_n_B} can further imply $|(\widetilde{\Theta}^{(a)}_{j,j})^{-1}-(\Theta^{(a)*}_{j,j})^{-1}|\leq \frac{1}{2\Theta^{(a))*}_{j,j}}$. The rest of the arguments for bounding $\|\widetilde{\Theta}^{(a)}_{j,:}-\Theta^{(a)*}_{j,:}\|_2$, $\|\widetilde{\Theta}^{(a)}_{j,:}-\Theta^{(a)*}_{j,:}\|_1$, $\|\widetilde{\Theta}^{(a)}_{j,:}-\Theta^{(a)*}_{j,:}\|_{\lambda^{(a,j)},1}$ follow similarly.
\end{proof}

\begin{proof}[Proof of Theorem~\ref{thm:nb_lasso_debias_decomp}]
First we start with some algebra to decompose $\widetilde{\theta}^{(a)}_b-\theta^{(a)*}_b=\widetilde{\theta}^{(a)}_b+\frac{\Theta^{*}_{a,b}}{\Theta^*_{a,a}}$ into terms $B$ and $E$. Here we applied the fact that $\theta^{(a)*}_{\backslash a}=-\Theta^{*-1}_{a,a}\Theta^*_{:,a}$. By definition \eqref{eq:debiased_nb_lasso},
\begin{align*}
    \widetilde{\theta}^{(a)}_{b}-\theta^{(a)*}_b=&\widehat{\theta}^{(a)}_b-\theta^{(a)*}_b-\widehat{\Theta}^{(a)}_{b,:}(\widehat{\Sigma}\widehat{\theta}^{(a)}-\widehat{\Sigma}_{:,a})\\
    =&(\Theta^{(a)*}_{b,:}\Sigma^*-\widehat{\Theta}^{(a)}_{b,:}\widehat{\Sigma})(\widehat{\theta}^{(a)}-\theta^{(a)*})-\widehat{\Theta}^{(a)}_{b,:}(\widehat{\Sigma}\theta^{(a)*}-\widehat{\Sigma}_{:,a})\\
    =&(\Theta^{(a)*}_{b,:}\Sigma^*-\widehat{\Theta}^{(a)}_{b,:}\widehat{\Sigma})(\widehat{\theta}^{(a)}-\theta^{(a)*})+\Theta^{*-1}_{a,a}(\widehat{\Theta}^{(a)}_{b,:}-\Theta^{(a)*}_{b,:})(\widehat{\Sigma}-\Sigma^*)\Theta^{*}_{:,a}\\
    &+\Theta^{*-1}_{a,a}\Theta^{(a)*}_{b,:}(\widehat{\Sigma}-\Sigma^*)\Theta^{*}_{:,a},
\end{align*}
where the second line is due to that $(\Theta^{(a)*}_{b,:}\Sigma^*)_{\backslash a} = e_b^\top$, and the third line holds since $\widehat{\Sigma}\theta^{(a)*}-\widehat{\Sigma}_{:,a}=-\Theta^{*-1}_{a,a}\widehat{\Sigma}\Theta^*_{:,a}$, and $\widehat{\Theta}^{(a)}_{b,:}\Sigma^*\Theta^*_{:,a}=\widehat{\Theta}^{(a)}_{b,:}e_a=0$.
Let 
\begin{equation}\label{eq:B_def}
    B=(\Theta^{(a)*}_{b,:}\Sigma^*-\widehat{\Theta}^{(a)}_{b,:}\widehat{\Sigma})(\widehat{\theta}^{(a)}-\theta^{(a)*})+\Theta^{*-1}_{a,a}(\widehat{\Theta}^{(a)}_{b,:}-\Theta^{(a)*}_{b,:})(\widehat{\Sigma}-\Sigma^*)\Theta^{*}_{:,a},
\end{equation}
and $E=\Theta^{*-1}_{a,a}\Theta^{(a)*}_{b,:}(\widehat{\Sigma}-\Sigma^*)\Theta^{*}_{:,a}$. We will show an upper bound for $|B|$ and a normal approximation for $E$.
\subsubsection{Bounding the bias term \texorpdfstring{$B$}{Lg}}
We first further decompose the bias term $B$ and upper bound it by functions of estimation errors $\widehat{\theta}^{(a)}-\theta^{(a)*}$ and $\widehat{\Theta}^{(a)}_{b,:}-\Theta^{(a)*}_{b,:}$. 
One can show that
\begin{align*}
   &\left|(\Theta^{(a)*}_{b,:}\Sigma^*-\widehat{\Theta}^{(a)}_{b,:}\widehat{\Sigma})(\widehat{\theta}^{(a)}-\theta^{(a)*})\right|\\
   =&\Big|\Theta^{(a)*}_{b,:}(\Sigma^*-\widehat{\Sigma})(\widehat{\theta}^{(a)}-\theta^{(a)*})+(\Theta^{(a)*}_{b,:}-\widehat{\Theta}^{(a)}_{b,:})\Sigma^*(\widehat{\theta}^{(a)}-\theta^{(a)*})\\
   &+(\Theta^{(a)*}_{b,:}-\widehat{\Theta}^{(a)}_{b,:})(\widehat{\Sigma}-\Sigma^*)(\widehat{\theta}^{(a)}-\theta^{(a)*})\Big|\\
   \leq&\|\Theta^{(a)*}_{b,:}\|_1\|\Sigma^*_{\overline{N}_b^{(a)},:}-\widehat{\Sigma}_{\overline{N}_b^{(a)},:}\|_{\infty}\|\widehat{\theta}^{(a)}-\theta^{(a)*}\|_1\\
   &+\|\Sigma^*\|\|\Theta^{(a)*}_{b,:}-\widehat{\Theta}^{(a)}_{b,:}\|_2\|\widehat{\theta}^{(a)}-\theta^{(a)*}\|_2\\
   &+\sum_{j,k}|\widehat{\Sigma}_{j,k}-\Sigma^*_{j,k}||\widehat{\Theta}^{(a)}_{b,j}-\Theta^*_{b,j}||\widehat{\theta}^{(a)}_k-\theta^{(a)*}_k|.
\end{align*}
Since 
\begin{align*}
    \overline{\cN}^{(a)}_b=\{j:\Theta^{(a)*}_{b,j}\neq 0\}=&\{j:\Theta^*_{b,j}-\frac{\Theta^*_{a,b}\Theta^*_{a,j}}{\Theta^*_{a,a}}\neq 0\}\\
    \subseteq&\{j:\Theta^*_{b,j}\neq 0\}\cup\{j\neq a:\Theta^*_{a,j}\neq 0\}=\cN_a\cup\overline{\cN}_b,
\end{align*} 
we have $d_b^{(a)}\leq d_a + d_b + 1$, and hence the sample size condition in Theorem~\ref{thm:nb_lasso_err} and Lemma~\ref{lem:debias_nb_lasso_err} can be implied by Assumption~\ref{assump:inference_n_B}. By Lemma~\ref{lem:SampleCov_entry_err}, Lemma~\ref{lem:debias_nb_lasso_err} and Theorem~\ref{thm:nb_lasso_err}, one can show that
\begin{equation}\label{eq:bias_1_bnds}
    \begin{split}
        &\|\Theta^{(a)*}_{b,:}\|_1\|\Sigma^*_{\overline{N}_b^{(a)},:}-\widehat{\Sigma}_{\overline{\cN}_b^{(a)},:}\|_{\infty}\|\widehat{\theta}^{(a)}-\theta^{(a)*}\|_1\\\leq &C\kappa_{\Sigma^*}(\sqrt{\gamma_a}+1)\|\Sigma^*\|_{\infty}^2\|\Theta^{(a)*}_{b,:}\|_1\|\Theta^*_{:,a}\|_1\frac{(d_a+d_b+1)\log p}{\mymin_{j\in \cN_a\cup \overline{\cN}^{(a)}_{b}}\mymin_{k\in[p]}n_{j,k}},\\
   &\|\Sigma^*\|\|\Theta^{(a)*}_{b,:}-\widehat{\Theta}^{(a)}_{b,:}\|_2\|\widehat{\theta}^{(a)}-\theta^{(a)*}\|_2\\
   \leq&C\kappa_{\Sigma^*}^3\|\Sigma^*\|_{\infty}^2\|\Theta^*_{:,a}\|_1\|\Theta^{(a)*}_{:,b}\|_1\frac{(d_a+d_b+1)\log p}{\mymin_{j\in \cN_a\cup\overline{\cN}^{(a)}_{b}}\mymin_{k}n_{j,k}},\\
   &\sum_{j,k}|\widehat{\Sigma}_{j,k}-\Sigma^*_{j,k}||\widehat{\Theta}^{(a)}_{b,j}-\Theta^*_{b,j}||\widehat{\theta}^{(a)}_k-\theta^{(a)*}_k|\\
   \leq&\frac{C\Theta^*_{a,a}}{\|\Theta^*_{:,a}\|_1}\|\widehat{\Theta}^{(a)}_{b,:}-\Theta^{(a)*}_{b,:}\|_2\sum_{k=1}^p\lambda^{(a)}_k|\widehat{\theta}^{(a)}_k-\theta^{(a)*}_k|\\
   \leq &\frac{C\kappa_{\Sigma^*}^2\|\Sigma^*\|_{\infty}^3}{\lambda_{\mymin}(\Sigma^*)}\|\Theta^*_{:,a}\|_1\|\Theta^{(a)*}_{:,b}\|_1\left(\frac{(d_a+d_b+1)\log p}{\mymin_{j\in \cN_a\cup\overline{\cN}^{(a)}_b}\mymin_{k\in [p]}n_{j,k}}\right)^{\frac{3}{2}}. 
    \end{split}
\end{equation}

Here $\kappa_{\Sigma^*}=\frac{\lambda_{\mymax}(\Sigma^*)}{\lambda_{\mymin}(\Sigma^*)}$, and we have applied the fact that $\Theta^*_{a,a}\geq \lambda_{\mymax}^{-1}(\Sigma^*)$, $\|\Theta^{(a)*}_{:,b}\|_2\leq \lambda^{-1}_{\mymin}(\Sigma^*)$. Also noting that we have the condition 
$$
\frac{(d_a+d_b+1)\log p}{n_1^{(a,b)}}\leq \frac{C\lambda^2_{\mymin}(\Sigma^*)}{\|\Sigma^*\|_{\infty}^2},
$$
suggesting the upper bound for the third term in \eqref{eq:bias_1_bnds} can be dominated by the second term.
Hence we have
\begin{align*}
    &\left|(\Theta^{(a)*}_{b,:}\Sigma^*-\widehat{\Theta}^{(a)}_{b,:}\widehat{\Sigma})(\widehat{\theta}^{(a)}-\theta^{(a)*})\right|\\
   \leq&C\kappa_{\Sigma^*}(\kappa_{\Sigma^*}^2+\sqrt{\gamma_a})\|\Sigma^*\|_{\infty}^2\|\Theta^*_{:,a}\|_1\|\Theta^{(a)*}_{:,b}\|_1\frac{(d_a+d_b+1)\log p}{n_1^{(a,b)}}.
\end{align*}

While for the second term in \eqref{eq:B_def}, one can show that
\begin{align*}
    &|\Theta^{*-1}_{a,a}(\widehat{\Theta}^{(a)}_{b,:}-\Theta^{(a)*}_{b,:})(\widehat{\Sigma}-\Sigma^*)\Theta^{*}_{:,a}|\\
    \leq& \frac{\|\Theta^*_{:,a}\|_1}{\Theta^*_{a,a}}\|\widehat{\Theta}^{(a)}_{b,:}-\Theta^{(a)*}_{b,:}\|_1\|\widehat{\Sigma}_{:,\overline{\cN}_a}-\Sigma^*_{:,\overline{\cN}_a}\|_{\infty}\\
    \leq &C\kappa_{\Sigma^*}(\kappa_{\Sigma^*}+\sqrt{\gamma^{(a)}_b})\|\Sigma^*\|_{\infty}^2\|\Theta^*_{:,a}\|_1\|\Theta^{(a)*}_{b,:}\|_1\frac{(d_a+d_b+1)\log p}{n_1^{(a,b)}}.
\end{align*}
Therefore, with probability at least $1-Cp^{-c}$, $$|B|\leq C_1(\Theta^*)\frac{(d_a+d_b+1)\log p}{n_1^{(a,b)}},$$ with $C_1(\Theta^*)=C\kappa_{\Sigma^*}(\kappa_{\Sigma^*}^2+\sqrt{\gamma_a}+\sqrt{\gamma_b^{(a)}})\|\Sigma^*\|_{\infty}^2\|\Theta^*_{:,a}\|_1\|\Theta^{(a)*}_{:,b}\|_1$.
\subsubsection{Normal approximation of term \texorpdfstring{$E$}{Lg}}\label{sec:proof_normalE}
Let $\delta^{(i)},N\in \bR^{p\times p}$ such that $\delta^{(i)}_{j,k}=\ind{j,k\in V_i}$ and $N_{j,k}=n_{j,k}$. Then we can write $\widehat{\Sigma}=\sum_{i=1}^nx_ix_i^\top\circ \delta^{(i)}\oslash N$, where $\circ$ represents Hadamard product and $\oslash$ represents Hadamard division. With a little abuse of notation, when $N_{j,k}=0$, we let $\frac{\delta^{(i)}_{j,k}}{N_{j,k}}=0$. 
To apply the central limit theorem and establish the normal approximation result, here we rewrite $E$ as follows:
\begin{equation*}
    E=\sum_{i=1}^n\Theta^{*-1}_{a,a}\Theta^{(a)*}_{b,:}[(x_ix_i^\top-\Sigma^*) \circ \delta^{(i)}\oslash N]\Theta^*_{:,a}.
\end{equation*}
Here we applied the fact that $(\sum_{i=1}^n \Sigma^* \circ \delta^{(i)}\oslash N)_{j,k}=\Sigma^*_{j,k}$ as long as $N_{j,k}>0$, and $\mymin_{j\in\overline{\cN}_a,k\in \overline{\cN}^{(a)}_{b}}n_{j,k}>0$. Since the $n$ terms inside the summation are independent, we can apply the central limit theorem. The variance $\sigma_n^2(a,b)$ can be calculated as follows:
\begin{align*}
    \sigma_n^2(a,b)=&\Theta^{*-2}_{a,a}\sum_{i=1}^n\bE(\Theta^{(a)*}_{b,:}[(x_ix_i^\top-\Sigma^*) \circ \delta^{(i)}\oslash N]\Theta^*_{:,a})^2\\
    =&\Theta^{*-2}_{a,a}\sum_{i=1}^n\mathrm{Var}\left(\sum_{j,k}\Theta^{(a)*}_{b,j}\Theta^*_{k,a}\delta^{(i)}_{j,k}n_{j,k}^{-1}x_{i,j}x_{i,k}\right)\\
    =&\Theta^{*-2}_{a,a}\sum_{j,k,j',k'}\frac{\Theta^{(a)*}_{b,j}\Theta^*_{k,a}\Theta^{(a)*}_{b,j'}\Theta^*_{k',a}}{n_{j,k}n_{j',k'}}\mathcal{T}^*_{j,k,j',k'}\sum_{i=1}^n\delta^{(i)}_{j,k}\delta^{(i)}_{j',k'}\\
    =&\frac{1}{(\Theta^{*}_{a,a})^2}\mathcal{T}^{(n)*}\times_1\Theta^{(a)*}_{:,b}\times_2\Theta^{*}_{:,a}\times_3\Theta^{(a)*}_{:,b}\times_4\Theta^{*}_{:,a},
\end{align*}
where we used the definition $$\mathcal{T}^*_{j,k,j',k'}=\mathrm{Cov}(x_{i,j}x_{i,k},x_{i,j'}x_{i,k'})=\Sigma^*_{j,j'}\Sigma^*_{k,k'}+\Sigma^*_{j,k'}\Sigma^*_{k,j'}$$ on the third line, and $\mathcal{T}^{(n)*}_{j,k,j',k'}=\mathcal{T}^*_{j,k,j',k'}\frac{n_{j,k,j',k'}}{n_{j,k}n_{j',k'}}$. In the following, we will verify the Lyapounov's condition so that we can apply the central limit theorem.

Define $U, U^{(\delta,i)}, \epsilon^{(i)}\in \bR^{p\times p}$ by \begin{align*}
    U_{j,k}=\frac{\Theta^{(a)*}_{j,b}\Theta^*_{k,a} + \Theta^{(a)*}_{k,b}\Theta^*_{j,a}}{2\Theta^*_{a,a}},
\end{align*}
$U^{(\delta,i)}_{j,k}=U_{j,k}\frac{\delta^{(i)}_{j,k}}{n_{j,k}}$, and $\epsilon^{(i)}_{j,k}=x_{i,j}x_{i,k}-\Sigma^*_{j,k}$. Then we can write $E=\sum_{i=1}^n\langle U^{(\delta,i)},\epsilon^{(i)}\rangle$. Now we will show that $\sum_{i=1}^n\bE|\langle U^{(\delta,i)},\epsilon^{(i)}\rangle|^{2+\delta}=o(\sigma_n^{2+\delta}(a,b))$ for $\delta=\frac{4}{\varepsilon}>0$. First we bound the sub-exponential norm $\|\cdot\|_{\psi_1}$ of $\langle U^{(\delta,i)},\epsilon^{(i)}\rangle$, which is defined in Definition~\ref{def:psi_alpha}. One can show that $$\left\|\sum_{j,k}U^{(\delta,i)}_{j,k}\epsilon^{(i)}_{j,k}\right\|_{\psi_1}\leq \sum_{j,k}|U_{j,k}|\frac{\delta^{(i)}_{j,k}}{n_{j,k}}\|x_{i,j}x_{i,k}-\Sigma^*_{j,k}\|_{\psi_1}\leq C\sum_{j,k}|U_{j,k}|\frac{\delta^{(i)}_{j,k}}{n_{j,k}}\|\Sigma^*\|_{\infty}.$$
Then by the definition of the norm $\|\cdot\|_{\psi_1}$, we have
\begin{align*}
    &\sum_{i=1}^n\bE|\langle U^{(\delta,i)},\epsilon^{(i)}\rangle|^{2+\delta}\\
    \leq &\sum_{i=1}^n\left(C(2+\delta)\|\Sigma^*\|_{\infty}\right)^{2+\delta}\left(\sum_{j,k}|U_{j,k}|\delta^{(i)}_{j,k}n^{-1}_{j,k}\right)^{2+\delta}\\
    \leq &\left(C(2+\delta)\|\Sigma^*\|_{\infty}\right)^{2+\delta}\left(\sum_{j,k}|U_{j,k}|^2n_{j,k}^{-1}\right)^{1+\delta/2}\sum_{i=1}^n\left(\sum_{j\in \overline{\cN}^{(a)}_b,k\in \overline{\cN}_a}\delta^{(i)}_{j,k}n_{j,k}^{-1}\right)^{1+\delta/2},
\end{align*}
where the third line is due to the Cauchy-Schwarz inequality and the fact that $U_{j,k}\neq 0$ only when $j\in \overline{\cN}^{(a)}_b$, $k\in \overline{\cN}_a$, or $j\in \overline{\cN}_a$, $k\in \overline{\cN}^{(a)}_b$. In addition, for the third term in the multiplication on the third line, we have
\begin{equation}
    \begin{split}
       &\sum_{i=1}^n\left(\sum_{j\in \overline{\cN}^{(a)}_b,k\in \overline{\cN}_a}\delta^{(i)}_{j,k}n_{j,k}^{-1}\right)^{1+\delta/2}\\
       =&(d_a+d_b+1)^{2+\delta}\sum_{i=1}^n\left(\frac{1}{(d_a+d_b+1)^2}\sum_{j\in \overline{\cN}^{(a)}_b,k\in \overline{\cN}_a}\delta^{(i)}_{j,k}n_{j,k}^{-1}\right)^{1+\delta/2}\\
    \leq&(d_a+d_b+1)^{\delta}\sum_{i=1}^n\sum_{j\in \overline{\cN}^{(a)}_b,k\in \overline{\cN}_a}\delta^{(i)}_{j,k}n_{j,k}^{-(1+\delta/2)}\\
    \leq &(d_a+d_b+1)^{2+\delta}(n_2^{(a,b)})^{-\delta/2},
    \end{split}
\end{equation}
where the third line is due to Jensen's inequality and the fact that $g(x)=x^{2+\delta}$ is a convex function; the fourth line is due to that $\sum_{i=1}^n\delta^{(i)}_{j,k}=n_{j,k}$, $|\overline{\cN}^{(a)}_b|, |\overline{\cN}_a|\leq d_a+d_b+1$, and $n_2^{(a,b)}=\mymin_{j\in \overline{\cN}_a,k\in \overline{\cN}^{(a)}_{b}}n_{j,k}$, $n_{j,k} = n_{k,j}$. On the other hand, we can lower bound $\sigma_n^2(a,b)$ as follows:
\begin{align*}
    \sigma_n^2(a,b)=&\sum_{i=1}^n\bE\langle U^{(\delta,i)},\epsilon^{(i)}\rangle^2\\
    = &\sum_{i=1}^n\sum_{j,j',k,k'}\mathcal{T}^*_{j,k,j',k'}U^{(\delta,i)}_{j,k}U^{(\delta,i)}_{j',k'}\\
    = &\sum_{i=1}^n\sum_{j,j'}\Sigma^*_{j,j'}U^{(\delta,i)}_{j,:}\Sigma^*(U^{(\delta,i)}_{j',:})^\top+\sum_{j,k'}\Sigma^*_{j,k'}U^{(\delta,i)}_{j,:}\Sigma^*U^{(\delta,i)}_{:,k'}.
\end{align*}    
Due to the symmetry of $U^{(\delta,i)}$, we further have
\begin{align*}
    \sigma_n^2(a,b)= &2\sum_{i=1}^n\langle\Sigma^*,U^{(\delta,i)}\Sigma^*(U^{(\delta,i)})\rangle\\
    =&2\sum_{i=1}^n\|\Sigma^{*\frac{1}{2}}U^{(\delta,i)}\Sigma^{*\frac{1}{2}}\|_F^2\\
    \geq&2\lambda_{\min}^2(\Sigma^*)\sum_{j,k}U_{j,k}^2n_{j,k}^{-1}.
\end{align*}
Therefore, combining the upper bounds for $\sum_{i=1}^n\bE|\langle U^{(\delta,i)},\epsilon^{(i)}\rangle|^{2+\delta}$ and $\sigma_n^{2+\delta}(a,b)$ leads to the following:
\begin{align*}
    \frac{\sum_{i=1}^n\bE|\langle U^{(\delta,i)},\epsilon^{(i)}\rangle|^{2+\delta}}{\sigma_n^{2+\delta}(a,b)}\leq &\left(\frac{C(2+\delta)\|\Sigma^*\|_{\infty}}{\lambda_{\mymin}(\Sigma^*)}\right)^{2+\delta}(d_a+d_b+1)^{2+\delta}(n_2^{(a,b)})^{-\delta/2}\\
    \leq&C_{\epsilon}^{\delta/2}(\Sigma^*)(d_a+d_b+1)^{2+\delta}(n_2^{(a,b)})^{-\delta/2}.
\end{align*}
Since $\delta=\frac{4}{\epsilon}$, $C_{\epsilon}(\Sigma^*)(d_a+d_b+1)^{2+\epsilon}=o(n_2(a,b))$, the inequality above implies $$\sum_{i=1}^n\bE|\langle U^{(\delta,i)},\epsilon^{(i)}\rangle|^{2+\delta}=o(\sigma_n^{2+\delta}(a,b)).$$ Now we can apply the central limit theorem to obtain the following convergence in distribution result:
\begin{align*}
    \sigma_n^{-1}(a,b)E\overset{d}{\rightarrow}\cN(0,1).
\end{align*}
\subsubsection{Normal Approximation Result}
Now that we have proved an upper bound for $B$ and the convergence in distribution result for $E$, a combination of these two results can lead to our final claim. To see this, note that
\begin{align*}
    \sigma_n^{-1}(a,b)\left(\theta^{(a)}_b+\frac{\Theta^*_{a,b}}{\Theta^*_{a,a}}\right)=\frac{B}{\sigma_n(a,b)}+\sigma_n^{-1}(a,b)E.
\end{align*}
As shown earlier, 
\begin{align*}
    |B|\leq &C_1(\Theta^*)\frac{(d_a+d_b+1)\log p}{n_1^{(a,b)}}\\
    \leq&C\kappa_{\Sigma^*}(\kappa_{\Sigma^*}^2+\sqrt{\gamma_a}+\sqrt{\gamma^{(a)}_b})\lambda^2_{\mymax}(\Sigma^*)\|\Theta^*_{:,a}\|_1\|\Theta^{(a)*}_{:,b}\|_1\frac{(d_a+d_b+1)\log p}{n_1^{(a,b)}},
\end{align*} 
and 
\begin{align*}
    \sigma_n(a,b)\geq&\sqrt{2}\lambda_{\mymin}(\Sigma^*)\sqrt{\sum_{j,k}|U_{j,k}|^2n_{j,k}^{-1}}\\
    \geq&\frac{\sqrt{2}\lambda_{\mymin}(\Sigma^*)}{2\Theta^*_{a,a}}\frac{\mymin_{(j,k)\in S_2(a, b)}\left|\Theta^{(a)*}_{b,j}\Theta^{*}_{a,k}+\Theta^{(a)*}_{b,k}\Theta^{*}_{a,j}\right|}{\sqrt{n_2^{(a,b)}}}\\
    \geq &\frac{\sqrt{2}}{2}\lambda^2_{\mymin}(\Sigma^*)\frac{\mymin_{(j,k)\in S_2(a, b)}\left|\Theta^{(a)*}_{b,j}\Theta^{*}_{a,k}+\Theta^{(a)*}_{b,k}\Theta^{*}_{a,j}\right|}{\sqrt{n_2^{(a,b)}}}.
\end{align*}
Hence with probability at least $1-Cp^{-c}$,
\begin{align*}
    \left|\frac{B}{\sigma_n(a,b)}\right|\leq C(\Theta^*;a,b)(\kappa_{\Theta^*}^2+\sqrt{\gamma_a}+\sqrt{\gamma^{(a)}_b})\frac{(d_a+d_b+1)\log p}{\sqrt{n_1^{(a,b)}}}\sqrt{\frac{n_2^{(a,b)}}{n_1^{(a,b)}}},
\end{align*}
where $$C(\Theta^*;a,b)=\frac{C\kappa_{\Theta^*}^3\|\Theta^*_{:,a}\|_1\|\Theta^{(a)*}_{:,b}\|_1}{\mymin_{(j,k)\in S_2(a, b)}\left|\Theta^{(a)*}_{b,j}\Theta^{*}_{a,k}+\Theta^{(a)*}_{b,k}\Theta^{*}_{a,j}\right|}$$ and $\kappa_{\Theta^*} = \frac{\lambda_{\mymax}(\Theta^*)}{\lambda_{\mymin}(\Theta^*)}=\frac{\lambda_{\mymax}(\Sigma^*)}{\lambda_{\mymin}(\Sigma^*)}=\kappa_{\Sigma^*}$. When Assumptions~\ref{assump:inference_n_B}-\ref{assump:inference_n_BE} hold, $\left|\frac{B}{\sigma_n(a,b)}\right|=o_p(1)$, and hence $$\sigma_n^{-1}(a,b)\left(\theta^{(a)}_b+\frac{\Theta^*_{a,b}}{\Theta^*_{a,a}}\right)\overset{d}{\rightarrow}\cN(0,1).$$
\end{proof}
\begin{proof}[Proof of Proposition \ref{prop:var_bnds}]
As has been shown in the proof of Theorem~\ref{thm:nb_lasso_debias_decomp}, 
\begin{align*}
    \sigma_n(a,b)=&\sqrt{2}\left(\sum_{i=1}^n\|\Sigma^{*\frac{1}{2}}U^{(\delta,i)}\Sigma^{*\frac{1}{2}}\|_F^2\right)^{\frac{1}{2}}\\
    \geq&\sqrt{2}\lambda_{\mymin}(\Sigma^*)\frac{\mymin_{(j,k)\in S_2(a, b)}\left|\Theta^{(a)*}_{b,j}\Theta^{*}_{a,k}+\Theta^{(a)*}_{b,k}\Theta^{*}_{a,j}\right|}{2\Theta^*_{a,a}\sqrt{n_2^{(a,b)}}}.
\end{align*}
Similarly, for the upper bound, one has
\begin{align*}
    \sigma_n^2(a,b)=&2\sum_{i=1}^n\|\Sigma^{*\frac{1}{2}}U^{(\delta,i)}\Sigma^{*\frac{1}{2}}\|_F^2\\
    \leq &2\lambda_{\max}^2(\Sigma^*)\sum_{(j,k)\in S_2(a, b)}U_{j,k}^2n_{j,k}^{-1}\\
    \leq &\lambda_{\max}^2(\Sigma^*)\frac{2\|\Theta^{(a)*}_{:,b}\|_2^2\|\Theta^{*}_{:,a}\|_2^2}{(\Theta_{a,a}^*)^2}(n_2^{(a,b)})^{-1},
\end{align*}
one can show that
\begin{align*}
    \sigma_n^2(a,b)\leq&\sum_{i=1}^n2\lambda_{\mymax}\left(\Sigma^*_{\overline{\cN}_a,\overline{\cN}_a}\otimes \Sigma^*_{\overline{\cN}^{(a)}_{b},\overline{\cN}^{(a)}_{b}}\right)\|U^{(\delta,i)}\|_F^2\\
    \leq &2\lambda^2_{\mymax}(\Sigma^*)\sum_{j,k}|U_{j,k}|^2n_{j,k}^{-1}\\
    \leq&\frac{2\lambda^2_{\mymax}(\Sigma^*)\|\Theta^{(a)*}_{:,b}\|_2^2\|\Theta^*_{:,a}\|_2^2}{(\Theta^*_{a,a})^2n_2^{(a,b)}}\\
    \leq &2\kappa_{\Sigma^*}^4(n_2^{(a,b)})^{-1}.
\end{align*}
\end{proof}
\subsection{Proof of Proposition~\ref{prop:var_est}: Consistency of Variance Estimator}
First we define the following $4$th order tensors that would be useful in our subsequent proof: $\widehat{\mathcal{T}}, \mathcal{U}^*, \widehat{\mathcal{U}}, \mathcal{N}\in \bR^{p\times p\times p\times p}$ that satisfy
\begin{align*}
    \widehat{\mathcal{T}}_{j,k,j',k'}=&\widetilde{\Sigma}_{j,j'}\widetilde{\Sigma}_{k,k'}+\widetilde{\Sigma}_{j,k'}\widetilde{\Sigma}_{k,j'},\\
    \mathcal{U}^*_{j,k,j',k'}=&\Theta^{(a)*}_{b,j}\overline{\theta}^{(a)}_k\Theta^{(a)*}_{b,j'}\overline{\theta}^{(a)}_{k'},\\
    \widehat{\mathcal{U}}_{j,k,j',k'}=&\widetilde{\Theta}^{(a)}_{b,j}\widehat{\overline{\theta}}^{(a)}_k\widetilde{\Theta}^{(a)}_{b,j'}\widehat{\overline{\theta}}^{(a)}_{k'},\\
    \mathcal{N}_{j,k,j',k'}=&\frac{n_{j,k,j',k'}}{n_{j,k}n_{j',k'}},
\end{align*}
where $\overline{\theta}^{(a)*}=\frac{\Theta^*_{:,a}}{\Theta^*_{a,a}}$.
Also recall the definition of tensor $\widehat{\mathcal{T}}^{(n)}$, then we can rewrite $\widehat{\sigma}_n^2(a,b)$ as follows
\begin{align*}
    \widehat{\sigma}_n^2(a,b) = \langle \widehat{\mathcal{T}}\circ \mathcal{N},\widehat{\mathcal{U}}\rangle = \langle \widehat{\mathcal{T}},\widehat{\mathcal{U}}\circ \mathcal{N}\rangle.
\end{align*}
Define $\mathcal{E}_1=\widehat{\mathcal{T}}-\mathcal{T}^*$, $\mathcal{E}_2=\widehat{\mathcal{U}}-\mathcal{U}^*$, then the estimation error for variance $\sigma_n^2(a,b)$ can be decomposed as follows:
\begin{align*}
    \widehat{\sigma}_n^2(a,b)-\sigma_n^2(a,b)=&\underbrace{\langle \cE_1\circ \cN, \cU^*\rangle}_{\mathrm{\RNum{1}}} + \underbrace{\langle\cT^*,\cE_2\circ\cN\rangle}_{\mathrm{\RNum{2}}} + \underbrace{\langle\cE_1\circ\cN,\cE_2\rangle}_{\mathrm{\RNum{3}}},
\end{align*}
To establish the consistency of the variance estimator $\widehat{\sigma}^2_n(a,b)$, we bound the three terms above separately.
\subsubsection{Bounding Error Term \texorpdfstring{$\mathrm{\RNum{1}}$}{Lg}}
First we provide an entry-wise estimation error bound for the $4$th order moment $\cT^*$.
\begin{lemma}[Entry-wise error bound for $\cE_1$]\label{lem:T_entry_err}
    With probability at least $1-Cp^{-c}$,
    \begin{align*}
        |(\cE_1)_{j,k,j',k'}|\leq C\|\Sigma^*\|_{\infty}^2\Big(&\sqrt{\frac{\log p}{n_{j,j'}}}+\sqrt{\frac{\log p}{n_{k,k'}}}+\frac{\log p}{\sqrt{n_{j,j'}n_{k,k'}}}\\
        &+\sqrt{\frac{\log p}{n_{j,k'}}}+\sqrt{\frac{\log p}{n_{k,j'}}}+\frac{\log p}{\sqrt{n_{j,k'}n_{k,j'}}}\Big).
    \end{align*}
\end{lemma}
Then by the definition of $\mathrm{\RNum{1}}$, one can show that with probability at least $1-Cp^{-c}$,
\begin{align*}
    |\mathrm{\RNum{1}}|\leq &C\|\Sigma^*\|_{\infty}^2\sum_{j,k,j',k'}\left(\frac{\sqrt{n_{j,k,j',k'}\log p}}{n_{j,k}n_{j',k'}}+\frac{\log p}{n_{j,k}n_{j',k'}}\right)\cU^*_{j,k,j',k'}\\
    \leq&C\|\Sigma^*\|_{\infty}^2\|\cU^*\|_1\mymax_{j,j'\in \overline{\cN}^{(a)}_b, k,k'\in \overline{\cN}_a}\left(\frac{\sqrt{\log p}}{n_{j,k}\sqrt{n_{j',k'}}}+\frac{\log p}{n_{j,k}n_{j',k'}}\right)\\
    \leq&C\|\Sigma^*\|_{\infty}^2\|\cU^*\|_1\left(\frac{\sqrt{\log p}}{(n_2^{(a,b)})^{3/2}}+\frac{\log p}{(n_2^{(a,b)})^2}\right)
\end{align*}
where the second line is due to that $\cU^*_{j,k,j',k'}$ is only nonzero for $j,j'\in\overline{\cN}^{(a)}_b$, $k,k'\in \overline{\cN}_a$.
\subsubsection{Bounding Error Term \texorpdfstring{$\mathrm{\RNum{2}}$}{Lg}}
Since $|\mathrm{\RNum{2}}|\leq \|\cT^*\|_{\infty}\|\cE_2\circ \cN\|_1\leq 2\|\Sigma^*\|_{\infty}^2\|\cE_2\circ \cN\|_1$, we will show an upper bound for $\|\cE_2\circ \cN\|_1$ in the following. By definition, 
\begin{align*}
    \cE_2\circ \cN=[\widetilde{\Theta}^{(a)}_{b,:}\otimes \widehat{\overline{\theta}}^{(a)}\otimes \widetilde{\Theta}^{(a)}_{b,:}\otimes \widehat{\overline{\theta}}^{(a)}-\Theta^{(a)*}_{b,:}\otimes \overline{\theta}^{(a)*}\otimes \Theta^{(a)*}_{b,:}\otimes \overline{\theta}^{(a)*}]\circ \cN.
\end{align*}
The following lemma provides a general upper bound for kronecker products.
\begin{lemma}\label{lem:kron_err}
Consider vectors $u^{(1)},\dots,u^{(4)}, \epsilon^{(1)},\dots,\epsilon^{(4)}\in \bR^p$. Let $i_l=2\lceil l/2\rceil -\ind{l\text{ is even}}$. If for all $1\leq l\leq 4$, $$\sum_{j=1}^p\left|\frac{\epsilon^{(l)}_j}{\sqrt{\mymin_{k}n_{j,k}}}\right|\leq \frac{C_1\|u^{(l)}\|_1}{\sqrt{\mymin_{j\in \mathrm{supp}(u^{(l)})}\mymin_{k}n_{j,k}}},\quad \|\epsilon^{(l)}\|_1\leq C_1\|u^{(l)}\|_1$$ for some universal constant $C_1$, 
then we have
\begin{align*}
    &\|[(u^{(1)}+\epsilon^{(1)})\otimes (u^{(2)}+\epsilon^{(2)})\otimes (u^{(3)}+\epsilon^{(3)})\otimes (u^{(4)}+\epsilon^{(4)})\\
    &-u^{(1)}\otimes u^{(2)}\otimes u^{(3)}\otimes u^{(4)}]\circ \cN\|_1\\
    \leq&C_2\sum_{l=1}^4\left[\sum_{j=1}^p\frac{\left|\epsilon^{(l)}_j\right|}{\sqrt{\mymin_{k}n_{j,k}}}\frac{\|u^{(i_{l})}\|_1}{\sqrt{\mymin_{j\in \mathrm{supp}(u^{(l)})}\mymin_{k}n_{j,k}}}\prod_{m\neq l,i_l}\|u^{(m)}\|_1\right],
\end{align*}
where $C_2$ is also a universal constant.
\end{lemma}
We can apply Lemma~\ref{lem:kron_err} with $u^{(1)}=u^{(3)}=\Theta^{(a)*}_{b,:}$, $u^{(2)}=u^{(4)}=\overline{\theta}^{(a)*}$, $\epsilon^{(1)}=\epsilon^{(3)}=\widetilde{\Theta}^{(a)}_{b,:}-\Theta^{(a)*}_{b,:}$, $\epsilon^{(2)}=\epsilon^{(4)}=\widehat{\overline{\theta}}^{(a)}-\overline{\theta}^{(a)*}=-\widehat{\theta}^{(a)}+\theta^{(a)*}$. By Lemma \ref{lem:debias_nb_lasso_err}, 
\begin{align*}
    \|\epsilon^{(1)}\|_1\leq& \frac{C(\kappa_{\Sigma^*}+\sqrt{\gamma_b^{(a)}})\|\Sigma^*\|_{\infty}\|\Theta^{(a)*}_{b,:}\|_1}{\lambda_{\mymin}(\Sigma^*)}d^{(a)}_b\sqrt{\frac{\log p}{\mymin_{k\in\overline{\cN}_b^{(a)}}\mymin_{i\in[p]}n_{i,k}}}\\
    \leq&C\|\Theta^{(a)*}_{b,:}\|_1=C\|u^{(1)}\|_1,
\end{align*}
where we have applied the sample size condition in Assumption~\ref{assump:inference_n_B}. In addition, Lemma~\ref{lem:debias_nb_lasso_err} implies
\begin{align*}
    \sum_{j=1}^p\frac{|\epsilon^{(1)}_j|}{\sqrt{\mymin_kn_{j,k}}}\leq& \frac{C\kappa_{\Sigma^*}\|\Sigma^*\|_{\infty}}{\lambda_{\mymin}(\Sigma^*)}\|\Theta^{(a)*}_{b,:}\|_1\frac{d_b^{(a)}\sqrt{\log p}}{\mymin_{k\in \cN_k^{(a)}}\mymin_{i\in [p]}n_{i,k}}\\
    \leq&\frac{\|u^{(1)}\|_1}{\sqrt{\mymin_{k\in \cN_k^{(a)}}\mymin_{i\in [p]}n_{i,k}}}.
\end{align*}
Furthermore, Theorem~\ref{thm:nb_lasso_err} suggests that
\begin{align*}
    \|\epsilon^{(2)}\|_1\leq& \frac{C(\sqrt{\gamma_a}+1)\|\Sigma^*\|_{\infty}\|\Theta^{*}_{a,:}\|_1}{\lambda_{\mymin}(\Sigma^*)\Theta^*_{a,a}}d_a\sqrt{\frac{\log p}{\mymin_{j\in\cN_a}\mymin_kn_{j,k}}}\\
    \leq&\frac{C\|\Theta^*_{a,:}\|_1}{\Theta^*_{a,a}},
\end{align*}
\begin{align*}
    \sum_{j=1}^p\frac{|\epsilon^{(2)}_j|}{\sqrt{\mymin_kn_{j,k}}}\leq& \frac{C\|\Sigma^*\|_{\infty}\|\Theta^*_{a,:}\|_1}{\lambda_{\mymin}(\Sigma^*)\Theta^*_{a,a}}\frac{d_a\sqrt{\log p}}{\mymin_{j\in \cN_a}\mymin_{k}n_{j,k}}\\
    \leq&\frac{\|u^{(2)}\|_1}{\sqrt{\mymin_{k\in \cN_a}\mymin_{i\in [p]}n_{i,k}}}.
\end{align*}
Therefore, 
\begin{align*}
    \|\cE_2\circ \cN\|_1\leq &\frac{C\kappa_{\Sigma^*}\|\Sigma^*\|_{\infty}\|\Theta^{(a)*}_{b,:}\|_1^2\|\theta^{(a)*}\|_1^2}{\lambda_{\mymin}(\Sigma^*)}\frac{(d_a+d_b+1)\sqrt{\log p}}{(n_1^{(a,b)})^{3/2}},
\end{align*}
which further implies that
\begin{align*}
    |\mathrm{\RNum{2}}|\leq \frac{C\kappa_{\Sigma^*}\|\Sigma^*\|_{\infty}^3\|\Theta^{(a)*}_{b,:}\|_1^2\|\theta^{(a)*}\|_1^2}{\lambda_{\mymin}(\Sigma^*)}\frac{(d_a+d_b+1)\sqrt{\log p}}{(n_1^{(a,b)})^{3/2}}.
\end{align*}
\subsubsection{Bounding Error Term \texorpdfstring{$\mathrm{\RNum{3}}$}{Lg}}
By Lemma~\ref{lem:T_entry_err}, one can bound $\mathrm{\RNum{3}}$ as follows:
\begin{align*}
    \mathrm{\RNum{3}}\leq&C\|\Sigma^*\|_{\infty}^2\sum_{j,k,j',k'}\left(\frac{\sqrt{(\log p)n_{j,k,j',k'}}}{n_{j,k}n_{j',k'}}+\frac{\log p}{n_{j,k}n_{j',k'}}\right)|(\cE_2)_{j,k,j',k'}|\\
    \leq&C\|\Sigma^*\|_{\infty}^2\sqrt{\log p}\sum_{j,k,j',k'}\frac{\sqrt{n_{j,k,j',k'}}}{n_{j,k}n_{j',k'}}|(\cE_2)_{j,k,j',k'}|\\
    &+C\|\Sigma^*\|_{\infty}^2\log p\sum_{j,k,j',k'}\frac{|(\cE_2)_{j,k,j',k'}|}{n_{j,k}n_{j',k'}}\\
    \leq&C\|\Sigma^*\|_{\infty}^2\sqrt{\log p}\|\cE_2\circ\cN^{(1)}\|_1+C\|\Sigma^*\|_{\infty}^2\log p\|\cE_2\circ \cN^{(2)}\|_1,
\end{align*}
where $\cN^{(1)},\cN^{(2)}\in \bR^{p\times p\times p\times p}$ satisfy $\cN^{(1)}_{j,k,j',k'}=\frac{\sqrt{n_{j,k,j',k'}}}{n_{j,k}n_{j',k'}}$ and $\cN^{(2)}_{j,k,j',k'}=\frac{1}{n_{j,k}n_{j',k'}}$. By Lemma~\ref{lem:kron_err2} (a similar result to Lemma~\ref{lem:kron_err} but with $\cN^{(1)}$ and $\cN^{(2)}$), we have
\begin{align*}
    \|\cE\circ \cN^{(1)}\|_1\leq &\frac{C\kappa_{\Sigma^*}\|\Sigma^*\|_{\infty}\|\Theta^{(a)*}_{b,:}\|_1^2\|\theta^{(a)*}\|_1^2}{\lambda_{\mymin}(\Sigma^*)}\frac{(d_a+d_b+1)\sqrt{\log p}}{(n_1^{(a,b)})^{2}},\\
    \|\cE\circ \cN^{(2)}\|_1\leq &\frac{C\kappa_{\Sigma^*}\|\Sigma^*\|_{\infty}\|\Theta^{(a)*}_{b,:}\|_1^2\|\theta^{(a)*}\|_1^2}{\lambda_{\mymin}(\Sigma^*)}\frac{(d_a+d_b+1)\sqrt{\log p}}{(n_1^{(a,b)})^{5/2}}.
\end{align*}
Therefore, 
\begin{align*}
    \mathrm{\RNum{3}}\leq&\frac{C\kappa_{\Sigma^*}\|\Sigma^*\|^3_{\infty}\|\Theta^{(a)*}_{b,:}\|_1^2\|\theta^{(a)*}\|_1^2}{\lambda_{\mymin}(\Sigma^*)}\frac{(d_a+d_b+1)\log p}{(n_1^{(a,b)})^{2}}.
\end{align*}

Combining the upper bounds for $\mathrm{\RNum{1}}$, $\mathrm{\RNum{2}}$, and $\mathrm{\RNum{3}}$, and applying Assumptions~\ref{assump:inference_n_B} and \ref{assump:inference_n_BE}, one can show that
\begin{align*}
    &|\widehat{\sigma}_n^2(a,b)-\sigma_n^2(a,b)|\\
    \leq &\frac{C\kappa_{\Sigma^*}\|\Sigma^*\|_{\infty}^3\|\Theta^{(a)*}_{b,:}\|_1^2\|\Theta^{*}_{:,a}\|_1^2}{\lambda_{\mymin}(\Sigma^*)(\Theta^*_{a,a})^2}\frac{(d_a+d_b+1)\sqrt{\log p}}{(n_1^{(a,b)})^{3/2}}.
\end{align*}
To interpret the bound above, note that $\|\Sigma^*\|_{\infty}^2\|\Theta^{(a)*}_{b,:}\|_1^2\frac{\|\Theta^{*}_{:,a}\|_1^2}{(\Theta^*_{a,a})^2}$ can be viewed as an upper bound for $|\langle \cT^*,\cU^*\rangle|\leq \|\cT^*\|_{\infty}\|\cU^*\|_1$, $\frac{\kappa_{\Sigma^*}\|\Sigma^*\|_{\infty}}{\lambda_{\mymin}(\Sigma^*)}\frac{(d_a+d_b+1)\sqrt{\log p}}{(n_1^{(a,b)})^{3/2}}$ is the error brought by estimating $\cU^*=\Theta^{(a)*}_{b,:}\otimes \overline{\theta}^{(a)*}\otimes \Theta^{(a)*}_{b,:}\otimes \overline{\theta}^{(a)*}$, which dominates the error brought by estimating $\cT^*$. As have been shown in the proof of Theorem~\ref{thm:nb_lasso_debias_decomp}, 
\begin{align*}
    \sigma_n^2(a,b)\geq \frac{\lambda_{\mymin}^2(\Sigma^*)\mymin_{(j,k)\in S_2(a, b)}\left|\Theta^{(a)*}_{b,j}\Theta^{*}_{a,k}+\Theta^{(a)*}_{b,k}\Theta^{*}_{a,j}\right|^2}{2(\Theta^*_{a,a})^2n_2^{(a,b)}},
\end{align*}
which implies
\begin{equation}\label{eq:var_est_ratio_err}
    \begin{split}
        &\frac{|\widehat{\sigma}_n^2(a,b)-\sigma_n^2(a,b)|}{\sigma_n^2(a,b)}\\
    \leq&\frac{C^2(\Theta^*;a,b)}{\kappa_{\Theta^*}^2}\frac{(d_a+d_b+1)n_2^{(a,b)}\sqrt{\log p}}{(n_1^{(a,b)})^{3/2}}\\
    =:&\varepsilon_n.
    \end{split}
\end{equation}
By Assumption \ref{assump:var_est}, $\lim_{n\rightarrow \infty}\varepsilon_n=0$. Note that when $\varepsilon_n\leq \frac{1}{2}$, with probability at least $1-Cp^{-c}$, one has
\begin{align*}
    |\widehat{\sigma}_n(a,b)-\sigma_n(a,b)|\leq &\frac{\widehat{\sigma}_n^2(a,b)-\sigma_n^2(a,b)}{2\sqrt{\sigma_n^2(a,b)-|\widehat{\sigma}_n^2(a,b)-\sigma_n^2(a,b)|}}\\
    \leq&\frac{\varepsilon_n\sigma_n(a,b)}{2\sqrt{1-\varepsilon_n}},
\end{align*}
and hence,
\begin{align*}
    \left|\frac{\sigma_n(a,b)}{\widehat{\sigma}_n(a,b)}-1\right|=&\frac{\sigma_n(a,b)-\widehat{\sigma}_n(a,b)}{\widehat{\sigma}_n(a,b)}\\
    \leq&\frac{|\widehat{\sigma}_n(a,b)-\sigma_n(a,b)|}{\sigma_n(a,b)-|\widehat{\sigma}_n(a,b)-\sigma_n(a,b)|}\\
    \leq&\frac{\varepsilon_n}{2\sqrt{1-\varepsilon_n}-\varepsilon_n}.
\end{align*}
Since $\lim_{n\rightarrow \infty}\frac{\varepsilon_n}{2\sqrt{1-\varepsilon_n}-\varepsilon_n}=0$, for any $\delta>0$, there exists $n_0$ such that if $n>n_0$, $\frac{\varepsilon_n}{2\sqrt{1-\varepsilon_n}-\varepsilon_n}\leq \delta$, and thus
\begin{align*}
    &\lim_{n\rightarrow\infty}\bP\left(\left|\frac{\sigma_n(a,b)}{\widehat{\sigma}_n(a,b)}-1\right|> \delta\right)\\
    \leq &\lim_{n\rightarrow\infty}\bP\left(\frac{|\widehat{\sigma}_n^2(a,b)-\sigma_n^2(a,b)|}{\sigma_n^2(a,b)}> \frac{\varepsilon_n}{2\sqrt{1-\varepsilon_n}-\varepsilon_n}\right)\\
    \leq &\lim_{n\rightarrow\infty}Cp_n^{-c}=0.
\end{align*}
Here we write $p=p(n)$ to reflect the fact that $p$ also tends to $\infty$ as $n$ tends to $\infty$. Therefore, $\frac{\sigma_n(a,b)}{\widehat{\sigma}_n(a,b)}\overset{p}{\rightarrow} 1$.
\begin{proof}[Proof of Theorem~\ref{thm:normal_approx_var_est}]
Combine the results in Theorem~\ref{thm:nb_lasso_debias_decomp} and Proposition~\ref{prop:var_est}, and apply Slutsky's theorem, the proof is then complete.
\end{proof}
\subsection{Proof of Theorem \ref{thm:typeI_power}}
Under the null hypothesis $\mathcal{H}_0:\Theta^*_{a,b}=0$, Theorem \ref{thm:normal_approx_var_est} suggests that $\widehat{z}(a,b)\overset{d}{\rightarrow}\cN(0,1)$, and hence $\lim_{n,p\rightarrow \infty}\bP(|\widehat{z}(a,b)|\geq z_{\alpha/2})=2F_Z(-z_{\alpha/2})=\alpha$.

While under the alternative hypothesis $\mathcal{H}_A:\frac{\Theta^*_{a,b}}{\Theta^*_{a,a}}=\delta_n$, we have
\begin{align*}
    \widehat{z}(a,b)=\frac{\delta_n}{\sigma_n(a,b)}\frac{\sigma_n(a,b)}{\widehat{\sigma}_n(a,b)}+\widehat{\sigma}_n^{-1}(a,b)(\widetilde{\theta}^{(a)}_b-\theta^{(a)*}_b).
\end{align*}
If $\lim_{n,p\rightarrow \infty}\frac{\delta_n}{\sigma_n(a,b)}=\delta$ for some $\delta\in \bR$, then Proposition \ref{prop:var_est} suggests that $\frac{\delta_n}{\sigma_n(a,b)}\frac{\sigma_n(a,b)}{\widehat{\sigma}_n(a,b)}\overset{p}{\rightarrow} \delta$. By Slutsky's theorem and Theorem \ref{thm:normal_approx_var_est}, $\widehat{z}(a,b)\overset{d}{\rightarrow} \cN(\delta,1)$, and hence
$$
\lim_{n,p\rightarrow \infty}\bP(|\widehat{z}(a,b)|\geq z_{\alpha/2})=F_Z(-z_{\alpha/2}+\delta)+F_Z(-z_{\alpha/2}-\delta)\geq F_Z(|\delta|-z_{\alpha/2}).
$$
In particular, when $\delta=0$, $F_Z(-z_{\alpha/2}+\delta)+F_Z(-z_{\alpha/2}-\delta)=\alpha$, and $
\lim_{n,p\rightarrow \infty}\bP(|\widehat{z}(a,b)|\geq z_{\alpha/2})=\alpha.$

Now we consider the last case: $\lim_{n,p\rightarrow \infty}\frac{|\delta_n|}{\sigma_n(a,b)}=+\infty$. Since $$|\widehat{z}(a,b)|\geq \left|\frac{\delta_n}{\widehat{\sigma}_n(a,b)}\right|-\left|\widehat{\sigma}_n^{-1}(a,b)(\widetilde{\theta}^{(a)}_b-\theta^{(a)*}_b)\right|,
$$
we have that for any $\delta>0$
\begin{align*}
    \bP(|\widehat{z}(a,b)|<z_{\alpha/2})\leq& \bP\left(|\widehat{\sigma}_n^{-1}(a,b)(\widetilde{\theta}^{(a)}_b-\theta^{(a)*}_b)|>\left|\frac{\delta_n}{\widehat{\sigma}_n(a,b)}\right|-z_{\alpha/2}\right)\\
    \leq&\bP(|\widehat{\sigma}_n^{-1}(a,b)(\widetilde{\theta}^{(a)}_b-\theta^{(a)*}_b)|>\left|\frac{\delta_n}{\sigma_n(a,b)}\right|(1-\delta)-z_{\alpha/2})\\
    &+\bP(|\frac{\sigma_n(a,b)}{\widehat{\sigma}_n(a,b)}-1|>\delta).
\end{align*}
Recall the error term $\varepsilon_n$ defined in \eqref{eq:var_est_ratio_err}. Let $\delta=\frac{\varepsilon_n}{2\sqrt{1-\varepsilon_n}-\varepsilon_n}$, and recall that in the proof of Proposition \ref{prop:var_est}, we have shown that $\bP(|\frac{\sigma_n(a,b)}{\widehat{\sigma}_n(a,b)}-1|>\delta)\leq Cp^{-c}$. Since $\lim_{n,p\rightarrow\infty}\delta=0$, $\lim_{n,p\rightarrow \infty}\frac{|\delta_n|}{\sigma_n(a,b)}=+\infty$, one has
$$
C_n:=\lim_{n,p\rightarrow}\left|\frac{\delta_n}{\sigma_n(a,b)}\right|(1-\delta)-z_{\alpha/2}=+\infty.
$$
Hence for any $\epsilon>0$, there exists $N_1$ such that if $n>N_1$, $C_n>z_{\epsilon/4}$; there also exists $N_2$ such that if $n>N_2$, $|\bP(|\widehat{\sigma}_n^{-1}(a,b)(\widetilde{\theta}^{(a)}_b-\theta^{(a)*}_b)|>z_{\epsilon/4})-\bP(|Z|>z_{\epsilon/4})|<\epsilon/2$, where $Z$ follows standard Gaussian distribution. Therefore, for $n\geq N_1\vee N_2$,
\begin{align*}
    \bP(|\widehat{\sigma}_n^{-1}(a,b)(\widetilde{\theta}^{(a)}_b-\theta^{(a)*}_b)|>C_n)\leq & \bP(|\widehat{\sigma}_n^{-1}(a,b)(\widetilde{\theta}^{(a)}_b-\theta^{(a)*}_b)|>z_{\epsilon/4})\\
    \leq&\bP(|Z|>z_{\epsilon/4})+\epsilon/2\\
    \leq&\epsilon,
\end{align*}
which implies that $\lim_{n,p\rightarrow \infty}\bP(|\widehat{\sigma}_n^{-1}(a,b)(\widetilde{\theta}^{(a)}_b-\theta^{(a)*}_b)|>C_n)=0$ and hence
$$
\lim_{n,p\rightarrow \infty}\bP(|\widehat{z}(a,b)|\geq z_{\alpha/2})=1.
$$
\subsection{Proof of Theorem \ref{thm:coverage_Theta}: Valid Inference for the Precision Matrix}
Now we show the coverage guarantee of our confidence interval $\widehat{\mathbb{C}}^{a,b}_{\Theta,\alpha}$ defined in \eqref{eq:CI_Theta} and \eqref{eq:CI_Theta_aa}. We will focus on the $a\neq b$ case first, and then present the proof for the $a=b$ case, which can be built on top of the previous proof steps. In particular, for the $a\neq b$ case, we first prove the normal approximation result for $\frac{\widetilde{\Theta}_{a,b}-\Theta^*_{a,b}}{\sigma_{n,\Theta}(a,b)}$, and then we show the variance estimation consistency of $\widehat{\sigma}_{n,\Theta}(a,b)$.
\subsubsection{Decomposition for \texorpdfstring{$\frac{\widetilde{\Theta}_{a,b}-\Theta^*_{a,b}}{\sigma_{n,\Theta}(a,b)}$}{Lg}}\label{sec:proof_Theta_ab_decomp}
By the definition of $\widetilde{\Theta}_{a,b}$, one can show that
\begin{equation*}
        \widetilde{\Theta}_{a,b}-\Theta^*_{a,b} = -\Theta^*_{a,a}(\widetilde{\theta}^{(a)}_b - \theta^{(a)*}_b)-(\widehat{\Theta}_{a,a}-\Theta^*_{a,a})(\widetilde{\theta}^{(a)}_b - \theta^{(a)*}_b)-(\widehat{\Theta}_{a,a}-\Theta^*_{a,a})\theta^{(a)*}_b.
\end{equation*}
Theorem \ref{thm:nb_lasso_debias_decomp} suggests that 
$$
-\Theta^*_{a,a}(\widetilde{\theta}^{(a)}_b - \theta^{(a)*}_b) = -\Theta^*_{a,a}B -\Theta^*_{a,a}E,
$$
where $|B|\leq C(\Theta^*,\gamma_a,\gamma^{(a)}_b)\frac{(d_a+d_b+1)\log p}{n_1^{(a,b)}}$ with probability at least $1-Cp^{-c}$, and $\sigma_n^{-1}(a,b)E\overset{\text{d}}{\rightarrow}\cN(0,1).$ 
Now we focus on the latter two terms and aim to decompose them into terms that are higher order biases and error terms that converges to a normal distribution. We start with a decomposition of $\widehat{\Theta}_{a,a}-\Theta^*_{a,a}$ as a function of $\widehat{\Theta}_{a,a}^{-1}-(\Theta^*_{a,a})^{-1}$. By the definition of $\widehat{\Theta}_{a,a}$, we have that
\begin{equation}\label{eq:tilde_Theta_inv_err}
    \begin{split}
        \widehat{\Theta}_{a,a}^{-1}-(\Theta^*_{a,a})^{-1} = &\widehat{\overline{\theta}}^{(a)\top}\widehat{\Sigma}\widehat{\overline{\theta}}^{(a)}-\overline{\theta}^{(a)*\top}\Sigma^*\overline{\theta}^{(a)*}\\
        =&\underbrace{\overline{\theta}^{(a)*\top}(\widehat{\Sigma}-\Sigma^*)\overline{\theta}^{(a)*}}_{\varepsilon^{(1)}_1}+\underbrace{\widehat{\overline{\theta}}^{(a)\top}\widehat{\Sigma}\widehat{\overline{\theta}}^{(a)}-\overline{\theta}^{(a)*\top}\widehat{\Sigma}\overline{\theta}^{(a)*}}_{\varepsilon^{(1)}_2}.
    \end{split}
\end{equation}
With a little abuse of notations, we decompose $\widehat{\Theta}_{a,a}^{-1}-(\Theta^*_{a,a})^{-1}$ into two error terms $\varepsilon^{(1)}_1$ and $\varepsilon^{(1)}_2$. We also let $\varepsilon^{(1)}=\varepsilon^{(1)}_1+\varepsilon^{(1)}_2$, while the three notations $\varepsilon^{(1)}_1$, $\varepsilon^{(1)}_2$, and $\varepsilon^{(1)}$ might be defined differently in the proofs of other theoretical statements. Then we can write 
\begin{equation}\label{eq:tilde_Theta_err}
    \begin{split}
        \widehat{\Theta}_{a,a}-\Theta^*_{a,a} = &-\Theta^*_{a,a}\widehat{\Theta}_{a,a}(\widehat{\Theta}_{a,a}^{-1}-(\Theta^*_{a,a})^{-1})\\
        =&-\Theta^*_{a,a}[\Theta^*_{a,a}-\Theta^*_{a,a}\widehat{\Theta}_{a,a}\varepsilon^{(1)}]\varepsilon^{(1)}\\
        =&-\underbrace{(\Theta^*_{a,a})^2\varepsilon^{(1)}_1}_{\varepsilon^{(2)}_1}-\underbrace{[(\Theta^*_{a,a})^2\varepsilon^{(1)}_2 - \varepsilon^{(1)2}\widehat{\Theta}_{a,a}(\Theta^*_{a,a})^2]}_{\varepsilon^{(2)}_2},
    \end{split}
\end{equation}
with $\varepsilon^{(2)}=\varepsilon^{(2)}_1+\varepsilon^{(2)}_2$.
With these decompositions in place, we have
\begin{equation}\label{eq:tilde_Theta_decomp}
    \begin{split}
        \widetilde{\Theta}_{a,b}-\Theta^*_{a,b} =\underbrace{\varepsilon^{(2)}_1\theta^{(a)*}_b-\Theta^*_{a,a}E}_{E^{(2)}}+ \underbrace{\varepsilon^{(2)}_2\theta^{(a)*}_b+\varepsilon^{(2)}(\widetilde{\theta}^{(a)}_b - \theta^{(a)*}_b)-\Theta^*_{a,a}B}_{B^{(2)}},
    \end{split}
\end{equation}
and in the following we will show an upper bound for $B^{(2)}$ based off bounds for $\varepsilon^{(1)}_1$, $\varepsilon^{(1)}_2$, and central limit theorem for
\begin{equation*}
\begin{split}
    E^{(2)}=&\left\langle\widehat{\Sigma}-\Sigma^*,-\Theta^{(a)*}_{:,b}\Theta^*_{a,:}-\frac{\Theta^*_{a,b}}{\Theta^*_{a,a}}\Theta^*_{:,a}\Theta^*_{a,:}\right\rangle\\
    =&\left\langle\widehat{\Sigma}-\Sigma^*,-\Theta^{*}_{:,b}\Theta^*_{a,:}\right\rangle.
\end{split}
\end{equation*}
\subsubsection{Bounding \texorpdfstring{$B^{(2)}$}{Lg}}\label{sec:proof_precision_B} 
For the first error term $\varepsilon^{(1)}_1$ in $\widehat{\Theta}_{a,a}^{-1}-(\Theta^*_{a,a})^{-1}$, one can show that
\begin{equation*}
    \begin{split}
        |\varepsilon^{(1)}_1|\leq& \|\overline{\theta}^{(a)*}\|_1^2\|\widehat{\Sigma}_{\overline{\mathcal{N}}_a,\overline{\mathcal{N}}_a}-\Sigma^*_{\overline{\mathcal{N}}_a,\overline{\mathcal{N}}_a}\|_{\infty}\\
        \leq&C\frac{\|\Sigma^*\|_{\infty}\|\Theta^*_{:,a}\|_1^2}{(\Theta^*_{a,a})^2}\sqrt{\frac{\log p}{n_1'^{(a,b)}}},
    \end{split}
\end{equation*}
where we have applied the entry-wise error bound for $\widehat{\Sigma}$ in Lemma \ref{lem:SampleCov_entry_err} in the second line, and $n_1'^{(a,b)}$ was defined as $n_1'^{(a,b)}=\min_{(j,k)\in S_1(a,b)\cup \{(a,a)\}}n_{j,k}$. While for $\varepsilon^{(1)}_2$, we show in the following that this is a higher-order error:
\begin{equation*}
    \begin{split}
        |\varepsilon^{(1)}_2|\leq&2|(\widehat{\overline{\theta}}^{(a)}-\overline{\theta}^{(a)*})^\top \widehat{\Sigma}\overline{\theta}^{(a)*}|+|(\widehat{\overline{\theta}}^{(a)}-\overline{\theta}^{(a)*})^\top\widehat{\Sigma}(\widehat{\overline{\theta}}^{(a)}-\overline{\theta}^{(a)*})|\\
        \leq&2|(\widehat{\overline{\theta}}^{(a)}-\overline{\theta}^{(a)*})^\top \Sigma^*\overline{\theta}^{(a)*}|+2|(\widehat{\overline{\theta}}^{(a)}-\overline{\theta}^{(a)*})^\top (\widehat{\Sigma}-\Sigma^*)\overline{\theta}^{(a)*}|\\
        &+|(\widehat{\overline{\theta}}^{(a)}-\overline{\theta}^{(a)*})^\top\Sigma^*(\widehat{\overline{\theta}}^{(a)}-\overline{\theta}^{(a)*})|+|(\widehat{\overline{\theta}}^{(a)}-\overline{\theta}^{(a)*})^\top(\widehat{\Sigma}-\Sigma^*)(\widehat{\overline{\theta}}^{(a)}-\overline{\theta}^{(a)*})|\\
        \leq &\|(\widehat{\Sigma}-\Sigma^*)(\widehat{\overline{\theta}}^{(a)}-\overline{\theta}^{(a)*})\|_{\infty}(2\|\overline{\theta}^{(a)*}\|_1+\|\widehat{\overline{\theta}}^{(a)}-\overline{\theta}^{(a)*}\|_1)+\lambda_{\max}(\Sigma^*)\|\widehat{\theta}^{(a)}-\theta^{(a)*}\|_2^2,
    \end{split}
\end{equation*}
where in the last line, we have applied the fact that $(\Sigma^*\overline{\theta}^{(a)*})_j=0$ for all $j\neq a$ and $\widehat{\overline{\theta}}^{(a)}_a=\overline{\theta}^{(a)*}_a=1$, which implies $(\widehat{\overline{\theta}}^{(a)}-\overline{\theta}^{(a)*})^\top \Sigma^*\overline{\theta}^{(a)*}=0$. Furthermore, our covariance error shown in Lemma \ref{lem:SampleCov_entry_err} suggests that
\begin{equation*}
    \begin{split}
        \|(\widehat{\Sigma} - \Sigma^*)(\widehat{\overline{\theta}}^{(a)}-\overline{\theta}^{(a)*})\|_{\infty}\leq &C\|\Sigma^*\|_{\infty}\max_j\sum_{k=1}^p|\widehat{\overline{\theta}}^{(a)}_k-\overline{\theta}^{(a)*}_k|\sqrt{\frac{\log p}{n_{j,k}}}\\
        \leq &\frac{C\|\widehat{\overline{\theta}}^{(a)}-\overline{\theta}^{(a)*}\|_{\lambda^{(a)},1}}{\|\overline{\theta}^{(a)*}\|_1}\\
        \leq &\frac{C\|\Sigma^*\|_{\infty}^2}{\lambda_{\min}(\Sigma^*)}\frac{\|\Theta^*_{:,a}\|_1}{\Theta^*_{a,a}}\frac{d_a\log p}{\min_{j\in \mathcal{N}_a}\min_kn_{j,k}},
    \end{split}
\end{equation*}
where we have applied Theorem \ref{thm:nb_lasso_err} in the last line. Furthermore, applying the $\ell_1$ and $\ell_2$ norm bounds in Theorem \ref{thm:nb_lasso_err} leads to the following:
\begin{equation}\label{eq:epsilon1_2_bnd}
    |\varepsilon^{(1)}_2|\leq\frac{C\|\Sigma^*\|_{\infty}^2\lambda_{\max}(\Sigma^*)}{\lambda_{\min}^2(\Sigma^*)}\frac{\|\Theta^*_{:,a}\|_1^2}{(\Theta^*_{a,a})^2}\frac{d_a\log p}{\min_{j\in \mathcal{N}_a}\min_kn_{j,k}}.
\end{equation}
While for $\varepsilon^{(1)}_1$, we have 
\begin{equation}\label{eq:epsilon1_1_bnd}
    \begin{split}
        |\varepsilon^{(1)}_1|\leq &\|\widehat{\Sigma}_{\overline{\mathcal{N}}_a,\overline{\mathcal{N}}_a}-\Sigma^*_{\overline{\mathcal{N}}_a,\overline{\mathcal{N}}_a}\|_{\infty}\frac{\|\Theta^*_{:,a}\|_1^2}{(\Theta^*_{a,a})^2}\\
        \leq &\frac{C\|\Sigma^*\|_{\infty}\|\Theta^*_{:,a}\|_1^2}{{(\Theta^*_{a,a})^2}}\sqrt{\frac{\log p}{\min_{j,k\in S_1'(a,b)}n_{j,k}}},
    \end{split}
\end{equation}
where $S_1'(a,b)$ in the last line was defined as $S_1'(a,b) = S_1(a,b)\cup \{(a,a)\}$ in Assumption \ref{assump:inference_n_B2}, and we have applied Lemma \ref{lem:SampleCov_entry_err}. Combining \eqref{eq:epsilon1_1_bnd} and \eqref{eq:epsilon1_2_bnd} together, and recalling Assumption \ref{assump:inference_n_B2}, we have $|\varepsilon^{(1)}|\leq \frac{C\|\Sigma^*\|_{\infty}\|\Theta^*_{:,a}\|_1^2}{{(\Theta^*_{a,a})^2}}\sqrt{\frac{\log p}{\min_{j,k\in S_1'(a,b)}n_{j,k}}}$. Note that Assumption \ref{assump:inference_n_B2} implies $|\varepsilon^{(1)}|\leq C\|\Sigma^*\|_{\infty}\kappa_{\Sigma^*}^2(d_a+1)\sqrt{\frac{\log p}{\min_{j,k\in S_1'(a,b)}n_{j,k}}}\leq \frac{1}{2}\lambda_{\min}(\Sigma^*)\leq \frac{1}{2}(\Theta^*_{a,a})^{-1}$, which further suggests that $|\widehat{\Theta}_{a,a}|\leq \frac{1}{(\Theta^*_{a,a})^{-1}-|\varepsilon^{(1)}|}\leq 2\Theta^*_{a,a}$. Therefore, by the definition of $\varepsilon^{(2)}_1$ and $\varepsilon^{(2)}_2$ in \eqref{eq:tilde_Theta_err}, one can show that $|\varepsilon^{(2)}_1|\leq C\|\Sigma^*\|_{\infty}\|\Theta^*_{:,a}\|_1^2\sqrt{\frac{\log p}{\min_{j,k\in S_1'(a,b)}n_{j,k}}}$, while for $\varepsilon^{(2)}_2$, some more calculations show that
\begin{equation*}
    \begin{split}
        |\varepsilon^{(2)}_2|\leq &\frac{C\|\Sigma^*\|_{\infty}^2\lambda_{\max}(\Sigma^*)}{\lambda_{\min}^2(\Sigma^*)}\|\Theta^*_{:,a}\|_1^2\frac{(d_a+1)\log p}{\min_{j\in \mathcal{N}_a}\min_kn_{j,k}}\\
        &+\frac{C\|\Sigma^*\|_{\infty}^2\|\Theta^*_{:,a}\|_1^4}{\Theta^*_{a,a}}\frac{\log p}{\min_{(j,k)\in S_1'(a,b)}n_{j,k}}\\
        \leq &\frac{C\|\Sigma^*\|_{\infty}^2\lambda_{\max}(\Sigma^*)}{\lambda_{\min}^2(\Sigma^*)}\|\Theta^*_{:,a}\|_1^2\frac{(d_a+1)\log p}{\min_{(j,k)\in S_1'(a,b)}n_{j,k}},
    \end{split}
\end{equation*}
where the last line is due to the fact that $\|\Theta^*_{:,a}\|_1^2\leq(d_a+1)\|\Theta^*_{:,a}\|_2^4\leq(d_a+1)^2\lambda_{\min}^{-2}(\Sigma^*)$ and $\Theta^*_{a,a}\geq \lambda_{\max}^{-1}(\Sigma^*)$. Combining bounds for $\varepsilon^{(2)}_{1}$ and $\varepsilon^{(2)}_2$ and applying Assumption \ref{assump:inference_n_B2}, we have 
\begin{equation*}
    |\varepsilon^{(2)}|\leq C\|\Sigma^*\|_{\infty}\|\Theta^*_{:,a}\|_1^2\sqrt{\frac{\log p}{\min_{(j,k)\in S_1'(a,b)}n_{j,k}}}\leq \frac{C\|\Sigma^*\|_{\infty}}{\lambda_{\min}^2(\Sigma^*)}(d_a+1)\sqrt{\frac{\log p}{\min_{(j,k)\in S_1'(a,b)}n_{j,k}}}\leq C\Theta^*_{a,a}.
\end{equation*}
Now we plug in these bounds into \eqref{eq:tilde_Theta_decomp} and apply Theorem \ref{thm:nb_lasso_debias_decomp} to bound $B^{(2)}$ as follows:
\begin{equation*}
    \begin{split}
        |B^{(2)}|\leq &|\varepsilon^{(2)}_2||\theta^{(a)*}_b|+|\varepsilon^{(2)}||E|+C\Theta^*_{a,a}|B|\\
        \leq&\frac{C\|\Sigma^*\|_{\infty}^2\lambda_{\max}(\Sigma^*)|\Theta^*_{a,b}|}{\lambda_{\min}^2(\Sigma^*)\Theta^*_{a,a}}\|\Theta^*_{:,a}\|_1^2\frac{(d_a+1)\log p}{\min_{(j,k)\in S_1'(a,b)}n_{j,k}}\\
        &+C\|\Sigma^*\|_{\infty}\|\Theta^*_{:,a}\|_1^2\sqrt{\frac{\log p}{\min_{(j,k)\in S_1'(a,b)}n_{j,k}}}|E|\\
        &+C\kappa_{\Sigma^*}(\kappa_{\Sigma^*}^2+\sqrt{\gamma_a}+\sqrt{\gamma^{(a)}_b})\|\Sigma^*\|_{\infty}^2\Theta^*_{a,a}\|\Theta^*_{:,a}\|_1\|\Theta^{(a)*}_{:,b}\|_1\frac{(d_a+d_b+1)\log p}{n_1^{(a,b)}},
    \end{split}
\end{equation*}
where $B$ is bounded in Theorem \ref{thm:nb_lasso_debias_decomp}; While for $E$, the specific form was defined in the proof of Theorem \ref{thm:nb_lasso_debias_decomp}: $E = \frac{\Theta^{(a)*}_{b,:}(\widehat{\Sigma}-\Sigma^*)\Theta^*_{:,a}}{\Theta^*_{a,a}}$, and we bound it by applying Lemma \ref{lem:SampleCov_entry_err}: with probability at least $1-p^{-c}$,
\begin{equation*}
\begin{split}
    |E|\leq &\frac{\left\|\widehat{\Sigma}_{\overline{\mathcal{N}}^{(a)}_b,\overline{\mathcal{N}}_a}-\Sigma^*_{\overline{\mathcal{N}}^{(a)}_b,\overline{\mathcal{N}}_a}\right\|_{\infty}\left\|\Theta^{(a)*}_{b,:}\right\|_1\left\|\Theta^*_{:,a}\right\|_1}{\Theta^*_{a,a}}\\
    \leq&\frac{C\|\Sigma^*\|_{\infty}\|\Theta^*_{:,a}\|_1\|\Theta^{(a)*}_{:,b}\|_1}{\Theta^*_{a,a}}\sqrt{\frac{\log p}{n_2^{(a,b)}}}.
\end{split}
\end{equation*}
Therefore, we can further simplify the bound for $|B^{(2)}|$ as 
\begin{equation*}
    |B^{(2)}|\leq C\kappa_{\Sigma^*}(\kappa_{\Sigma^*}^2+\sqrt{\gamma_a}+\sqrt{\gamma^{(a)}_b})\|\Sigma^*\|_{\infty}^2\|\Theta^*_{:,a}\|_1^2(\|\Theta^{(a)*}_{:,b}\|_1+|\Theta^*_{a,b}|)\frac{(d_a+d_b+1)\log p}{n_1'^{(a,b)}},
\end{equation*}
with probability at least $1-Cp^{-c}$.

\subsubsection{Central limit theorem for \texorpdfstring{$E^{(2)}$}{Lg}}\label{sec:proof_CLT_Theta}
To show the normal approximation for $E^{(2)}$, the proof is mostly the same as that of Theorem \ref{thm:nb_lasso_debias_decomp}. First note that we can also write $E^{(2)}=\sum_{i=1}^n\Theta^*_{a,:}[(x_ix_i^\top - \Sigma^*)\circ \delta^{(i)}\oslash N]\Theta^*_{:,b}$, sum of independent random variables. The variance $\sigma_{n,\Theta}^2(a,b)$ can be similarly computed as 
\begin{equation*}
\begin{split}
    \sigma_{n,\Theta}^2(a,b)=&\sum_{i=1}^n\mathbb{E}\left(\Theta^*_{a,:}[(x_ix_i^\top - \Sigma^*)\circ \delta^{(i)}\oslash N]\Theta^*_{:,b}\right)^2\\
    =&\sum_{i=1}^n\mathrm{Var}\left(\sum_{j,k}\Theta^*_{a,j}\Theta^*_{k,b}(x_{i,j}x_{i,k}\delta^{(i)}_{j,k}n_{j,k}^{-1}\right)\\
    =&\sum_{j,k,j',k'}\frac{\Theta^*_{a,j}\Theta^*_{k,b}\Theta^*_{a,j'}\Theta^*_{k',b}}{n_{j,k}n_{j',k'}}\mathcal{T}_{j,k,j',k'}\sum_{i=1}^n\delta^{(i)}_{j,k}\delta^{(i)}_{j',k'}\\
    =&\mathcal{T}^{(n)*}\times_1\Theta^*_{:,a}\times_2\Theta^*_{:,b}\times_3\Theta^*_{:,a}\times_4\Theta^*_{:,b}.
\end{split}
\end{equation*}
Similar to the proof of Theorem \ref{thm:nb_lasso_debias_decomp}, now we verify the Lyapounov's condition. We define $U, U^{(\delta,i)}, \epsilon^{(i)}\in \bR^{p\times p}$ by $U_{j,k}=\frac{1}{2}(\Theta^{*}_{j,b}\Theta^*_{k,a} + \Theta^{*}_{k,b}\Theta^*_{j,a})$,
$U^{(\delta,i)}_{j,k}=U_{j,k}\frac{\delta^{(i)}_{j,k}}{n_{j,k}}$, and $\epsilon^{(i)}_{j,k}=x_{i,j}x_{i,k}-\Sigma^*_{j,k}$. Then we can write $E^{(2)}=\sum_{i=1}^n\langle U^{(\delta,i)},\epsilon^{(i)}\rangle$. Note that the definition of $U, U^{(\delta,i)}$ are different from those in the proof of Theorem \ref{thm:nb_lasso_debias_decomp}. Let $\delta=\frac{4}{\epsilon}>0$ for the constant $\epsilon>0$ in Assumption \ref{assump:inference_n_E2}. Following the same arguments as in Section \ref{sec:proof_normalE}, one can show that
\begin{equation*}
    \sum_{i=1}^n\bE|\langle U^{(\delta,i)},\epsilon^{(i)}\rangle|^{2+\delta}\leq \left(C(2+\delta)\|\Sigma^*\|_{\infty}(d_a+d_b+1)\right)^{2+\delta}\left(\sum_{j,k}|U_{j,k}|^2n_{j,k}^{-1}\right)^{1+\delta/2}(n_2'^{(a,b)})^{-\delta/2}.
\end{equation*}
Furthermore, it can be shown that 
\begin{equation}\label{eq:precision_var_lowerbnd}
    \sigma_{n,\Theta}^2(a,b)\geq 2\lambda_{\min}^2(\Sigma^*)\sum_{j,k}U_{j,k}^2n_{j,k}^{-1}\geq \frac{1}{2}\lambda_{\min}^2(\Sigma^*)\min_{(j,k)\in S_2'(a,b)}|\Theta^*_{j,b}\Theta^*_{k,a}+\Theta^*_{k,b}\Theta^*_{j,a}|^2(n_2'^{(a,b)})^{-1}.
\end{equation} 
Therefore, we have 
\begin{equation*}
    \frac{\sum_{i=1}^n\bE|\langle U^{(\delta,i)},\epsilon^{(i)}\rangle|^{2+\delta}}{\sigma_{n,\Theta}^{2+\delta}(a,b)}\leq C_{\epsilon}^{\delta/2}(\Sigma^*)(d_a+d_b+1)^{2+\delta}(n_2'^{(a,b)})^{-\delta/2},
\end{equation*}
and hence Lyapounov's condition holds:
$\sum_{i=1}^n\bE|\langle U^{(\delta,i)},\epsilon^{(i)}\rangle|^{2+\delta}=o(\sigma_{n,\Theta}^{2+\delta}(a,b))$. Now we can apply the central limit theorem to obtain the following convergence in distribution result:
\begin{align*}
    \sigma_{n,\Theta}^{-1}(a,b)E^{(2)}\overset{d}{\rightarrow}\cN(0,1).
\end{align*}
Recall the definition $S_2'^{(a,b)}= (\overline{\mathcal{N}}_a\times \overline{\mathcal{N}}_b)\cup (\overline{\mathcal{N}}_b\times \overline{\mathcal{N}}_a)$ in Assumption \ref{assump:inference_n_E2}.

Recall the decomposition \ref{eq:tilde_Theta_decomp}, we have 
\begin{equation*}
    \sigma_{n,\Theta}^{-1}(a,b)(\widetilde{\Theta}_{a,b}-\Theta^*_{a,b})=\frac{B^{(2)}}{\sigma_{n,\Theta}(a,b)}+\frac{E^{(2)}}{\sigma_{n,\Theta}(a,b)},
\end{equation*}
where the second term converges to $\mathcal{N}(0,1)$, and the first term is bounded as follows with probability at least $1-Cp^{-c}$:
\begin{equation*}
    \left|\frac{B^{(2)}}{\sigma_{n,\Theta}(a,b)}\right|\leq C'(\Theta^*;a,b)(\kappa_{\Theta^*}^2+\sqrt{\gamma_a}+\sqrt{\gamma^{(a)}_b})\frac{(d_a+d_b+1)\log p}{\sqrt{n_1'^{(a,b)}}}\sqrt{\frac{n_2'^{(a,b)}}{n_1'^{(a,b)}}},
\end{equation*}
where $C'(\Theta^*;a,b) = \frac{C\kappa_{\Theta^*}^2\|\Theta^*_{:,a}\|_1^2[\|\Theta^{*}_{:,b}\|_1+\|\Theta^*_{:,a}\|_1(1+\frac{|\Theta^*_{a,b}|}{\Theta^*_{a,a}})]}{\lambda_{\min}(\Theta^*)\min_{(j,k)\in S_2'(a,b)}|\Theta^*_{a,j}\Theta^*_{b,k}+\Theta^*_{a,k}\Theta^*_{b,j}|}$, and we have applied the fact that $\|\Theta^{(a)*}_{:,b}\|_1=\|\Theta^*_{:,b}-\frac{\Theta^*_{a,b}}{\Theta^*_{a,a}}\Theta^*_{:,a}\|_1\leq \|\Theta^*_{:,b}\|_1+\|\Theta^*_{:,a}\|_1\frac{|\Theta^*_{a,b}|}{\Theta^*_{a,a}}$.
Assumption \ref{assump:inference_n_BE2} implies that $\frac{B^{(2)}}{\sigma_{n,\Theta}(a,b)}=o_p(1)$ and hence $\sigma_{n,\Theta}^{-1}(a,b)(\widetilde{\Theta}_{a,b}-\Theta^*_{a,b})\rightarrow \mathcal{N}(0,1)$.
\subsubsection{Consistency of variance estimate}\label{sec:proof_varest_Theta} 
Now it remains to show that $\frac{\widehat{\sigma}_{n,\Theta}(a,b)}{\sigma_{n,\Theta}(a,b)}\overset{p}{\rightarrow} 1$. Recall that we have already proved the consistency of variance estimate $\widehat{\sigma}_n^2(a,b)=\widehat{\mathcal{T}}^{(n)}\times_1\widehat{\Theta}^{(a)}_{b,:}\times_2\widehat{\overline{\theta}}^{(a)}\times_3\widehat{\Theta}^{(a)}_{b,:}\times_4\widehat{\overline{\theta}}^{(a)}$ for $\sigma_n^2(a,b)=\mathcal{T}^{(n)*}\times_1\Theta^{(a)*}_{b,:}\times_2\overline{\theta}^{(a)*}\times_3\Theta^{(a)*}_{b,:}\times_4\overline{\theta}^{(a)*}$, the variance associated with $\widetilde{\theta}^{(a)}_b$. While in this proof, the true variance parameter associated with $\widetilde{\Theta}_{a,b}$ is $\sigma_{n,\Theta}^2(a,b)=\mathcal{T}^{(n)*}\times_1\Theta^{*}_{b,:}\times_2\Theta^{*}_{a,:}\times_3\Theta^{*}_{b,:}\times_4\Theta^{*}_{a,:}$, and its estimate is $\widehat{\sigma}^2_{n,\Theta}(a,b)=\widehat{\mathcal{T}}^{(n)}\times_1\widehat{\Theta}_{b,:}\times_2\widehat{\Theta}_{a,:}\times_3\widehat{\Theta}_{b,:}\times_4\widehat{\Theta}_{a,:}$. The only difference lies in the parameters multiplied with the covariance matrix $\mathcal{T}^{(n)}$ or its estimate $\widehat{\mathcal{T}}^{(n)}$. Hence we can follow the same arguments as in the proof of Proposition \ref{prop:var_est} to show the consistency of $\widehat{\sigma}^2_{n,\Theta}(a,b)$ for $\sigma^2_{n,\Theta}(a,b)$, and we only highlight the key steps that differ from that proof.

Define $\mathcal{U}^*, \widehat{\mathcal{U}}, \mathcal{N}\in \bR^{p\times p\times p\times p}$ as $4$th order tensors that satisfy
\begin{align*}
    \mathcal{U}^*_{j,k,j',k'}=&\Theta^{*}_{b,j}\Theta^*_{a,k}\Theta^{*}_{b,j'}\Theta^*_{a,k'},\\
    \widehat{\mathcal{U}}_{j,k,j',k'}=&\widehat{\Theta}_{b,j}\widehat{\Theta}_{a,k}\widehat{\Theta}_{b,j'}\widehat{\Theta}_{a,k'},\\
    \mathcal{N}_{j,k,j',k'}=&\frac{n_{j,k,j',k'}}{n_{j,k}n_{j',k'}}.
\end{align*}
Recall that we proposed several options for $\widehat{\Theta}_{b,:}$ and $\widehat{\Theta}_{a,:}$. For simplicity, here we focus on the setting where $\widehat{\Theta}_{b,:}=\widetilde{\Theta}^{(a)}_{b,:}-\widehat{\Theta}_{a,a}\widehat{\theta}^{(a)}_b\widehat{\overline{\theta}}^{(a)}$, and $\widehat{\Theta}_{a,:}=\widehat{\Theta}_{a,a}\widehat{\overline{\theta}}^{(a)}$, while noting that all proof steps still hold if other options of $\widehat{\Theta}_{b,:}$ and $\widehat{\Theta}_{a,:}$ are considered.

Also recall the definition of tensor $\widehat{\mathcal{T}}^{(n)}$, then we can rewrite $\widehat{\sigma}_{n,\Theta}^2(a,b)$ as follows
\begin{align*}
    \widehat{\sigma}_{n,\Theta}^2(a,b) = \langle \widehat{\mathcal{T}}\circ \mathcal{N},\widehat{\mathcal{U}}\rangle = \langle \widehat{\mathcal{T}},\widehat{\mathcal{U}}\circ \mathcal{N}\rangle,
\end{align*} 
as an estimate of 
\begin{align*}
    \sigma_{n,\Theta}^2(a,b) = \langle \mathcal{T}^*\circ \mathcal{N},\mathcal{U}^*\rangle = \langle \mathcal{T}^*,\mathcal{U}^*\circ \mathcal{N}\rangle.
\end{align*}
Let $\cE_1= \widehat{\cT}-\cT^*$, $\cE_2=\widehat{\cU} - \cU^*$. Then the estimation error for variance can be decomposed as follows:
\begin{equation}\label{eq:precision_var_err_decomp}
    \widehat{\sigma}_{n,\Theta}^2(a,b)-\sigma_{n,\Theta}^2(a,b)=\underbrace{\langle \cE_1\circ \cN, \cU^*\rangle}_{\mathrm{\RNum{1}}} + \underbrace{\langle\cT^*,\cE_2\circ\cN\rangle}_{\mathrm{\RNum{2}}} + \underbrace{\langle\cE_1\circ\cN,\cE_2\rangle}_{\mathrm{\RNum{3}}},
\end{equation}
This is the same as the decomposition of $\widehat{\sigma}_n^2(a,b)-\sigma_n^2(a,b)$ but with different definitions of $\mathcal{U}^*$ and $\widehat{\cU}$. Following the same arguments as in the proof of Proposition \ref{prop:var_est}, $|\mathrm{\RNum{1}}|\leq C\|\Sigma^*\|_{\infty}^2\|\cU^*\|_1\left(\frac{\sqrt{\log p}}{(n_2'^{(a,b)})^{3/2}}+\frac{\log p}{(n_2'^{(a,b)})^2}\right)$, where $\|\cU^*\|_{1}=\|\Theta^*_{a,:}\|_1^2\|\Theta^*_{b,:}\|_1^2$.

While for error terms $\mathrm{\RNum{2}}$ and $\mathrm{\RNum{3}}$, following the same arguments as in the proof of Proposition \ref{prop:var_est}, the problem still reduces to bounding $\|\cE_2\circ \cN\|_1$, $\|\cE_2\circ \cN^{(1)}\|_1$, and $\|\cE_2\circ \cN^{(2)}\|_1$, where $\cN^{(1)},\cN^{(2)}\in \bR^{p\times p\times p\times p}$ satisfy $\cN^{(1)}_{j,k,j',k'}=\frac{\sqrt{n_{j,k,j',k'}}}{n_{j,k}n_{j',k'}}$ and $\cN^{(2)}_{j,k,j',k'}=\frac{1}{n_{j,k}n_{j',k'}}$. To achieve this, here we first present general versions of \ref{lem:kron_err} and Lemma \ref{lem:kron_err2}. 
\begin{lemma}\label{lem:kron_err3}
    Consider vectors $u^{(1)},\dots,u^{(4)}, \epsilon^{(1)},\dots,\epsilon^{(4)}\in \bR^p$. Let $i_l=l+1-2\ind{l\text{ is even}}$, $j_l=l+2-4\ind{l>2}$.
Then we have
\begin{align*}
    &\|[(u^{(1)}+\epsilon^{(1)})\otimes (u^{(2)}+\epsilon^{(2)})\otimes (u^{(3)}+\epsilon^{(3)})\otimes (u^{(4)}+\epsilon^{(4)})\\
    &-u^{(1)}\otimes u^{(2)}\otimes u^{(3)}\otimes u^{(4)}]\circ \cN\|_1\\
    \leq&C\sum_{l=1}^4\left[\sum_{j=1}^p\frac{\left|\epsilon^{(l)}_j\right|}{\sqrt{\mymin_{k}n_{j,k}}}\left(\frac{\|u^{(i_{l})}\|_1}{\sqrt{\mymin_{j\in \mathrm{supp}(u^{(l)})}\mymin_{k}n_{j,k}}}+ \sum_{j=1}^p\frac{\left|\epsilon^{(i_l)}_j\right|}{\sqrt{\mymin_{k}n_{j,k}}}\right)\prod_{m\neq l,i_l}\left(\|u^{(m)}\|_1+\|\epsilon^{(m)}\|_1\right)\right],
\end{align*}
\begin{align*}
    &\|[(u^{(1)}+\epsilon^{(1)})\otimes (u^{(2)}+\epsilon^{(2)})\otimes (u^{(3)}+\epsilon^{(3)})\otimes (u^{(4)}+\epsilon^{(4)})\\
    &-u^{(1)}\otimes u^{(2)}\otimes u^{(3)}\otimes u^{(4)}]\circ \cN^{(1)}\|_1\\
    \leq&C\sum_{l=1}^4\Bigg[\sum_{j=1}^p\frac{\left|\epsilon^{(l)}_j\right|}{\sqrt{\mymin_{k}n_{j,k}}}\left(\frac{\|u^{(i_{l})}\|_1}{\mymin_{j\in \mathrm{supp}(u^{(i_l)})}\mymin_{k}n_{j,k}}+\sum_{j=1}^p\frac{\left|\epsilon^{(i_l)}_j\right|}{\sqrt{\mymin_{k}n_{j,k}}}\right)\\
    &\left(\frac{\|u^{(j_l)}\|_1}{\mymin_{j\in \mathrm{supp}(u^{(j_l)})}\mymin_{k}n_{j,k}}+\sum_{j=1}^p\frac{\left|\epsilon^{(j_l)}_j\right|}{\sqrt{\mymin_{k}n_{j,k}}}\right)\prod_{m\neq l,i_l, j_l}(\|u^{(m)}\|_1+\|\epsilon^{(m)}\|_1)\Bigg],
\end{align*}
\begin{align*}
    &\|[(u^{(1)}+\epsilon^{(1)})\otimes (u^{(2)}+\epsilon^{(2)})\otimes (u^{(3)}+\epsilon^{(3)})\otimes (u^{(4)}+\epsilon^{(4)})\\
    &-u^{(1)}\otimes u^{(2)}\otimes u^{(3)}\otimes u^{(4)}]\circ \cN^{(2)}\|_1\\
    \leq&C\sum_{l=1}^4\left[\sum_{j=1}^p\frac{\left|\epsilon^{(l)}_j\right|}{\sqrt{\mymin_{k}n_{j,k}}}\prod_{m\neq l}\left(\frac{\|u^{(m)}\|_1}{\mymin_{j\in \mathrm{supp}(u^{(m)})}\mymin_{k}n_{j,k}}+\sum_{j=1}^p\frac{\left|\epsilon^{(m)}_j\right|}{\sqrt{\mymin_{k}n_{j,k}}}\right)\right],
\end{align*}
where $C$ is a universal constant.
\end{lemma}
We would like to apply Lemma \ref{lem:kron_err3} with $u^{(1)}=u^{(3)}=\Theta^*_{b,:}$, $u^{(2)}=u^{(4)}=\Theta^*_{a,:}$, $\epsilon^{(1)}=\epsilon^{(3)} = \widehat{\Theta}_{b,:}-\Theta^*_{b,:}$, and $\epsilon^{(2)}=\epsilon^{(4)}=\widehat{\Theta}_{a,:}-\Theta^*_{a,:}$. Now what remain to show are the $\ell_1$ and weighted $\ell_1$ bounds for $\widehat{\Theta}_{a,:}-\Theta^*_{a,:}$ and $\widehat{\Theta}_{b,:}-\Theta^*_{b,:}$. First note that
\begin{equation*}
    \begin{split}
        |\widehat{\Theta}_{a,j} - \Theta^*_{a,j}| \leq |\widehat{\Theta}_{a,a}-\Theta^*_{a,a}||\overline{\theta}^{(a)*}_j| + |\widehat{\overline{\theta}}^{(a)}_j-\overline{\theta}^{(a)*}_j||\Theta^*_{a,a}|+|\widehat{\Theta}_{a,a}-\Theta^*_{a,a}||\widehat{\overline{\theta}}^{(a)}_j-\overline{\theta}^{(a)*}_j|,
    \end{split}
\end{equation*}
and hence
\begin{equation*}
    \begin{split}
        \|\widehat{\Theta}_{a,:} - \Theta^*_{a,:}\|_1 \leq& |\widehat{\Theta}_{a,a}-\Theta^*_{a,a}|\|\overline{\theta}^{(a)*}\|_1 + \|\widehat{\theta}^{(a)}-\theta^{(a)*}\|_1|\Theta^*_{a,a}|+|\widehat{\Theta}_{a,a}-\Theta^*_{a,a}|\|\widehat{\theta}^{(a)}-\theta^{(a)*}\|_1,\\
        \sum_{j=1}^p\left|\frac{\widehat{\Theta}_{a,j} - \Theta^*_{a,j}}{\sqrt{\min_{k}n_{j,k}}}\right| \leq& |\widehat{\Theta}_{a,a}-\Theta^*_{a,a}|\sum_{j=1}^p\left|\frac{\overline{\theta}^{(a)*}_j}{\sqrt{\min_{k}n_{j,k}}}\right| + \sum_{j=1}^p\left|\frac{\widehat{\theta}^{(a)}_j-\theta^{(a)*}_j}{\sqrt{\min_{k}n_{j,k}}}\right||\Theta^*_{a,a}|\\
        &+|\widehat{\Theta}_{a,a}-\Theta^*_{a,a}|\sum_{j=1}^p\left|\frac{\widehat{\theta}^{(a)}_j-\theta^{(a)*}_j}{\sqrt{\min_k n_{j,k}}}\right|,
    \end{split}
\end{equation*}
where we have applied the fact that $\widehat{\overline{\theta}}^{(a)}-\overline{\theta}^{(a)*} = -(\widehat{\theta}^{(a)}-\theta^{(a)*})$.
Recall from Theorem \ref{thm:nb_lasso_err} that
\begin{equation*}
    \begin{split}
        \|\widehat{\theta}^{(a)}-\theta^{(a)*}\|_1\leq &\frac{C\|\Sigma^*\|_{\infty}(\sqrt{\gamma_a}+1)}{\lambda_{\mymin}(\Sigma^*)}\frac{\|\Theta^*_{:,a}\|_1}{\Theta^*_{a,a}}\sqrt{\frac{d_a^2\log p}{\mymin_{j\in \cN_a}\mymin_k n_{j,k}}},\\
        \sum_{j=1}^p\left|\frac{\widehat{\theta}^{(a)}_j-\theta^{(a)*}_j}{\sqrt{\min_k n_{j,k}}}\right|\leq&\frac{C\|\Sigma^*\|_{\infty}}{\lambda_{\mymin}(\Sigma^*)}\frac{\|\Theta^*_{:,a}\|_1}{\Theta^*_{a,a}}\frac{d_a\sqrt{\log p}}{\mymin_{j\in \cN_a}\mymin_k n_{j,k}}.
    \end{split}
\end{equation*}
Also, we just showed the bounds for $\varepsilon^{(2)}=|\widehat{\Theta}_{a,a}-\Theta^*_{a,a}|$ in Section \ref{sec:proof_precision_B}:
\begin{equation}\label{eq:proof_precision_epsilon2}
    |\varepsilon^{(2)}|\leq \frac{C\|\Sigma^*\|_{\infty}}{\lambda_{\min}^2(\Sigma^*)}(d_a+1)\sqrt{\frac{\log p}{n_1'^{(a,b)}}}\leq C\Theta^*_{a,a}.
\end{equation}
Therefore, 
\begin{equation*}
    \begin{split}
        \|\widehat{\Theta}_{a,:} - \Theta^*_{a,:}\|_1 \leq& \frac{C\|\Sigma^*\|_{\infty}\|\Theta^*_{:,a}\|_1}{\lambda_{\min}(\Sigma^*)}(\kappa_{\Sigma^*}+\sqrt{\gamma_a}+1)(d_a+1)\sqrt{\frac{\log p}{n_1'^{(a,b)}}}\leq C\|\Theta^*_{a,:}\|_1,\\
        \sum_{j=1}^p\left|\frac{\widehat{\Theta}_{a,j} - \Theta^*_{a,j}}{\sqrt{\min_{k}n_{j,k}}}\right| \leq& \frac{C\|\Sigma^*\|_{\infty}\|\Theta^*_{:,a}\|_1}{\lambda_{\min}(\Sigma^*)}\kappa_{\Sigma^*}(d_a+1)\sqrt{\frac{\log p}{n_1'^{(a,b)}}}\frac{1}{\sqrt{\min_{j\in \overline{\cN}_a}\min_kn_{j,k}}}\leq \frac{\|\Theta^*_{a,:}\|_1}{\sqrt{\min_{j\in \overline{\cN}_a}\min_k n_{j,k}}}.
    \end{split}
\end{equation*}
where we have applied Assumption \ref{assump:inference_n_B2} in the second inequalities.

While for $\widehat{\Theta}_{b,:}-\Theta^*_{b,:}$, note that $\Theta^*_{b,:}=\Theta^{(a)*}_{b,:}+\frac{\Theta^*_{a,b}}{\Theta^*_{a,a}}\Theta^*_{a,:}$, and hence
\begin{equation*}
\begin{split}
    |\widehat{\Theta}_{b,j}-\Theta^*_{b,j}|\leq &|\widetilde{\Theta}^{(a)}_{b,j}-\Theta^{(a)*}_{b,j}|+|\widehat{\Theta}_{a,a}\widehat{\theta}^{(a)}\widehat{\overline{\theta}}^{(a)}_j+\Theta^*_{a,b}\overline{\theta}^{(a)*}_j|\\
    \leq& |\widetilde{\Theta}^{(a)}_{b,j}-\Theta^{(a)*}_{b,j}|+|\widehat{\Theta}_{a,a}\widehat{\theta}^{(a)}_b + \Theta^*_{a,b}||\overline{\theta}^{(a)*}_j|+|\Theta^*_{a,b}||\widehat{\overline{\theta}}^{(a)}_j - \overline{\theta}^{(a)*}_j|\\
    &+|\widehat{\Theta}_{a,a}\widehat{\theta}^{(a)}_b + \Theta^*_{a,b}||\widehat{\overline{\theta}}^{(a)}_j - \overline{\theta}^{(a)*}_j|,
\end{split}
\end{equation*}
which implies that
\begin{equation*}
    \begin{split}
        \|\widehat{\Theta}_{b,:}-\Theta^*_{b,:}\|_1\leq&\|\widetilde{\Theta}^{(a)}_{b,:}-\Theta^{(a)*}_{b,:}\|_1+|\widehat{\Theta}_{a,a}\widehat{\theta}^{(a)}_b + \Theta^*_{a,b}|\|\overline{\theta}^{(a)*}\|_1\\
        &+|\Theta^*_{a,b}|\|\widehat{\overline{\theta}}^{(a)} - \overline{\theta}^{(a)*}\|_1+|\widehat{\Theta}_{a,a}\widehat{\theta}^{(a)}_b + \Theta^*_{a,b}|\|\widehat{\overline{\theta}}^{(a)} - \overline{\theta}^{(a)*}\|_1\\
        \sum_{j=1}^p\left|\frac{\widehat{\Theta}_{b,j} - \Theta^*_{b,j}}{\sqrt{\min_{k}n_{j,k}}}\right| \leq& \sum_{j=1}^p\left|\frac{\widetilde{\Theta}^{(a)}_{b,j} - \Theta^{(a)*}_{b,j}}{\sqrt{\min_{k}n_{j,k}}}\right|+|\widehat{\Theta}_{a,a}\widehat{\theta}^{(a)}_b + \Theta^*_{a,b}||\sum_{j=1}^p\left|\frac{\overline{\theta}^{(a)*}_j}{\sqrt{\min_{k}n_{j,k}}}\right|\\
        &+ \sum_{j=1}^p\left|\frac{\widehat{\theta}^{(a)}_j-\theta^{(a)*}_j}{\sqrt{\min_{k}n_{j,k}}}\right||\Theta^*_{a,b}|
        +|\widehat{\Theta}_{a,a}\widehat{\theta}^{(a)}_b + \Theta^*_{a,b}|\sum_{j=1}^p\left|\frac{\widehat{\theta}^{(a)}_j-\theta^{(a)*}_j}{\sqrt{\min_k n_{j,k}}}\right|.
    \end{split}
\end{equation*}
The associated upper bounds for $\widetilde{\Theta}^{(a)}_{b,:}-\Theta^{(a)*}_{b,:}$ and $\widehat{\theta}^{(a)}-\theta^{(a)*}$ are established in Lemma \ref{lem:debias_nb_lasso_err} and Theorem \ref{thm:nb_lasso_err}. While for the error term $|\widehat{\Theta}_{a,a}\widehat{\theta}^{(a)}_b + \Theta^*_{a,b}|$, we bound it as follows:
\begin{equation*}
    \begin{split}
        |\widehat{\Theta}_{a,a}\widehat{\theta}^{(a)}_b + \Theta^*_{a,b}|\leq &|\widehat{\Theta}_{a,a}-\Theta^*_{a,a}||\theta^{(a))*}_b|+\Theta^*_{a,a}|\widehat{\theta}^{(a)}_b + \theta^{(a)*}_b| + |\widehat{\Theta}_{a,a}-\Theta^*_{a,a}||\widehat{\theta}^{(a)}_b + \theta^{(a)*}_b|\\
        \leq&\frac{C\|\Sigma^*\|_{\infty}}{\lambda_{\min}(\Sigma^*)}\left(\kappa_{\Sigma^*}|\Theta^*_{a,b}|+\|\Theta^*_{:,a}\|_2\right)(d_a+1)\sqrt{\frac{\log p}{n_1'^{(a,b)}}}\\
        \leq&\frac{C\|\Sigma^*\|_{\infty}\kappa_{\Sigma^*}}{\lambda_{\min}(\Sigma^*)}(|\Theta^*_{a,b}|+\Theta^*_{a,a})(d_a+1)\sqrt{\frac{\log p}{n_1'^{(a,b)}}},
    \end{split}
\end{equation*}
where we have applied \eqref{eq:proof_precision_epsilon2} and the $\ell_2$-norm err bound for $\widehat{\theta}^{(a)} - \theta^{(a)*}$ in Theorem \ref{thm:nb_lasso_err} in the second line, and the third line is due to the fact that $\frac{\|\Theta^*_{:.a}\|_2^2}{\Theta^*_{a,a}}\leq\frac{\lambda_{\max}(\Theta^*)}{\lambda_{\min}(\Theta^*)}=\kappa_{\Sigma^*}$. 
Combining the bounds above with Lemma \ref{lem:debias_nb_lasso_err} and Theorem \ref{thm:nb_lasso_err}, we have
\begin{equation*}
    \begin{split}
        &\|\widehat{\Theta}_{b,:}-\Theta^*_{b,:}\|_1\\
        \leq&\frac{C\|\Sigma^*\|_{\infty}(\kappa_{\Sigma^*}+\sqrt{\gamma_a}+\sqrt{\gamma^{(a)}_b})}{\lambda_{\mymin}(\Sigma^*)}\left[\|\Theta^*_{:,b}\|_1+\|\Theta^*_{:,a}\|_1\left(1+\frac{|\Theta^*_{a,b}|}{\Theta^*_{a,a}}\right)\right](d_a+d_b+1)\sqrt{\frac{\log p}{n_1'^{(a,b)}}}\\
        \leq&C\left[\|\Theta^*_{:,b}\|_1+\|\Theta^*_{:,a}\|_1\left(1+\frac{|\Theta^*_{a,b}|}{\Theta^*_{a,a}}\right)\right],\\
        &\sum_{j=1}^p\left|\frac{\widehat{\Theta}_{b,j} - \Theta^*_{b,j}}{\sqrt{\min_{k}n_{j,k}}}\right|\\
        \leq& \frac{C\kappa_{\Sigma^*}\|\Sigma^*\|_{\infty}}{\lambda_{\mymin}(\Sigma^*)}\left[\|\Theta^{*}_{b,:}\|_1+\|\Theta^*_{:,a}\|_1\left(1+\frac{|\Theta^*_{a,b}|}{\Theta^*_{a,a}}\right)\right]\frac{(d_a+d_b+1)\sqrt{\log p}}{n_1'^{(a,b)}}\\
        \leq&\frac{C\left[\|\Theta^*_{:,b}\|_1+\|\Theta^*_{:,a}\|_1\left(1+\frac{|\Theta^*_{a,b}|}{\Theta^*_{a,a}}\right)\right]}{\sqrt{n_1'^{(a,b)}}},
    \end{split}
\end{equation*}
where we have applied Assumption \ref{assump:inference_n_B2} in the last inequalities. 
Now we apply Lemma \ref{lem:kron_err3} and obtain the following bounds for $\|\cE\circ \cN\|_1$, $\|\cE\circ \cN^{(1)}\|_1$, $\|\cE\circ \cN^{(2)}\|_1$:
\begin{equation*}
    \begin{split}
        \|\cE\circ \cN\|_1\leq& \frac{C\kappa_{\Sigma^*}\|\Sigma^*\|_{\infty}\|\Theta^*_{:,a}\|_1^2}{\lambda_{\min}(\Sigma^*)}\left[\|\Theta^*_{:,b}\|_1+\|\Theta^*_{:,a}\|_1\left(1+\frac{|\Theta^*_{a,b}|}{\Theta^*_{a,a}}\right)\right]^2\frac{(d_a+d_b+1)\sqrt{\log p}}{(n_1'^{(a,b)})^{3/2}},\\
        \|\cE\circ \cN^{(1)}\|_1\leq& \frac{C\kappa_{\Sigma^*}\|\Sigma^*\|_{\infty}\|\Theta^*_{:,a}\|_1^2}{\lambda_{\min}(\Sigma^*)}\left[\|\Theta^*_{:,b}\|_1+\|\Theta^*_{:,a}\|_1\left(1+\frac{|\Theta^*_{a,b}|}{\Theta^*_{a,a}}\right)\right]^2\frac{(d_a+d_b+1)\sqrt{\log p}}{(n_1'^{(a,b)})^{2}},\\
        \|\cE\circ \cN^{(2)}\|_1\leq& \frac{C\kappa_{\Sigma^*}\|\Sigma^*\|_{\infty}\|\Theta^*_{:,a}\|_1^2}{\lambda_{\min}(\Sigma^*)}\left[\|\Theta^*_{:,b}\|_1+\|\Theta^*_{:,a}\|_1\left(1+\frac{|\Theta^*_{a,b}|}{\Theta^*_{a,a}}\right)\right]^2\frac{(d_a+d_b+1)\sqrt{\log p}}{(n_1'^{(a,b)})^{5/2}}.
    \end{split}
\end{equation*}
Therefore, using the same steps as in the proof of Proposition \ref{prop:var_est}, we have
\begin{equation*}
    \begin{split}
        \mathrm{\RNum{2}}\leq& \frac{C\kappa_{\Sigma^*}\|\Sigma^*\|_{\infty}^3\|\Theta^*_{:,a}\|_1^2}{\lambda_{\min}(\Sigma^*)}\left[\|\Theta^*_{:,b}\|_1+\|\Theta^*_{:,a}\|_1\left(1+\frac{|\Theta^*_{a,b}|}{\Theta^*_{a,a}}\right)\right]^2\frac{(d_a+d_b+1)\sqrt{\log p}}{(n_1'^{(a,b)})^{3/2}},\\
        \mathrm{\RNum{3}}\leq& \frac{C\kappa_{\Sigma^*}\|\Sigma^*\|_{\infty}^3\|\Theta^*_{:,a}\|_1^2}{\lambda_{\min}(\Sigma^*)}\left[\|\Theta^*_{:,b}\|_1+\|\Theta^*_{:,a}\|_1\left(1+\frac{|\Theta^*_{a,b}|}{\Theta^*_{a,a}}\right)\right]^2\frac{(d_a+d_b+1)\log p}{(n_1'^{(a,b)})^{2}}.
    \end{split}
\end{equation*}
Finally, plugging in the bounds for $\mathrm{\RNum{1}}$, $\mathrm{\RNum{2}}$, and $\mathrm{\RNum{3}}$ into \eqref{eq:precision_var_err_decomp}, one directly has
\begin{align*}
    &|\widehat{\sigma}_{n,\Theta}^2(a,b)-\sigma_{n,\Theta}^2(a,b)|\\
    \leq &\frac{C\kappa_{\Sigma^*}\|\Sigma^*\|_{\infty}^3\|\Theta^*_{:,a}\|_1^2}{\lambda_{\min}(\Sigma^*)}\left[\|\Theta^*_{:,b}\|_1+\|\Theta^*_{:,a}\|_1\left(1+\frac{|\Theta^*_{a,b}|}{\Theta^*_{a,a}}\right)\right]^2\frac{(d_a+d_b+1)\sqrt{\log p}}{(n_1'^{(a,b)})^{3/2}}.
\end{align*}
Recall the lower bound for $\sigma_{n,\Theta}^2(a,b)$ we have shown earlier in \eqref{eq:precision_var_lowerbnd}, which implies the following:
\begin{equation*}
    \frac{|\widehat{\sigma}_{n,\Theta}^2(a,b)-\sigma_{n,\Theta}^2(a,b)|}{\sigma_{n,\Theta}^2(a,b)}\leq C'(\Theta^*;a,b)^2\frac{(d_a+d_b+1)\sqrt{\log p}}{\sqrt{n_1'^{(a,b)}}}\frac{n_2'^{(a,b)}}{n_1'^{(a,b)}}.
\end{equation*}
Following the same arguments as in the proof of Proposition \ref{prop:var_est}, and applying Assumption \ref{assump:var_est2}, we arrive at our conclusion that $\frac{\widehat{\sigma}_{n,\Theta}(a,b)}{\sigma_{n,\Theta}(a,b)}-1\overset{p}{\rightarrow} 0$. Therefore, the proof is now complete.
\subsubsection{Proof for the \texorpdfstring{$a=b$}{Lg} case}
Now we show that for the $a=b$ case, the confidence interval $\widehat{\mathbb{C}}_{\Theta,\alpha}^{a,a}$ is also asymptotically valid under Assumptions \ref{assump:inference_n_B2}-\ref{assump:var_est2}. Let $\sigma_{n,\Theta}^2(a,a)= \mathcal{T}^{(n)*}\times_1\Theta^*_{:,a}\times_2\Theta^*_{:,a}\times_3\Theta^*_{:,a}\times_4\Theta^*_{:,a}$. We will show that $\frac{\widehat{\Theta}_{a,a}-\Theta^*_{a,a}}{\sigma_{n,\Theta}(a,a)}\rightarrow\mathcal{N}(0,1)$ and $\widehat{\sigma}_{n,\Theta}^2(a,a)$ is a consistent estimator for $\sigma_{n,\Theta}^2(a,a)$.

Recall that we have shown the following decomposition in Section \ref{sec:proof_Theta_ab_decomp}:
\begin{equation*}
    \widehat{\Theta}_{a,a}-\Theta^*_{a,a} = -(\Theta^*_{a,a})^2\varepsilon^{(1)}_1 - [(\Theta^*_{a,a})^2\varepsilon^{(1)}_2 - \varepsilon^{(1)2}\widehat{\Theta}_{a,a}(\Theta^*_{a,a})^2].
\end{equation*}
Following the same steps as in Section \ref{sec:proof_precision_B} and Section \ref{sec:proof_CLT_Theta}, one can show that with probability at least $1-Cp^{-c}$,
\begin{equation*}
    \left|(\Theta^*_{a,a})^2\varepsilon^{(1)}_2 - \varepsilon^{(1)2}\widehat{\Theta}_{a,a}(\Theta^*_{a,a})^2\right|\leq \frac{C\|\Sigma^*\|_{\infty}^2\lambda_{\max}(\Sigma^*)\|\Theta^*_{:,a}\|_1}{\lambda_{\min}(\Sigma^*)}\frac{(d_a+1)\log p}{n_1'^{(a,a)}},
\end{equation*}
and $\sigma_{n,\Theta}^2(a,a)\geq \frac{2\lambda_{\min}^2(\Sigma^*)\min_{j\in \overline{\mathcal{N}}_a}(\Theta^*_{j,a})^4}{n_2'^{(a,a)}}$. Therefore, Assumption \ref{assump:inference_n_BE2} suggests that
\begin{equation*}
    \frac{\left|(\Theta^*_{a,a})^2\varepsilon^{(1)}_2 - \varepsilon^{(1)2}\widehat{\Theta}_{a,a}(\Theta^*_{a,a})^2\right|}{\sigma_{n,\Theta}(a,a)}\leq C'(\Theta^*;a,a)\kappa_{\Sigma^*}(da+1)\log p\frac{\sqrt{n_2'^{(a,a)}}}{n_1'^{(a,a)}}=o(1).
\end{equation*}
Furthermore, following the same argument as in Section \ref{sec:proof_CLT_Theta}, one can show that $\frac{-(\Theta^*_{a,a})^2\varepsilon_1^{(1)}}{\sigma_{n,\Theta}(a,a)}\rightarrow \mathcal{N}(0,1)$ based on Assumption \ref{assump:inference_n_E2}. Therefore, the final step is to show that $\frac{\widehat{\sigma}_{n,\Theta}(a,a)}{\sigma_{n,\Theta}^2(a,a)}\overset{p}{\rightarrow}1$. Still, we follow the same arguments as in Section \ref{sec:proof_CLT_Theta}, and achieve the following bound:
\begin{equation*}
    |\widehat{\sigma}_{n,\Theta}^2(a,a)-\sigma_{n,\Theta}^2(a,a)|\leq C\kappa^2_{\Sigma^*}\lambda_{\max}^2(\Sigma^*)\|\Theta^*_{:,a}\|_1^4\frac{(d_a+1)\sqrt{\log p}}{(n_1'^{(a,a)})^{{3/2}}}.
\end{equation*}
Putting this error bound above and the lower bound $\sigma_{n,\Theta}^2(a,a)\geq \frac{2\lambda_{\min}^2(\Sigma^*)\min_{j\in \overline{\mathcal{N}}_a}(\Theta^*_{j,a})^4}{n_2'^{(a,a)}}$ together, one has
\begin{equation*}
    \frac{|\widehat{\sigma}^2_{n,\Theta}(a,a)-\sigma_{n,\Theta}^2(a,a)|}{\sigma_{n,\Theta}^2(a,a)}\leq C'^2(\Theta;a,a)\frac{(d_a+1)\sqrt{\log p}}{\sqrt{n_1'^{(a,a)}}}\frac{n_2'^{(a,a)}}{n_1'^{(a,a)}},
\end{equation*}
with high probability. A direct application of Assumption \ref{assump:var_est2} verifies the final step: $\frac{\widehat{\sigma}_{n,\Theta}(a,a)}{\sigma_{n,\Theta}^2(a,a)}\overset{p}{\rightarrow}1$. The proof is now complete.
\subsection{Proof of Theorem \ref{thm:FDR_valid}}
Here, we provide the detailed proofs of Theorem \ref{thm:FDR_valid} with Assumptions \ref{assump:FDR_n} and \ref{assump:FDR_correlation}. Let $m=p(p-1)/2$ be the number of edge-wise tests, $m_0=\{(j,k)\in [p]:\Theta^*_{j,k}=0\}$ be the number of node pairs for which the null hypothesis holds true, and the Gaussian tail probability function $G(t)=2(1-\Phi(t))$ for $t\geq 0$. Recall that $t_p=\sqrt{2\log m-2\log\log m}$,
$$
\mathrm{FDP}=\frac{\sum_{(a,b)\in \mathcal{H}_0}\ind{p_{(a,b)}\leq \rho_0}}{R(\rho_0)\vee 1}=\frac{\sum_{(a,b)\in \mathcal{H}_0}\ind{|\widehat{z}(a,b)|\geq t_0}}{R(\rho_0)\vee 1},
$$
where $t_0=G^{-1}(\rho_0)$, $\rho_0=\sup\left\{G(t_p)\leq \rho\leq 1: \frac{m\rho}{R(\rho)\vee 1}\leq \alpha\right\}$, if there exists $G(t_p)\leq \rho\leq 1$ such that $\frac{m\rho}{R(\rho)\vee 1}\leq \alpha$;
otherwise, $\rho_0=G(\sqrt{2\log m})$. 

\paragraph{Case I: $\rho_0=G(\sqrt{2\log m})$} We first show that when $\rho_0=G(\sqrt{2\log m})$ and $t_0=\sqrt{2\log m}$, FDP$\overset{p}{\rightarrow} 0$. Applying the decomposition in Theorem \ref{thm:nb_lasso_debias_decomp} to edge $(a,b)$, we have $\widehat{z}(a,b) = \frac{B(a,b)+E(a,b)}{\widehat{\sigma}_n(a,b)}$ if $(a,b)\in \mathcal{H}_0$, where 
$$
\frac{B(a,b)}{\sigma_n(a,b)}\leq C(\Theta^*;a,b)(\kappa_{\Theta^*}^2+\sqrt{\gamma_a}+\sqrt{\gamma^{(a)}_b})\frac{(d_a+d_b+1)\log p}{\sqrt{n_1^{(a,b)}}}\sqrt{\frac{n_2^{(a,b)}}{n_1^{(a,b)}}}:=\varepsilon^{(1)}_{a,b}
$$ 
with probability at least $1-Cp^{-c}$, and the proof of Proposition \ref{prop:var_est} implies that with the same probability, $$\left|\frac{\sigma_n(a,b)}{\widehat{\sigma}_n(a,b)}-1\right|\leq \frac{C^2(\Theta^*;a,b)}{\kappa_{\Theta^*}^2}\frac{(d_a+d_b+1)n_2^{(a,b)}\sqrt{\log p}}{(n_1^{(a,b)})^{3/2}}:=\varepsilon^{(2)}_{a,b}.$$
Hence some calculations show that
\begin{align*}
    \mathbb{P}(\mathrm{FDP}>0)\leq& \sum_{(a,b)\in \mathcal{H}_0}\mathbb{P}\left(|\widehat{z}(a,b)|\geq \sqrt{2\log m}\right)\\
    \leq& \sum_{(a,b)\in \mathcal{H}_0}\mathbb{P}\left(\left|\frac{B(a,b)}{\widehat{\sigma}_n(a,b)}+\frac{E(a,b)}{\widehat{\sigma}_n(a,b)}\right|\geq \sqrt{2\log m}\right)\\
    \leq &\sum_{(a,b)\in \mathcal{H}_0}\mathbb{P}\left(\left|\frac{B(a,b)}{\sigma_n(a,b)}\right|>\varepsilon^{(1)}_{a,b}\right)+\mathbb{P}\left(\left|\frac{\sigma_n(a,b)}{\widehat{\sigma}_n(a,b)}-1\right|>\varepsilon^{(2)}_{a,b}\right)+\mathbb{P}\left(\left|\frac{E(a,b)}{\sigma_n(a,b)}\right|\geq \frac{t_0}{1+\varepsilon^{(2)}_{a,b}} - \varepsilon^{(1)}_{a,b}\right)\\
    \leq &Cp^{-c}+ \sum_{(a,b)\in \mathcal{H}_0}\mathbb{P}\left(\left|\frac{E(a,b)}{\sigma_n(a,b)}\right|\geq \frac{t_0}{1+\varepsilon^{(2)}_{a,b}} - \varepsilon^{(1)}_{a,b}\right).
\end{align*}
The following lemma concentrates the tail bound of $\frac{E(a,b)}{\sigma_n(a,b)}$ around the Gaussian tail.
\begin{lemma}\label{lem:tail_approx_one_pair}
    For any node pair $(a,b)$, $t\geq 0$, $\varepsilon>0$
    \begin{align*}
        &\max\{\mathbb{P}\left(\left|\frac{E(a,b)}{\sigma_n(a,b)}\right|>t\right) - G(t-\varepsilon), G(t+\varepsilon)-\mathbb{P}\left(\left|\frac{E(a,b)}{\sigma_n(a,b)}\right|>t\right)\}\\
   \leq &C\exp\left\{-\frac{c\varepsilon \lambda_{\min}^3(\Sigma^*)\sqrt{n_2^{(a,b)}}}{C\|\Sigma^*\|_{\infty}^3(d_a+d_b+1)^3}\right\}.
    \end{align*}
\end{lemma}
Applying Lemma \ref{lem:tail_approx_one_pair} with $\varepsilon = (\log p)^{-2}$ gives us
\begin{align*}
    \mathbb{P}(\mathrm{FDP}>0)\leq& Cp^{-c}+ \sum_{(a,b)\in \mathcal{H}_0}G\left(\frac{t_0}{1+\varepsilon^{(2)}_{a,b}} - \varepsilon^{(1)}_{a,b}-(\log p)^{-2}\right)\\
    &+ Cp^2\exp\left\{-\frac{c \lambda_{\min}^3(\Sigma^*)\sqrt{n_2^{(a,b)}}}{C\|\Sigma^*\|_{\infty}^3(d_a+d_b+1)^3(\log p)^2}\right\}\\
    \leq &Cp^{-c}+ \sum_{(a,b)\in \mathcal{H}_0}G\left(\frac{t_0}{1+\varepsilon^{(2)}_{a,b}} - \varepsilon^{(1)}_{a,b}-(\log p)^{-2}\right),
\end{align*}
where the last line is due to Assumption \ref{assump:FDR_n}. While for the last term, let $x=t_0-(\frac{t_0}{1+\varepsilon^{(2)}_{a,b}} - \varepsilon^{(1)}_{a,b}-(\log p)^{-2})$, and we will also need the following lemma on Gaussian tail bound:
\begin{lemma}\label{lem:gaussian_tail}
For any $t>0$,
    $$
    \frac{2}{\sqrt{2\pi}(t+1/t)}e^{-\frac{t^2}{2}}\leq G(t)\leq \frac{2}{\sqrt{2\pi}t}e^{-\frac{t^2}{2}}.
    $$
\end{lemma}
Lemma \ref{lem:gaussian_tail} is a standard Gaussian tail bound, which also appears in \cite{javanmard2019false}.
Hence we have 
\begin{align*}
    G(t_0-x)\leq &\frac{2}{\sqrt{2\pi}(t_0-x)}e^{-\frac{(t_0-x)^2}{2}}\\
    \leq &\frac{2}{\sqrt{2\pi}(t_0-x)}e^{-\frac{t_0^2-2t_0x}{2}}\\
    \leq &\frac{1}{\sqrt{2\pi}(t_0-x)p^2}e^{\sqrt{\log p}x}.
\end{align*}
Further note that 
\begin{align*}
    x=&t_0-\left(\frac{t_0}{1+\varepsilon^{(2)}_{a,b}} - \varepsilon^{(1)}_{a,b}-(\log p)^{-2}\right)\\ 
    \leq &\sqrt{2\log m}\varepsilon_{a,b}^{(2)} + \varepsilon_{a,b}^{(1)} + (\log p)^{-2}\\
    \leq &C_1\frac{(d_a+d_b+1)n_2^{(a,b)}\log p}{(n_1^{(a,b)})^{3/2}} + C_2\frac{(d_a+d_b+1)\log p\sqrt{n_2^{(a,b)}}}{n_1^{(a,b)}} + (\log p)^{-2},
\end{align*}
and applying Assumption \ref{assump:FDR_n} leads to
$\sqrt{\log p}x\leq \frac{C}{\log p}$, which implies
$$
G(t_0-x)\leq\frac{C}{p^2\sqrt{\log p}}.
$$
Therefore, when $t_0=\sqrt{2\log m}$, 
$$
\mathbb{E}(\mathrm{FDP})\leq \mathbb{P}(\mathrm{FDP}>0)\leq \frac{C}{\sqrt{\log p}},
$$
where we have applied the boundedness of $\mathrm{FDP}$ in the inequality above.

\paragraph{Case II: $\rho_0 = \sup\left\{G(t_p)\leq \rho\leq 1: \frac{m\rho}{R(\rho)\vee 1}\leq \alpha\right\}$} Now suppose that $$\rho_0 = \sup\left\{G(t_p)\leq \rho\leq 1: \frac{m\rho}{R(\rho)\vee 1}\leq \alpha\right\},$$
and $t_0=G^{-1}(\rho_0)$, then 
\begin{align*}
    \mathrm{FDP} = &\frac{\sum_{(a,b)\in \mathcal{H}_0}\ind{|\widehat{z}(a,b)|\geq t_0}}{R(\rho_0)\vee 1}\\
    \leq &\frac{\left|\sum_{(a,b)\in \mathcal{H}_0}(\ind{|\widehat{z}(a,b)|\geq t_0} - \rho_0)\right|+\rho_0m}{R(\rho_0)\vee1}\\
    \leq &\alpha\left(1+\frac{\left|\sum_{(a,b)\in \mathcal{H}_0}(\ind{|\widehat{z}(a,b)|\geq t_0} - G(t_0))\right|}{mG(t_0)}\right)\\
    \leq &\alpha\left(1+\sup_{0\leq t\leq t_p}\frac{\left|\sum_{(a,b)\in \mathcal{H}_0}(\ind{|\widehat{z}(a,b)|\geq t} - G(t))\right|}{mG(t)}\right).
\end{align*}
In fact, it suffices to show that 
\begin{equation}\label{eq:FDR_proof_err_sup_interval}
    \sup_{0\leq t\leq t_p}\left|\frac{\sum_{(a,b)\in \mathcal{H}_0}\ind{|\widehat{z}(a,b)|\geq t}}{m_0G(t)}-1\right|\rightarrow 0
\end{equation}
in probability, which immediately implies that
$\lim_{n,p\rightarrow \infty}\mathbb{P}(\mathrm{FDP}>\alpha+\epsilon)=0$; In addition, by the boundedness of the quantity above, convergence in probability also implies convergence in expectation and hence we would have $\limsup_{n,p\rightarrow \infty}\mathrm{FDR}\leq \alpha$.

We now discretize the interval $[0,t_p]$ into $0=x_1<x_2<\dots<x_K=t_p$, where $x_k-x_{k-1} = \frac{1}{\sqrt{\log m\log\log m}}$. In the following, we show that \eqref{eq:FDR_proof_err_sup_interval} can be implied by 
\begin{equation}\label{eq:FDR_proof_err_max}
    \max_{1\leq k\leq K}\left|\frac{\sum_{(a,b)\in \mathcal{H}_0}\ind{|\widehat{z}(a,b)|\geq x_k}}{m_0G(x_k)}-1\right|\rightarrow 0
\end{equation}
in probability.
Suppose for now that \eqref{eq:FDR_proof_err_max} holds, then for any $t\in [0, t_p]$, there exists $1\leq l\leq K$
such that $x_l\leq t \leq x_{l+1}$, and hence
\begin{align*}
    \frac{\sum_{(a,b)\in \mathcal{H}_0}\ind{|\widehat{z}(a,b)|\geq x_{l+1}}}{m_0 G(x_{l})} - 1\leq\frac{\sum_{(a,b)\in \mathcal{H}_0}\ind{|\widehat{z}(a,b)|\geq t}}{m_0G(t)} - 1\leq\frac{\sum_{(a,b)\in \mathcal{H}_0}\ind{|\widehat{z}(a,b)|\geq x_l}}{m_0 G(x_{l+1})} - 1\\
    \left|\frac{\sum_{(a,b)\in \mathcal{H}_0}\ind{|\widehat{z}(a,b)|\geq t}}{m_0G(t)} - 1\right|\leq\max_{1\leq k\leq K}\left|\frac{\sum_{(a,b)\in \mathcal{H}_0}\ind{|\widehat{z}(a,b)|\geq x_k}}{m_0 G(x_{k})} - 1\right|\frac{G(x_l)}{G(x_{l+1})} +\left|\frac{G(x_l)}{G(x_{l+1})} - 1\right|.
\end{align*}
Let $\phi(\cdot)$ be the density function of standard Gaussian distribution. Now note that
$\frac{G(x_{l+1})}{G(x_{l})}\geq 1-\frac{2(x_{l+1}-x_l)\phi(x_l)}{G(x_{l})}$, and Lemma \ref{lem:gaussian_tail} suggests
$$
\frac{\phi(x_l)}{G(x_{l})}\leq \begin{cases}
\frac{x_l+1}{2}\leq t_p,&x_l\geq 1,\\
\frac{\phi(0)}{G(1)}, &0\leq x_l\leq 1.
\end{cases}
$$
Hence $|\frac{G(x_{l+1})}{G(x_{l})}-1|\leq c\sqrt{\log \log m}\rightarrow 0$, and \eqref{eq:FDR_proof_err_max} implies \eqref{eq:FDR_proof_err_sup_interval}.

The rest of the proof is devoted to showing \eqref{eq:FDR_proof_err_sup_interval}. For any $\epsilon>0$, some calculations show that
\begin{equation}\label{eq:FDR_proof_prob_ratio_bound}
\begin{split}
    &\mathbb{P}\left(\max_{1\leq k\leq K}\left|\frac{\sum_{(a,b)\in \mathcal{H}_0}\ind{|\widehat{z}(a,b)|\geq x_k}}{m_0G(x_k)}-1\right|>\epsilon\right)\\
    \leq &\sum_{k=1}^K\mathbb{P}\left(\left|\frac{\sum_{(a,b)\in \mathcal{H}_0}\ind{|\widehat{z}(a,b)|\geq x_k}}{m_0G(x_k)}-1\right|>\epsilon\right)\\
    \leq&\frac{1}{\epsilon^2}\sum_{k=1}^K\mathbb{E}\left(\left|\frac{\sum_{(a,b)\in \mathcal{H}_0}\ind{|\widehat{z}(a,b)|\geq x_k}}{m_0G(x_k)}-1\right|^2\right)\\
    =&\frac{1}{\epsilon^2}\frac{1}{m_0^2}\sum_{(a,b)\in \mathcal{H}_0,(a',b')\in \mathcal{H}_0}\sum_{k=1}^K\left(\frac{\mathbb{P}(|\widehat{z}(a,b)|\geq x_k,|\widehat{z}(a',b')|\geq x_k)}{G^2(x_k)}-1\right)\\
    &-\frac{1}{\epsilon^2}\frac{2}{m_0}\sum_{(a,b)\in \mathcal{H}_0}\sum_{k=1}^K\left(\frac{\mathbb{P}(|\widehat{z}(a,b)|\geq x_k)}{G(x_k)}-1\right),
\end{split}
\end{equation}
and we will lower bound $\mathrm{\RNum{1}}(a,b)=\sum_{k=1}^K\left(\frac{\mathbb{P}(|\widehat{z}(a,b)|\geq x_k)}{G(x_k)}-1\right)$, and upper bound $\mathrm{\RNum{2}}(a,b,a',b')=\sum_{k=1}^K\left(\frac{\mathbb{P}(|\widehat{z}(a,b)|\geq x_k,|\widehat{z}(a',b')|\geq x_k)}{G^2(x_k)}-1\right)$, respectively.

\paragraph{Bounding $\mathrm{\RNum{1}}(a,b)$:} We first follow similar arguments to case I. Recall the decompositions and error bounds in the proof of Theorem \ref{thm:nb_lasso_debias_decomp} and Proposition \ref{prop:var_est}, we have
\begin{align*}
    \mathbb{P}(|\widehat{z}(a,b)|\geq x_k)=&\mathbb{P}\left(\left|\frac{B(a,b)}{\widehat{\sigma}_n(a,b)}+\frac{E(a,b)}{\widehat{\sigma}_n(a,b)}\right|\geq x_k\right)\\
    \geq &\mathbb{P}\left(\left|\frac{E(a,b)}{\sigma_n(a,b)}\right|\geq x_k(1+\varepsilon_{a,b}^{(2)}) + \varepsilon_{a,b}^{(1)}\right)-Cp^{-c}\\
    \geq &G(x_k(1+\varepsilon_{a,b}^{(2)}) + \varepsilon_{a,b}^{(1)}+(\log p)^{-2})-Cp^{-c}\\
    &-C\exp\left\{-\frac{c \lambda_{\min}^3(\Sigma^*)\sqrt{n_2^{(a,b)}}}{C\|\Sigma^*\|_{\infty}^3(d_a+d_b+1)^3(\log p)^2}\right\},
\end{align*}
where we applied Lemma \ref{lem:tail_approx_one_pair} with $\varepsilon=(\log p)^{-2}$ on the last line. Since $G(x_k)\geq G(t_p)\geq \frac{c}{t_p}e^{-\frac{t_p^2}{2}}\geq\frac{c}{p^2\sqrt{\log p}}$, $K=t_p\sqrt{\log m\log\log m}+1\leq C\log p\sqrt{\log\log p}$, 
one can show that 
\begin{align*}
\mathrm{\RNum{1}}(a,b)\geq &\sum_{k=1}^K\frac{G(x_k(1+\varepsilon_{a,b}^{(2)}) + \varepsilon_{a,b}^{(1)}+(\log p)^{-2})-G(x_k)}{G(x_k)} \\
&- C\exp\left\{3\log p - \frac{c \lambda_{\min}^3(\Sigma^*)\sqrt{n_2^{(a,b)}}}{C\|\Sigma^*\|_{\infty}^3(d_a+d_b+1)^3(\log p)^2}\right\} - Cp^{-c}.
\end{align*}
Assumption \ref{assump:FDR_n} implies that the last two terms both converge to zero, thus it suffices to show that
$\sum_{k=1}^K\frac{G(x_k(1+\varepsilon_{a,b}^{(2)}) + \varepsilon_{a,b}^{(1)}+(\log p)^{-2})-G(x_k)}{G(x_k)}\rightarrow 0$. In fact,
\begin{align*}
    0\leq&\sum_{k=1}^K\frac{G(x_k)-G(x_k(1+\varepsilon_{a,b}^{(2)}) + \varepsilon_{a,b}^{(1)}+(\log p)^{-2})}{G(x_k)}\\
    \leq &\sum_{k=1}^K\frac{(x_k\varepsilon_{a,b}^{(2)} + \varepsilon_{a,b}^{(1)}+(\log p)^{-2}))\phi(x_k)}{G(x_k)}\\
    \leq&Ct_p\sum_{k=1}^K(t_p\varepsilon_{a,b}^{(2)} + \varepsilon_{a,b}^{(1)}+(\log p)^{-2}))\\
    \leq&C_1\frac{(d_a+d_b+1)(\log p)^{\frac{5}{2}}\sqrt{\log\log p}n_2^{(a,b)}}{(n_1^{(a,b)})^{3/2}} + C_2\frac{(d_a+d_b+1)(\log p)^{\frac{5}{2}}\sqrt{\log\log p}\sqrt{n_2^{(a,b)}}}{n_1^{(a,b)}} + \sqrt{\frac{\log\log p}{\log p}},
\end{align*}
and Assumption \ref{assump:FDR_n} suggests the right hand side of the inequality above would converge to zero.

\paragraph{Bounding $\mathrm{\RNum{2}}(a,b,a',b')$:}
For any node pairs $(a,b)$ and $(a',b')$, we can utilize the same decomposition for the test statistics $\widehat{z}(a,b)$ and $\widehat{z}(a',b')$. Following similar arguments to bounding $\mathrm{\RNum{1}}(a,b)$, one can show that
\begin{align*}
    &\mathbb{P}(\widehat{z}(a,b)\geq x_k,\widehat{z}(a',b')\geq x_k)\\
    \leq&Cp^{-c}+\mathbb{P}\left(\left|\frac{E(a,b)}{\sigma(a,b)}\right|\geq\frac{x_k}{1+\varepsilon_{a,b}^{(2)}}-\varepsilon^{(1)}_{a,b},\left|\frac{E(a',b')}{\sigma(a',b')}\right|\geq\frac{x_k}{1+\varepsilon_{a',b'}^{(2)}}-\varepsilon^{(1)}_{a',b'}\right).
\end{align*}
The following lemma suggests that the tail bound for $(\frac{E(a,b)}{\sigma(a,b)},\frac{E(a',b')}{\sigma(a',b')})$ can be well approximated by the tail bound for a bivariate normal distribution with the same covariance matrix.
\begin{lemma}\label{lem:tail_approx_two_pairs}
    For any node pairs $(a,b)$ and $(a',b')$, $t_1, t_2\geq 0$, $\varepsilon>0$,
    \begin{align*}
        &\mathbb{P}\left(\left|\frac{E(a,b)}{\sigma_n(a,b)}\right|>t_1, \left|\frac{E(a',b')}{\sigma_n(a',b')}\right|>t_2\right)\\
        \leq &\mathbb{P}\left(\left|z_1\right|>t_1-\varepsilon, \left|z_2\right|>t_2-\varepsilon\right)+C\exp\left\{-\frac{c\varepsilon \lambda_{\min}(\Sigma^*)\sqrt{n_2^{(a,b)}\vee n_2^{(a',b')}}}{C\|\Sigma^*\|_{\infty}(d+1)\alpha(\Theta^*,\{V_i\}_{i=1}^n)}\right\},
    \end{align*}
    where $(z_1, z_2)$ follows bivariate normal distribution with the same mean and covariance matrix as $(\frac{E(a,b)}{\sigma_n(a,b)},\frac{E(a',b')}{\sigma_n(a',b')})$.
\end{lemma}
We can apply Lemma \ref{lem:tail_approx_two_pairs} with $\varepsilon=(\log p)^{-2}$, and obtain the following: 
\begin{align*}
    \mathrm{\RNum{2}}(a,b,a',b')=&\sum_{k=1}^K\left(\frac{\mathbb{P}(|\widehat{z}(a,b)\geq x_k, |\widehat{z}(a',b')|\geq x_k)}{G^2(x_k)}-1\right)\\
    \leq &Cp^{-c}+C\exp\left\{
    C\log p-\frac{c\lambda_{\min}(\Sigma^*)\sqrt{n_2^{(a,b)}\vee n_2^{(a',b')}}}{C\|\Sigma^*\|_{\infty}(d+1)\alpha(\Theta^*,\{V_i\}_{i=1}^n)}(\log p)^2\right\} \\
    &+ \sum_{k=1}^K\left(\frac{\mathbb{P}\left(|z_1|\geq x_k', |z_2|\geq x_k'\right)}{G^2(x_k)} - 1\right),
\end{align*}
where $x_k'=\frac{x_k}{1+\max\{\varepsilon_{a,b}^{(2)}, \varepsilon_{a',b'}^{(2)}\}}-\max\{\varepsilon^{(1)}_{a,b},\varepsilon^{(1)}_{a',b'}\} - (\log p)^{-2}$.
Assumption \ref{assump:FDR_n} implies that the first two terms both converge to zero, hence we will focus on the last term in the following analysis. We will discuss three different cases separately: (i) $(a, b, a', b')$ $\in \mathcal{A}_1(\rho_0)$; (ii) $(a, b, a', b')$ $\in \mathcal{A}_2(\rho_0,\gamma)$; (iii) $(a, b, a',b')$ $\in (\mathcal{H}_0\times \mathcal{H}_0)\backslash (\mathcal{A}_1(\rho_0)\cup \mathcal{A}_2(\rho_0,\gamma))$.

For the case (i), one can show that
\begin{align*}
    &\sum_{k=1}^K\left(\frac{\mathbb{P}\left(|z_1|\geq x_k', |z_2|\geq x_k'\right)}{G^2(x_k)} - 1\right)\\
    \leq &\sum_{k=1}^K\frac{1}{G^2(x_k)}\\
    \leq &\int_{0}^{t_p}G^{-1}(x)dx\sqrt{\log m\log\log m}\\
    \leq&\left[\sqrt{\frac{\pi}{2}}\int_{1}^{t_p}(x+\frac{1}{x})e^{\frac{x^2}{2}}dx + \frac{1}{G(1)}\right]\sqrt{\log m\log\log m}\\
    \leq &\left[\sqrt{2\pi}\int_{1}^{t_p}xe^{\frac{x^2}{2}}dx + \sqrt{2\pi e}\right]\sqrt{\log m\log\log m}\\
    \leq &\frac{\sqrt{2\pi}m\sqrt{\log \log m}}{\sqrt{\log m}}.
\end{align*}
Since $|\mathcal{A}_1(\rho_0)|\leq Cp^2$,
\begin{align*}
    \frac{1}{m_0^2}\sum_{(a,b,a',b')\in \mathcal{A}_1(\rho_0)}\mathrm{\RNum{2}}(a,b,a',b')\leq C\frac{\sqrt{\log \log p}}{\sqrt{\log p}}+Cp^{-c}\rightarrow 0.
\end{align*}

Now we present a lemma on the tail bound of bivariate Gaussian distribution, which will be useful for the other two cases.
\begin{lemma}\label{lem:bivariate_normal_tail}
    Suppose that $(z_1,z_2)^\top\sim \mathcal{N}(0,\begin{pmatrix}1 & \rho\\\rho & 1\end{pmatrix})$.
    \begin{itemize}
        \item If $|\rho|\leq C(\log p)^{-2-\gamma}$ for some $\gamma>0$, then 
        \begin{align}\label{eq:bivariate_vanishing_correlation}
           \sup_{0\leq x\leq \sqrt{C\log p}}\left|\frac{\mathbb{P}(|z_1|\geq x,|z_2|\geq x)}{G^2(x)} - 1\right|\leq C(\log p)^{-1-\min\{\gamma,\frac{1}{2}\}}. 
        \end{align}
        \item If $|\rho|\leq \rho_0<1$ for some $\rho_0>0$, then 
        \begin{align}\label{eq:bivariate_rho} \mathbb{P}(|z_1|\geq t,|z_2|\geq t)\leq C(t+1)^{-2}\exp\{-t^2/(1+\rho_0)\}  
        \end{align} 
        holds uniformly over $0\leq t \leq \sqrt{C\log p}$.
    \end{itemize}
\end{lemma}
Lemma \ref{lem:bivariate_normal_tail} is a direct implication of Lemma 6.1 and Lemma 6.2 in \cite{liu2013gaussian}.
For the case (ii), we apply \eqref{eq:bivariate_rho} to show that
\begin{align*}
    &\sum_{k=1}^K\left(\frac{\mathbb{P}\left(|z_1|\geq x_k', |z_2|\geq x_k'\right)}{G^2(x_k)} - 1\right)\\
    \leq & \sum_{k=1}^K\frac{(x_k+1)^2}{(x_k'+1)^2}\exp\left\{\frac{\rho_0}{1+\rho_0}x_k^2+\frac{x_k^2-x_k'^2}{1+\rho_0}\right\},
\end{align*}
where the second line is due to that $G(x_k)\geq \frac{C}{(x_k + 1)}e^{-\frac{x_k^2}{2}}$. Further note that
\begin{align*}
    x_k^2-x_k'^2 \leq 2(x_k-x_k')x_k\leq 2t_p^2\max\{\varepsilon_{a,b}^{(2)},\varepsilon_{a',b'}^{(2)}\} + 2t_p\max\{\varepsilon_{a,b}^{(1)},\varepsilon_{a',b'}^{(1)}\} + 2t_p(\log p)^{-2}\rightarrow 0,
\end{align*}
which then implies 
\begin{align*}
    &\sum_{k=1}^K\left(\frac{\mathbb{P}\left(|z_1|\geq x_k', |z_2|\geq x_k'\right)}{G^2(x_k)} - 1\right)\\
    \leq &C\sum_{k=1}^K\exp\left\{\frac{\rho_0}{1+\rho_0}x_k^2\right\}\\
    \leq &C\sqrt{\log m\log\log m}\int_{0}^{t_p}\exp\left\{\frac{\rho_0}{1+\rho_0}x^2\right\}dx\\
    \leq &Cm^{\frac{2\rho_0}{1+\rho_0}}(\log m)^{\frac{1}{2}-\frac{2\rho_0}{1+\rho_0}}\sqrt{\log\log m}.
\end{align*}
Therefore, Assumption \ref{assump:FDR_correlation} suggests
$$
\sum_{(a,b,a',b')\in \mathcal{A}_2(\rho_0,\gamma)}\mathrm{\RNum{2}}(a,b,a',b')=o(m^2).
$$

While for the case (iii), we have 
\begin{align*}
    &\sum_{k=1}^K\left(\frac{\mathbb{P}\left(|z_1|\geq x_k', |z_2|\geq x_k'\right)}{G^2(x_k)} - 1\right)\\
    \leq &\sum_{k=1}^K\left|\frac{\mathbb{P}\left(|z_1|\geq x_k', |z_2|\geq x_k'\right)}{G^2(x_k')} - 1\right|\left|\frac{G^2(x_k')}{G^2(x_k)} - 1\right|+\left|\frac{\mathbb{P}\left(|z_1|\geq x_k', |z_2|\geq x_k'\right)}{G^2(x_k')} - 1\right|+\left|\frac{G^2(x_k')}{G^2(x_k)} - 1\right|\\
    \leq &C(\log p)^{-1-\min\{\gamma,\frac{1}{2}\}}+C\frac{\phi(x_k')(x_k-x_k')}{G(x_k)}\\
    \leq &C\sqrt{\log\log p}(\log p)^{-\min\{\gamma,\frac{1}{2}\}}+C\log p\sqrt{\log \log p}(t_p\max\{\varepsilon_{a,b}^{(2)},\varepsilon_{a',b'}^{(2)}\} + \max\{\varepsilon_{a,b}^{(1)},\varepsilon_{a',b'}^{(1)}\} + (\log p)^{-2}),
\end{align*}
which converges to zero by Assumption \ref{assump:FDR_n}. Therefore, 
$$
\sum_{(a,b,a',b')\notin \mathcal{A}_1(\rho_0)\cup\mathcal{A}_2(\rho_0,\gamma)}\mathrm{\RNum{2}}(a,b,a',b')=o(m^2).
$$
Now we have shown that \eqref{eq:FDR_proof_prob_ratio_bound} converges to zero for any $\epsilon>0$, and hence \eqref{eq:FDR_proof_err_sup_interval} holds. The proof of Theorem \ref{thm:FDR_valid} is complete.
\subsection{Proof of Supporting Lemmas and Auxiliary Propositions}
Before presenting the proofs of the supporting lemmas, here we first introduce the notions of sub-Gaussian and sub-exponential random variables, $\psi_{\alpha}$ norm, which was introduced in \cite{adamczak2011restricted} as a generalization of the sub-Gaussian and sub-exponential norms when $\alpha = 2$ or $\alpha = 1$. 
\begin{definition}[$\psi_{\alpha}$-norm \citep{adamczak2011restricted,xia2022inference}]\label{def:psi_alpha}
    The $\psi_{\alpha}$-norm of any random variable $X$ and $\alpha>0$ is defined as $\|X\|_{\psi_{\alpha}}:=\inf\{C\in (0,\infty): \bE[\exp\{(|X|/C)^{\alpha}\}]\leq 2\}$
\end{definition}
The The following two lemmas from \cite{hao2020sparse} provide useful properties for product of random variables with bounded $\psi_{\alpha}$ norm.
\begin{lemma}[$\psi_{\alpha}$ norm of product of r.v.s \citep{hao2020sparse}]\label{lem:prod_psi_alpha}
Suppose $X_1,\dots,X_m$ are $m$ random variables (not necessarily independent) with $\psi_{\alpha}$-norm bounded by $\|X_j\|_{\psi_{\alpha}}\leq K_j$. Then the $\psi_{\alpha/m}$-norm of $\prod_{j=1}^mX_j$ is bounded as $\|\prod_{j=1}^mX_j\|_{\psi_{\alpha/m}}\leq \prod_{j=1}^m K_j$.
\end{lemma}
\begin{lemma}[Concentration of sum of r.v.s with bounded $\psi_{\alpha}$-norm \cite{hao2020sparse,vershynin2010introduction}]\label{lem:concentration_psi_alpha}
Suppose $0<\alpha\leq 1$, $X_1,\dots,X_n$ are independent random variables satisfying $\|X_i\|_{\psi_{\alpha}}\leq b$. Then there exists absolute constant $C(\alpha)$ only depending on $\alpha$ such that for any $a=(a_1,\dots,a_n)\in \bR^n$ and $0<\delta<1/e^2$, 
$$
\left|\sum_{i=1}^na_iX_i-\bE(\sum_{i=1}^n a_iX_i)\right|\leq C(\alpha)b\|a\|_2(\log \delta^{-1})^{1/2}+C(\alpha)b\|a\|_{\infty}(\log \delta^{-1})^{1/\alpha}
$$
with probability at least $1-\delta$.
\end{lemma}
Lemma \ref{lem:concentration_psi_alpha} is a direct combination of Lemma 6 in \cite{hao2020sparse} and Proposition 5.16 in \cite{vershynin2010introduction}.
\begin{proof}[Proof of Lemma~\ref{lem:SampleCov_entry_err}]
We first provide an entry-wise error bound for $$\left|\widehat{\Sigma}_{j,k}-\Sigma^*_{j,k}\right|=\left|\frac{1}{n_{j,k}}\sum_{i:j,k\in V_i}(x_{i,j}x_{i,k}-\Sigma^*_{j,k})\right|,$$ over $1\leq j,k \leq p$. 
		
Since $\frac{x_{i,j}}{\sqrt{\Sigma^*_{j,j}}}\sim \mathcal{N}(0,1)$, we have $\|x_{i,j}\|_{\psi_2}\leq C_1\sqrt{\Sigma^*_{j,j}}$ for some universal constant $C_1$~\cite[see e.g., Lemma 5.14 in][]{vershynin2010introduction}. By Lemma~\ref{lem:prod_psi_alpha}, for each $(j,k)$, $\|x_{i,j}x_{i,k}\|_{\psi_{1}}\leq C_1^2\|\Sigma^*\|_{\infty}$. 
Therefore, by Lemma~\ref{lem:concentration_psi_alpha}, we have
		$$
		|\widehat{\Sigma}_{j,k}-\Sigma^*_{j,k}|\leq C\|\Sigma^*\|_{\infty}\sqrt{\frac{\log p}{n_{j,k}}},
		$$
		with probability at least $1-\exp\{-c\log p\}$. 
		
		By the definition of $\widetilde{\Sigma}$ in (1), we know that 
		\begin{equation*}
		\begin{split}
		\sqrt{n_{j,k}}|\widetilde{\Sigma}_{j,k}-\widehat{\Sigma}_{j,k}|\leq& \mymax_{j',k'}\sqrt{n_{j',k'}}|\widetilde{\Sigma}_{j',k'}-\widehat{\Sigma}_{j',k'}|\\
		\leq&\mymax_{j',k'}\sqrt{n_{j',k'}}|\Sigma^*_{j',k'}-\widehat{\Sigma}_{j',k'}|,
		\end{split}
		\end{equation*}
hence we can bound $|\widetilde{\Sigma}_{j,k}-\Sigma^*_{j,k}|$ as follows: 
	\begin{equation*}
	\begin{split}
|\widetilde{\Sigma}_{j,k}-\Sigma^*_{j,k}|\leq& |\widehat{\Sigma}_{j,k}-\Sigma^*_{j,k}|+|\widetilde{\Sigma}_{j,k}-\widehat{\Sigma}_{j,k}|\\
	\leq&\frac{2}{\sqrt{n}_{j,k}}\mymax_{j',k'}\sqrt{n_{j',k'}}|\widehat{\Sigma}_{j',k'}-\Sigma^*_{j',k'}|\\
	\leq&C\|\Sigma^*\|_{\infty}\sqrt{\frac{\log p}{n_{j,k}}}.
	\end{split}
	\end{equation*}
\end{proof}
\begin{proof}[Proof of Lemma~\ref{lem:T_entry_err}]
    By the definition of $\cE_1$, 
    \begin{align*}
        (\cE_1)_{j,k,j',k'}=&\widetilde{\Sigma}_{j,j'}\widetilde{\Sigma}_{k,k'}-\Sigma^*_{j,j'}\Sigma^*_{k,k'}\\
        =&(\widetilde{\Sigma}_{j,j'}-\Sigma^*_{j,j'})\Sigma^*_{k,k'}+\Sigma^*_{j,j'}(\widetilde{\Sigma}_{k,k'}-\Sigma^*_{k,k'})+(\widetilde{\Sigma}_{j,j'}-\Sigma^*_{j,j'})(\widetilde{\Sigma}_{k,k'}-\Sigma^*_{k,k'}).
    \end{align*}
    By Lemma~\ref{lem:SampleCov_entry_err}, with probability at least $1-Cp^{-c}$, 
    \begin{align*}
        |(\cE_1)_{j,k,j',k'}|\leq &C\|\Sigma^*\|_{\infty}^2\left(\sqrt{\frac{\log p}{n_{j,j'}}}+\sqrt{\frac{\log p}{n_{k,k'}}}+\frac{\log p}{\sqrt{n_{j,j'}n_{k,k'}}}\right)
    \end{align*}
    holds for all $j,k,j',k'\in [p]$.
\end{proof}
\begin{proof}[Proof of Lemma \ref{lem:kron_err}]
    First we can decompose the error tensor as follows:
    \begin{align*}
        &[(u^{(1)}+\epsilon^{(1)})\otimes (u^{(2)}+\epsilon^{(2)})\otimes (u^{(3)}+\epsilon^{(3)})\otimes (u^{(4)}+\epsilon^{(4)})\\
    &-u^{(1)}\otimes u^{(2)}\otimes u^{(3)}\otimes u^{(4)}]\circ \cN\\
    =&\mathcal{D}_1+\mathcal{D}_2+\mathcal{D}_3+\mathcal{D}_4,
    \end{align*}
    where $\mathcal{D}_i$ is the sum of error terms of the $i$th order (including product of $i$ error terms). In particular, 
    \begin{equation}\label{eq:D1_decomp}
        \begin{split}
        \mathcal{D}_1=&(\epsilon^{(1)}\otimes u^{(2)}\otimes u^{(3)}\otimes u^{(4)})\circ \cN\\
        &+(u^{(1)}\otimes \epsilon^{(2)}\otimes u^{(3)}\otimes u^{(4)})\circ \cN\\
        &+(u^{(1)}\otimes u^{(2)}\otimes \epsilon^{(3)}\otimes u^{(4)})\circ \cN\\
        &+(u^{(1)}\otimes u^{(2)}\otimes u^{(3)}\otimes \epsilon^{(4)})\circ \cN.
        \end{split}
    \end{equation}
    We will prove an upper bound for $\|\mathcal{D}_1\|_1$ and show that all higher-order terms can be controlled by $C\|\mathcal{D}_1\|_1$ for some universal constant $C$.
    
    Now we bound each term in \eqref{eq:D1_decomp}. For the first term,
    \begin{align*}
        &\|(\epsilon^{(1)}\otimes u^{(2)}\otimes u^{(3)}\otimes u^{(4)})\circ \cN\|_1\\
        = &\sum_{j,k,j',k'}\frac{n_{j,k,j',k'}}{n_{j,k}n_{j',k'}}|\epsilon^{(1)}_j||u^{(2)}_k||u^{(3)}_{j'}||u^{(4)}_{k'}|\\
        \leq& \sum_{j,k,j',k'}\frac{1}{n_{j,k}}|\epsilon^{(1)}_j||u^{(2)}_k||u^{(3)}_{j'}||u^{(4)}_{k'}|\\
        \leq&\sum_{j}\frac{|\epsilon^{(1)}_j|}{\sqrt{\mymin_kn_{j,k}}}\sum_{k}\frac{|u^{(2)}_k|}{\sqrt{\mymin_jn_{j,k}}}\|u^{(3)}\|_1\|u^{(4)}\|_1\\
        \leq &\sum_{j}\frac{|\epsilon^{(1)}_j|}{\sqrt{\mymin_kn_{j,k}}}\frac{\|u^{(2)}\|_1}{\sqrt{\mymin_{k\in \mathrm{supp}(u^{(2)})}\mymin_{k}n_{j,k}}}\|u^{(3)}\|_1\|u^{(4)}\|_1
    \end{align*}
    where the third line is due to that $n_{j,k,j',k'}\leq n_{j,k}, n_{j',k'}$. Following similar arguments to the above, we have
    \begin{align*}
        &\|(u^{(1)}\otimes \epsilon^{(2)}\otimes u^{(3)}\otimes u^{(4)})\circ \cN\|_1\\
        \leq&\frac{\|u^{(1)}\|_1}{\sqrt{\mymin_{k\in \mathrm{supp}(u^{(1)})}\mymin_{k}n_{j,k}}}\sum_{k}\frac{|\epsilon^{(2)}_k|}{\sqrt{\mymin_jn_{j,k}}}\|u^{(3)}\|_1\|u^{(4)}\|_1,
    \end{align*}
    \begin{align*}
        &\|(u^{(1)}\otimes u^{(2)}\otimes \epsilon^{(3)}\otimes u^{(4)})\circ \cN\|_1\\
        \leq&\|u^{(1)}\|_1\|u^{(2)}\|_1\sum_{j'}\frac{|\epsilon^{(3)}_{j'}|}{\sqrt{\mymin_{k'}n_{j',k'}}}\sum_{k'}\frac{\|u^{(4)}\|_1}{\sqrt{\mymin_{k\in \mathrm{supp}(u^{(4)})}\mymin_{k}n_{j,k}}},
    \end{align*}
    and
    \begin{align*}
        &\|(u^{(1)}\otimes u^{(2)}\otimes u^{(3)}\otimes \epsilon^{(4)})\circ \cN)\circ \cN\|_1\\
        \leq&\|u^{(1)}\|_1\|u^{(2)}\|_1\frac{\|u^{(3)}\|_1}{\sqrt{\mymin_{k\in \mathrm{supp}(u^{(3)})}\mymin_{k}n_{j,k}}}\sum_{k'}\frac{|\epsilon^{(4)}_{k'}|}{\sqrt{\mymin_{j'}n_{j',k'}}}.
    \end{align*}
    Combining the bounds above, we have
    \begin{align*}
        |\mathcal{D}_1|\leq \sum_{l=1}^4\left[\sum_{j=1}^p\frac{|\epsilon^{(l)}_j|}{\sqrt{\mymin_{k\in [p]}n_{j,k}}}\frac{\|u^{(i_l)}\|_1}{\sqrt{\mymin_{j\in \mathrm{supp}(u^{(i_l)})}\mymin_{k}n_{j,k}}}\prod_{m\neq l,i_l}\|u^{(m)}\|_1\right].
    \end{align*}
    While for the higher order terms, we can still apply the same argument. For instance, $(\epsilon^{(1)}\otimes \epsilon^{(2)}\otimes u^{(3)}\otimes u^{(4)})\circ \cN$ is one term in $\mathcal{D}_2$, and it satisfies
    \begin{align*}
        &\|(\epsilon^{(1)}\otimes \epsilon^{(2)}\otimes u^{(3)}\otimes u^{(4)})\circ \cN\|_1\\
        \leq&\sum_{j}\frac{|\epsilon^{(1)}_j|}{\sqrt{\mymin_k n_{j,k}}}\sum_{k}\frac{\epsilon^{(2)}_k}{\sqrt{\mymin_j n_{j,k}}}\|u^{(3)}\|_1\|u^{(4)}\|_1\\
        \leq&C_1\sum_{j}\frac{|\epsilon^{(1)}_j|}{\sqrt{\mymin_k n_{j,k}}}\frac{\|u^{(2)}\|_1}{\sqrt{\mymin_{k\in \mathrm{supp}(u^{(2)})}\mymin_{k}n_{j,k}}}\|u^{(3)}\|_1\|u^{(4)}\|_1, 
    \end{align*}
    where in the last line we have utilized the condition that 
    $$\sum_{j=1}^p\left|\frac{\epsilon^{(l)}_j}{\sqrt{\mymin_{k\in [p]}n_{j,k}}}\right|\leq C_1\frac{\|u^{(l)}\|_1}{\sqrt{\mymin_{k\in \mathrm{supp}(u^{(l)})}\mymin_{k}n_{j,k}}}.$$
    While for $(\epsilon^{(1)}\otimes u^{(2)}\otimes \epsilon^{(3)}\otimes u^{(4)})\circ \cN$, 
    \begin{align*}
        &\|(\epsilon^{(1)}\otimes u^{(2)}\otimes \epsilon^{(3)}\otimes u^{(4)})\circ \cN\|_1\\
        \leq&\sum_{j}\frac{|\epsilon^{(1)}_j|}{\sqrt{\mymin_k n_{j,k}}}\sum_{k}\frac{u^{(2)}_k}{\sqrt{\mymin_j n_{j,k}}}\|\epsilon^{(3)}\|_1\|u^{(4)}\|_1\\
        \leq&C_1\sum_{j}\frac{|\epsilon^{(1)}_j|}{\sqrt{\mymin_k n_{j,k}}}\sum_{k}\frac{u^{(2)}_k}{\sqrt{\mymin_j n_{j,k}}}\|u^{(3)}\|_1\|u^{(4)}\|_1, 
    \end{align*}
    where the last line is due to that $\|\epsilon^{(l)}\|_1\leq C_1\|u^{(l)}\|_1$. By applying similar arguments to all terms in $\mathcal{D}_2$, $\mathcal{D}_3$ and $\mathcal{D}_4$, our proof is complete.
\end{proof}
\begin{lemma}\label{lem:kron_err2}
Assume that all conditions in Lemma~\ref{lem:kron_err} hold, and let $j_l\in\{1,\dots,4\}$ satisfy $j_l\equiv l+2(\mathrm{mod}\ 4)$, then we have
\begin{align*}
    &\|[(u^{(1)}+\epsilon^{(1)})\otimes (u^{(2)}+\epsilon^{(2)})\otimes (u^{(3)}+\epsilon^{(3)})\otimes (u^{(4)}+\epsilon^{(4)})\\
    &-u^{(1)}\otimes u^{(2)}\otimes u^{(3)}\otimes u^{(4)}]\circ \cN^{(1)}\|_1\\
    \leq&C_2\sum_{l=1}^4\left[\sum_{j=1}^p\frac{\left|\epsilon^{(l)}_j\right|}{\sqrt{\mymin_{k}n_{j,k}}}\frac{\|u^{(i_{l})}\|_1\|u^{(j_l)}\|_1}{\mymin_{j\in \mathrm{supp}(u^{(i_l)})\cup \mathrm{supp}(u^{(j_l)})}\mymin_{k}n_{j,k}}\prod_{m\neq l,i_l, j_l}\|u^{(m)}\|_1\right],
\end{align*}
\begin{align*}
    &\|[(u^{(1)}+\epsilon^{(1)})\otimes (u^{(2)}+\epsilon^{(2)})\otimes (u^{(3)}+\epsilon^{(3)})\otimes (u^{(4)}+\epsilon^{(4)})\\
    &-u^{(1)}\otimes u^{(2)}\otimes u^{(3)}\otimes u^{(4)}]\circ \cN^{(2)}\|_1\\
    \leq&C_2\sum_{l=1}^4\left[\sum_{j=1}^p\frac{\left|\epsilon^{(l)}_j\right|}{\sqrt{\mymin_{k}n_{j,k}}}\prod_{m\neq l}\frac{\|u^{(m)}\|_1}{\mymin_{j\in \mathrm{supp}(u^{(m)})}\mymin_{k}n_{j,k}}\right],
\end{align*}
where $C_2$ is also a universal constant.
\end{lemma}

\begin{proof}[Proof of Lemma \ref{lem:kron_err2}]
    The proof is very similar to the proof of Lemma~\ref{lem:kron_err}. In the following, we only show the detailed upper bounds for the first-order error terms (4 terms for each of the two cases, with $\cN^{(1)}$ and $\cN^{(2)}$ respectively), while the higher-order error terms can be bounded as a constant factor multiplying the first-order error bounds using the same arguments as in the the proof of Lemma~\ref{lem:kron_err}.
    
    In particular, one can show that one of the first-order error term satisfies
    \begin{align*}
        &\|(\epsilon^{(1)}\otimes u^{(2)}\otimes u^{(3)}\otimes u^{(4)})\circ \cN^{(1)}\|_1\\
        = &\sum_{j,k,j',k'}\frac{\sqrt{n_{j,k,j',k'}}}{n_{j,k}n_{j',k'}}|\epsilon^{(1)}_j||u^{(2)}_k||u^{(3)}_{j'}||u^{(4)}_{k'}|\\
        \leq& \sum_{j,k,j',k'}\frac{1}{n_{j,k}\sqrt{n_{j',k'}}}|\epsilon^{(1)}_j||u^{(2)}_k||u^{(3)}_{j'}||u^{(4)}_{k'}|\\
        \leq&\sum_{j}\frac{|\epsilon^{(1)}_j|}{\sqrt{\mymin_kn_{j,k}}}\sum_{k}\frac{|u^{(2)}_k|}{\sqrt{\mymin_jn_{j,k}}}\sum_{j'}\frac{|u^{(3)}_{j'}|}{\sqrt{\mymin_{k'}n_{j',k'}}}\|u^{(4)}\|_1\\
        \leq &\sum_{j}\frac{|\epsilon^{(1)}_j|}{\sqrt{\mymin_kn_{j,k}}}\frac{\|u^{(2)}\|_1\|u^{(3)}\|_1}{\mymin_{k\in \mathrm{supp}(u^{(2)})\cup\mathrm{supp}(u^{(3)})}\mymin_{k}n_{j,k}}\|u^{(4)}\|_1,
    \end{align*}
    where we applied $n_{j,k,j',k'}\leq n_{j',k'}$ on the third line, and
    \begin{align*}
        &\|(\epsilon^{(1)}\otimes u^{(2)}\otimes u^{(3)}\otimes u^{(4)})\circ \cN^{(2)}\|_1\\
        = &\sum_{j,k,j',k'}\frac{1}{n_{j,k}n_{j',k'}}|\epsilon^{(1)}_j||u^{(2)}_k||u^{(3)}_{j'}||u^{(4)}_{k'}|\\
        \leq&\sum_{j}\frac{|\epsilon^{(1)}_j|}{\sqrt{\mymin_kn_{j,k}}}\sum_{k}\frac{|u^{(2)}_k|}{\sqrt{\mymin_jn_{j,k}}}\sum_{j'}\frac{|u^{(3)}_{j'}|}{\sqrt{\mymin_{k'}n_{j',k'}}}\sum_{k'}\frac{|u^{(4)}_{k'}|}{\sqrt{\mymin_{j'}n_{j',k'}}}\\
        \leq &\sum_{j}\frac{|\epsilon^{(1)}_j|}{\sqrt{\mymin_kn_{j,k}}}\prod_{m\neq 1}\frac{\|u^{(m)}\|_1}{\mymin_{k\in \mathrm{supp}(u^{(m)})}\mymin_{k}n_{j,k}}.
    \end{align*}
    Similarly, one can show that
    \begin{align*}
        &\|(u^{(1)}\otimes \epsilon^{(2)}\otimes u^{(3)}\otimes u^{(4)})\circ \cN^{(1)}\|_1\\
        \leq&\frac{\|u^{(1)}\|_1\|u\|_3}{\sqrt{\mymin_{k\in \mathrm{supp}(u^{(1)})\cup \mathrm{supp}(u^{(3)})}\mymin_{k}n_{j,k}}}\sum_{k}\frac{|\epsilon^{(2)}_k|}{\sqrt{\mymin_jn_{j,k}}}\|u^{(4)}\|_1,
    \end{align*}
    \begin{align*}
        &\|(u^{(1)}\otimes u^{(2)}\otimes \epsilon^{(3)}\otimes u^{(4)})\circ \cN\|_1\\
        \leq&\|u^{(2)}\|_1\sum_{j'}\frac{|\epsilon^{(3)}_{j'}|}{\sqrt{\mymin_{k'}n_{j',k'}}}\sum_{k'}\frac{\|u^{(4)}\|_1\|u^{(1)}\|_1}{\sqrt{\mymin_{k\in \mathrm{supp}(u^{(4)})\cup \mathrm{supp}(u^{(2)})}\mymin_{k}n_{j,k}}},
    \end{align*}
    and
    \begin{align*}
        &\|(u^{(1)}\otimes u^{(2)}\otimes u^{(3)}\otimes \epsilon^{(4)})\circ \cN)\circ \cN\|_1\\
        \leq&\|u^{(1)}\|_1\frac{\|u^{(3)}\|_1\|u^{(2)}\|_1}{\sqrt{\mymin_{k\in \mathrm{supp}(u^{(3)})\cup \mathrm{supp}(u^{(2)})}\mymin_{k}n_{j,k}}}\sum_{k'}\frac{|\epsilon^{(4)}_{k'}|}{\sqrt{\mymin_{j'}n_{j',k'}}};
    \end{align*}
    \begin{align*}
        &\|(u^{(1)}\otimes \epsilon^{(2)}\otimes u^{(3)}\otimes u^{(4)})\circ \cN^{(2)}\|_1\\
        \leq&\sum_{k}\frac{|\epsilon^{(2)}_k|}{\sqrt{\mymin_jn_{j,k}}}\prod_{m\neq 2}\frac{\|u^{(m)}\|_1}{\mymin_{k\in \mathrm{supp}(u^{(m)})}\mymin_{k}n_{j,k}},
    \end{align*}
    \begin{align*}
        &\|(u^{(1)}\otimes u^{(2)}\otimes \epsilon^{(3)}\otimes u^{(4)})\circ \cN^{(2)}\|_1\\
        \leq&\sum_{k}\frac{|\epsilon^{(3)}_k|}{\sqrt{\mymin_jn_{j,k}}}\prod_{m\neq 3}\frac{\|u^{(m)}\|_1}{\mymin_{k\in \mathrm{supp}(u^{(m)})}\mymin_{k}n_{j,k}},
    \end{align*}
    \begin{align*}
        &\|(u^{(1)}\otimes u^{(2)}\otimes u^{(3)}\otimes \epsilon^{(4)})\circ \cN^{(2)}\|_1\\
        \leq&\sum_{k}\frac{|\epsilon^{(4)}_k|}{\sqrt{\mymin_jn_{j,k}}}\prod_{m\neq 4}\frac{\|u^{(m)}\|_1}{\mymin_{k\in \mathrm{supp}(u^{(m)})}\mymin_{k}n_{j,k}}.
    \end{align*}
    Hence the first-order error terms for both cases ($\cN^{(1)}$ and $\cN^{(2)}$) are bounded accordingly.
\end{proof}
\begin{proof}[Proof of Lemma \ref{lem:kron_err3}]
    Lemma \ref{lem:kron_err3} is a more general version of Lemma \ref{lem:kron_err} and \ref{lem:kron_err2}, where the upper bound conditions for $\epsilon^{(l)}$, $l=1,2,3,4$ are not assumed. It can be directly proved using exactly the same arguments as in the proof of Lemma \ref{lem:kron_err} and \ref{lem:kron_err2}.
\end{proof}
\begin{proof}[Proof of Proposition \ref{prop:indexsets}]
    Recall the definition of $S_2(a, b)$ in \eqref{eq:index_sets}. Note that if $(j,k)\notin \overline{\cN}_a\times \overline{\cN}^{(a)}_b$ and $(j,k)\notin \overline{\cN}^{(a)}_b\times \overline{\cN}_a$, one has
    $\Theta^{(a)*}_{j,b}\Theta^*_{k,a}+\Theta^{(a)*}_{k,b}\Theta^*_{j,a}=0$ and hence $(j,k)\notin S_2(a, b)$. This implies that $S_2(a, b)\subseteq (\overline{\cN}_a\times \overline{\cN}^{(a)}_b)\times (\overline{\cN}^{(a)}_b\times \overline{\cN}_a)$. If $S_2(a, b)\neq (\overline{\cN}_a\times \overline{\cN}^{(a)}_b)\times (\overline{\cN}^{(a)}_b\times \overline{\cN}_a)$, then there exists $(j,k)\in (\overline{\cN}_a\times \overline{\cN}^{(a)}_b)$ such that $\Theta^{(a)*}_{j,b}\Theta^*_{k,a}=-\Theta^{(a)*}_{k,b}\Theta^*_{j,a}\neq 0$, which defines a set of measure zero.
    
    While for $\overline{\cN}_a^{(b)}$, by definition, $j\in \overline{\cN}_a^{(b)}$ if and only if
    $\Theta^*_{b,j}\Theta^*_{a,a}-\Theta^*_{a,j}\Theta^*_{a,b}\neq 0$. If $b\notin \cN_a$, $\Theta^*_{a,b}=0$, and this condition is equivalent to $\Theta^*_{b,j}\neq 0$ and hence $\overline{\cN}_a^{(b)}=\overline{\cN}_b$. 
    
    Otherwise, if $j\notin \overline{\cN}_b\cup \overline{\cN}_a$, one has $\Theta^*_{b,j}\Theta^*_{a,a}-\Theta^*_{a,j}\Theta^*_{a,b}\neq 0$ which implies $j\notin \overline{\cN}_a^{(b)}$. Thus $\overline{\cN}_a^{(b)}\subseteq \overline{\cN}_b\cup \overline{\cN}_a$. In this case, $\overline{\cN}_a^{(b)}\neq \overline{\cN}_b\cup \overline{\cN}_a$ only happens when there exists $j\in \overline{\cN}_b\cup \overline{\cN}_a$ such that $\Theta^*_{b,j}\Theta^*_{a,a}=\Theta^*_{a,j}\Theta^*_{a,b}$, which also defines a set of measure zero.
\end{proof}
\begin{proof}[Proof of Lemma \ref{lem:tail_approx_one_pair}]
    First note that we can write $\frac{E(a,b)}{\sigma_n(a,b)} = \sum_{i=1}^n\frac{\epsilon_i(a,b)}{\sigma_n(a,b)}$, where 
    $$
    \epsilon_i(a,b) = \sum_{j,k}(x_{i,j}x_{i,k}-\Sigma^*_{j,k})\frac{\delta^{(i)}_{j,k}}{n_{j,k}}\frac{\Theta^{(a)*}_{j,b}\Theta^*_{k,a}+\Theta^{(a)*}_{k,b}\Theta^*_{j,a}}{2\Theta^*_{a,a}},
    $$
    and hence $\{\frac{\epsilon_i(a,b)}{\sigma_n(a,b)}\}_{i=1}^n$ are independent mean zero random variables. The following Gaussian approximation result for sum of independent r.v.s is useful for our proof:
    \begin{lemma}[Theorem 1.1 in \cite{zaitsev1987gaussian}]\label{lem:gauss_approx_tool}
        Let $\tau>O$ and $\xi_1,\dots,\xi_n\in \mathbb{R}^k$ be independent mean zero random vectors such that 
        \begin{equation}\label{eq:moment_condition}
            |\mathbb{E}(\langle \xi_i, t\rangle^2\langle \xi_i, u\rangle^{m-2})|\leq \frac{1}{2}m!\tau^{m-2}\|u\|^{m-2}\mathbb{E}(\langle \xi_i, t\rangle^2)
        \end{equation}
        holds for $m\geq 3$, $t,u\in\mathbb{R}^k$, and $i=1,\dots,n$. Let $S=\sum_{i=1}^n\xi_i$. Denote by $Z\in \mathbb{R}^k$ a multivariate Gaussian random variable with the zero mean and the same covariance as that of $S$. Then for all $\varepsilon>0$,
        $$
        \sup_{B\subset \mathbb{R}^k}\max\{\mathbb{P}(S\in B)-\mathbb{P}(Z\in B^{\varepsilon}), \mathbb{P}(Z\in B)-\mathbb{P}(S\in B^{\varepsilon})\}\leq C_1(k)\exp\{-\frac{\varepsilon}{C_2(k)\tau}\},
        $$
        where $B^{\epsilon}=\{x\in \mathbb{R}^k: \inf_{y\in B}\|x-y\|_2\leq\varepsilon\}$ for any set $B\subset \mathbb{R}^k$.
    \end{lemma}
    In the following, we will let $\xi_i(a,b)=\frac{\epsilon_i(a,b)}{\sigma_n(a,b)}$ and show that $\xi_i(a,b)$ satisfies \eqref{eq:moment_condition}. Recall the $\|\cdot\|_{\psi_{\alpha}}$ norm in Definition \ref{def:psi_alpha}, and by the equivalence between some properties of sub-exponential random variables \citep[see][Definition 5.13]{vershynin2010introduction}, we have
    \begin{align*}
        |\mathbb{E}(\langle \xi_i(a,b), t\rangle^2\langle \xi_i(a,b), u\rangle^{m-2})|\leq &t^2|u|^{m-2}\mathbb{E}|\xi_i(a,b)|^m\\
        \leq&t^2|u|^{m-2}m^m\|\xi_i(a,b)\|_{\psi_1}^m\\
        \leq &|u|^{m-2}t^2\frac{1}{2}m!\tau_i^{m-2}\mathbb{E}(\xi_i^2(a,b)),
    \end{align*}
    where $\tau_i = 12\|\xi_i(a,b)\|_{\psi_1}\left(\frac{\|\xi_i(a,b)\|_{\psi_1}^2}{\mathbb{E}(\xi_i^2(a,b))}\vee 1\right)$. The last line is due to that 
    $$
    \frac{2m^m}{m!}\leq \frac{2e^{m}}{\sqrt{2\pi m}}\leq\frac{e^m}{\sqrt{m}}\leq (\frac{e^3}{\sqrt{3}})^{m-2}\leq 12^{m-2},
    $$
    where the first inequality is due to Stirling's formula: $m!\geq \sqrt{2\pi m}(\frac{m}{e})^m$, and 
    $$
    m^m\|\xi_i(a,b)\|_{\psi_1}^m\leq  \frac{1}{2}m!12^{m-2}\|\xi_i(a,b)\|_{\psi_1}^{m-2}\frac{\|\xi_i(a,b)\|_{\psi_1}^2}{\mathbb{E}(\xi_i^2(a,b))}\mathbb{E}(\xi_i^2(a,b))\leq \frac{1}{2}m!\left(12\|\xi_i(a,b)\|_{\psi_1}\left(\frac{\|\xi_i(a,b)\|_{\psi_1}^2}{\mathbb{E}(\xi_i^2(a,b))}\vee 1\right)\right)^{m-2}\mathbb{E}(\xi_i^2(a,b)).
    $$
   Now it suffices to show that each $\tau_i\leq \frac{C\|\Sigma^*\|_{\infty}^3(d_a+d_b+1)^3}{\lambda_{\min}^3(\Sigma^*)\sqrt{n_2^{(a,b)}}}$. Recall that we have defined $U^{(\delta,i)}_{j,k}=\frac{U_{j,k}\delta^{(i)}_{j,k}}{n_{j,k}}$, where $U_{j,k}=\frac{\Theta^{(a)*}_{j,b}\Theta^*_{k,a}+\Theta^{(a)*}_{k,b}\Theta^*_{j,a}}{2\Theta^*_{a,a}}$ in the proof of Theorem \ref{thm:nb_lasso_debias_decomp}. Then we can also write $\xi_i(a,b) = \frac{1}{\sigma_n(a,b)}\sum_{j,k}(x_{i,j}x_{i,k}-\Sigma^*_{j,k})U^{(\delta,i)}_{j,k}$. It has been shown in the proof of Theorem \ref{thm:nb_lasso_debias_decomp} that $\sigma_n(a,b) \geq \sqrt{2}\lambda_{\min}(\Sigma^*)\sqrt{\sum_{j,k}U_{j,k}^2n_{j,k}^{-1}}$. Since $\|x_{i,j}x_{i,k}-\Sigma^*_{j,k}\|_{\psi_1}\leq C\|\Sigma^*\|_{\infty}$, as has been shown in the proof of Lemma \ref{lem:SampleCov_entry_err}, we have
    \begin{align*}
        \|\xi_i(a,b)\|_{\psi_1}\leq &\frac{C\|\Sigma^*\|_{\infty}\|U^{(\delta,i)}\|_1}{\lambda_{\min}(\Sigma^*)\sqrt{\sum_{j,k}U_{j,k}^2n_{j,k}^{-1}}}\\
        \leq&\frac{C\|\Sigma^*\|_{\infty}(d_a+d_b+1)\|U^{(\delta,i)}\|_F}{\lambda_{\min}(\Sigma^*)\sqrt{\sum_{j,k}U_{j,k}^2n_{j,k}^{-1}}}\\
        \leq&\frac{C\|\Sigma^*\|_{\infty}(d_a+d_b+1)}{\lambda_{\min}(\Sigma^*)}\sqrt{\frac{\sum_{j,k}U_{j,k}^2n_{j,k}^{-2}}{\sum_{j,k}U_{j,k}^2n_{j,k}^{-1}}}\\
        \leq &\frac{C\|\Sigma^*\|_{\infty}(d_a+d_b+1)}{\lambda_{\min}(\Sigma^*)\sqrt{n_2^{(a,b)}}}.
    \end{align*}    
    Furthermore, we can upper bound $\frac{\|\xi_i\|_{\psi_1}^2}{\mathbb{E}(\xi_i^2)}$ as follows:
    \begin{align*}
        \frac{\|\xi_i(a,b)\|_{\psi_1}^2}{\mathbb{E}(\xi_i^2(a,b))}=&\frac{\|\epsilon_i(a,b)\|_{\psi_1}^2}{\mathbb{E}(\epsilon_i(a,b)^2)}\\
        \leq &\frac{C\|\Sigma^*\|^2_{\infty}(d_a+d_b+1)^2\|U^{(\delta,i)}\|_F^2}{\lambda_{\min}^2(\Sigma^*)\|U^{(\delta,i)}\|_F^2}\\
        =&\frac{C\|\Sigma^*\|^2_{\infty}(d_a+d_b+1)^2}{\lambda_{\min}^2(\Sigma^*)}.
    \end{align*}
    Therefore, we have shown that $\tau_i = 12\|\xi_i(a,b)\|_{\psi_1}\left(\frac{\|\xi_i(a,b)\|_{\psi_1}^2}{\mathbb{E}(\xi_i^2(a,b))}\vee 1\right)\leq \frac{C\|\Sigma^*\|^3_{\infty}(d_a+d_b+1)^3}{\lambda_{\min}^3(\Sigma^*)\sqrt{n_2^{(a,b)}}}$, and the proof is now complete.
\end{proof}
\begin{proof}[Proof of Lemma \ref{lem:tail_approx_two_pairs}]
    Similar to the proof of Lemma \ref{lem:tail_approx_one_pair}, we will apply the Gaussian approximation bound in Lemma \ref{lem:gauss_approx_tool} by validating \eqref{eq:moment_condition} for $(\xi_i(a,b),\xi_i(a',b'))^{\top}\in \mathbb{R}^2$.
    
    For any two-dimensional random vector $X = (X_1,X_2)^\top$, and $t,u\in \mathbb{R}^2$, we can show that 
    \begin{align*}
        \mathbb{E}(\langle X, t\rangle^2 \langle X, u\rangle^{m-2})\leq &\sqrt{\mathbb{E}\|X\|_2^4}\sqrt{\mathbb{E}\|X\|_2^{2m-4}}\|u\|_2^{m-2}\|t\|_2^2,
    \end{align*}
     and for any integer $k$, 
     \begin{align*}
         \mathbb{E}\|X\|_2^{2k}\leq \mathbb{E}(X_1^2+X_2^2)^{k}\leq 2^{k-1}(\mathbb{E}|X_1|^{2k}+\mathbb{E}|X_2|^{2k})\leq &2^{k-1}(2k)^{2k}(\|X_1\|_{\psi_1}^{2k}+\|X_2\|_{\psi_1}^{2k})\\
         \leq &(2\sqrt{2}k)^{2k}(\|X_1\|_{\psi_1}\vee\|X_2\|_{\psi_1})^{2k}
     \end{align*}
    where the second inequality is due to the Jensen's inequality, the third inequality is due to the property of $\|\cdot\|_{\psi_1}$ norm. Hence we have
    \begin{align*}
        \mathbb{E}(\langle X, t\rangle^2 \langle X, u\rangle^{m-2})\leq &\sqrt{\mathbb{E}\|X\|_2^4}\sqrt{\mathbb{E}\|X\|_2^{2m-4}}\|u\|_2^{m-2}\|t\|_2^2\\
        \leq &32(2\sqrt{2}m)^{m-2}(\|X_1\|_{\psi_1}\vee\|X_2\|_{\psi_1})^{m}\|u\|_2^{m-2}\|t\|_2^2\\
        \leq &\frac{32}{m^2}[2\sqrt{2}(\|X_1\|_{\psi_1}\vee\|X_2\|_{\psi_1})]^{m-2}\frac{\|X_1\|_{\psi_1}^2\vee\|X_2\|_{\psi_1}^2}{\lambda_{\min}(\mathrm{Cov}(X))}\lambda_{\min}(\mathrm{Cov}(X))m^m\|u\|_2^{m-2}\|t\|_2^2\\
        \leq&\frac{1}{2}m!\|u\|_2^{m-2}\tau^{m-2}\mathbb{E}(\langle X, t\rangle^2),
    \end{align*}
     where $\tau = \frac{120(\|X_1\|_{\psi_1}\vee\|X_1\|_{\psi_1})^3}{\lambda_{\min}(\mathrm{Cov}(X))}$, and we have applied the fact that $\frac{2m^m}{m!}\leq (\frac{e^3}{\sqrt{3}})^{m-2}$ and $\mathbb{E}(\langle X, t\rangle^2)\geq \|t\|_2^2\lambda_{\min}(\mathrm{Cov}(X))$ on the last line.
     
     Recall the proof of Lemma \ref{lem:tail_approx_one_pair} and the definition of $\alpha(\Theta^*,\{V_i\}_{i=1}^n)$, we have 
     $$
     \|\xi_i(a,b)\|_{\psi_1},\|\xi_i(a',b')\|_{\psi_1}\leq \frac{C\|\Sigma^*\|_{\infty}(d+1)}{\lambda_{\min}(\Sigma^*)\sqrt{n_2^{(a,b)}\vee n_2^{(a',b')}}},
     $$
     and $(\xi_i(a,b),\xi_i(a',b'))^{\top}$ satisfies \eqref{eq:moment_condition} with $\tau = \frac{C\|\Sigma^*\|_{\infty}(d+1)\alpha(\Theta^*,\{V_i\}_{i=1}^n)}{\lambda_{\min}(\Sigma^*)\sqrt{n_2^{(a,b)}\vee n_2^{(a',b')}}}$. The proof completes with the application of Lemma \ref{lem:gauss_approx_tool}.
\end{proof}
\begin{proof}[Proof of Proposition \ref{prop:assump7_simultaneous}]
   We first derive the closed form solution of $\rho_n(a,b,a',b')$ under the simultaneous measurement setting, and then find appropriate super-sets of $\mathcal{A}_1(\rho_0)$ and $\mathcal{A}_2(\rho_0,\gamma)$ with bounded sizes. 
   
   Recall that $\rho_n(a,b,a',b') = \frac{\sigma_n^2(a,b,a',b')}{\sigma_n(a,b)\sigma_n(a',b')}$, where $\sigma_n^2(a,b,a',b')=\frac{1}{\Theta^*_{a,a}\Theta^*_{a',a'}}\mathcal{T}^{(n)*}\times_1\Theta^{(a)*}_{:,b}\times_2\Theta^*_{:,a}\times_3\Theta^{(a)*}_{:,b'}\times_4\Theta^*_{:,a'}$, $\sigma_n^2(a,b)=\frac{1}{(\Theta^*_{a,a})^2}\mathcal{T}^{(n)*}\times_1\Theta^{(a)*}_{:,b}\times_2\Theta^*_{:,a}\times_3\Theta^{(a)*}_{:,b}\times_4\Theta^*_{:,a}$. In the simultaneous measurement setting, $n_{j,k,j',k'}=n_{j,k}=n_{j',k'}=n$ for all node $j,k,j',k'\in [p]$, and hence $\mathcal{T}^{(n)*}_{j,k,j',k'}= \frac{1}{n}\mathcal{T}^*_{j,k,j',k'}=\frac{1}{n}(\Sigma^*_{j,j'}\Sigma^*_{k,k'}+\Sigma^*_{j,k'}\Sigma^*_{k,j'})$. Therefore, one can compute  
   \begin{equation*}
   \begin{split}
       \sigma_n^2(a,b)=&\frac{1}{n(\Theta^*_{a,a})^2}\left[\Theta^{(a)*}_{b,:}\Sigma^*\Theta^{(a)*}_{:,b}\Theta^*_{a,:}\Sigma^*\Theta^*_{:,a}+\Theta^{(a)*}_{b,:}\Sigma^*\Theta^{*}_{:,a}\Theta^{(a)*}_{b,:}\Sigma^*\Theta^*_{:,a}\right]\\
       =&\frac{\Theta^*_{a,a}\Theta^*_{b,b}-(\Theta^*_{a,b})^2}{n(\Theta^*_{a,a})^2}.
   \end{split}
   \end{equation*}
   Similarly, one has $\sigma_n^2(a',b')=\frac{\Theta^*_{a',a'}\Theta^*_{b',b'}-(\Theta^*_{a',b'})^2}{n(\Theta^*_{a',a'})^2}$. While for $\sigma_n^2(a,b,a',b')$, note that when $(a,b), (a',b')\in \mathcal{H}_0$, $\Theta^*_{a,b}=\Theta^*_{a',b'}=0$. Thus we have
   \begin{equation*}
       \begin{split}
           \sigma_n^2(a,b,a',b') = &\frac{1}{n\Theta^*_{a,a}\Theta^*_{a',a'}}\left[\Theta^{(a)*}_{b,:}\Sigma^*\Theta^{(a')*}_{:,b'}\Theta^*_{a,:}\Sigma^*\Theta^*_{:,a'}+\Theta^{(a)*}_{b,:}\Sigma^*\Theta^*_{:,a'}\Theta^{(a')*}_{b',:}\Sigma^*\Theta^*_{:,a}\right]\\
           = &\frac{1}{n\Theta^*_{a,a}\Theta^*_{a',a'}}\left[\Theta^{*}_{b,:}\Sigma^*\Theta^{*}_{:,b'}\Theta^*_{a,:}\Sigma^*\Theta^*_{:,a'}+\Theta^{*}_{b,:}\Sigma^*\Theta^*_{:,a'}\Theta^{*}_{b',:}\Sigma^*\Theta^*_{:,a}\right]\\
           =&\frac{\Theta^{*}_{b,b'}\Theta^*_{a,a'}+\Theta^{*}_{b,a'}\Theta^{*}_{a,b'}}{n\Theta^*_{a,a}\Theta^*_{a',a'}}.
       \end{split}
   \end{equation*}
   where we have applied the fact that $\Theta^{(a)*}_{:,b} = \Theta^*_{:,b}$ when $\Theta^*_{a,b}=0$ on the second line. Therefore, for $(a,b), (a',b')\in \mathcal{H}_0$,
   \begin{equation*}
       \begin{split}
           \rho_n(a,b,a',b')=&\frac{\Theta^{*}_{b,b'}\Theta^*_{a,a'}+\Theta^{*}_{b,a'}\Theta^{*}_{a,b'}}{\sqrt{\Theta^*_{a,a}\Theta^*_{b,b}\Theta^*_{a',a'}\Theta^*_{b',b'}}}\\
       =&\Omega^*_{a,a'}\Omega^*_{b,b'}+\Omega^*_{a,b'}\Omega^*_{b,a'}.
       \end{split}
   \end{equation*}
   Therefore, we have 
   \begin{equation}
   \begin{split}
       \mathcal{A}_1(\rho_0)\subset&\{(a,b,a',b'): |\Omega^*_{a,a'}\Omega^*_{b,b'}+\Omega^*_{a,b'}\Omega^*_{b,a'}|>\rho_0\}\\
       \subset&\{(a,b,a',b'): |\Omega^*_{a,a'}\Omega^*_{b,b'}|>\frac{\rho_0}{2}\text{ or }|\Omega^*_{a,b'}\Omega^*_{b,a'}|>\frac{\rho_0}{2}\}\\
       \subset&\{(a,b,a',b'), (a,b,b',a'): |\Omega^*_{a,a'}|>\frac{\rho_0}{2}\text{ and }|\Omega^*_{b,b'}|>\frac{\rho_0}{2}\}.
   \end{split} 
   \end{equation}
   Since for $1\leq j\leq p$, $|\{k:|\Omega^*_{j,k}|>\frac{\rho_0}{2}|\leq \frac{C}{2}$, it is straightforward to see that $|\{(a,b,a',b'): |\Omega^*_{a,a'}|>\frac{\rho_0}{2}\text{ and }|\Omega^*_{b,b'}|>\frac{\rho_0}{2}\}|\leq \frac{C}{2}p^2$, and hence $|\mathcal{A}_1(\rho_0)|\leq Cp^2$. While for $\mathcal{A}_2(\rho_0,\gamma)$, note that
   \begin{equation}
   \begin{split}
       \mathcal{A}_2(\rho_0,\gamma)\subset&\{(a,b,a',b'): |\Omega^*_{a,a'}\Omega^*_{b,b'}+\Omega^*_{a,b'}\Omega^*_{b,a'}|>0\}\\
       \subset&\{(a,b,a',b'): \Theta^*_{a,a'}\Theta^*_{b,b'}\neq 0\text{ or }\Theta^*_{a,b'}\Theta^*_{b,a'}\neq 0\},
   \end{split} 
   \end{equation}
   which implies $|\mathcal{A}_2(\rho_0,\gamma)|\leq 2d^2p^2\ll p^{\frac{4}{1+\rho_0}}(\log p)^{\frac{2\rho_0}{1+\rho_0}-\frac{1}{2}}(\log\log p)^{-\frac{1}{2}}$. Now the proof is complete.
\end{proof}
\end{appendices}
\bibliographystyle{plain}
\bibliography{reference.bib}

\end{document}